\documentclass[a4paper,14pt]{extreport}
\PassOptionsToPackage{hyperindex}{hyperref} 

\usepackage{iftex}
\newif\ifxetexorluatex   
\ifXeTeX
    \xetexorluatextrue
\else
    \ifLuaTeX
        \xetexorluatextrue
    \else
        \xetexorluatexfalse
    \fi
\fi

\usepackage{pdflscape}                              
\usepackage{geometry}                               

\usepackage{amsthm,amsfonts,amsmath,amssymb,amscd}  
\usepackage{mathtools}                              

\newlength{\curtextsize}
\newlength{\bigtextsize}

\ifxetexorluatex
    \usepackage{polyglossia}                        
\else
    \RequirePDFTeX                                  
    \usepackage{cmap}                               
    \usepackage[utf8]{inputenc}                     
    \usepackage[english]{babel}
    \IfFileExists{pscyr.sty}{\usepackage{pscyr}}{}  
\fi

\usepackage{indentfirst}                            

\usepackage[dvipsnames,usenames]{xcolor}

\usepackage{colortbl}

\usepackage{longtable}                              
\usepackage{multirow,makecell,array}                
\usepackage{booktabs}                               

\usepackage{soulutf8}                               
\usepackage{icomma}                                 

\usepackage{hyperref}

\usepackage{graphicx}                               

\usepackage{enumitem}

\usepackage{caption}                                
\usepackage{subcaption}                             

\usepackage[onehalfspacing]{setspace}               

\usepackage[figure,table]{totalcount}               
\usepackage{totcount}                               
\usepackage{totpages}                               

\ifxetexorluatex
    \usepackage{cleveref}                           
\else
    \usepackage[english]{cleveref}                  
\fi
\creflabelformat{equation}{#2#1#3}                  

\usepackage{fancyhdr}

\usepackage{calc}               

\usepackage {interfaces-base}   

\usepackage{titlesec}           

\usepackage{tocloft}

\usepackage{chngcntr}           
\usepackage{tabularx,tabulary}  

\usepackage{fancyvrb}
\usepackage{listings}

\usepackage{floatrow}

\usepackage{titling}

\usepackage{ifthen}                 
\newcounter{intvl}
\newcounter{otstup}
\newcounter{contnumeq}
\newcounter{contnumfig}
\newcounter{contnumtab}
\newcounter{pgnum}
\newcounter{bibliosel}
\newcounter{chapstyle}
\newcounter{headingdelim}
\newcounter{headingalign}
\newcounter{headingsize}
\newcounter{tabcap}
\newcounter{tablaba}
\newcounter{tabtita}

\definecolor{linkcolor}{rgb}{0,0,0} 
\definecolor{citecolor}{rgb}{0,0,0} 
\definecolor{urlcolor}{rgb}{0,0,0} 

\renewcommand{\chaptername}{Chapter}

\newcommand{\thesisAuthor}             
{%
    \texorpdfstring{
        Jacob Daniel Biamonte
    }{%
        Biamonte, Jacob Daniel
    }%
}
\newcommand{\thesisUdk}                
{\todo{xxx.xxx}}
\newcommand{\thesisTitle}              
{\texorpdfstring{\MakeUppercase{On the Theory of Modern Quantum Algorithms}}{On the Theory of Modern Quantum Algorithms}}
\newcommand{\thesisSpecialtyNumber}    
{\texorpdfstring{01.04.02}{01.04.02}}
\newcommand{\thesisSpecialtyTitle}     
{\texorpdfstring{Applied Mathematics and Informatics}{Applied Mathematics and Informatics}}
\newcommand{\thesisDegree}             
{Doctor of Physical and Mathematical Sciences (DS.c.)}
\newcommand{\thesisCity}               
{Moscow}
\newcommand{\thesisYear}               
{2020}
\newcommand{\thesisOrganization}       
{Moscow Institute of Physics and Technology}

\newcommand{\thesisInOrganization}       
{Moscow Institute of Physics and Technology}

\newcommand{\supervisorFio}            
{\todo{Name Surname}}
\newcommand{\supervisorRegalia}        
{\todo{academic degree, academic rank}}

\newcommand{\opponentOneFio}           
{\todo{Name Surname}}
\newcommand{\opponentOneRegalia}       
{\todo{academic degree}}
\newcommand{\opponentOneJobPlace}      
{\todo{work place name}}
\newcommand{\opponentOneJobPost}       
{\todo{position}}

\newcommand{\opponentTwoFio}           
{\todo{Name Surname}}
\newcommand{\opponentTwoRegalia}       
{\todo{academic degree}}
\newcommand{\opponentTwoJobPlace}      
{\todo{work place name}}
\newcommand{\opponentTwoJobPost}       
{\todo{position}}

\newcommand{\leadingOrganizationTitle} 
{\todo{Федеральное государственное бюджетное образовательное учреждение высшего профессионального образования с~длинным длинным длинным длинным названием}}

\newcommand{\defenseDate}              
{\todo{DD mmmmmmmm YYYY~г.~at~XX}}
\newcommand{\defenseCouncilNumber}     
{\todo{NN}}
\newcommand{\defenseCouncilTitle}      
{\todo{committee institution name}}
\newcommand{\defenseCouncilAddress}    
{\todo{address}}

\newcommand{\defenseSecretaryFio}      
{\todo{Name Surname}}
\newcommand{\defenseSecretaryRegalia}  
{\todo{academic degree}}            

\newcommand{\synopsisLibrary}          
{\todo{library name}}
\newcommand{\synopsisDate}             
{\todo{DD mmmmmmmm YYYY}}

\newcommand{\keywords}
{}      
\ifxetexorluatex
    \setmainlanguage[babelshorthands=true]{english}  
    \setotherlanguage{english}                       
    \ifXeTeX
        \defaultfontfeatures{Ligatures=TeX,Mapping=tex-text}
    \else
        \defaultfontfeatures{Ligatures=TeX}
    \fi
    \newfontfamily\cyrillicfont{Times New Roman}
    \newfontfamily\cyrillicfontsf{Arial}
    \newfontfamily\cyrillicfonttt{Courier New}
\else
    \IfFileExists{pscyr.sty}{}{}
\fi

\DeclareCaptionLabelSeparator*{emdash}{~--- }             
\ifthenelse{\equal{\thetabcap}{0}}{%
    \newcommand{\tabcapalign}{\raggedright}  
}

\ifthenelse{\equal{\thetablaba}{0} \AND \equal{\thetabcap}{1}}{%
    \newcommand{\tabcapalign}{\raggedright}  
}

\ifthenelse{\equal{\thetablaba}{1} \AND \equal{\thetabcap}{1}}{%
    \newcommand{\tabcapalign}{\centering}    
}

\ifthenelse{\equal{\thetablaba}{2} \AND \equal{\thetabcap}{1}}{%
    \newcommand{\tabcapalign}{\raggedleft}   
}

\ifthenelse{\equal{\thetabtita}{0} \AND \equal{\thetabcap}{1}}{%
    \newcommand{\tabtitalign}{\raggedright}  
}

\ifthenelse{\equal{\thetabtita}{1} \AND \equal{\thetabcap}{1}}{%
    \newcommand{\tabtitalign}{\centering}    
}

\ifthenelse{\equal{\thetabtita}{2} \AND \equal{\thetabcap}{1}}{%
    \newcommand{\tabtitalign}{\raggedleft}   
}

\DeclareCaptionFormat{tablenocaption}{\tabcapalign #1\strut}        
\ifthenelse{\equal{\thetabcap}{0}}{%
    \DeclareCaptionFormat{tablecaption}{\tabcapalign #1#2#3}
    \captionsetup[table]{labelsep=emdash}                       
}{%
    \DeclareCaptionFormat{tablecaption}{\tabcapalign #1#2\par
        \tabtitalign{#3}}                                       
    \captionsetup[table]{labelsep=space}                        
}
\DeclareCaptionLabelFormat{continued}{Continuation of table~#2}

\ifLuaTeX
    \hypersetup{
        unicode,                
    }
\fi

\hypersetup{
    linktocpage=true,           
    plainpages=false,           
    colorlinks,                 
    linkcolor={linkcolor},      
    citecolor={citecolor},      
    urlcolor={urlcolor},        
    pdftitle={\thesisTitle},    
    pdfauthor={\thesisAuthor},  
    pdfsubject={\thesisSpecialtyNumber\ \thesisSpecialtyTitle},      
    pdfkeywords={\keywords},    
    pdflang={ru},
}

\DeclareRobustCommand{\todo}{\textcolor{red}}       
\setlist{nosep,
    labelindent=\parindent,leftmargin=*
}

\graphicspath{{images/}{Dissertation/images/}}         

\LoadInterface {titlesec}                   

\newlength{\otstuplen}
\ifthenelse{\equal{\theheadingalign}{0}}{
    \newcommand{\hdngalign}{\filcenter}                
}{%
    \newcommand{\hdngalign}{\filright}                 
} 

\cftsetrmarg{2.55em plus1fil}                       

\ifthenelse{\theheadingdelim > 0}{%
    \renewcommand\cftchapaftersnum{.\ }   
}{%
\renewcommand\cftchapaftersnum{\quad}     
}
\ifthenelse{\theheadingdelim > 1}{%
    \renewcommand\cftsecaftersnum{.\ }    
    \renewcommand\cftsubsecaftersnum{.\ } 
}{%
\renewcommand\cftsecaftersnum{\quad}      
\renewcommand\cftsubsecaftersnum{\quad}   
}

\ifthenelse{\equal{\thepgnum}{1}}{%
\addtocontents{toc}{~\hfill{Page}\par}
}

\ifnum\curtextsize>\bigtextsize     
\else
\fi

\makeatletter
\let\ps@plain\ps@fancy              
\makeatother
\makeatletter
\titleformat{\chapter}[block]                                
        {\hdngalign\fontsize{14pt}{16pt}\selectfont\bfseries}%
        {\thechapter\cftchapaftersnum}                       
        {0em}
        {}%

\titleformat{\section}[block]                                
        {\hdngalign\fontsize{14pt}{16pt}\selectfont\bfseries}%
        {\thesection\cftsecaftersnum}                        
        {0em}
        {}%

\titleformat{\subsection}[block]                             
        {\hdngalign\fontsize{14pt}{16pt}\selectfont\bfseries}%
        {\thesubsection\cftsubsecaftersnum}                  
        {0em}
        {}%

\ifthenelse{\equal{\thechapstyle}{1}}{%
    \sectionformat{\chapter}{
        label=\chaptername\ \thechapter\cftchapaftersnum,
        labelsep=0em,
    }
    \renewcommand{\cftchappresnum}{\chaptername\ }
    \setlength{\cftchapnumwidth}{\widthof{\cftchapfont\cftchappresnum\thechapter\cftchapaftersnum}}
}%
\makeatother 
\beforetitleunit=\curtextsize
\aftertitleunit=\curtextsize

\titlespacing{\chapter}{\theotstup\parindent}{-1.7em}{*\theintvl}       
\titlespacing{\section}{\theotstup\parindent}{*\theintvl}{*\theintvl}
\titlespacing{\subsection}{\theotstup\parindent}{*\theintvl}{*\theintvl}
\titlespacing{\subsubsection}{\theotstup\parindent}{*\theintvl}{*\theintvl}

\ifthenelse{\equal{\theheadingsize}{1}}{
    \setlength{\cftbeforetoctitleskip}{-1.2\curtextsize}        
    \setlength{\cftbeforeloftitleskip}{-1.4\curtextsize}        
    \setlength{\cftbeforelottitleskip}{-1.4\curtextsize}        
    \sectionformat{\chapter}{
        format=\hdngalign\Large\bfseries, 
        top-=0.4em,                       
    }
    \sectionformat{\section}{
        format=\hdngalign\large\bfseries, 
    }
}

\ifthenelse{\equal{\theheadingsize}{1}\AND \curtextsize < \bigtextsize}{
    \sectionformat{\chapter}{
        top-=0.2em, 
    }
}

\ifthenelse{\equal{\thecontnumeq}{1}}{%
    \counterwithout{equation}{chapter} 
}
\ifthenelse{\equal{\thecontnumfig}{1}}{%
    \counterwithout{figure}{chapter}   
}
\ifthenelse{\equal{\thecontnumtab}{1}}{%
    \counterwithout{table}{chapter}    
}

\makeatletter
\def\formbytotal#1#2#3#4#5{%
    \newcount\@c
    \@c\totvalue{#1}\relax
    \newcount\@last
    \newcount\@pnul
    \@last\@c\relax
    \divide\@last 10
    \@pnul\@last\relax
    \divide\@pnul 10
    \multiply\@pnul-10
    \advance\@pnul\@last
    \multiply\@last-10
    \advance\@last\@c
    \total{#1}~#2%
    \ifnum\@pnul=1#5\else%
    \ifcase\@last#5\or#3\or#4\or#4\or#4\else#5\fi
    \fi
}
\makeatother

\AtBeginDocument{
    \regtotcounter{totalcount@figure}
    \regtotcounter{totalcount@table}       
    \regtotcounter{TotPages}               
}           

\def\zz{\ifx\[$\else\aftergroup\zzz\fi}
\def\zzz{\setbox0\lastbox
\dimen0\dimexpr\extrarowheight + \ht0-\dp0\relax
\setbox0\hbox{\raise-.5\dimen0\box0}%
\ht0=\dimexpr\ht0+\extrarowheight\relax
\dp0=\dimexpr\dp0+\extrarowheight\relax 
\box0
}

\lstdefinelanguage{Renhanced}%
{keywords={abbreviate,abline,abs,acos,acosh,action,add1,add,%
        aggregate,alias,Alias,alist,all,anova,any,aov,aperm,append,apply,%
        approx,approxfun,apropos,Arg,args,array,arrows,as,asin,asinh,%
        atan,atan2,atanh,attach,attr,attributes,autoload,autoloader,ave,%
        axis,backsolve,barplot,basename,besselI,besselJ,besselK,besselY,%
        beta,binomial,body,box,boxplot,break,browser,bug,builtins,bxp,by,%
        c,C,call,Call,case,cat,category,cbind,ceiling,character,char,%
        charmatch,check,chol,chol2inv,choose,chull,class,close,cm,codes,%
        coef,coefficients,co,col,colnames,colors,colours,commandArgs,%
        comment,complete,complex,conflicts,Conj,contents,contour,%
        contrasts,contr,control,helmert,contrib,convolve,cooks,coords,%
        distance,coplot,cor,cos,cosh,count,fields,cov,covratio,wt,CRAN,%
        create,crossprod,cummax,cummin,cumprod,cumsum,curve,cut,cycle,D,%
        data,dataentry,date,dbeta,dbinom,dcauchy,dchisq,de,debug,%
        debugger,Defunct,default,delay,delete,deltat,demo,de,density,%
        deparse,dependencies,Deprecated,deriv,description,detach,%
        dev2bitmap,dev,cur,deviance,off,prev,,dexp,df,dfbetas,dffits,%
        dgamma,dgeom,dget,dhyper,diag,diff,digamma,dim,dimnames,dir,%
        dirname,dlnorm,dlogis,dnbinom,dnchisq,dnorm,do,dotplot,double,%
        download,dpois,dput,drop,drop1,dsignrank,dt,dummy,dump,dunif,%
        duplicated,dweibull,dwilcox,dyn,edit,eff,effects,eigen,else,%
        emacs,end,environment,env,erase,eval,equal,evalq,example,exists,%
        exit,exp,expand,expression,External,extract,extractAIC,factor,%
        fail,family,fft,file,filled,find,fitted,fivenum,fix,floor,for,%
        For,formals,format,formatC,formula,Fortran,forwardsolve,frame,%
        frequency,ftable,ftable2table,function,gamma,Gamma,gammaCody,%
        gaussian,gc,gcinfo,gctorture,get,getenv,geterrmessage,getOption,%
        getwd,gl,glm,globalenv,gnome,GNOME,graphics,gray,grep,grey,grid,%
        gsub,hasTsp,hat,heat,help,hist,home,hsv,httpclient,I,identify,if,%
        ifelse,Im,image,\%in\%,index,influence,measures,inherits,install,%
        installed,integer,interaction,interactive,Internal,intersect,%
        inverse,invisible,IQR,is,jitter,kappa,kronecker,labels,lapply,%
        layout,lbeta,lchoose,lcm,legend,length,levels,lgamma,library,%
        licence,license,lines,list,lm,load,local,locator,log,log10,log1p,%
        log2,logical,loglin,lower,lowess,ls,lsfit,lsf,ls,machine,Machine,%
        mad,mahalanobis,make,link,margin,match,Math,matlines,mat,matplot,%
        matpoints,matrix,max,mean,median,memory,menu,merge,methods,min,%
        missing,Mod,mode,model,response,mosaicplot,mtext,mvfft,na,nan,%
        names,omit,nargs,nchar,ncol,NCOL,new,next,NextMethod,nextn,%
        nlevels,nlm,noquote,NotYetImplemented,NotYetUsed,nrow,NROW,null,%
        numeric,\%o\%,objects,offset,old,on,Ops,optim,optimise,optimize,%
        options,or,order,ordered,outer,package,packages,page,pairlist,%
        pairs,palette,panel,par,parent,parse,paste,path,pbeta,pbinom,%
        pcauchy,pchisq,pentagamma,persp,pexp,pf,pgamma,pgeom,phyper,pico,%
        pictex,piechart,Platform,plnorm,plogis,plot,pmatch,pmax,pmin,%
        pnbinom,pnchisq,pnorm,points,poisson,poly,polygon,polyroot,pos,%
        postscript,power,ppoints,ppois,predict,preplot,pretty,Primitive,%
        print,prmatrix,proc,prod,profile,proj,prompt,prop,provide,%
        psignrank,ps,pt,ptukey,punif,pweibull,pwilcox,q,qbeta,qbinom,%
        qcauchy,qchisq,qexp,qf,qgamma,qgeom,qhyper,qlnorm,qlogis,qnbinom,%
        qnchisq,qnorm,qpois,qqline,qqnorm,qqplot,qr,Q,qty,qy,qsignrank,%
        qt,qtukey,quantile,quasi,quit,qunif,quote,qweibull,qwilcox,%
        rainbow,range,rank,rbeta,rbind,rbinom,rcauchy,rchisq,Re,read,csv,%
        csv2,fwf,readline,socket,real,Recall,rect,reformulate,regexpr,%
        relevel,remove,rep,repeat,replace,replications,report,require,%
        resid,residuals,restart,return,rev,rexp,rf,rgamma,rgb,rgeom,R,%
        rhyper,rle,rlnorm,rlogis,rm,rnbinom,RNGkind,rnorm,round,row,%
        rownames,rowsum,rpois,rsignrank,rstandard,rstudent,rt,rug,runif,%
        rweibull,rwilcox,sample,sapply,save,scale,scan,scan,screen,sd,se,%
        search,searchpaths,segments,seq,sequence,setdiff,setequal,set,%
        setwd,show,sign,signif,sin,single,sinh,sink,solve,sort,source,%
        spline,splinefun,split,sqrt,stars,start,stat,stem,step,stop,%
        storage,strstrheight,stripplot,strsplit,structure,strwidth,sub,%
        subset,substitute,substr,substring,sum,summary,sunflowerplot,svd,%
        sweep,switch,symbol,symbols,symnum,sys,status,system,t,table,%
        tabulate,tan,tanh,tapply,tempfile,terms,terrain,tetragamma,text,%
        time,title,topo,trace,traceback,transform,tri,trigamma,trunc,try,%
        ts,tsp,typeof,unclass,undebug,undoc,union,unique,uniroot,unix,%
        unlink,unlist,unname,untrace,update,upper,url,UseMethod,var,%
        variable,vector,Version,vi,warning,warnings,weighted,weights,%
        which,while,window,write,\%x\%,x11,X11,xedit,xemacs,xinch,xor,%
        xpdrows,xy,xyinch,yinch,zapsmall,zip},%
    otherkeywords={!,!=,~,$,*,\%,\&,\%/\%,\%*\%,\%\%,<-,<<-},%
    alsoother={._$},%
    sensitive,%
    morecomment=[l]\#,%
    morestring=[d]",%
    morestring=[d]'
}%

\newlength{\twless}
\newlength{\lmarg}
    

\DefineVerbatimEnvironment
{Verb}{Verbatim}
{fontsize=\fontsize{12pt}{14pt}\selectfont}

\RawFloats[figure,table]            

\DeclareNewFloatType{ListingEnv}{
    placement=htb,
    within=chapter,
    fileext=lol,
    name=Листинг,
}

\makeatletter
\AtBeginDocument{%
    \let\c@ListingEnv\c@lstlisting
    
    \let\ftype@lstlisting\ftype@ListingEnv 
}
\makeatother

\ifthenelse{\equal{\thebibliosel}{0}}{%

\usepackage{cite}                                   

\bibliographystyle{BibTeX-Styles/utf8gost71u}    

\makeatletter
\renewcommand{\@biblabel}[1]{#1.}   
\makeatother

\newtotcounter{citenum}
\def\oldcite{}
\let\oldcite=\bibcite
\def\bibcite{\stepcounter{citenum}\oldcite}
}{

\usepackage[%
backend=biber,
bibencoding=utf8,
sorting=none,
style=gost-numeric,
language=autobib,
autolang=other,
clearlang=true,
defernumbers=true,
sortcites=true,
doi=false,
isbn=false,
]{biblatex}

\DeclareSourcemap{ 
    \maps{
        \map{
            \step[fieldsource=language, fieldset=langid, origfieldval, final]
            \step[fieldset=language, null]
        }
        \map{
            \step[fieldsource=numpages, fieldset=pagetotal, origfieldval, final]
            \step[fieldset=pagestotal, null]
        }
        \map{
            \step[fieldsource=medium,
            match=\regexp{Электронный\s+ресурс},
            final]
            \step[fieldset=media, fieldvalue=eresource]
        }
        \map[overwrite]{
            \step[fieldset=issn, null]
        }
        \map[overwrite]{
            \step[fieldsource=abstract]
            \step[fieldset=abstract,null]
        }
        \map[overwrite]{ 
            \step[fieldsource=urldate,
            match=\regexp{([0-9]{2})\.([0-9]{2})\.([0-9]{4})},
            replace={$3-$2-$1$4}, 
            final]
        }
        \map[overwrite]{ 
            \perdatasource{biblio/othercites.bib}
            \step[fieldset=keywords, fieldvalue={biblioother,bibliofull}]
        }
        \map[overwrite]{ 
            \perdatasource{biblio/authorpapersVAK.bib}
            \step[fieldset=keywords, fieldvalue={biblioauthorvak,biblioauthor,bibliofull}]
        }
        \map[overwrite]{ 
            \perdatasource{biblio/authorpapers.bib}
            \step[fieldset=keywords, fieldvalue={biblioauthornotvak,biblioauthor,bibliofull}]
        }
        \map[overwrite]{ 
            \perdatasource{biblio/authorconferences.bib}
            \step[fieldset=keywords, fieldvalue={biblioauthorconf,biblioauthor,bibliofull}]
        }
        \map[overwrite]{
            \step[typesource=online, fieldsource=howpublished, fieldset=organization, origfieldval, final]
            \step[fieldset=howpublished, null]
        }
    }
}

\addbibresource{biblio/othercites.bib}
\addbibresource{biblio/authorpapersVAK.bib}
\addbibresource{biblio/authorpapers.bib}
\addbibresource{biblio/authorconferences.bib}

\newtotcounter{citenum}
\makeatletter
\defbibenvironment{counter} 
  {\setcounter{citenum}{0}%
  \renewcommand{\blx@driver}[1]{}%
  } 
  {} 
  {\stepcounter{citenum}} 
\makeatother
\defbibheading{counter}{}

\newtotcounter{citeauthorvak}
\makeatletter
\defbibenvironment{countauthorvak} 
{\setcounter{citeauthorvak}{0}%
    \renewcommand{\blx@driver}[1]{}%
} 
{} 
{\stepcounter{citeauthorvak}} 
\makeatother
\defbibheading{countauthorvak}{}

\newtotcounter{citeauthornotvak}
\makeatletter
\defbibenvironment{countauthornotvak} 
{\setcounter{citeauthornotvak}{0}%
    \renewcommand{\blx@driver}[1]{}%
} 
{} 
{\stepcounter{citeauthornotvak}} 
\makeatother
\defbibheading{countauthornotvak}{}

\newtotcounter{citeauthorconf}
\makeatletter
\defbibenvironment{countauthorconf} 
{\setcounter{citeauthorconf}{0}%
    \renewcommand{\blx@driver}[1]{}%
} 
{} 
{\stepcounter{citeauthorconf}} 
\makeatother
\defbibheading{countauthorconf}{}

\newtotcounter{citeauthor}
\makeatletter
\defbibenvironment{countauthor} 
{\setcounter{citeauthor}{0}%
    \renewcommand{\blx@driver}[1]{}%
} 
{} 
{\stepcounter{citeauthor}} 
\makeatother
\defbibheading{countauthor}{}

}

\usepackage{longtable}
 
\usepackage{amsmath,amsthm,amsfonts,amssymb,amscd}
\usepackage{amsfonts}

\usepackage{dsfont}
\usepackage{physics}
\usepackage{calc}
\usepackage{multicol}
\usepackage{mathrsfs}
\usepackage{floatrow}
\usepackage{qcircuit} 
\usepackage{imakeidx}
\makeindex[columns=3, title=Alphabetical Index, intoc]
\usepackage[bottom]{footmisc} 
\usepackage{overpic}
\usepackage{common/Karnaugh}

\usepackage{pdfpages}

\usepackage[bottom]{footmisc}
\usepackage{float}
\usepackage{framed}
\usepackage[most]{tcolorbox}
\theoremstyle{definition}

\usepackage{authblk} 
\usepackage{mdframed}
\usepackage{mathtools}
\usepackage{tabularx}
\usepackage{multirow} 
\usepackage{bm}
\usepackage{subcaption}

\usepackage{mathrsfs}
\usepackage{tikz-cd} 
\usepackage{circuitikz-0.8.3} 

\newcolumntype{R}{>{\raggedleft\arraybackslash}X}
\newcolumntype{L}{>{\raggedright\arraybackslash}X}

\usepackage{listings}
\usepackage{color}
 \usepackage{csquotes}
\definecolor{codegreen}{rgb}{0,0.6,0}
\definecolor{codegray}{rgb}{0.5,0.5,0.5}
\definecolor{codepurple}{rgb}{0.58,0,0.82}
\definecolor{backcolour}{rgb}{0.95,0.95,0.92}
 
\lstdefinestyle{mystyle}{
    backgroundcolor=\color{backcolour},   
    commentstyle=\color{codegreen},
    keywordstyle=\color{magenta},
    numberstyle=\tiny\color{codegray},
    stringstyle=\color{codepurple},
    basicstyle=\footnotesize,
    breakatwhitespace=false,         
    breaklines=true,                 
    captionpos=b,                    
    keepspaces=true,                 
    numbers=left,                    
    numbersep=5pt,                  
    showspaces=false,                
    showstringspaces=false,
    showtabs=false,                  
    tabsize=2
}
 
\newcommand{\XOR}{\mathsf{XOR}}
\newcommand{\B}{\mathbb{B}}

\def\LH{{\sc local Hamiltonian}}
\def\2LH{{\sc $2$-local Hamiltonian}}
\def\5LH{{\sc $5$-local Hamiltonian}}
\def\AH{{\sc $2$-local ZX Hamiltonian}}
\def\SH{{\sc $2$-local ZZXX Hamiltonian}}
\def\RH{{\sc $2$-local real Hamiltonian}}
\def\RFLH{{\sc $5$-local real Hamiltonian}}
\def\clock{{\mathrm{clock}}}
\def\YES{{\sc Yes}}
\def\AND{{\sf AND}}
\def\OR{{\sf OR}}

\newcommand\myarrow{\mathrel{\stackrel{\makebox[0pt]{\mbox{\tiny -1}}}{\longrightarrow}}}

\newcommand{\QMA}{{\sf QMA}}
\newcommand{\BQP}{{\sf BQP}}
\newcommand{\SAT}{{\sc 3-SAT}}
\newcommand{\NP}{{\sf{NP}}}
\renewcommand{\P}{{\sf{P}}}
\def\LH{{\sc local Hamiltonian}}
\def\2LH{{\sc $2$-local Hamiltonian}}
\def\5LH{{\sc $5$-local Hamiltonian}}
\def\AH{{\sc $2$-local ZX Hamiltonian}}
\def\SH{{\sc $2$-local ZZXX Hamiltonian}}
\def\RH{{\sc $2$-local real Hamiltonian}}
\def\RFLH{{\sc $5$-local real Hamiltonian}}
\def\clock{{\mathrm{clock}}}
\def\YES{{\sc Yes}}

\usepackage{dsfont}
\newcommand{\openone}{\mathds{1}}

\newtheorem{theorem}{Theorem}
\newtheorem{lemma}{Lemma}
\newtheorem{example}{Example}
\newtheorem{remark}{Remark}
\newtheorem{definition}{Definition}
\newtheorem{proposition}{Proposition}

\usepackage{mathtools} 

\newmdtheoremenv[outerlinewidth=2,leftmargin=40, rightmargin=20,backgroundcolor=green!7,
outerlinecolor=blue,innertopmargin=0pt,
splittopskip=\topskip,skipbelow=\baselineskip, skipabove=\baselineskip,ntheorem]{myexample}{Example}[section]

\newmdtheoremenv[outerlinewidth=2,leftmargin=40, rightmargin=20,backgroundcolor=white,
outerlinecolor=blue,innertopmargin=0pt,
splittopskip=\topskip,skipbelow=\baselineskip, skipabove=\baselineskip,ntheorem]{mydefinition}{Definition}[section]

\newmdtheoremenv[outerlinewidth=2,leftmargin=40, rightmargin=20,backgroundcolor=white,
outerlinecolor=blue,innertopmargin=0pt,
splittopskip=\topskip,skipbelow=\baselineskip, skipabove=\baselineskip,ntheorem]{myremark}{Remark}[section]

\newmdtheoremenv[outerlinewidth=2,leftmargin=40, rightmargin=20,backgroundcolor=white,
outerlinecolor=blue,innertopmargin=0pt,
splittopskip=\topskip,skipbelow=\baselineskip, skipabove=\baselineskip,ntheorem]{mytheorem}{Theorem}[section]

\newmdtheoremenv[outerlinewidth=2,
backgroundcolor=pink!20,outerlinecolor=blue,innertopmargin=0pt,
splittopskip=\topskip,skipbelow=\baselineskip, skipabove=\baselineskip,ntheorem]{myproblem}{Problem}

\newmdtheoremenv[outerlinewidth=2,
backgroundcolor=pink!20,outerlinecolor=blue,innertopmargin=0pt,
splittopskip=\topskip,skipbelow=\baselineskip, skipabove=\baselineskip,ntheorem]{exercise}{Exercise}

\def\ketbra#1#2{{\vert#1\rangle\!\langle#2\vert}}

\newcommand\bydef{\stackrel{\mathclap{\normalfont\mbox{\normalfont\tiny def}}}{=}}
\def\eye{\mathds{1}}

\graphicspath{{./pics/}}

\title{{\bf \huge \sc {\bf On the Theory of Modern Quantum Algorithms}}\\
{\small Dissertation submitted for the degree of} \\ 
{\small \it Doctor of Physical and Mathematical Sciences} \\
{\small Mathematical Physics 01.01.03} \\~\\ 
Skolkovo Institute of Science and Technology 
}

\DeclareMathOperator{\spn}{span}
\DeclareMathOperator*{\argmin}{arg\,min}
 
 \let\leq\leqslant
 
 \let\geq\geqslant
 
 \let\oldvarepsilon\varepsilon
 \let\epsilon\oldvarepsilon
 \let\oldint\int
 \def\int{\oldint\limits}

 \makeatletter
\let\blx@rerun@biber\relax
\makeatother

\begin{document}

\begin{titlepage}
\begin{center}~\\~\\~\\
Autonomous non-profit organization for higher education\\ SKOLKOVO INSTITUTE OF SCIENCE AND TECHNOLOGY
\\[3\baselineskip] 
\hfill {\small {\it \copyright~manuscript author's signature}}
 \hfill\includegraphics[width=0.35\textwidth]{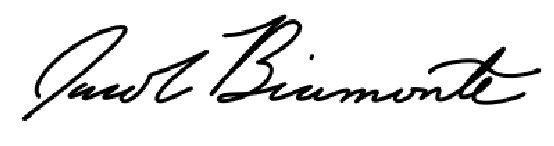}  
\\[2\baselineskip] 

~
\\~\\~\\
{\Large \bf Jacob Daniel Biamonte}
\\~\\~\\
{\Large \sc \bf ON THE MATHEMATICAL STRUCTURE OF QUANTUM MODELS OF COMPUTATION BASED ON HAMILTONIAN MINIMISATION}
\\~\\~\\
{\small  01.01.03 - Mathematical Physics}
\\~\\~\\

{\small Dissertation submitted for the degree of} \\ 
{\small \itshape Doctor of Physical and Mathematical Sciences}
\\
\vfill \bf {Moscow --- 2021}

\end{center}
\end{titlepage}


\includepdf{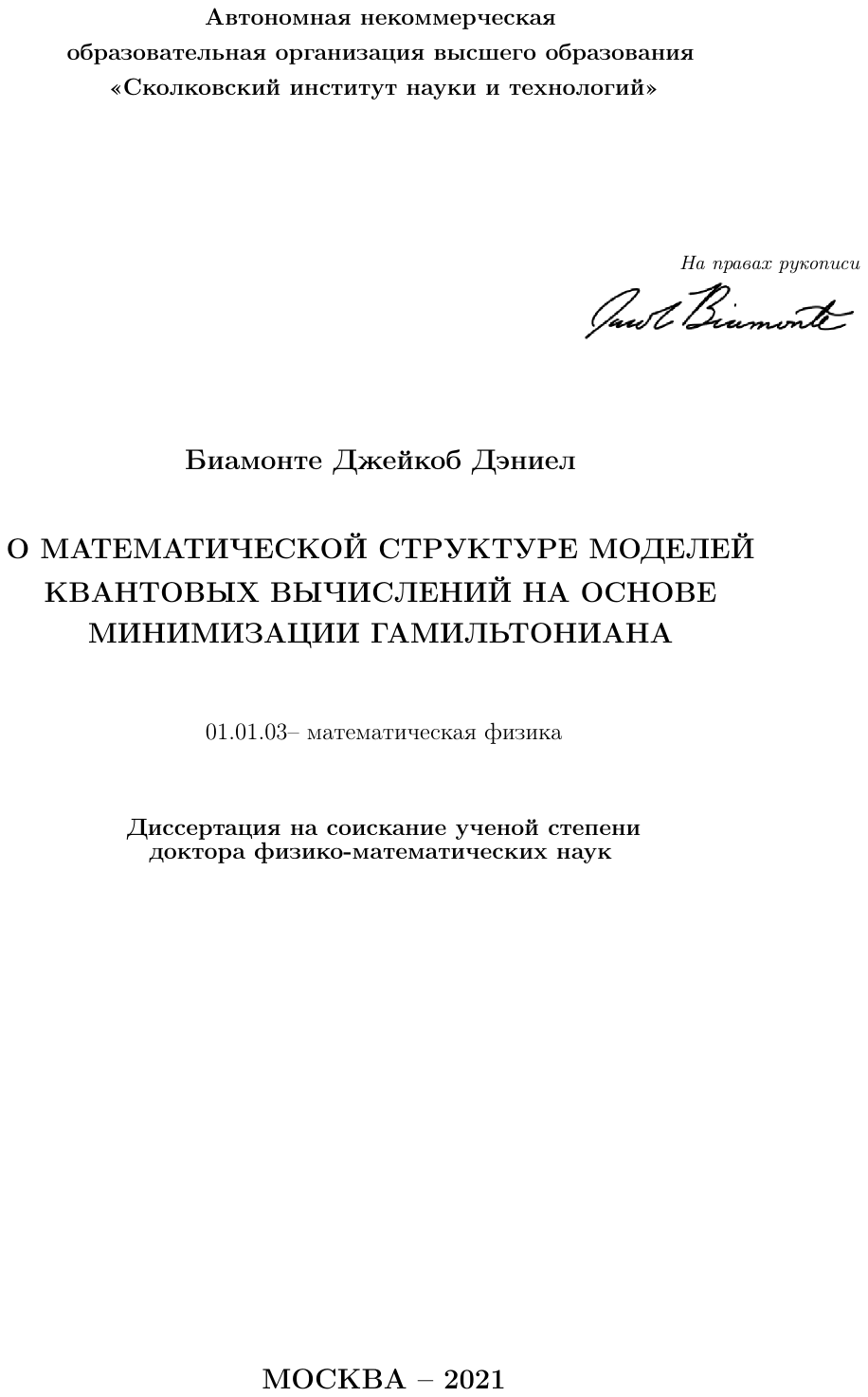}

\newpage



\begingroup
\let\cleardoublepage\clearpage

\setcounter{page}{2}
\tableofcontents
\makeatletter
\makeatother
\newpage 
\endgroup

 

\chapter*{Acknowledgements}

This thesis is my act to participate in the Russian tradition of The Doktor Nauk degree (Doctor of Mathematical and Physical Sciences).  The degree has no academic equivalent in North America, as it is a post-doctoral degree.
It might be compared to the German Habilitation.  While the process and the thesis styles are different, in 1999 Russia and Germany signed a {\it Statement on Mutual Academic Recognition}. This statement included agreeing on the equivalence of the Russian Doctor of Science and the German Habilitation.\footnote{According to the International Standard Classification of Education, for purposes of international educational statistics, a {\it Doctor Nauk in Physical and Mathematical Sciences} is abbreviated as D.Sc.---see paragraph 262 International Standard Classification of Education (ISCED) UNESCO 2011.}  

The most significant portion of this thesis was completed in Moscow after joining the Skolkovo Institute of Science and Technology in 2017.  Around the time when I moved to Moscow, there was a transformation in the field of quantum information processing (see e.g.~this popular account we coauthored~\cite{Biamonte2019}).  The research topic of quantum algorithms began to rapidly leave the university research lab and become a subject newly explored by industry.  This spawned a new era of development: it lead to variational quantum algorithms. From a research point of view, the theory of these new algorithms contained large gaps.  So I decided to explore this emerging domain and to build a research group centered around the mathematical theory of variational quantum algorithms.  Moreover, we approached the topic with tools from ordered mathematical structures and the formal theory of computation as well as condensed matter physics.  

I thank Maxim Fedorov, Grigory Kabatyansky 
and Alexander Kuleshov in particular for their support in hiring me at Skoltech.  I further thank Alexander Kuleshov, Maxim Fedorov, Grigory Kabatyansky, Keith Stevenson, Maxim Fedorov, Lawrence Stein, Pavel Dorozhkin and Alexey Ponomarev for helping to shape the vision and supporting the groups development in various ways, culminating in what it has become today.  The research group turned into an official laboratory structure in 2019 and has grown to approach 20 people.  We are applying some traditional methods of mathematical physics to modern quantum science and technology.  Our focus is to develop the mathematical theory of variational quantum algorithms.  


Large parts of this thesis---including the universality proof of the variational model---represent independent research conducted alone.  While other parts represent successful collaborations, which are credited in the text.  I thankfully have had many strong collaborators over the years.  Here I have included for the most part, my own primary contributions of these joint works.

I thank  Richik Sengupta and Sergey Filippov who kindly offered their advise as I needed it many times, especially when trying to figure out Russian thesis regulations.  The Russian Quantum Center, courtesy in particular of Alex Fedorov and Ruslan Yunusov, kindly agreed to review my work through public presentation as did Sergey Kulik and Stanislav Straupe of M.~V.~Lomonosov Moscow State University Quantum Technology Center.  

I humbly tip my hat to the following readers.  These brave souls (listed alphabetically) found and reported typos, errors or omissions, improving this monograph for all future readers. In alphabetical order I thank: 
Soumik Adhikary,  
Ernesto Campos,  
Ignacio Cirac, 
Alex Fedorov,  
Alexandr Holevo, 
Andrey Kardashin,  
Vladimir Korepin,  
Johannes Jakob Meyer, 
Albert Nasibulin, 
Alex Pechen,  
Walter Pogosov, 
Sergey Rykovanov,  
Richik Sengupta, 
Antoine Tilloy, 
Beat Toedtli, 
Alexey Uvarov and 
Alexander Vlasov. 

I also gratefully acknowledge Aly Nasrallah and Przemys{\l}aw Scherwentke for assisting with \LaTeX~typesetting. I thank  Andrey Kardashin, Dina Fedotova and Ksenia Samburskaya for assisting in understanding the myriad of rules in both the document format and the defense process.  Last but certainly not least, I also thank my degree consultant, Alex Pechen, for his patience and careful advice as well as the official opponents and members of the committee for their time and consideration.

Yours sincerely, ~\hfil \\ \null\hfil

 \includegraphics[width=0.3\textwidth]{signature}\hfil \\ ~\hfil \null\hfil

Jacob Biamonte \hfill Moscow, 2021 

\chapter*{INTRODUCTION}\label{chap:intro}

The Turing machine is one of several abstract models of the computers we are accustomed to. Today's computers which we know and love---albeit smartphones or the mainframes behind the internet---are all built from billions of transistors.  While transistors utilize quantum mechanical effects (such as tunneling: in which an electron can both penetrate and bounce off an energy barrier concurrently), the composite operation of today's computers is purely deterministic or {\it classical}.   By classical, we mean {\it classical mechanics} which is exactly the  physics (a.k.a.~mechanics) we'd anticipate day to day in our lives.  The term {\it quantum mechanics} dates to 1925 in work~\cite{Born1925} by Born and Jordan (in German, {\it quanten mechanik}---without the space of course) and comprises the physics governing atomic systems. Quantum mechanics contains principles and rules that appear to contradict the classical mechanics we are so intuitively familiar with.  Such counter intuitive phenomena provide new possibilities to store and manipulate (quantum) information.  

Quantum computing dates back at least to 1979, when physicist Paul Benioff~\cite{Benioff1980} proposed a quantum mechanical model of the Turing machine. Richard Feynman~\cite{Fey82} and independently Yuri Manin\cite{manin} suggested that a quantum computer had the potential to simulate physical processes that a classical computer could not. Such ideas were further formulated and developed in the work of David Deutsch~\cite{deutsch1985quantum}---Deutsch formulated a quantum Turing machine and applied a sort of {\it anthropic principle} to the plausible computations allowed by the laws of physics.  Namely, what we now call the Church–Turing–Deutsch principle asserts that a universal (quantum) computing device can simulate any physical process.  Yet even the most elementary quantum systems appear impossible to fully emulate using classical computers.  Whereas quantum computers would readily emulate other quantum systems~\cite{Fey82,Lloyd1996}. 

Recent developments in quantum information processing have fostered a global research effort to understand and develop applications for noisy real-world quantum information processors (often called NISQ: Noisy Intermediate-Scale Quantum (NISQ).  Unlike traditional textbook quantum algorithms, quantum algorithms executed on NISQ devices operate in the presence of systematic and random errors.  In practice this limits the depth of the circuit that can be executed.  Experimental developments have lead to a novel utilitarian means of quantum computation enabled by an iterative classical-to-quantum feedback process called, variational quantum computation. 


We aim to present a consistent and general framework, which conceptually binds many of the tools used across contemporary quantum programming.  The unifying focus is on properties of ground states of Hamiltonians. Programming ground states is required in adiabatic quantum computation and other models of ground state annealing while Hamiltonian minimization is also central to physics and chemistry simulation algorithms that are widely anticipated future quantum computing applications.

The {\bf goal} is then simply stated.  We present a coherent view that develops mathematical structures and connects the core ideas across the areas of: 
\begin{enumerate}

    \item[(i)] Ground state and adiabatic quantum computation. 
    
    \item[(ii)] The quantum simulation of ground state properties of physical systems.
    
    \item[(iii)] The variational approach to effective Hamiltonian minimization. 
\end{enumerate}

Indeed, the variational model of quantum computation is stated by means of a Hamiltonian minimization problem that utilizes a classical-to-quantum feedback loop.  We further model and formalize this algorithmic process.

\section*{MAIN STATEMENTS DEFENDED} 

\begin{enumerate}
    \item The formulation of the Ising and quantum kernel problem statements and development of a mathematical apparatus to program parent Hamiltonian models with specific ground state properties.  
    
    \item The development of specific and improved $k$-body to $2$-body Hamiltonian reductions.  
    
    
    \item  Proof that the von Neumann entropy of stochastic propagators on a graph is subadditive. 
    
    \item Showing that (i) $|y_+\rangle=|0\rangle+\imath|1\rangle$, (ii, iii) cups and caps, (iv) Hadamard and (v) {\sf COPY} generate any Clifford tensor network and hence that the ZX tensor rewrite system admits a poly-time terminating rewrite sequence establishing the Gottesman–Knill theorem.
    
    \item The combinatorial quantum circuit area law bounds the maximum possible entanglement across any bipartition of qubits acted on by a quantum circuit comprised of local unitaries and CNOT gates. 
    
    \item Utilization of the parent Hamiltonian mathematical apparatus and gadgets to embed quantum and classical circuits into the low energy sector of Hamiltonians, thereby contributing mappings of {\bf MA}- and {\bf QMA}-hard problem instances to {\bf MA}- and {\bf QMA}-hard Hamiltonian ground state energy decision problems. 
    
    \item Utilization of the mathematical apparatus to embed quantum and classical circuits into the lowest energy state of Hamiltonians thereby mapping {\bf MA}- and {\bf QMA}-hard problem instances to {\bf MA}- and {\bf QMA}-hard Hamiltonian ground state energy decision problems. 
    
    
    \item Proving that physically relevant Hamiltonians---including the tunable Ising model with additional XX-interactions---can embed universal quantum computational resources for ground state quantum computation. 
    
    
    \item The development of the mathematical model to describe variational quantum computation and the establishment of the computational universality of the {\it variational model of quantum computation}.  
    
\end{enumerate}

\section*{CONSISTENCY OF RESULTS} 

This thesis considers qubits.  For short (non-error corrected) circuits this model is physically justified as follows (see the experimental summary in Table \ref{tab:qubitmodel}): 

\begin{enumerate}
    \item NISQ Era variational quantum algorithms consider a fixed error tolerance and tune a short quantum circuit to minimize an objective function. 
    
    \item Circuits with dozens of gates can now be realized with negligible accumulated total error: 
\end{enumerate}

\begin{table}
\begin{center}
\begin{tabular}{|p{3.cm}|p{3.cm}|p{2.cm}|p{3.cm}|p{2.cm}|p{1cm}|}
  \hline
    \textbf{Experiment} & \textbf{Organization} & {\bf Qubits}  & {\bf Ansatz} & \textbf{Depth}& \textbf{year}\\
      \hline
      {\sc Ising model} & Cornell/IBM & 20 & Alternating & 25 & 2021 \\\hline
      {\sc QAOA} & Google & 23 & Split operator & 4 & 2021 \\\hline
      {\sc High energy model} & MSU/SkT & 2 & Checkerboard  & 3 & 2021\\\hline
      {\sc Supremacy} & Google & 53 & HEA & 20 & 2019\\\hline
      {\sc Lattice model} & Innsbruck  & 20 & Split operator & 6 & 2019 \\\hline
      {\sc Chemistry} & IBM & 6 & HEA & 2 & 2017 \\\hline
      
\end{tabular}
\end{center}\caption{Justification of the qubit model in terms of ansatz demonstrations of specified depth/qubit counts.}\label{tab:qubitmodel}
\end{table}

The validity of the results are confirmed by consistency with prior art and rigorous mathematical proofs wherever appropriate.  Numerical experiments were sometimes also employed which reconfirm analytical findings. 

Results forming this dissertation date back several years and appeared in peer reviewed journal articles.  Several of these results now comprise parts of the accepted literature on the topic.  This includes work on Ising model embeddings, work on stochastic versus quantum walks, developing more general perturbation gadgets as well as results on using phase estimation for quantum simulation.  

This so-called {\it variational} approach to quantum computation was formally proven (in the noise free setting) to represent a universal model of quantum computation by this thesis.  This extended and built on several known results appearing in the related topic of Hamiltonian complexity theory.  Many recent studies have not quantified the number of terms needed in the penalty function to implement a variational algorithm.  We hence define a cardinality measure and quantify the number of Pauli terms in the sigma basis.  This is consistent with past findings but presents a new focus to quantify penalty functions.  

In addition, many studies have presented various penalty functions to illustrate that variational algorithms are capable of algorithmic tasks.  A universality proof shows that penalty functions in principle are more general.  This is again consistent with the state of the art.  The original published papers which present these results have become accepted parts of the literature, some over a decade old.

\section*{PRESENTATION OF THE RESULTS} 
 
Contents and results from this dissertation were presented by the author to peers as follows (talks entirely dedicated to the presentation of the DSc thesis are denoted with `[Thesis presentation]' preceding the title of the thesis): \\

\begin{enumerate}

\item {\bf [Thesis presentation] On the mathematical structure of quantum models of computation based on Hamiltonian minimisation} \\
I.E.~Tamm Theory Department, P.N.~Lebedev Institute of Physics, the Russian Academy of Sciences, Moscow, Russian Federation, 22 September 2021 

\item {\bf [Thesis presentation] On the mathematical structure of quantum models of computation based on Hamiltonian minimisation} \\
Laboratory of Quantum Optics and Quantum Information, Center for Advanced Studies, Peter the Great St.~Petersburg Polytechnic University, St. Petersburg, Russian Federation, 15 September 2021 
\item 
{\bf [Thesis presentation] On the mathematical structure of quantum models of computation based on Hamiltonian minimisation}\\
Department of Supercomputers and Quantum Informatics, The Faculty of Computational Mathematics and Cybernetics, Lomonosov Moscow State University, Moscow, Russian Federation, 14 September 2021  

\item {\bf [Thesis presentation] On the mathematical structure of quantum models of computation based on Hamiltonian minimisation} \\
Department of Higher Mathematics, Moscow Institute of Physics and Technology, Moscow, Russian Federation, 8 September 2021  

\item {\bf [Thesis presentation] On the mathematical structure of quantum models of computation based on Hamiltonian minimisation} \\
Skolkovo Institute of Science and Technology, Moscow, Russian Federation, 7 September 2021 

\item {\bf [Thesis presentation] On the mathematical structure of quantum models of computation based on Hamiltonian minimisation} \\
Kazan Quantum Center, Kazan National Research Technical University named after A.N. Tupolev, Kazan, Russian Federation, 4 September 2021

\item {\bf [Thesis presentation] On the mathematical structure of quantum models of computation based on Hamiltonian minimisation} \\
Max Planck Institute of Quantum Optics, Hans-Kopfermann-Str.~1
85748 Garching, 21 July 2021

\item {\bf On variational quantum computation}\\
(General Institutional Seminar), P.N.~Lebedev Institute of Physics, the Russian Academy of Sciences, Moscow, Russian Federation, 17 March 2021 

\item {\bf [Thesis presentation] On quantum computation by variation of a quantum circuits parameters to minimise an effective Hamiltonian iteratively realised by local measurements}\\
Department of Mathematical Methods for Quantum Technologies, Steklov Mathematical Institute of the Russian Academy of Sciences,  Moscow, Russian Federation, 25 March 2021

\item {\bf [Thesis presentation] On the mathematical structure of quantum models of computation based on Hamiltonian minimisation}\\
Skolkovo Institute of Science and Technology, Skolkovo, Russian Federation, 25 September 2020 

\item {\bf [Thesis presentation] On the mathematical structure of quantum models of computation based on Hamiltonian minimisation} \\
The Russian Quantum Center, Skolkovo, Russian Federation, 26 Aug 2020

\item {\bf [Thesis presentation] On the mathematical structure of quantum models of computation based on Hamiltonian minimisation}\\
M.V.~Lomonosov Moscow State University Quantum Technologies Center, Moscow, Russian Federation, 14 July 2020

\item {\bf Variational Models of Quantum Computation}\\ Episode IX, Google Research Series on Quantum Computing \\
Google Poland, Warsaw Poland, 10 October 2019 
    
    \item {\bf A Universal Model of Variational Quantum Computation}\\
Quantum Machine Learning and Data Analytics Workshop \\
Purdue University, Discovery Park, West Lafayette Indiana \\ United States, September 2019 

\item {\bf Quantum Enhanced Machine Learning} \\ 
Physics Challenges in Machine Learning for Network Science \\
Queen Mary University of London\\
London, United Kingdom, September 2019 

\item {\bf Quantum Machine Learning for Quantum Simulation}\\
Machine Learning for Quantum Matter \\ 
Nodita, Stockholm, Sweden, August 2019 

\item {\bf Recent Results in the Theory of Variational Quantum Computation}\\ 
the 5th International Conference on Quantum Technologies \\
The Russian Quantum Center, Moscow Russia 2019 

\item {\bf Variational Quantum Computation in Photonics}\\ 
The 28th Annual International Laser Physics Workshop \\
Gyeongju, South Korea, July 2019

\item {\bf Trends in Variational Quantum Algorithms}\\
Overview style talk given (multiple times) at 
\begin{enumerate}
    \item Riken Institute (Japan)
    \item NTT laboratories (Tokyo, Japan)
    \item CIIRC Institute (Prague)
\end{enumerate}

\item {\bf Quantum Machine Learning Matrix Product States}\\ 
Keynote talk at the Workshop on Quantum Information\\
Harvard, USA, April 23-24, 2018

\item {\bf Quantum Complex Networks}\\
Keynote Lighting Talk at International school and conference on network science (NetSci) \\
Paris, France 2018

\end{enumerate}

\section*{PUBLICATIONS}
The author has sixty two papers listed in Scopus [September 2021]. The thesis compiles results from twenty primary research articles, one book and two review articles. A list of twenty publications is given at the end of this synopsis.  

\section*{AUTHOR CONTRIBUTION}
The author has had many successful collaborations.  The main results of the dissertation were published in small teams or as single author manuscripts.  Results derived with collaborators are clearly indicated as such, either in the body of the text or in reference to the result/theorem.  The focus has been on the authors own contribution to these joint works.  

\section*{DISSERTATION STRUCTURE}

The dissertation consists of an introduction, six chapters, a conclusion, a list of symbols, a list of abbreviations, a glossary of terms, a bibliography, a list of figures, a list of tables and finally an alphabetical index.  

\section*{DISSERTATION CONTENTS} 


To present the most central portions of the theory underpinning contemporary quantum algorithms, we focus on Hamiltonian ground states.  The rudimentary though still non-trivial starting point is understanding how to program ground states of Ising type models.  

We state the following properties of quantum theory stated in terms of quantum bits (qubits).  


\begin{definition}[Complex Euclidean space] 
\begin{equation}
    V_n=[\mathbb{C}^2]^{\otimes n} \cong [\mathbb{C}]^{2^n}
\end{equation}
\end{definition}

We will equivalently write $[\mathbb{C}^2]^{\otimes n}$, $\mathbb{C}_2^{\otimes n}$. 

\begin{remark}
${\mathscr L}(\mathbb{C}_2^{\otimes n})$ denotes the space of linear maps from $\mathbb{C}_2^{\otimes n}$ to itself.
\end{remark}

The dissertation considers the following linear maps:
\begin{remark}[Linear qubit maps]
~
\begin{enumerate}
    \item States: $\psi \in V_n$ 
    \item Effects: $\psi^\dagger \in V_n^\star = (V_n \rightarrow \mathbb{C})$. 
    \item Hamiltonians $A$ in  $\text{herm}_{\mathbb{C}}(2^n) \equiv \lbrace A \in \mathcal{L}(V_n) \: | \: A = A^\dagger \rbrace$. 
    \item Propagators $U$ in $\text{\bf U}_{\mathbb{C}}(2^n)\equiv
            ~\{U\in\mathcal{L}(V_n) ~|~U^\dagger U=\eye\}$. 
\end{enumerate}
\end{remark}

\begin{remark}[Inner product]
The standard inner product is used:
        \begin{equation}
        \bra{\cdot}\cdot\rangle: V_n^* \otimes V_n \rightarrow \mathbb{C},~~(\phi,\psi)\rightarrow \bra{\phi}\psi\rangle =\sum_{j}\bar{\phi}_j \psi^j \in \mathbb{C}. \notag
        \end{equation}
States/effects are unit $\ell_2$ vectors.
\end{remark}

\begin{remark}[Computational basis] 
The dissertation tends to fix the so called, computational basis:
\begin{enumerate}
\item $n$-qubit basis: $\mathcal{B}_n = \{\ket{0},\ket{1}\}^{\otimes n}$ with $2^n$ orthonormal basis vectors
\item Single qubit basis: $\ket{0},\ket{1}\in \mathcal{B}_1$ 
\item $\text{span}_{\mathbb C}\{\mathcal{B}_n\} \cong V_n = [{\mathbb C}^2]^{\otimes n}$
\end{enumerate}
\end{remark}

\begin{remark}[Properties of Sigma matrices] This thesis makes use of the following properties of Sigma matrices: 
\begin{enumerate}

    \item $\sigma^l \sigma^m=\imath \epsilon_{lmn}\sigma^n+\delta_{lm}\eye$ 
    
    \item $R_{\bm n}(\theta) = e^{-\imath\theta(\bm n.\bm \sigma)} = \cos\theta-\imath(\bm n.\bm\sigma)\sin\theta$,
where $\bm n.\bm \sigma\equiv n_1\sigma^1+n_2\sigma^2+n_3\sigma^3$ 
    

    \item  $\text{herm}_{\mathbb C}(2^n) = \text{span}_\mathbb{R}\left\{ \bigotimes_{l=1}^{n} \sigma^{\alpha_{l}}_l  \mid \alpha_{l} = 0,1,2,3 \right\}$ 
    
    \item $\sigma^a_{k} \equiv \eye_{1}\otimes\ldots\otimes \eye_{k-1} \otimes \sigma^a_{k}\otimes \eye_{k+1}\otimes\ldots\otimes \eye_{n}$

    
\end{enumerate}
\end{remark}

\begin{definition}[Pauli group]  $\mathbf{P}_n=\left\{ e^{\imath\theta\pi/2} \bigotimes_{j=1}^{n} \sigma^{\alpha_{j}}  \mid \theta, \alpha_{j} = 0,1,2,3 \right\}$. 
\end{definition}

\begin{remark}
$\text{span}\{\Re(\  \mathbf{P}_n)\}=\text{span}_\mathbb{R}\left\{ \bigotimes_{l=1}^{n} \sigma^{\alpha_{l}}_l  \mid \alpha_{l} = 0,1,2,3 \right\}$. 
\end{remark}

\begin{definition}[Clifford group] $\mathbf{C}_n=\{C\in {\bf U}(2^n)\mid C\mathbf{P}_nC^\dagger = \mathbf{P}_n\}$. 
\end{definition}

\begin{remark}
For $C\in \mathbf{C}_n$, 
\begin{equation}
    C \left(\bigotimes_{j=1}^{n} \sigma^{\alpha_{j}} \right) C^\dagger = \pm \bigotimes_{j=1}^{n} \sigma^{\gamma_{j}}
\end{equation}
for $\alpha_{j}, \gamma_j \in \lbrace 0,1,2,3 \rbrace$. 
\end{remark}

\begin{remark}[The expected value of a Hamiltonian relative a state]
The dissertation will consider the expected value as:
        \begin{equation}
            (A,\psi, \psi^\dagger)\rightarrow \bra{\psi}A\ket{\psi}=\sum_{l,m} A_{l,m}\bar{\psi}_m\psi_l\in \mathbb{C} \notag
        \end{equation} 
        for $A\in\text{herm}_{\mathbb{C}}(2^n)$. 
\end{remark}

The dissertation works with Hamiltonian operators.  The simplest case is the generalized Ising model.  

\begin{remark}[Generalized Ising model]
 A generalized Ising model is an energy function of a symmetric graph $G=(E, V)$.  The energy (Hamiltonian) function is given as: 
\begin{equation}
{H}_{\text{Ising}}=\sum_{j\in V}h_js_j+\frac{1}{2}\sum_{l,m\in E}J_{lm}s_ls_m. 
\end{equation}
where $s_j\in \{\pm 1\}$
\end{remark}

The dissertation relies on connections between problems in mathematical physics and the theory of complexity. 

\begin{remark}
We assume all numbers are defined to some fixed but arbitrary finite precision to avoid pathologies.  
\end{remark}

\begin{definition}[The class \textbf{NP}]
 A problem class $\Gamma$ is said to be inside \textbf{NP} if candidate solutions to instances $\omega\in\Gamma$ can be verified in time $\mathcal{O}(\text{poly}(|\omega|))$.
\end{definition}

The concept of minimisation problems where the inputs are {\it easy} to evaluate is one of the concepts motivating the dissertation.  For example: one can determine the energy of a given spin configuration with respect to the following Hamiltonian using an algorithm that is polynomial in the number of Hamiltonian terms/size of the input. 
\begin{equation}
{H}_{\text{Ising}}=\sum_{i}h_is_i+\sum_{i,j}J_{ij}s_is_j. 
\end{equation}

\begin{remark}
The minimisation of generalized Ising Hamiltonians is {\bf NP}-hard.
\end{remark}

\begin{definition}
 A problem is {\bf NP}-hard if all problems inside {\bf NP} can be reduced to it (Karp reduction).
\end{definition}

\begin{definition}
 A problem is {\bf NP}-complete when it is in {\bf NP} and also {\bf NP}-hard.
\end{definition} 

\begin{definition}
A language $L \in$ \textbf{MA}[a,b] if there exists a probabilistic polynomial time verifier $V$, such that:
\begin{enumerate}
    \item $\forall x \in L ~~~ \exists y: ~~ \vert y \vert = poly(\vert x \vert), P(V(x, y) = 1) \geq a $
    
    \item $\forall x \not\in L ~~~ \forall y: ~~ \vert y \vert =poly(\vert x \vert), P(V(x, y) = 1) \leq b$
\end{enumerate}
\end{definition}

\begin{remark}~
\begin{enumerate}
    \item The numbers $a,b \in [0,1]$ are such that $a-b \geq poly(\vert x \vert^{-1})$
    \item One would consider instance $x$ to be the description of a probabilistic circuit taking input $y$ and outputting $V(x,y) \in [0,1]$
    \item \textbf{NP} = \textbf{MA}$[1,0]$
\end{enumerate}
\end{remark}

\begin{definition}
A language $L \in$ \textbf{QMA}$[a,b]$ if there exists a polynomial time quantum verifier $V$ such that:
\begin{enumerate}
    \item $\forall x \in L ~~~ \exists \ket{\xi}\in [\mathbb{C}^2]^{ \otimes \text{poly}(|x|)}: P(V(x,\ket{\xi}) = 1) \geq a $
    \item $\forall x \not\in L ~~~ \forall  \ket{\xi}\in [\mathbb{C}^2]^{ \otimes \text{poly}(|x|)} ~~ P(V(x, \ket{\xi}) = 1) \leq b$
\end{enumerate}
\end{definition}

\begin{remark}~

\begin{enumerate}
    \item The numbers $a,b \in [0,1]$ are such that $a-b \geq \text{poly}(\vert x \vert^{-1})$
    \item One would consider instance $x$ to be the description of a quantum circuit taking input state $\ket{\xi}$ and outputting on the first qubit $V(x,y) \in [0,1]$
    \item It is assumed that the verifier has access to a slack register initially in the state  $\ket{0}^{{ \otimes \text{poly}(\vert x \vert)}}$
\end{enumerate}
\end{remark}

\subsection*{Chapter 1} 

The dissertation begins by recalling several established results related to programming the ground states of generalised Ising systems.  This presents and builds on my own work as well as the work of others---see the dissertation for citations.  

The first chapter begins by considering the relationship between qubit quantum states and Ising penalty functions.  

The dissertation begins by defining the field extension: 
\begin{equation}
    \mathbb{C}[x_1, x_2, \dots, x_n] \Big/ x_1, x_2, \dots, x_n \in \{0,1\} 
\end{equation}
where $x_1, x_2, \dots, x_n \in \{0,1\}$ the quotient constraint is equivalent to $x_i x_i = x_i$ (idempotence). We arrive at the ring of (qubit) polynomials of type: 
\begin{equation}
    \{0, 1\}^n \rightarrow \mathbb{C} 
\end{equation}
by means of the following mapping
\begin{equation}
    f({\bf x}) = \sum_{I\in \{0, 1\}^n} a_I {\bf x}^I 
\end{equation}
where 
\begin{equation}
    {\bf x}^I \bydef (x_1)^{i_1} (x_2)^{i_2} \cdots (x_n)^{i_n}
\end{equation}
and we abuse notation as 
$$
(x)^0 \bydef (1-x) 
$$
with $x^1 = x$. 

It goes on to state the following propositions. 
\begin{proposition}[Biamonte (2008)]\label{prop:grade1}
The ring $\mathbb{C}[x_1, x_2, \dots, x_n] \Big/ \forall i, x_i^2 = x_i$ is graded as 
\begin{equation}
    \mathbb{C} \oplus \mathbb{C}[x_1]\oplus \cdots\oplus \mathbb{C}[x_n] \oplus \mathbb{C}[x_1, x_2] \oplus \cdots \oplus \mathbb{C}[x_{n-1}, x_n] \oplus \cdots \oplus \mathbb{C}[x_1, x_2, \dots, x_n] 
\end{equation}
where the quotients are omitted for brevity of notation.  
\end{proposition}

We call an expansion canonical when it is unique up to labeling variables. 

\begin{proposition}
The expansion 
\begin{equation}
    f({\bf x}) = a_0 + \sum a_i x_i + \sum a_{ij}x_i x_j + \cdots + \sum a_{ij\dots n}x_i x_j \dots x_n 
\end{equation}
is canonical. 
\end{proposition}

More generally, the early chapter presents the following: 

\begin{lemma}\label{lemma:iso1} 
The follow isomorphisms hold. 
\begin{equation}
    \mathbb{C}[x_1, x_2, \dots, x_n] \Big/ \forall i, x_i^2 = x_i \simeq \mathbb{C}_2^{\otimes n} \simeq \text{diag}\text{Mat}_{\mathbb{C}}(2^n) 
\end{equation}
\end{lemma}

By considering the real valued restriction from Proposition \ref{prop:grade1} and hence Lemma \ref{lemma:iso1}, this concept formally connects pseudo Boolean and Ising minimization problems: 

\begin{proposition}[Operator embedding of Pseudo Boolean forms]
Any Pseudo Boolean function 
\begin{equation}
    f({\bf x})  = \sum_I a_I {\bf x}^I
\end{equation}
gives rise to an operator embedding 
\begin{equation}
    [f] = \sum_I a_I \ketbra{I}{I}  \bydef {\hat f}
\end{equation}
by Lemma \ref{lemma:iso1}.  The minimisation problems are evidently related as:
\begin{equation}
    \min_{x\in \{0, 1\}^n}f(x) = x', 
\end{equation}
then 
\begin{equation}
    \min_{\psi \in {\mathcal A}} \bra{\psi}{\hat f}\ket{\psi}  = \bra{x'}{\hat f}\ket{x'} 
\end{equation}
for the appropriate vector space ${\mathcal A}$. 
\end{proposition}

The dissertation then details the practical codomain extension of Karnaugh maps.  This is used to derive penalty functions for logical operations.  In particular, a deductive method is presented based on Karnaugh maps to derive the following penalty functions. The method to derive these appears novel whereas various penalty functions exist in the literature. 

\begin{theorem}
The following penalty functions embed the logical product $-x_1x_2x_3$ into their lowest energy sector as:  
\begin{equation}
-x_1x_2x_3=\min_{z\in\{0,1\}}z(2-x_1-x_2-x_3), 
\end{equation}
and  
\begin{equation}
-x_1x_2x_3=\min_{z\in\{0,1\}}z(-x_1+x_2+x_3)-x_1x_2-x_1x_3+x_1. 
\end{equation}
\end{theorem}

\begin{theorem}[Boolean function embedding, Biamonte (2008)]\label{thm:booleantoH1}
Any Boolean function $f(x_1, x_2, \dots, x_n)$ expressed over the basis $\{\vee, \wedge, \neg\}$ embeds into the spectrum of a Hermitian operator formed by the linear extension of $\{P_0, P_1, \eye\}$ by means of the following maps \eqref{eqn:m11} and \eqref{eqn:m21}. %
\begin{align}\label{eqn:m11}
&\wedge \longrightarrow \otimes 
\\
\label{eqn:m21}
&\vee \longrightarrow + 
\end{align}
For every (positive polarity, a.k.a.~non-negated) Boolean variable $x_j$ we apply 
\begin{equation}\label{eqn:m31} 
x_j \longrightarrow P_1^j. 
\end{equation}
For negated variable $\neg x_j$  we apply 
\begin{equation}\label{eqn:m41}
\neg x_j \longrightarrow P_0^j. 
\end{equation}
In both cases \eqref{eqn:m31} and \eqref{eqn:m41},  $1\leq j\leq n$ becomes a spin label index which $P^j$ acts on. Moreover 
the above mapping induces an operator $\mathcal{H}$ such that 
\begin{equation}
    \mathcal{H}\ket{{\bf x}} = f({\bf x}) \ket{{\bf x}}
\end{equation}
for Boolean function $f({\bf x})$ and bit string ${\bf x}$. 
\end{theorem}

\begin{theorem}[Kernel embedding]
A Boolean function $f(x)$ embeds into the kernel of a non-negative Ising penalty function by applying the map from Theorem \ref{thm:booleantoH1} to the function $g(x, f(x)) = 0$, $g(x, 1-f(x)) = 1$. 
\end{theorem}

\begin{remark}
The condition $g(x, 1-f(x)) = 1$ can readily be modified to $g(x, 1-f(x)) \geq 1$ leaving the operators constructed by Theorem \ref{thm:booleantoH1} non-negative with identical kernals. 
\end{remark}

\begin{definition} The set of all two-body Ising Hamiltonians on $n$ spins is defined as: 
$\begin{aligned}[t]
  \Omega_n &= \{a_0+a_1x_1+a_2x_2+\dots +a_{12}x_1x_2 + a_{13}x_1x_3+\dots\\
    &a_{n-1,n}x_{n-1}x_n| \forall j, k, a_{jk}\in [-l, l]\subset \mathbb{R}\}.
\end{aligned}$
\end{definition}

\begin{proposition}
$\nexists H\in \Omega_3$ $|$ $\text{Ker}\{H\} = \text{span}\{x, y, z\in \mathbb{B}|z = x\oplus y\}$. 
\end{proposition}

We will also show that the orbits of embedded functions in $\text{Ker}\{H\}$ separate under conjugation of $H$ by $\sigma_x$ into equivalency classes:  {(\sf AND $\sim$ OR $\sim$ NAND $\sim$ NOR) $\in \Omega_3$} and  {(\sf XOR $\sim$ EQV) $\in \Omega_4$}. 

\subsection*{Chapter 2} 

The second chapter presents a detailed mathematical (structural) comparison between quantum and stochastic mechanics.  Table \ref{tab:cvspvsqbits1} is presented. Then the contents of Table 2 are developed. 

\begin{table}
    \centering
    \begin{tabular}{m{0.15\textwidth} m{0.15\textwidth} m{0.30\textwidth} m{0.30\textwidth}}
    & bits & probabilistic bits & qubits 
    \end{tabular}
    \noindent
    \begin{tabular}{m{0.15\textwidth} |m{0.15\textwidth}| m{0.30\textwidth}| m{0.30\textwidth}}
    \hline
    state (single unit) & bit $\in \lbrace 0,1 \rbrace $ & real vector \newline $a, b \in \mathbb{R}_{+}$ \hfill $a+b=1$ \newline $\Vec{p} = a\Vec{0} + b\Vec{1}$ or $a \ket{0} + b\ket{1}$  & complex vector \newline $\alpha , \beta \in \mathbb{C}$ \hfill $\abs{\alpha}^{2} + \abs{\beta}^{2} = 1$ \newline 
    $\Vec{\psi} = \alpha\Vec{0} + \beta\Vec{1}$ or $\alpha \ket{0} + \beta \ket{1}$
    \\
    \hline
    
    state (multi-unit) & bitstring \newline $x \in \lbrace 0,1 \rbrace^{n}$ & prob.distribution (stochastic vector) \newline $\Vec{p} = \sum_{x\in\{0, 1\}^n} a_x \ket{x} \in [\mathbb{R}^2_+]^{\otimes n}$ & wavefunction (complex vector) \newline $\Vec{\psi} =\sum_{x\in\{0, 1\}^n} \alpha_x \ket{x} \in [\mathbb{C}^2]^{\otimes n}$
    \\
    \hline
    
    operations & Boolean logic & stochastic matrices \newline $\sum_{j} P_{ij} =1, $ $P_{ij}\geq 0$  & unitary matrices \newline $U^{\dagger}U = \bf{1}$ 
    \\
    \hline
    
    component ops & Boolean gates & tensor product of matrices & tensor product of matrices
    \\
    \hline
    
    \end{tabular}
    \caption{Summary of deterministic, probabilistic and quantum bits. We use the standard notation that $\mathbb{R}^2_+$ denotes the two-dimensional real vector space with non-negative entries. Likewise, $\mathbb{C}^2$ is the two-dimensional complex vector space.  The space of $n$ pbits,  $n$ qubits are respectively given by the tensor product of spaces, $[\mathbb{R}^2_+]^{\otimes n}$ and $[\mathbb{C}^2]^{\otimes n}$. }
    \label{tab:cvspvsqbits1}
    
\end{table}

\begin{table}[htbp!]
\begin{center}
\renewcommand{\arraystretch}{1.5}
\resizebox{0.75\textwidth}{!}{\begin{minipage}{\textwidth}
\begin{tabular}{ p{3cm}|p{5cm}|p{5cm} }
         & {\bf quantum mechanics} & {\bf stochastic mechanics} \\\hline 
  {\bf state}  
& vector $ \psi \in \mathbb{C}^n$ with
$$ \sum_i |\psi_i|^2 = 1 $$
  &
vector $\psi \in \mathbb{R}^n$ with $$ \sum_i \psi_i = 1 $$ 
and we typically insist that,  $$ \psi_i \ge 0 $$ 
\\\hline 
{\bf observable} & $n \times n$ matrix ${\mathcal O}$ with 
$$  {\mathcal O}^\dagger = {\mathcal O}$$
where $({\mathcal O}^\dagger)_{i j} \bydef \overline{{\mathcal O}}_{j i}$
  & vector ${\mathcal O} \in \mathbb{R}^n$  \\\hline 
{\bf expected value} &
$$ \langle \psi |{\mathcal O}| \psi \rangle \bydef \sum_{i,j} \overline{\psi}_i {\mathcal O}_{i j} \psi_j $$
&  $$ \langle {\mathcal O} \psi \rangle \bydef \sum_i {\mathcal O}_i \psi_i $$ 
\\\hline
  {\bf symmetry}\break (linear map sending states to states) & unitary $n \times n$ matrix: $$ U U^\dagger = U^\dagger U = \eye $$
 & stochastic $n \times n$ matrix: $$ \sum_i U_{i j} = 1 , \quad U_{i j} \ge 0 $$
\\\hline  
  {\bf symmetry \break generator} & self-adjoint $n \times n$ matrix: $${\mathcal H}={\mathcal H}^\dagger$$ 
 & infinitesimal stochastic $n \times n$ \break matrix:
$$ \sum_i {\mathcal H}_{i j}=0 , \quad  i\neq j \;$$ 
\newline $$\Rightarrow \; {\mathcal H}_{i j} \le 0 $$ 
\\\hline
{\bf symmetries from symmetry \break generators} &
$$ U(t) = \exp(-\imath t{\mathcal H}) $$ &
$$ U(t) = \exp(-t{\mathcal H}) $$
\\\hline
{\bf equation of \hfill \break motion} & $$\imath \frac{d}{dt} \psi(t) = {\mathcal H} \psi(t)$$ with solution $$\psi(t) = \exp(-\imath t{\mathcal H})\psi(0)$$ & $$\frac{d}{dt} \psi(t) = -{\mathcal H} \psi(t)$$ with solution $$\psi(t) = \exp(-t{\mathcal H})\psi(0)$$ 
 \end{tabular} 
\end{minipage} }
\caption{Summary of quantum versus statistical mechanics.}
\end{center}
\end{table}

\begin{remark}
Every finite dimensional quantum or stochastic process can be viewed as a (spinless single particle) walk on a graph given by the support of the corresponding time propagator. 
\end{remark}

For the purpose of comparison, the following definitions are all given in the dissertation.  

\begin{remark}[Summary of stochastic versus quantum walks]
$G$ is a simple graph.  Labeling the nodes of $G$ lifts to specify:
\begin{enumerate}
    \item $A$ the adjacency matrix (generator of a quantum walk). 
    \item $D$ the diagonal matrix of the degrees. 
    \item ${\mathcal L}$ the symmetric Laplacian (generator of stochastic and quantum walks), which when normalized by $D$ returns both: 
    \item[3.1] $S$ the generator of the uniform escape stochastic walk and
    \item[3.2] $Q$ the quantum walk generator to which ${\mathcal L}$ is similar. 
\end{enumerate}
\end{remark}

Several results are derived, leading to the subadditivity of entropy of stochastic generators: 

\begin{remark}
A simple undirected graph with edges weighted by real numbers gives rise to a {\it generalized symmetric adjacency matrix}.  For edges labeled $l$ and $m$ weighted by $w \in \mathbb{R}$, the $l$-$m^{\textup th}$ entry of the corresponding adjacency matrix is $w$.
\end{remark}

\begin{definition}
A {\it generalized Laplacian} arises as 
\begin{equation}
    \mathcal{L} = \mathcal{D} - \mathcal{A} 
\end{equation}
where $A$ is a generalized symmetric adjacency matrix and $D$ stores on its diagonal entries the sums of the corresponding rows of $A$. 
\end{definition}

\begin{theorem}[Biamonte-DeDomenico 2016]
	Given two generalized Laplacians and their sum  $\mathcal{L}_C=\mathcal{L}_A+\mathcal{L}_B$, and corresponding Gibbs state density matrices $ \mathcal{\rho_{C}} = e^{\beta(\mathcal{L_{A}+L_{B}})}/\mathcal{Z} $, the von Neumann entropy $S(\rho) = Tr\{\rho \ln_2 \rho\}$ is subadditive as, 
	\begin{equation}
	S(\mathbf{\rho_{C}}) \leq S(\mathbf{\rho_{A}}) + S(\mathbf{\rho_{B}}). 
	\end{equation}
\end{theorem}

\begin{remark}
We adopt the notation that $ S(\rho_A) \equiv S_A $, $ S(\rho_B) \equiv S_B $, etc.
\end{remark}

The second chapter concludes by presenting several methods to find minimal graph properties on a quantum processor.  

\section*{Chapter 3} 

Techniques from the theory of tensor networks can apply to quantum circuits.  In chapter 3 the following theorem on generating families of tensor networks is proven. 

\begin{theorem}[Minimal Stabilizer Tensor Generators]
The following generating tensors are sufficient to simulate any stabilizer quantum circuit: \begin{enumerate}
    \item[(a)] a vector $\ket{t}\bydef \ket{0}+\imath \ket{1}$,
    \item[(b)] the Hadamard gate and 
    \item[(c)] the {\sf XOR}- and {\sf COPY} tensors and 
    \item[(d)] a covector $\bra{+}\bydef \bra{0}+\bra{1}$. 
\end{enumerate}
\end{theorem}

\begin{remark}
The Gottesman–Knill theorem states that stabilizer circuits---circuits that only consist of gates from the normalizer of the qubit Pauli group, a.k.a.~Clifford group---can be simulated in polynomial time on a probabilistic classical computer. 
\end{remark}

The dissertation constructs a sequence of graphical rewrites to establish this theorem by algebraic properties of tensor contraction, namely: 

\begin{theorem}[Graphical Proof of the Gottesman--Knill Theorem] 
For $n$-qubits acted on by $L$ Clifford gates, there exists a confluent sequence of rewrites, that establishes the Gottesman--Knill theorem in ${\mathcal O}(\text{poly}(n, L))$ steps. 
\end{theorem}

\subsection*{Chapter 4}
We then consider the minimisation of Hamiltonians by parameterised quantum circuits. This provides an illustrative connection between computational and physical complexity, stated and defined in the early chapter step wise. The variational model contains the following ingredients which will be further defined: 
\begin{enumerate}
    \item {\bf States}. A vector of real parameters $\boldsymbol{\theta}$ sets a circuit to produce $\ket{\psi ( \boldsymbol{\theta} )}$. 
    
    \item {\bf Measurements}. Expected values of a Pauli strings, $\bigotimes_{j=1}^{n} \sigma_j^{\alpha_{j}}$ for $\alpha_{j} \in \lbrace 0,1,2,3 \rbrace$ can be computed for each $\ket{\psi ( \boldsymbol{\theta} )}$. 
    
    \item {\bf Compute cost function}. A cost function defined by a weighted sum of expected values is computed for each $\ket{\psi ( \boldsymbol{\theta} )}$.  
    
    \item {\bf Outer-loop optimization}. Classical optimization routines update parameters $\boldsymbol{\theta}\rightarrow \boldsymbol{\theta}_\star$. 
\end{enumerate}

\begin{definition}[Variational Statespace---Biamonte (2021)]
The variational statespace of a $p$-parameterized $n$-qubit state preparation process is the union of $\ket{\psi({\boldsymbol \theta})}$ over all possible assignments of real numbers ${\boldsymbol \theta}$:
\begin{equation}
   \Gamma = \bigcup_{\boldsymbol{\theta}\in (0, 2 \pi]^{\times p}}\ket{\psi(\boldsymbol{\theta})}. 
\end{equation}
\end{definition} 

\begin{definition}[Variational Sequence] 
A variational sequence specifies parameters to prepare a state in a variational statespace.  It can be given by defining a specific sequence of gates or by specifying control parameter values. 
\end{definition}

\begin{definition}[Variational principle]
 A {\it variational principle} is a problem specific reduction to that of finding extrema of an objective function.   Variational quantum computation considers the normalized minimization:
\begin{equation}
    \min_{\ket{\psi(\boldsymbol{\theta})}\in \Gamma\subset V_n}\bra{\psi(\boldsymbol{\theta})} H\ket{\psi(\boldsymbol{\theta})} \geq \min_{\ket{\psi} \in V_n}\bra{\psi} H\ket{\psi}.
\end{equation}
\end{definition}

\begin{remark}
Alternative NISQ approaches might minimise the variance 
\begin{equation}
\min_{}(\langle\mathcal{H}^2\rangle - \langle\mathcal{H}\rangle ^2) \geq 0 
\end{equation}
which vanishes if and only if $\ket{\psi}$ is an eigenstate of $\mathcal H$.
\end{remark}

Cost function implementation proceeds by applying the fact that an expected value of a sum is a sum of expected values.   
\begin{equation} 
\bra{\psi}{{\mathcal H}}\ket{\psi} = \bra{\psi}{\sum_k h_k\bigotimes_{j=1}^{n} \sigma_j^{\alpha_{j}(k)} }\ket{\psi} = \sum_k h_k \bra{\psi}{\bigotimes_{j=1}^{n} \sigma_j^{\alpha_{j}(k)}}\ket{\psi}
\end{equation}
where $h_k$ is a real number and $\bigotimes_{j=1}^{n} \sigma_j^{\alpha_{j}(k)}$ is a Pauli string for $\alpha_{j} \in \lbrace 0,1,2,3 \rbrace$.  

\begin{remark}[Iteration]
Given copies of $\ket{\psi}$, measuring $\bigotimes_{j=1}^{n} \sigma_j^{\alpha_{j}(k)}$ repeatedly gives an estimate for each $\bra{\psi}{\bigotimes_{j=1}^{n} \sigma_j^{\alpha_{j}(k)}}\ket{\psi}$ separately.
\end{remark}




\begin{remark}
Whereas the objective function can be evaluated term-wise, achieving tolerance $\sim\epsilon$ requires $\sim\epsilon^{-2}$ measurements---see Hoeffding's inequality. 
\end{remark} 

\begin{definition}[Objective Function Cardinality] 
The number of terms in the Pauli basis $\{\openone , X, Y, Z\}^{\otimes n}$ needed to express an objective function. 
\end{definition}

\begin{example}
Let ${\mathcal H}=\sum_k h_k\bigotimes_{j=1}^{n} \sigma_j^{\alpha_{j}(k)}$ for coefficients $h_k$ and Pauli strings $\bigotimes_{j=1}^{n} \sigma_j^{\alpha_{j}(k)}$.  Then $|{\mathcal H}|_{\text{card}}= \sum_k (h_k)^0$.
\end{example}


\begin{definition}[Bounded Objective Function---Biamonte (2021)] 
A family of objective functions is {\it efficiently computable} when uniformly generated by calculating the expected value of an operator with $\text{poly}(n)$ bounded cardinality over 
\begin{equation} 
\begin{split}
\Omega &  \subset \{\openone , X, Y, Z\}^{\otimes n}.
 \end{split}
\end{equation}
\end{definition}

\begin{definition}[Poly-Computable Objective Function---Biamonte (2021)]
An objective function
\begin{equation}
    f\colon\ket{\phi}^{\times \mathcal{O}(\text{poly}(n))} \rightarrow \mathbb{R}_+
\end{equation}
is called poly-computable provided $\text{poly}(n)$ independent physical copies of $\ket{\phi}$ can be efficiently prepared to evaluate a bounded objective function.  
\end{definition}  

\begin{definition}[Accepting a Quantum State---Biamonte (2021)]
An objective function $f$ {\it accepts} $\ket{\phi}$ when given $\mathcal{O}(\text{poly}~n)$ copies of $\ket{\phi}$, 
\begin{equation}
    f(\ket{\phi}^{\times \mathcal{O}(\text{poly}(n)}) = f(\ket{\phi}, \ket{\phi}, \cdots, \ket{\phi})< \Delta
\end{equation}
 evaluates strictly less than a chosen real parameter $\Delta > 0$.  
  \end{definition}

\begin{theorem}[Energy to Overlap Theorem---Biamonte (2021)]
Let non-negative $\mathcal{H}= \mathcal{H}^\dagger\in \mathscr{L}(\mathbb{C}_d)$ have spectral gap $\Delta$ and non-degenerate ground eigenvector $\ket{\psi}$ of eigenvalue $0$.  
Consider then a unit vector $\ket{\phi}\in \mathbb{C}_d$ such that 
\begin{equation}
\bra{\phi }\mathcal{H}\ket{\phi } < \Delta 
\end{equation} 
it follows that 
\begin{equation}
1 - \frac{\bra{\phi}\mathcal{H}\ket{\phi} }{\Delta} \leq | \braket{\phi}{\psi}|^2 \leq 1 - \frac{\bra{\phi}\mathcal{H}\ket{\phi} }{\text{Tr}\{ \mathcal{H}\}}.
\end{equation} 
\end{theorem}

Several constructions related to quantum approximate optimization using short parameterised quantum circuits are subsequently developed.  A general bound applicable to short circuits is then given. 

Consider a pure $n$-qubit state $\ket{\psi}$. 
\begin{definition}
Bipartite Rank is the Schmidt number (the number of non-zero singular values) across any reduced bipartite density state from $\ket{\psi}$ (i.e.~$\lceil n/2 \rceil$ qubits). 
\end{definition}
 
\begin{definition}
 An ebit is a unit of entanglement contained in a maximally entangled two-qubit (Bell) state.
\end{definition}

\begin{remark}
A quantum state with $q$ ebits of entanglement (quantified by any entanglement measure) contains the same amount of entanglement (in that measure) as $q$ Bell states. \item If a task requires $r$ ebits, it can be done with $r$ or more Bell states, but not with fewer.  Maximally entangled states in $\mathbb{C}^d\otimes \mathbb{C}^d$ have $\log_2(d)$ ebits of entanglement. 
\end{remark}

The dissertation then presents and proves the following: 

\begin{theorem}[Combinatorial quantum circuit area law---Biamonte-Morales-Koh (2020)]
Let $c$ be the depth of 2-qubit controlled {\sf NOT} gates in a ansatz circuit. Then the maximum possible number of ebits accross any bi-partition is $\min \{ \left \lfloor{n/2}\right \rfloor, c \}$.
\end{theorem}

Finally, the chapter presents the definition of an effect the dissertation author discovered and published with coauthors. 

\begin{definition}
 Let $\ket{\psi}$, be the ansatz states generated from a \textit{p}--depth {\sf QAOA} circuit. Then 
\begin{equation}\label{eqn:reachabilitydef1}
    f = \min_{\psi \in \Gamma \subset \mathcal{H}} \bra{\psi}\mathcal{V}\ket{\psi} - \min_{\phi \in \mathcal{H}} \bra{\phi}\mathcal{V}\ket{\phi}\geq 0,
\end{equation}
characterises the limiting performance of {\sf QAOA}. 
\end{definition}

The R.H.S.~of equation \eqref{eqn:reachabilitydef1} can be expressed as a function, $f(p,\alpha,n)$.

\begin{proposition}[Reachability Deficit---with Akshay et al.~2020]
For $p \in \mathbb{N}$ and fixed problem size, $\exists$ $ \alpha > \alpha_c$ such that $f$ from \eqref{eqn:reachabilitydef1} is non-vanishing. This is a reachability deficit.
\end{proposition}

\subsection*{Chapter 5}

Chapter 5 develops a universal model of variational quantum computation.  The early chapter related to programming diagonal Hamiltonin ground states.  Chapter's 5 and 6 focus on the non-diagonal case.  

The dissertation then goes on to construct Hermitian $\mathcal{H} \in \mathcal{L}(\mathbb C_2^{\otimes n})$ with $\mathcal{H}\geq 0$ and non-degenerate  $\ket{\psi}\in \mathbb C_2^{\otimes n}$  as $\mathcal{H}\ket{\psi}=0$.   Define  
\begin{equation}\label{eqn:proj1}
P_\phi = \sum_{i=1}^n \ket{1}\bra{1}^{(i)} = \frac{n}{2}\left(\openone - \frac{1}{n}\sum_{i=1}^n   Z^{(i)} \right) 
\end{equation} 
and consider \eqref{eqn:proj1} as the initial Hamiltonian, preparing state $\ket{0}^{\otimes n}$. 

We will act on \eqref{eqn:proj1} with a sequence of gates $ \prod_{l=1}^L U_l$ corresponding to the circuit being simulated as 
\begin{equation}\label{eqn:isoaffine1}
h(k) = \left(\prod_{l=1}^{k\leq L} U_l \right)P_\phi \left(\prod_{l=1}^{k\leq L} U_l\right)^\dagger \geq 0
\end{equation}
which isospectral on \eqref{eqn:proj1}.


\begin{lemma}[Clifford Gate Cardinality Invariance] \label{lemma:invariance1} 
For $C$ a Clifford gate and $h\in \text{span}_\mathbb{R}\left\{ \bigotimes_{l=1}^{n} \sigma^{\alpha_{l}}_l  \mid \alpha_{l} = 0,1,2,3 \right\}$, $|h|_{\text{card}} = |C h C^\dagger|_{\text{card}}$. 
\end{lemma}

\begin{remark}
The algebraic $k$-locality of \eqref{eqn:isoaffine1} is not invariant under Clifford conjugation.  
\end{remark}

\begin{remark}
Non-Clifford gates increase the cardinality of \eqref{eqn:isoaffine1} by exponentially and so must be logarithmically bounded from above,  restricting to $p$ gate circuit's with $\mathcal{O}(\textup{poly} \ln p)$ non-Clifford single qubit gates.
\end{remark}

We will then consider embedding general quantum circuits into Hamiltonian ground states. 

Two notions of universality are common in the literature: 
\begin{enumerate}
    \item Strongly universal means a system is fully controllable and able to approximate any state.  
    
    \item Computationally universal means that any quantum circuit can be efficiently simulated by this model. 
\end{enumerate}

\begin{remark}[with Morales and Zimboras QIP {\bf 19}:291 (2020)]
One can simulate general $p$-depth circuits containing two-qubit gates with ansatze circuits of $\mathcal{O}\left(\text{poly}(p)\right)$ depth.
\end{remark}

\begin{theorem}[Biamonte PRA 103:L030401 (2021)] 
Let $\Pi_{l=1}^L U_l \ket{0^n}$ be an $L$-gate quantum circuit preparing state $\ket{\psi}$ on $n$-qubits and containing $L-c$ non-Clifford gates. Then there exists a non-negative Hamiltonian ${\mathcal H}$ on n-qubits with $|{\mathcal H}|_{\text{card}}= \mathcal{O}\left(\text{poly}(c, e^{L-c})\right)$, gap $\Delta$ and $\ker\{ {\mathcal H} \}= \text{span}\{ \Pi_{l=1}^L U_l \ket{0^n}\}$. In particular, if $\ket{\phi}$ is such that 
\begin{equation}
0\leq \bra{\phi} {\mathcal H} \ket{\phi}  < \Delta 
\end{equation} 
then it follows that 
\begin{equation}
1 - \frac{\bra{\phi} {\mathcal H}\ket{\phi} }{\Delta} \leq | \braket{\phi}{\psi}|^2 \leq 1 - \frac{\bra{\phi}{\mathcal H}\ket{\phi} }{\text{Tr}\{ {\mathcal H}\}}.
\end{equation} 
\end{theorem}


For some $U$ a Clifford gate, Lemma \ref{lemma:invariance1} shows that the cardinality is invariant.  Non-Clifford gates increase the cardinality by factors $\mathcal{O}(e^{L-c})$ and so must be logarithmically bounded from above.  Hence, telescopes bound the number of expected values by restricting to circuit's with 
$$
k\sim\mathcal{O}(\text{poly} \ln n)
$$ 
general single qubit gates. Clifford gates do however modify the locality of terms appearing in the expected values. 

Chapter 5 then presents then proves the following theorem (\ref{thm:history1}) which establishes universality of the variational model of quantum computation. 

\begin{theorem}[Universal Objective Function---Biamonte (2021)]\label{thm:history1}
Consider a quantum circuit of $L$ gates on $n$-qubits producing state $\prod_l U_l \ket{0}^{\otimes n}$.  Then there exists an objective function (Hamiltonian, $\mathcal{H}$) with non-degenerate ground state, cardinality $\mathcal{O}(L^2)$ and spectral gap $\Delta\geq \mathcal{O}(L^{-2})$ acting on $n+\mathcal{O}(\ln L)$ qubits such that acceptance implies efficient preparation of the state $\prod_l U_l \ket{0}^{\otimes n}$.  Moreover, a variational sequence exists causing the objective function to accept.  
\end{theorem} 

The proof follows from several lemma.  Degeneracy is first lifted. We let $P_0 = \ket{0}\bra{0}$. 
\begin{lemma}[Degeneracy Lifting] 
A tensor product of a projector on the first clock qubit with a telescope  
\begin{equation}
\mathcal{H}_{\text{in}} = V\left( \sum_{i = 1}^n P_1^{(i)} \right)V^\dagger \otimes P_0
\end{equation} 
lifts the degeneracy of $\mathcal{H}_{\text{prop}}$ and the history state with fixed input as 
\begin{equation}
\frac{1}{\sqrt{L+1}} \sum_{t=0}^L \prod_{l = 1}^t U_l(V\ket{0}^{\otimes n}) \otimes \ket{t} 
\end{equation}
becomes the non-degenerate ground state of $J\cdot \mathcal{H}_{\text{in}} + K\cdot  \mathcal{H}_{\text{prop}}$ for real $J, K>0$.
\end{lemma}

The penalty function is gaped and omits a log-space embedding. 

\begin{lemma}[Gap Existence]
For appropriate non-negative $J$ and $K$, the operator $J\cdot \mathcal{H}_{\text{in}} + K\cdot  \mathcal{H}_{\text{prop}}$ is gapped as
\begin{equation}
\Delta \geq \max\{ J,  \frac{K \pi^2}{2(L+1)^2} \}. 
\end{equation} 
\end{lemma}

\begin{lemma}[Logspace Embedding $\mathcal{H}_{\text{prop}}$] 
The clock space of $\mathcal{H}_{\text{prop}}$ embeds into $\mathcal{O}(\ln L)$ slack qubits, leaving the ground space of $J\cdot \mathcal{H}_{\text{in}} + K\cdot  \mathcal{H}_{\text{prop}}$ and the gap invariant. 
\end{lemma}

The dissertation then proves acceptance and derives the bound, noting that one must add $M$ identity gates to boost the probability of the desired circuit output state $\ket{\phi } = \Pi _{l=1}^LU_l \ket{0}^{\otimes n}$. The telescoping construction, we have that 
\begin{equation}\label{eqn:application1}
1 - \frac{\bra{\phi}\mathcal{H}\ket{\phi} }{\max\{ J,  \frac{K \pi^2}{2(L+1)^2} \}} \leq | \braket{\phi}{\psi_{\text{hist}}}|^2 = \frac{1}{1+\frac{L+1}{M}}
\end{equation} 
whenever $\bra{\phi }\mathcal{H}\ket{\phi } < \max\{ J,  \frac{K \pi^2}{2(L+1)^2} \}$.  For large enough $M>L$, the right hand side of \eqref{eqn:application1} approaches unity, implying acceptance. 

\subsection*{Chapter 6}

\begin{remark}
Kitaev et al.~established that sparse Hamiltonian's restricted to have at most 5-body bounded strength interactions have a ground state energy problem which is complete for the quantum analog of the complexity class {\bf NP} ({\sf QMA}-hard).
\end{remark}

\begin{definition}
 The $k$-local Hamiltonian problem: The input is a $k$-local Hamiltonian acting on n qubits, which is the sum of poly many Hermitian matrices that act on only $k$ qubits. The input also contains two numbers $a< b \in [0,1]$, such that $\frac{1}{b-a}=\mathcal{O}(n^{-c})$ for some constant $c$. The problem is to determine whether the smallest eigenvalue of this Hamiltonian is less than $a$ or greater than $b$, promised that one of these is the case.
\end{definition}
    
\begin{remark}
The $k$-local Hamiltonian admits an energy decision problems with is {\sf QMA}-complete for $k \geq 2$.  The minimisation of $k$-local Hamiltonians is {\sf QMA}-hard for $k \geq 2$. We seek to determine the simplest $2$-local {\sf QMA}-hard Hamiltonian to embed computational problems into a Hamiltonian for practical means.  
\end{remark}

The dissertation then develops and proves the following theorems. 

\begin{remark}[Real Hamiltonians]
We call Hamiltonian's expressed in the real subset of the Pauli basis, {\it real Hamiltonians}.  That is, qubit Hamiltonians that contain no tensor product terms with odd numbers of $Y$ operator(s). The corresponding ground state energy problem is called {\scshape Real Hamiltonian}. 
\end{remark}

\begin{lemma}
The ground state energy decision problem {\scshape Real Hamiltonian} is {\sf QMA}-hard. 
\end{lemma}

\begin{remark}[Complexity (Sketch)]
Given a Hamiltonian on $n$ qubits, determine if $\min_{\ket{\psi}\in V_n} \bra{\psi}H\ket{\psi}$ is below $b$ or above $a$ for $a, b\in [0,1]$ and $b-a\geq \text{poly}(n^{-1})$. 
\end{remark}

\begin{remark}[Universality (Sketch)]
A computationally universal set of real valued gates is embedded to act in ground states of \eqref{eqn:zzxx0} and \eqref{eqn:zx0}. 
\end{remark} 

\begin{theorem}[Biamonte-Love (2008)]
The ground energy decision problem {\it ZZXX Hamiltonian} is {\sf QMA}-hard, given as: 
\begin{equation}\label{eqn:zzxx0}
H_{\text{ZZXX}}=\sum_{i}h_i Z_i+\sum_{i,j}J_{ij} Z_i Z_j+\sum_{i,j}K_{ij} X_i X_j.
\end{equation}
\end{theorem}

\begin{theorem}[Biamonte-Love (2008)]
The ground energy decision problem {\it ZX Hamiltonian} is {\sf QMA}-hard, given as: 
\begin{equation}\label{eqn:zx0}
{H}_{\text{ZX}}=\sum_{i}h_i Z_i+\sum_{i,j}J_{ij} Z_i X_j.
\end{equation}
\end{theorem}

The Hamiltonian \eqref{eqn:zzxx1} (that is, \ref{eqn:zx1}) can create effective $Y\otimes Y$ (that is, $Z\otimes Z\otimes Z$) interactions with error $\epsilon$ using one slack bit acted on by a term $\sim \epsilon^{-4} Z$ (that is, $\sim \epsilon^{-5} X$).  

\begin{theorem}[with Cao-et al.~(2015)]
The Hamiltonian
\begin{equation}\label{eqn:zzxx1}
   H_{\text{ZZXX}}=\sum_{i}h_i Z_i+\sum_{i}\Delta_i X_i+\sum_{i,j}J_{ij} Z_i Z_j+\sum_{i,j}K_{ij} X_i X_j. 
\end{equation}
emulates a $Y\otimes Y$ interaction with $\delta=\mathcal{O}(\epsilon^{-4})$  given one slack qubit. 
\end{theorem}

\begin{theorem}[with Cao-et al.~(2015)]
The Hamiltonian
\begin{equation}\label{eqn:zx1}
    H_{\text{Ising,X}}=\sum_{i}h_i Z_i+\sum_{i}\Delta_i X_i+\sum_{i,j}J_{ij} Z_i Z_j.
\end{equation}
emulates the $Z\otimes Z \otimes Z$ interaction with $\delta=\mathcal{O}(\epsilon^{-5})$ given one slack qubit.
\end{theorem}

\subsection*{Conclusion}

The conclusion presents and discusses the implications of efficiently checkable quantum versus classical minimization problems.  It also presents some future research directions. 

Anticipated computational resources to determine ground state energy and calculate energy relative to a state have been conjectured.  In Table \ref{table:zoo1} I have summarized what is known/conjectured regarding efficiently checkable minimisation problems.  Therein `Restricted Ising' denotes problems known to be in {\sf P}. ($^\star$) denotes conjectures. Electronic structure problem instances have constant maximum size so are assumed to be in {\sf BQP} whereas the ZZXX model is {\sf QMA}-hard.

\begin{center}
\begin{table}[H]\caption{{ Hamiltonian complexity micro zoo}}
  \begin{tabular}{|p{5.4cm}|p{5.4cm}|p{5.4cm}|}
  \hline
      \textbf{Problem Hamiltonian} & {\bf Finding Ground Energy} (Classical / Quantum) & {\bf Calculating State Energy}  (Classical / Quantum) \\
      \hline
      {\sc 1-Local Hamiltonian} & Polynomial & Polynomial\\\hline
      {\sc 2-Local Ising} & Exp  & Polynomial \\\hline
      Electronic Structure & $^\star$Exp & $^\star$Exp / Polynomial\\\hline
      ZZXX Model & Exp  & $^\star$Exp / $^\star$Polynomial\\\hline 
  \end{tabular} \label{table:zoo1}
  \end{table}
\end{center}

\chapter{The Algebra of Programming Hamiltonian Ground States}\label{chap:progGS}

A universal model of quantum computation is an abstraction of a physical process.  The abstract model is then proven to---in principle---be able to emulate any quantum circuit efficiently in the circuits size. Early ideas in quantum computation~\cite{Feynman59, Fey82, deutsch1985quantum, Feynman1986, Deutsch73} lead to the so called, circuit (a.k.a.~gate) model of quantum computation (see the book \cite{NC}). In the history of quantum computation, several other models have been proven to be universal models of quantum computation through their computational equivalence to the defacto quantum circuit model~\cite{NC}. This includes adiabatic quantum computation~\cite{farhi2014quantum, 2004quant.ph..5098A} both discrete and continuous quantum walks~\cite{PhysRevLett.102.180501, Lovett_2010}, measurement based quantum computation~\cite{PhysRevLett.97.150504} as well as this authors installment proving universality of the variational model~\cite{UVQC}.  

To study contemporary quantum computing applications, we will adopt the view of computation in terms of ground states of Hamiltonians \cite{KSV02, 2004quant.ph..5098A}.  This will later be used as a foundation to understand the modern class of variational quantum algorithms \cite{2014NatCo...5E4213P, UVQC}.  The approach taken provides an elegant and practical connection between theoretical computer science and mathematical physics.  On one hand, computational complexity can classify ground state problems.  On the other hand, physical systems can be constructed and their ground states can be accessed and used as a computational resource.

The computational properties of ground states are the unifying theme of this thesis and indeed, offer a golden thread connecting the contemporary fundamental underpinnings behind advanced techniques to program quantum enhanced processors, of many types.  We want to begin by explaining the core ideas as simply and as plainly as possible.  Our starting point is the generalized Ising model: proven to be universal for classical computation.  We will then develop these ideas piece-wise as our journey together through these pages commences.  We will later extend the techniques developed in this chapter to the case of Hamiltonians with non-diagonal (a.k.a.~non-commuting terms) which are proven to be universal models for ground state quantum computation.  

We are concerned with instances of two general problems, where a problem is defined as a {\it set} (class) in complexity theory.  The first is a decision problem, which serves essentially as a theoretical tool to study the limits of computation.  

This chapter requires the following elementary properties of quantum theory.  Quantum computation functions with quantum bits (qubits). Qubits should both be isolated from their surroundings yet also be made to interact. In practice, design imperfections and random noise can not be avoided, meaning that the ideal qubit can never exist.  Such noise processes serve to restrict quantum circuit depth. 

\begin{remark}[$n$-qubits]
 We work with the complex Euclidean space $\mathcal{H} \bydef [\mathbb{C}^2]^{\otimes n}$.
\end{remark}

\begin{remark}
By  ${\mathscr L}(\mathbb{C}_2^{\otimes n})$ we will denote the space of linear maps from $\mathbb{C}_2^{\otimes n}$ to itself.
\end{remark}

\begin{remark}[Linear qubit maps]
We will consider the following linear maps:
\begin{enumerate}
    \item States: $\psi \in \mathcal{H} \simeq (\mathbb{C}\rightarrow \mathcal{H})$ where $\psi$ takes the complex number $c$ into $\mathcal{H}$ trivially as $\psi(c)= c\cdot \psi \in \mathcal{H}$. 
    \item Effects: $\psi^\dagger \in \mathcal{H}^\star \simeq (\mathcal{H} \rightarrow \mathbb{C})$. 
    \item Hamiltonians $A$ in  $\text{herm}_{\mathbb{C}}(2^n) \equiv \lbrace A \in \mathcal{L}(\mathcal{H}) \: | \: A = A^\dagger \rbrace$. 
    \item Propagators $U$ in $\text{\bf U}_{\mathbb{C}}(2^n)\equiv
            ~\{U\in\mathcal{L}(\mathcal{H}) ~|~U^\dagger U=\eye\}$. 
\end{enumerate}
\end{remark}

\begin{remark}
We will work with the standard inner product:
        \begin{equation}
        \bra{\cdot}\cdot\rangle: \mathcal{H}^* \otimes \mathcal{H} \rightarrow \mathbb{C},~~(\phi,\psi)\rightarrow \bra{\phi}\psi\rangle =\sum_{j}\bar{\phi}_j \psi^j \in \mathbb{C}. \notag
        \end{equation}
Here, states/effects are unit $\ell_2$ vectors. 
\end{remark}

\begin{remark}[Computational basis]
We will typically fix the so called, computational basis: 
\begin{enumerate}
    \item Single qubit basis vectors are given as $\ket{0},\ket{1}$.  
    \item Composite $n$-qubit basis are taken from $\{\ket{0},\ket{1}\}^{\otimes n}$
    \item The $2^n$ basis vectors satisfy $\braket{l}{m}=\delta_{lm}$.  
\end{enumerate}
\end{remark}

\begin{remark}[The expected value of a Hamiltonian relative a state]
We consider the expected value as
        \begin{equation}
            (A,\psi, \psi^\dagger)\rightarrow \bra{\psi}A\ket{\psi}=\sum_{l,m} A_{l,m}\bar{\psi}_m\psi_l\in \mathbb{C} \notag
        \end{equation} 
        for $A\in\text{herm}_{\mathbb{C}}(2^n)$. 
\end{remark}

We will work with Hamiltonian operators.  The simplest case is the generalized Ising model.  

\begin{remark}[Generalized Ising model]\label{def:ising}
 A generalized Ising model is an energy function of a symmetric graph $G=(E, V)$.  Each edge in $E$ is weighted by a real number and each vertex in $V$ is assigned a binary variable $s\in \pm 1$. Each vertex in $V$ can further be associated with an {\it onsite energy}.   
 
 To calculate the (pseudo) energy of a graph relative to an edge assignment ${\bf s}$, we will consider the edge weight matrix $J$ where entry the $l,m$th entry in $E$ is the so called interaction energy between vertex $l$ and $m$ in $V$ and the onsite energy vector $h$.  The energy function is given as: 
\begin{equation}\label{eqn:is}
{H}_{\text{Ising}}=\sum_{j\in V}h_js_j+\frac{1}{2}\sum_{l,m\in E}J_{lm}s_ls_m. 
\end{equation}
\end{remark}

\begin{remark}
Two mathematical problems arise in the literature related to the Ising model \eqref{eqn:is}.  Here and in \S~\ref{chap:varintro} and \ref{chap:gadgets} we consider the calculation of the ground state configurations of \eqref{eqn:is}.  We calculate various forms of the partition function: 
\begin{equation}
    {\mathcal Z} = \text{Tr}\{e^{-\beta {H}_{\text{Ising}}} \}, 
\end{equation}
for $\beta$ playing the role of inverse temperature in \S~\ref{chap:qvspc}, paying particular attention to the Ising model in \S~\ref{chap:varintro}. 
\end{remark}

\section{{\sf P}- vs.~{\sf NP} problems and mathematical physics}

Let us continue with a more formal definition of efficiency. First we consider the following definition.  

\begin{definition}(Decision Problem)
 A decision problem instance can be posed as a {\scshape Yes}-{\scshape No} question of input values. 
\end{definition}

The following represent {\scshape Yes} or {\scshape No} decision problems. 

\begin{example}(Primality Testing)
Is a given natural number prime? 
\end{example}

\begin{example}
Given two numbers $x$ and $y$, does $x$ evenly divide $y$?  The answer is either {\scshape Yes} or {\scshape No} depending upon the values of $x$ and $y$.
\end{example}

\begin{remark}[Efficient algorithm]
A method for solving a decision problem, given in the form of an algorithm, is called a {\bf decision procedure} for that problem.  If the procedure or algorithm terminates in time bounded above by some polynomial in the problem size (i.e.~the size of the problems description), then we call the procedure/algorithm/process {\bf efficient}. 
\end{remark}

We will now consider {\sf P} as containing the set of {\bf tractable} decision problems: decision problems for which we have polynomial-time algorithms.

\begin{definition}
 The complexity class {\sf P} contains all decision problems that can be solved with worst-case polynomial time-complexity.
\end{definition}

\begin{remark}
In other words, a decision problem is in the class {\sf P} if there exists an algorithm that solves any instance of size $n$ in ${\mathcal O}(n^k)$ time, for some $n$ independent integer $k$.
Here $n$ is the number of bits needed for encoding the input. 
\end{remark}

The second class of decision problems that we are concerned with here in the class {\sf NP}, which stands for {\bf non-deterministic polynomial time}.  This class was among the first to be connected with concepts appearing in mathematical physics. In 1982 Barahona proved that finding the ground state of reasonably simplistic Ising spin glass models is {\bf NP}-hard \cite{Barahona82}.  

\begin{definition}[The class \textbf{NP}]
 A problem class $\Gamma$ is said to be inside \textbf{NP} if candidate solutions to instances $\omega\in\Gamma$ can be verified in time $\mathcal{O}(\text{poly}(|\omega|))$.
\end{definition}

\begin{example}[Generalized Ising model is in {\bf NP}]
The problem is to determine if a generalized Ising Hamiltonian 
\begin{equation}\label{eqn:zx}
{H}_{\text{Ising}}=\sum_{j}h_js_j+\sum_{j, k}J_{jk}s_js_k,  
\end{equation}
has a ground eigenvalue less than real $a$ or if all eigenvalues are greater than $b$ for $b-a=\mathcal{O}(1)$ as large as the lowest spectral gap.  This problem is evidently in {\bf NP} as one can determine the energy of a given spin configuration using an algorithm that is polynomial in the number of Hamiltonian terms/size of the input. 
\end{example}

\begin{definition}[The class \textbf{NP}-hard]
A problem is {\bf NP}-hard when all problems inside {\bf NP} can be reduced to it (Karp reduction).  The minimisation of generalized Ising Hamiltonians is {\bf NP}-hard.
\end{definition}

\begin{example}
Minimising ${H}_{\text{Ising}}=\sum_{j, k}J_{jk}s_js_k$ for $J_{jk}\in\{0, 1, -1\}$ is {\bf NP}-hard \cite{Barahona82}. 
\end{example}

\begin{definition}[The class \textbf{NP}-complete]
 A problem is {\bf NP}-complete when it is in {\bf NP} and also {\bf NP}-hard.
\end{definition}

\begin{remark}[Problems versus instances]
Many optimisation problems, such as restricted forms of function minimization, are in the class {\sf NP}: more generally optimisation by minimization is {\sf NP}-hard.  Problems in {\sf NP} can be mapped (by a many-to-one mapping) onto an $\sf NP$-hard optimisation incarnation.  
\end{remark}

\begin{remark}
{\sf NP}-complete problems represent a subclass of {\sf NP} containing the hardest problems inside {\sf NP}: each of the problems inside the {\sf NP}-complete subclass are readily mapped from one to another (by polynomial, i.e.~Karp, reduction). The minimisation of Ising models is readily cast to an {\sf NP}-complete decision problem. If any {\sf NP}-complete problem has a polynomial time algorithm, all problems in {\sf NP} do. Such an algorithm is widely conjectured not to exist.
\end{remark}

\begin{remark}
We assume that $\mathbb{N} \bydef \mathbb{N}\cup \{0\}$. This sometimes appears in the literature as $\mathbb{N}_0$. 
\end{remark}

Consider then a pseudo Boolean function $f$ .  That is, $f$ is a map from $n$-tuples of $0$ and $1$ to the integers between $0$ and some natural number (possibly defined to be ${\mathcal O}(\text{poly}~n)$), that is for $f$ pseudo Boolean we have type 
\begin{equation}
    f\colon\{0, 1\}^n \rightarrow \mathbb{N}.  
\end{equation}
If we consider the class of all such functions, under the strict condition that $f$ can be evaluated for all $y\in \{0, 1\}^n$ in some time not exceeding ${\mathcal O}(\text{poly}~n)$, then we arrive at the following decision problem. 

\begin{definition}({\scshape Integer Decision {\sf SAT}})
 Consider 
 \begin{equation}
    f\colon\{0, 1\}^n \rightarrow \mathbb{N},   
\end{equation}
such that $f(y)$ can be evaluated in ${\mathcal O}(\text{poly}~n)$ time for all $y\in \{0, 1\}^n$.  We want to decide the following: 
\begin{enumerate}
    \item[(1)] if there exists a $z$ such that $f(z)=0$, or  
    \item[(2)] if the function $f(z)\geq 1$ for all $z\in \{0, 1\}^n$ 
\end{enumerate}
where $f$ is promised to be either (1)~a {\scshape Yes} instance or otherwise (2)~a {\scshape No} instance.  
\end{definition}

This problem is a standard {\it decision problem}, the optimization variant is given as 
\begin{definition}({\scshape Min Integer {\sf SAT}})
 Given  
 \begin{equation}
    f\colon\{0, 1\}^n \rightarrow \mathbb{N},   
\end{equation}
such that $f(y)$ can be evaluated in ${\mathcal O}(\text{poly}~n)$ time for all $y\in \{0, 1\}^n$. Find find $z'$ such that 
\begin{equation}
   z' = \argmin_{z\in \{0,1\}^n} f(z),  
\end{equation}
where $\min_{z\in \{0,1\}^n}\limits f(z) = f(z')$. 
\end{definition}

We will see that {\scshape Integer Decision {\sf SAT}} relates to physics of actual systems and provides a bridge between the theory of computation and that of physics.  In terms of practice, physical systems exist which embed and evolve to approximately solve {\scshape Min Integer {\sf SAT}}.  Function minimization is conceptually easy to understand. It would seem that quantum mechanics provides a richer class of minimization problems, wherein the target function is typically replaced with a Hermitian operator where one is subsequently tasked with determining the ground eigenvalue.   

Let us start with an explanatory version of a {\it quantum} problem, inspired by Kitaev's {\scshape Local Hamiltonian} which we will develop and apply to questions of modern relevance later in \S~\ref{sec:variational}. For now consider variants simplified for illustrative purposes. For this, we need to rely on quantum bits. 

\begin{definition}
 We call an operator ${\mathcal H} \in {\mathscr L}(\mathbb{C}_2^{\otimes n})$ a $k$-local Hamiltonian if ${\mathcal H}$ is expressible as ${\mathcal H}= \sum_{l=1}^{\text{poly}(n)} {\mathcal A}_l$ where each term ${\mathcal A}_l$ is Hermitian and acts non-trivially on at most $k$ qubits. 
\end{definition}

\begin{definition}({\scshape Decision Hamiltonian}---ignores locality)
  
 \begin{equation}\label{eqn:dh}
     {\mathcal H} = \sum_{l=1}^{\text{poly}(n)} {\mathcal A}_l \in {\mathscr L}(\mathbb{C}_2^{\otimes n})
 \end{equation}
 be a non-negative Hamiltonian on $n$ qubits.  Decide  
 \begin{enumerate}
     \item[(1)] if ${\mathcal H}$ has a zero eigenvalue, or 
     \item[(2)] if all eigenvalues of ${\mathcal H}$ are greater than or equal to some $f(n)> 0$. 
 \end{enumerate}
 $f(n)$ will be determined later.
\end{definition}

\begin{remark}[State preparation]
State preparation creates a quantum state by some gate sequence (or another physical process). 
\end{remark}

Deciding instances of {\scshape Decision Hamiltonian} has some evident practical merit.  Consider being given access to a quantum computer that simulates \eqref{eqn:dh} and prepares an arbitrary quantum state $\ket{\psi}$ (a witness) such that we can calculate the expected value $\bra{\psi}{\mathcal H}\ket{\psi}=0$.  While appearing from the outset as artificial, we will develop {\scshape Decision Hamiltonian} in a sequence of steps as a conceptual building block behind powerful mathematical tool(s) to probe the power and limitations of quantum enhanced information processing.  {\scshape Decision Hamiltonian} is closely related to the more practically encountered variant (Definition \ref{eqn:mh}), which we will study in tandem as follows.

\begin{remark}
We consider a quantum state $\psi$ and a Hilbert space ${\mathscr A}$. We adopt the slight abuse of notation that:  
\begin{enumerate}
    \item $\psi \subset {\mathscr A}$ is shorthand for  $\psi \in {\mathscr B} \subset {\mathscr A}$ where $\psi$ is necessarily restricted to a proper subset ${\mathscr B}$ of ${\mathscr A}$, 
    \item $\psi \subseteq {\mathscr A}$ is shorthand for $\psi \in {\mathscr B} \subseteq {\mathscr A}$ where $\psi$ is restricted to a subset ${\mathscr B}$ that might be equivalent to ${\mathscr A}$, 
    \item $\psi \in {\mathscr A}$ means that $\psi$ can take any value in ${\mathscr A}$. 
\end{enumerate}
\end{remark}

\begin{definition}({\scshape Min Hamiltonian}---ignores locality)
 Let 
 \begin{equation}\label{eqn:mh}
     {\mathcal H} = \sum_{l=1}^{\text{poly}(n)} B_l \in {\mathscr L}(\mathbb{C}_2^{\otimes n})
 \end{equation}
 be a non-negative Hamiltonian on $n$ qubits.  Determine 
  \begin{equation}
     \min_{\psi \subseteq {\mathscr A}} \bra{\psi}{\mathcal H}\ket{\psi} = E^\star 
  \end{equation}
 where the domain ${\mathscr A} \subseteq \mathbb{C}_2^{\otimes n}$ is given as a possibly restricted subset of $\mathbb{C}_2^{\otimes n}$ and so $E^\star \geq E_\star$. 
\end{definition}

We readily establish that 
\begin{equation}
  0\leq E_\star = \min_{\psi \in \mathbb{C}_2^{\otimes n}} \bra{\psi}{\mathcal H}\ket{\psi} \leq E^\star = \min_{\psi \subseteq {\mathscr A}} \bra{\psi}{\mathcal H}\ket{\psi}.
\end{equation}
Hence, knowledge of $E_\star$ provided ${\mathscr A} =  \mathbb{C}_2^{\otimes n}$ readily lifts {\scshape Min Hamiltonian} to partition $\{${\scshape Yes}, {\scshape No}$\}$ instances of {\scshape Decision Hamiltonian}. Furthermore, {\scshape Min Hamiltonian} has practical applications as an eigenvalue problem, where similar and restricted forms arise in many areas of engineering and applied science, including determining the ground state energy of electronic structure Hamiltonians \cite{WBA11}.  We then consider the following variant of {\scshape Min Hamiltonian}. 

\begin{remark}

In words, we use equality in \eqref{eqn:arg1} to denote a vector in the linear span of the solution space \eqref{eqn:arg2} where $E_\star$ is the smallest eigenvalue of $\bra{\psi}{\mathcal H}\ket{\psi}$. 
\end{remark}

\begin{definition}({\scshape Argmin Hamiltonian}---ignores locality)
 Let 
 \begin{equation}\label{eqn:amh}
     {\mathcal H} = \sum_{l=1}^{\text{poly}(n)} B_l \in {\mathscr L}(\mathbb{C}_2^{\otimes n})
 \end{equation}
 be a non-negative Hamiltonian on $n$ qubits.  Determine 
  \begin{equation}
     \argmin_{\psi \subseteq {\mathscr A}} \bra{\psi}{\mathcal H}\ket{\psi} = \ket{\psi'}  
  \end{equation}
 where ${\mathscr A} \subseteq \mathbb{C}_2^{\otimes n}$ is given. 
\end{definition}

Provided the domain ${\mathscr A}$ is appropriately restricted, we will be able to store $\ket{\psi'}$ on a classical computer.  In fact, this is often the case.  For example, provided that \eqref{eqn:amh} represents a binary constrained optimization problem (which we will consider later), then ${\mathscr A}$ can safely be restricted to the domain $\{0, 1\}^n$ so that $\ket{\psi'}$ is readily stored as a bit string.  The lack of a tangible description of general (quantum) $\ket{\psi'}$ using a classical computer is one of the foreseen advantages of quantum processors, and a key element of quantum supremacy demonstrations.  

\begin{remark}[The memory argument]
Early arguments for quantum advantage considered an ideal state of interacting qubits, requiring about $2^{n+1}\cdot 16$ bytes of information to store assuming $32$ bit precision. This reaches $80$ terabytes (TB) at just less than $43$ qubits and $2.2$ petabytes (PB) at just under $47$: e.g.~the world’s largest memory of the supercomputer Trinity. Hence applications with $\geq 47$ qubits might already outperform classical computers at certain tasks.  While this argument didn't account for noise and approximations/compression schemes to reduce required memory, similar arguments are considered valid lines of reasoning today.
\end{remark}

While we have presented three rather generic problems, these three problems will be further refined to form the theoretical backbone of Hamiltonian complexity and due in part to their close connection to actual physical processes, we will see that such problems underpin the vast majority of modern quantum programming techniques.  

Going forward we will modify these problems to become either more abstract or otherwise more physical.  We will tailor these problems to apply specifically to restricted settings that are closer to what is available on today's quantum processors.  We will also push the limitations of what can be said about the computational complexity of various Hamiltonian energy problems.  

The starting place for all of this is to develop a language and the intuition to program ground states.  Subsequent chapters will take these ideas in a variety of directions but the starting place begins with one of the most basic, yet most applicable models of statistical mechanics.  We will develop techniques to fully control the ground states of Ising Hamiltonians (see Definition \ref{def:ising}).  

\section{Mathematical structures connecting Ising models and quantum states} 

The vector space representation of quantum states and operators is not always the most suitable for certain circumstances.  Here we will adopt methods from order theory and algebra to define the quotient ring structure representing qubit quantum states and interrelates states (under real valued quadratic restriction) with generalized Ising models. 

A field extension adjoins a field with an element(s) outside of the field.  Stemming from Galios theory, a field can be extended with an indeterminate.  For example, $\mathbb{R}[x]$ is the free ring of all polynomials in a single indeterminate $x$.  A quotient ring construction can then append additional roots to a field extension.  For example, $\mathbb{R}[x]\Big/ x^2=-1$, truncates the free generation of $\mathbb{R}[x]$ at second order.  In this case, $x$ plays the role of the complex number and $\mathbb{C}\simeq \mathbb{R}[x]\Big/ x^2=-1$.


We consider the field extension: 
\begin{equation}\label{eqn:qubitring}
    \mathbb{C}[x_1, x_2, \dots, x_n] \Big/ x_1, x_2, \dots, x_n \in \{0,1\} 
\end{equation}
where $x_1, x_2, \dots, x_n \in \{0,1\}$ the quotient constraint is equivalent to $x_i x_i = x_i$ (idempotence). We arrive at the ring of (qubit) polynomials of type: 
\begin{equation}
    \{0, 1\}^n \rightarrow \mathbb{C} 
\end{equation}
by means of the following mapping
\begin{equation}\label{eqn:fstandard}
    f({\bf x}) = \sum_{I\in \{0, 1\}^n} a_I {\bf x}^I 
\end{equation}
where 
\begin{equation}
    {\bf x}^I \bydef (x_1)^{i_1} (x_2)^{i_2} \cdots (x_n)^{i_n}
\end{equation}
and we abuse notation as 
$$
(x)^0 \bydef (1-x) 
$$
with $x^1 = x$. 

\begin{proposition}[Biamonte 2008, \cite{B08}]\label{prop:grade}
The ring $\mathbb{C}[x_1, x_2, \dots, x_n] \Big/ \forall i, x_i^2 = x_i$ is graded as 
\begin{equation}
    \mathbb{C} \oplus \mathbb{C}[x_1]\oplus \cdots \oplus\mathbb{C}[x_n] \oplus \mathbb{C}[x_1, x_2] \oplus \cdots \oplus \mathbb{C}[x_{n-1}, x_n] \oplus \cdots \oplus \mathbb{C}[x_1, x_2, \dots, x_n] 
\end{equation}
where the quotients are omitted for brevity of notation.  
\end{proposition}
We call an expansion canonical when it is unique up to labeling variables. Proposition \ref{prop:grade} follows from Proposition \ref{prop:can}. 
\begin{proposition}\label{prop:can}
The expansion 
\begin{equation}\label{eqn:fcan}
    f({\bf x}) = a_0 + \sum a_i x_i + \sum a_{ij}x_i x_j + \cdots + \sum a_{ij\dots n}x_i x_j \dots x_n 
\end{equation}
is canonical \cite{B08}. 
\end{proposition}
\begin{proof}
Proposition \ref{prop:can} follows from the point-wise decomposition of $f({\bf x})$ in \eqref{eqn:fstandard}. The identity $(x)^0 \bydef (1-x)$ is applied and then terms of like-powers are grouped, relabeling and grouping coefficients.
\end{proof}

We have hence established that the ring \eqref{eqn:qubitring} can faithfully represent any qubit state in canonical form \eqref{eqn:fcan}. 
This maps algebraic tools to the study of qubit states.  We further see that this algebraic structure is closely related to generalized Ising Hamiltonians.  

Let us then consider elementary examples.  Note that alternative approaches to represent quantum states and gates using Boolean and pseudo Boolean algebras can be found in the literature, including \cite{Fastovets_2019}. 

\begin{definition}[Graph or cluster state]
Given a graph $G(V,  E)$, with the set of vertices $V$ and  edges $E$, the corresponding graph state is defined as
\begin{equation}
    \left| G \right\rangle =\prod _{(a,b)\in E}U^{\{ a,b\} } {\left| + \right\rangle} ^{\otimes V}
\end{equation}
where $\left| + \right\rangle = ({\left| 0 \right\rangle} +{\left| 1 \right\rangle} )/\sqrt{2}$ and the operator $U^{\{ a,b\} }$ is the controlled-$Z$ interaction between the two vertices (qubits) $a$, $b$. 
\end{definition}

\begin{example}
If $G = P_3$ is a three-vertex path graph, then the $S_v$ stabilizers are 
\begin{align}
X \otimes & Z \otimes I, \\
Z \otimes & X \otimes Z, \\
       I \otimes & Z \otimes X
\end{align}
The corresponding quantum state is 
\begin{equation}
    \sqrt{8}\ket{P_3}= \ket{000} + \ket{100} + \ket{010} - \ket{110} + \ket{001} + \ket{101} - \ket{011} + \ket{111}
\end{equation}

The point-wise expansion of the polynomial representing $\ket{P_3}$ in the form \eqref{eqn:fstandard} is 
\begin{align}
    P_3(x, y,  z) &= (1-x) (1-y) (1-z)+x (1-y) (1-z)+(1-x) y (1-z)+ \\\nonumber -&x y (1-z)+(1-x)(1-y)z+x (1-y) z-(1-x) y z+x y z.
\end{align}
The corresponding canonical form as \eqref{eqn:fcan} is
\begin{equation}
    P_3(x,  y,  z) = 1-2 x y-2 y z+4 x y z. 
\end{equation}
\end{example}

\begin{remark}[The coalgebra dual space]
Covectors can be expanded by partial derivatives.  A caveat is that the constant indeterminate independent term (corresponding to a global phase in the case of quantum states; corresponding to a global energy shift in the case of generalized Ising Hamiltonians) is set to zero. Any costate can then be expanded with the coring expansion 
\begin{equation}\label{eqn:fcan}
    f({\bf x}) = \sum a_i \frac{\partial}{\partial x_i}     + \sum a_{ij}\frac{\partial}{\partial x_i}\frac{\partial}{\partial x_j}
 + \cdots + \sum a_{ij\dots n}\frac{\partial}{\partial x_i}\frac{\partial}{\partial x_j}
      \dots \frac{\partial}{\partial x_n}. 
\end{equation}
Operators can also be expanded, for example the $X$ Pauli matrix becomes 
\begin{equation}
    X \bydef {\bar x}\frac{\partial}{\partial x} +  x \frac{\partial}{\partial {\bar x}}
\end{equation}
which can be verified by direct calculation on the single qubit state $\psi(x) = \sin(\theta)\cdot {\bar x} + \cos(\theta) e^{-\imath \phi} \cdot x$. 
\end{remark}

\begin{lemma}\label{lemma:iso} 
The follow isomorphisms hold. 
\begin{equation}
    \mathbb{C}[x_1, x_2, \dots, x_n] \Big/ \forall i, x_i^2 = x_i \simeq \mathbb{C}_2^{\otimes n} \simeq \text{diag}\text{Mat}_{\mathbb{C}}(2^n) 
\end{equation}
\end{lemma}
\begin{proof}
The proof of Lemma \ref{lemma:iso} follows by constructing the invertible (linear) maps. 
\begin{center}
\begin{equation}
\begin{CD}
f({\bf x)}   @>^{-1}>>   \ket{f}\\
@VV^{-1}V       @VV^{-1}V\\
[f]   @=   [f]
\end{CD}
\end{equation}
\end{center}
We write 
\begin{equation}
    f({\bf x})  = \sum_I a_I {\bf x}^I
\end{equation}
\begin{equation}
    \ket{f} = \sum_I a_I \ket{I} 
\end{equation}
\begin{equation}
    [f] = \sum_I a_I \ket{I}\bra{I} 
\end{equation}
and the result follows by explicit construction.  
\end{proof} 

\section{Computation and the Ising model}

We will begin with the methodology to program the ground state of a physical system to embed logic functions.  This leads directly to the fundamental result establishing that the ground state of the general (tunable or adjustable) Ising model\index{Ising Model} represents a computationally meaningful system: specifically that finding the ground state can be shown to be \NP-hard in the language of computational complexity.  We will build on these results.    

Readers should come away with the basic tools needed to embed and sequence logic gates in the ground states of Hamiltonians.  They should also be familiar with the idea of a penalty function and the problem of reducing such functions to quadratic form for physical implementation(s).  Readers should further become familiar with the problem of three satisfiability (3-{\sf SAT}), its embedding into physical spin systems and its corresponding computational phase transition signature.  

This chapter utilizes results from my past work in (2008) and  my work with Whitfield, Faccin in (2012), namely \cite{B08, spinlogic2}.  We will consider interacting binary units called {\it{classical spins}} and adopt the matrix presentation of Boolean bits in which 
\begin{center}
Logical 0 $\mapsto \ket{0}$,  \hspace{60pt} Logical 1 $\mapsto \ket{1}$.
\end{center}
\noindent
The concatenation of bits is defined pairwise using Kronecker's tensor ($\otimes$) where the symbol ($\otimes$) is often omitted (e.g.~$\ket{q}\otimes\ket{r}$ is written equivalently as $\ket{q}\ket{r}$ or $\ket{q,r}$). 


\begin{example}\label{ex:phi}
We can define the vector corresponding to a Boolean switching function~$f$ as 
\begin{equation}
\ket{f} \bydef \sum_{x \in \{0,1\}^n} f(x) \ket{x}. 
\end{equation}
For instance, the following is proportional to the Bell state
\begin{equation} \label{eqn:phi}
\sqrt{2}\ket {\Phi^+} =\sum_{x_1,x_2} \left[1-(x_1-x_2)^2\right] \ket{x_1, x_2}= \ket{00}+\ket{11}, 
\end{equation}
 with $x_1,x_2\in \{0,1\}$. 
The Bell state is used in a range of quantum protocols and takes its name after pioneering quantum physicist, John Bell. 
\end{example}

There exists a useful method to embed Boolean equations (as well as their pseudo Boolean generalization discussed later) into the low energy configuration of a physical system.  To develop these methods, we introduce some machinery. 

In this thesis, each spin is considered to be acted on by a matrix in $\spn\{\mathds{1}, Z\}$ over the reals, in other words by a matrix
\begin{equation}
\spn\{\mathds{1}, Z\} = \alpha \mathds{1} + \beta Z \hspace{20pt} 	\forall \alpha, \beta \in \mathbb{R}, \hspace{20pt} 
\end{equation}
with 
\begin{equation}
\mathds{1}\bydef\ketbra{0}{0}+\ketbra{1}{1}, \hspace{10pt} Z\bydef\ketbra{0}{0}-\ketbra{1}{1}.
\end{equation}
%
We typically assign $\beta=\pm1/2$, $\alpha= 1/2$ and define orthogonal projectors \eqref{eqn:P1} and~\eqref{eqn:P0}. %
\begin{align}\label{eqn:P1}
P_{1}\bydef \ketbra{1}{1} = \frac{1}{2} (\mathds{1} - Z)
\\[6pt]
\label{eqn:P0}
P_{0}\bydef  \ketbra{0}{0} = \frac{1}{2} (\mathds{1} + Z)
\end{align}
Where equations \eqref{eqn:P1} and \eqref{eqn:P0} are defined by the relations \eqref{eqn:P0vec} and \eqref{eqn:P1vec}. 
\begin{align}\label{eqn:P0vec}
 P_0 \lvert 0 \rangle& = \ket{0} \qquad P_{0}\lvert 1 \rangle = 0
\\
\label{eqn:P1vec}
 P_1 \lvert 1 \rangle& = \ket{1} \qquad P_{1}\lvert 0 \rangle = 0
\end{align}
These equations \eqref{eqn:P1}, \eqref{eqn:P0}, \eqref{eqn:P0vec} and \eqref{eqn:P1vec} can be succinctly summarized respectively as \eqref{eqn:proj1} and \eqref{eqn:proj2}. 
\begin{align}\label{eqn:proj1}
P_a &= \frac{1}{2} (\mathds{1} + (-1)^a Z)
\\
\label{eqn:proj2}
P_a \ket{b}&= \delta_{ab}\ket{b}
\end{align}

We will emulate logic operations using the lowest eigenstates of operators formed from the real-linear extension taken over 
\begin{equation}
    \Omega_n \bydef \{P_{1}, P_{0}\}^{\otimes n}
\end{equation}
for fixed finite $n$---see Remark \ref{remark:sle}---in other words we will devise operators in 
\begin{equation}\label{eqn:basis}
\spn\{\Omega_n \}. 
\end{equation}

\begin{remark}[Notation---span of linear extension]\label{remark:sle}
The notation $\{A, B, \dots \}^{\otimes n}$ corresponds to the set of all $n$-word products of $A, B, \dots$ with tensor ($\otimes$) as concatenation.  The real-linear extension of \eqref{eqn:basis} over the set $\Omega_n$ implies that we might consider the span of e.g.~operators $P_{1}, P_{0}$ as well as the span of their composition using the tensor product such as e.g.~$\alpha\cdot P_{1}\otimes P_{0}$ acting on two systems with $\alpha\in \mathbb{R}$. 
\end{remark}

\begin{example}
We will now establish the following elementary properties.   
\begin{enumerate}
    \item[(a)] For $Z^2=\eye$, $Z$ has eigenvalues $\pm 1$.
    \item[(b)] For $P_a^2=P_a$, $P_a$ has eigenvalues $0,1$.
    \item[(c)] Consider $\det(Z-\lambda\eye) = \lambda^2 - 1 = 0$ and $\det(P-\lambda\eye) = \lambda(\lambda-1)=0$. By substituting $\lambda \mapsto P_a$ (or $Z$) and sending scalars $c$ to $c\eye$, it follows that $P_a$ and $Z$ satisfy their own characteristic equation (Cayley-Hamilton theorem).
\end{enumerate}
Consider $(Z-\eye)(Z+\eye)\ket{j}=0$ and if $\ket{j}$ is an eigenstate with eigenvalue $\lambda$ then $\lambda^2=1$ and hence $\lambda = \pm 1$, which establishes (a). For the operators from (a) and (b), their eigenvalues and eigenvectors are expressed as $Z\ket{j}=(-1)^j\ket{j}$ and $P_a\ket{j}=\delta_{aj}\ket{j}$ for $j = 0,1$. Similar arguments hold for $P_a$. Point (c) follows by direct calculation.   
\end{example}

The method to program and engineer the ground states of physical systems functions by adding a so called penalty (a.k.a.~energy penalty $\geq \Delta$) to undesirable spin configurations.  For example, to set a bit to logical zero, we can add a penalty $P_1$.  To set a pair of bits to logical 11 ($\ket{11}$), we will add the penalty $P_0\otimes \eye + \eye \otimes P_0$ (see Remark \ref{remark:probits}). 

\begin{definition}
It is common to sometimes adopt the notation as in \eqref{eqn:notation1}.  
\begin{equation}\label{eqn:notation1}
P_a\otimes P_b \otimes \cdots \otimes P_c \bydef P_{ab\dots c} 
\end{equation}
\end{definition}
\noindent 
\begin{remark}\label{remark:probits}
(Assigning spins to represent specific bit strings). 
The operators $\mathcal Q$ and $\mathcal Q'$ in equations \eqref{eqn:assign1} and \eqref{eqn:assign2} minimize to set $k$ spins to represent the $k$-long bit string $\omega_1\omega_2\cdots \omega_k$. Note that the over bar is the logical complement (negation), sending Boolean variable $x$ to $\bar{x}=1-x$. The constant $\Delta$ is chosen to be as large as possible ($\Delta >1$). 
\begin{align}\label{eqn:assign1}
{\mathcal Q} &= \Delta\sum_{j=1}^k P_{\bar {\omega}_j}=\Delta\sum_{j=1}^k\left(\eye- P_{ \omega_j} \right)
\\
\label{eqn:assign2}
{\mathcal Q'} &= \Delta\biggl(\eye - \bigotimes_{j=1}^k P_{\omega_j}\biggr)
\end{align}
\end{remark}

    \begin{table}[htbp]\renewcommand{\arraystretch}{1.5}
        \centering
         \caption{Contrasting binary, polarity and matrix embeddings of Boolean bits. }
        \begin{tabularx}{0.9\textwidth}{@{} | L | R |}
 
  \hline
	{\bf Variables} & {\bf Matrix Embedding} \\ [5pt]
    \hline
    \multirow{2}{*}{Boolean bit $x_i \in \{0, 1\}$} &  projector $P_{\overline{x}_i} \bydef \ketbra{\overline{x}_i}{\overline{x}_i}$ \\
 	& with spectrum $\in \{0, 1\}$ \\[5pt]
    \hline
    \multirow{2}{*}{Spin variable $s_i \in \{\pm 1\}$ } & $Z$ matrix $ \ketbra{0}{0} - \ketbra{1}{1}$\\
    & with spectrum $\in \{\pm1\}$ \\
    \hline
    Affine transformation relating $x_i, s_i$& Matrix relation between $P_a$ and $Z$\\
    $s_i=1-2x_i$ & $P_a = \frac{1}{2}(1+(-1)^a Z)$ \\[5pt]
    \hline
  \end{tabularx}
    \end{table}

\begin{remark}(Locality of an operator). The locality of an operator is the highest number of non-trivial terms in a tensor product describing that operator.  For example, the locality of the operator \eqref{eqn:assign1} is called 1-local, or local or one-body whereas the locality of the operator in \eqref{eqn:assign2} is $k$-local.  

In \eqref{eqn:assign1} we project onto the complement of the bits and form a sum over the projectors. Equation \eqref{eqn:assign2} is given by a tensor product of projectors.  Both of these operators share the same low-energy space.  Physical systems implement two-body interactions and so $k$-body interactions must be emulated.  This translates into computational resources. 
\end{remark}

We will continue explaining penalty functions by developing several examples.  

\begin{example}(Equality and Inequality Penalties). 
\label{ex:eqineq}
Two binary variables are equal when they both evaluate to logical $0$ (logical $1$). Logical equivalence is defined in this way.  And similarly for inequivalence.  

We wish to construct a non-negative operator $h_{=}$ with the property that 
$$
\ker \{h_{=}\} =\spn \{ \ket{x, y}| x=y\} 
$$ 
is the zero eigenspace and where all other eigenvectors are in an eigenspace $\geq \Delta$.  Such an operator is constructed directly from considering a modified truth table for the logical operation.  For equality and inequality (right most), the corresponding (energy) truth table is given as follows. 
\begin{center}
\begin{tabular}{c|c|c|c}
$x$ & $y$ & $x\stackrel{?}{=}y$ & $x\stackrel{?}{\neq} y$\\
\hline
$\ket{0}$ & $\ket{0}$ & 0 & $\geq\Delta$ \\
$\ket{0}$ & $\ket{1}$ & $\geq\Delta$ & 0 \\
$\ket{1}$ & $\ket{0}$ & $\geq\Delta$ & 0 \\
$\ket{1}$ & $\ket{1}$ & 0 & $\geq\Delta$ 
\end{tabular}
\end{center}
The penalty functions have an evident expression in terms of the projectors \eqref{eqn:P1} and \eqref{eqn:P0} as follows:   
\begin{align*}
h_{=}&\bydef \Delta(P_{0}\otimes P_{1} + P_{1}\otimes P_{0}),
\\
h_{\neq}&\bydef \Delta(P_{0}\otimes P_{0} + P_{1}\otimes P_{1}). 
\end{align*}
Here the defining relation is $h_{=}\bydef \Delta-h_{\neq}$.  We can check these formula by noting that the tensor product of scalars reduces to the usual product viz., \[ 
(P_a\otimes P_b) \ket{q, r} = (P_a\ket{q})\otimes (P_b\ket{r}) = \delta_{aq}\cdot \delta_{br} \ket{q,r}.
\]
\end{example}

\begin{remark}[Using $\cdot$ for multiplication by a scalar]
Though typically omitted, based on aesthetics we sometimes use $\cdot$ to denote multiplication by a scalar.  
\end{remark}

The objective is to embed logical operations into the ground states of tunable Ising Hamiltonians.  To that end, we must define a family of logical operations that form a universal generating basis from which we can express any logical operation.  So far we have only defined penalties that act on two spins.  Going further we will develop a penalty to embed the {\sf AND} gate.  It is well known that {\sf AND, COPY} and {\sf NOT} \index{Logic Gates} form a universal basis for Boolean logic. Logic gates and Boolean algebra will be further discussed in Section \ref{sec:switch}. 

\begin{example}The Logical {\sf AND} operation \cite{B08, spinlogic2} is constructed similarly to the procedure in Example \ref{ex:eqineq}.  We want to develop a penalty function such that the zero eigenspace is in
$$\spn \{ \ket{x, y, z}| z=x \cdot y\}$$ 
and the orthogonal space 
$$\spn \{ \ket{x, y, z}| z=1-x \cdot y\}$$ corresponds to eigenvalues of at least $\Delta$.

From the energy-truth table 
\begin{center}
\begin{tabular}{c|c|c|c}
$x$ & $y$ & $z$ & $z \stackrel{?}{=} x\cdot y$\\
\hline
0 & 0 & 0 & 0 \\
0 & 0 & 1 & $\geq\Delta$ \\
0 & 1 & 0 & 0 \\
0 & 1 & 1 & $\geq\Delta$ \\
1 & 0 & 0 & 0 \\
1 & 0 & 1 & $\geq\Delta$ \\
1 & 1 & 0 & $\geq\Delta$ \\
1 & 1 & 1 & 0 
\end{tabular}
\end{center}
we arrive at single penalty for each term $\geq \Delta$ as 
\begin{equation}\label{eqn:Hand3}
h_{\wedge}\bydef \Delta(P_{001}+P_{011}+P_{101} + P_{110}). 
\end{equation}
Expanding and simplifying \eqref{eqn:Hand3} yields $$
h_{\wedge} = \Delta (\eye\otimes\eye \otimes P_1 + P_1\otimes P_1 \otimes \eye - 2 P_1\otimes P_1 \otimes P_1). 
$$ 
\end{example}

To use the {\sf AND} gate penalty function \eqref{eqn:Hand3} in practice, one can set the input bits using the projectors $P_0$, $P_1$.  Alternatively, one could force the output bit to be logical 1 by projecting onto $\ket{0}$ 
\begin{equation}\label{eqn:andout}
\Delta(P_{001}+P_{011}+P_{101} + P_{110}) + \epsilon \mathds{1}\otimes\mathds{1}\otimes P_0
\end{equation}
where $0< \epsilon$.  

Now minimization of the penalty function \eqref{eqn:andout} would provide input conditions to satisfy the {\sf AND} function.  This approach becomes more interesting when considering a sequence of gates.  Hence, penalty functions for a universal set of classical logic gates should be developed.  The {\sf AND} gate by itself is not universal for classical logic: yet {\sf AND}, together with {\sf OR} and {\sf COPY} is.  However, the {\sf NAND}\index{Logic Gates} gate is universal provided one can also copy bits.  

\section{Low-energy subspace embedding}\label{sec:switch}

Here we will consider embedding switching functions into low-energy subspaces.  We will proceed by recalling some basic properties of Boolean algebra. 

\begin{definition}\label{def:diagHam}
(A note for mathematicians and computer scientists). We consider the diagonal matrix $\mathcal{H}$ acting on states $\ket{\psi}\in \ket{\{0,1\}^n}$ such that $\mathcal{H}\ket{\psi} = k \ket{\psi}$ where $k$ is a real number.  For each such $\ket{\psi}$, the quantity $\bra{\psi}{\mathcal H}\ket{\psi} = k$ is called the energy of $\ket{\psi}$ relative to $\mathcal{H}$.  Operators such as $\mathcal{H}$ are called Hamiltonians, or energy functions.  We will extend and refine this definition.
\end{definition}

\begin{remark}
The ground state or low-energy subspace of $\mathcal{H}$ from Definition \ref{def:diagHam} is given by the span of the vectors with minimal $k$.  We have been engineering non-negative Hamiltonians such that their low-energy (zero) eigenspace embeds logical functions. 
\end{remark}

\begin{definition}
A Boolean (or switching) function is an $n$-ary map
\begin{equation}
f\colon \mathbb{B}^n \rightarrow \mathbb{B}
\end{equation}
where $\mathbb{B} = \{ 0,1\}$ is the Boolean field and non-negative $n$ is called the \textbf{arity} of $f.$ The case $n=0$ formally defines the constant elements of $\mathbb{B},$ $0$ and $1$ (false and true respectively). 
\end{definition}

\begin{remark}
The total number of Boolean functions $f\colon \mathbb{B}^n \rightarrow \mathbb{B}$ for each $n$ is $2^{2^n}$.  
\end{remark}

\begin{example}
(Majority function). The \textbf{majority function} is false when $n/2$ or more input arguments are false and true otherwise. It can be written as
\begin{equation}
M(x_1,\dots,x_n) = \left\lfloor \frac{1}{2} + \frac{1}{n}\biggl(\sum_{\,i=1}^n x_i - \frac{1}{2}\biggr)\right\rfloor,
\end{equation}
where $\lfloor f \rfloor$ is a floor function of $f$.
\end{example}

\begin{example}
For $n=2$ the majority function becomes equivalent to the {\sf AND} function, which takes bit pairs $x,y \in \mathbb{B}$ to their logical product. The {\sf AND} of two bits $x,y$ is denoted equivalently as $x\wedge y,~x.y,~x\cdot y,~xy$ and is $1$ iff $x=y=1$ and else $0$.
\end{example}

The logical {\sf OR} function of Boolean variables $x$, $y$  is written
\begin{equation}
x \vee y = x + y - xy
\end{equation}
The {\sf AND}, {\sf OR}, {\sf NOT} gates have the following respective graphical representations 
\begin{center}
\begin{circuitikz} \draw
(0,0) node[and port] (myand1) {}
	(myand1.in 1) node [anchor=east] {$x$}
    (myand1.in 2) node [anchor=east] {$y$}
    (myand1.out) node [anchor=west] {$x\wedge y$ }
(4,0) node[or port] (myor1) {}
	(myor1.in 1) node [anchor=east] {$x$}
    (myor1.in 2) node [anchor=east] {$y$}
    (myor1.out) node [anchor=west] {$x\vee y$ }

(7,0) node{$x$}
(7.25,0) -- (7.75,0)
(7.75,-0.2) -- (7.75,0.2)
(7.75,0) circle [radius=0.2]
(7.75,0) -- (8.25,0)
(8.55,0) node{$~\neg x$};
    
\end{circuitikz}
\end{center}
where the rightmost gate negates its input bit $x$, sending it to $1-x$. Logical negation is written equivalently as $\neg x$, $\bar{x}$ and sometimes $x'$.

\begin{proposition}[Operator embedding of Pseudo Boolean forms]
Any Pseudo Boolean function 
\begin{equation}
    f({\bf x})  = \sum_I a_I {\bf x}^I
\end{equation}
gives rise to an operator embedding 
\begin{equation}
    [f] = \sum_I a_I \ket{I}\bra{I}  \bydef {\hat f}
\end{equation}
by Lemma \ref{lemma:iso}.  The minimisation problems are evidently related as:
\begin{equation}
    \min_{x\in \{0, 1\}^n}f(x) = x', 
\end{equation}
then 
\begin{equation}
    \min_{\psi \in {\mathcal A}} \bra{\psi}{\hat f}\ket{\psi}  = \bra{x'}{\hat f}\ket{x'} 
\end{equation}
for the appropriate vector space ${\mathcal A}$. 
\end{proposition}

\begin{example}[Operator embedding $3$-{\sf SAT}]
We will consider the partition function of $3$-{\sf SAT} instances in \S~\ref{chap:varintro}. 
For an explicit example of turning 3-{\sf SAT} into a ground state Hamiltonian problem, consider \eqref{eqn:sat1}. 
\begin{equation}\label{eqn:sat1}
  f =  (x_{i} \lor \overline{x}_{j} \lor x_{k})\land  (\overline{x}_{g} \lor x_{r} \lor \overline{x}_{s}) \land(\cdots)\land \cdots
\end{equation}
The first clause can be rewritten as an energy penalty viz., 
\begin{equation}
x_{i} \lor \overline{x}_{j} \lor x_{k} \mapsto 
\ketbra{0}{0}_i \otimes \ketbra{1}{1}_j \otimes \ketbra{0}{0}_k  
\end{equation}
and likewise for other clauses. 
\end{example}

\begin{theorem}[Boolean function embedding, Biamonte (2008)~\cite{B08}]\label{thm:booleantoH}
Any Boolean switching function $f(x_1, x_2, \dots, x_n)$ expressed over the basis $\{\vee, \wedge, \neg\}$ embeds into the spectrum of a Hermitian operator formed by the linear extension of $\{P_0, P_1, \eye\}$ by means of the following maps \eqref{eqn:m1} and \eqref{eqn:m2}. %
\begin{align}\label{eqn:m1}
&\wedge \longrightarrow \otimes 
\\
\label{eqn:m2}
&\vee \longrightarrow + 
\end{align}
For every (positive polarity, a.k.a.~non-negated) Boolean variable $x_j$ we apply 
\begin{equation}\label{eqn:m3} 
x_j \longrightarrow P_1^j. 
\end{equation}
For negated variable $\neg x_j$  we apply 
\begin{equation}\label{eqn:m4}
\neg x_j \longrightarrow P_0^j. 
\end{equation}
In both cases \eqref{eqn:m3} and \eqref{eqn:m4},  $1\leq j\leq n$ becomes a spin label index which $P^j$ acts on. Moreover 
the above mapping induces an operator $\mathcal{H}$ such that 
\begin{equation}
    \mathcal{H}\ket{{\bf x}} = f({\bf x}) \ket{{\bf x}}
\end{equation}
for Boolean function $f({\bf x})$ and bit string ${\bf x}$. 
\end{theorem}

\begin{theorem}[Kernel embedding]\label{theorem:kern}
A Boolean function $f(x)$ embeds into the kernel of a non-negative Ising penalty function by applying the map from Theorem \ref{thm:booleantoH} to the function $g(x, f(x)) = 0$, $g(x, 1-f(x)) = 1$. 
\end{theorem}

\begin{remark}
The condition $g(x, 1-f(x)) = 1$ can readily be modified to $g(x, 1-f(x)) \geq 1$ leaving the operators constructed by Theorem \ref{thm:booleantoH} non-negative with identical kernals. 
\end{remark}


\begin{example}[Equality penalty]
Recall the Boolean function from Example \ref{ex:phi}: 
\begin{equation}\label{ex:phi2}
1 - (x_1-x_2)^2. 
\end{equation}
Expanding \eqref{ex:phi2} we have $1 - x_1 - x_2 + 2x_1x_2$.  Using the techniques presented in the pseudo Boolean to ground state mapping, we arrive at the penalty: 
\begin{equation}
\eye - P_0\otimes \eye - \eye\otimes P_0 + 2 P_0\otimes P_0.  
\end{equation}
\end{example}




\begin{definition}
As is typical in quantum physics, we will often omit identity operators by writing $Z_a$ to mean 
\begin{equation}
Z_a \bydef  \eye_1 \otimes \cdots \eye_{a-1}\otimes Z_a \otimes \eye_{a+1} \otimes \cdots
\end{equation}
where the operator $Z_a$ acts on the $a^{th}$ spin. Likewise for pairs of operators with $a\neq b$, $Z_a\otimes Z_b = Z_a Z_b$ and again the identity is omitted. 
\end{definition}

Here we have developed a method to simulate (in principle) any pseudo Boolean penalty function using Ising Hamiltonians.  Physical realizations of classical or quantum annealers and of quantum Ising machines \cite{kirkpatrick1983optimization, utsunomiya2011mapping, inagaki2016coherent, pierangeli2019large, marandi2014network, nixon2013observing, dung2017variable, kalinin2018global} are however limited to two-body interactions.  Three-body terms (and higher) will be emulated using two-body terms through a construction involving the introduction of slack qubits.   

Before considering a method to reduce (or quadratrize) penalty Hamiltonians, let us state the following complexity result. 

\begin{remark}
Note that for a problem to be in the class {\sf NP}, a verification procedure must exist as follows.  Given an instance, the output of each input can be determined in polynomial time. In the case of problems defined by classes of Hamiltonians on $n$ spins/qubits, we consider that the Hamiltonian's description should be bounded in size by some polynomial in the size of the input.  With slight abuse of notation, we express this in \eqref{eqn:Hran3sat} by bounding the sum 
\begin{equation}\label{eqn:Hran3sat}
    {\mathcal H} = \sum_{\substack{(a, b, c)\in A \subseteq{\{0,1\}}^3 \\ (\alpha, \beta, \gamma) \in B \subseteq{\{1,2,\ldots,n\}^3} \\ \alpha \neq \beta \neq \gamma}} P^{\alpha \beta \gamma}_{a b c} 
\end{equation}
with upper limit ${\text{poly}(n)}$ as 
\begin{equation}\label{eqn:Hran3sat}
    {\mathcal H} = \sum^{{\mathcal O}(\text{poly}~n)}_{\substack{(a, b, c)\in \{0,1\} \\ (\alpha, \beta, \gamma) \in \{1,2,\ldots,n\} \\ \alpha \neq \beta \neq \gamma}} P^{\alpha \beta \gamma}_{a b c}. 
\end{equation}
or more compactly as 
\begin{equation}\label{eqn:Hran3sat}
    {\mathcal H} = \sum^{{\mathcal O}(\text{poly}~n)} P^{\alpha \beta \gamma}_{a b c} 
\end{equation}
\end{remark}
    
\begin{proposition}[{\scshape Decision Three-body Projector Ising Hamiltonian}]\label{thm:3bph}
Given non-negative 
\begin{equation}\label{eqn:Hran3sat}
    {\mathcal H} = \sum^{{\mathcal O}(\text{poly}~n)} P^{\alpha \beta \gamma}_{a b c}
\end{equation}
acting on $n$ spins where 
\begin{equation}
    P^{\alpha \beta \gamma}_{a b c} = P^\alpha_a \otimes P^\beta_b \otimes P^\gamma_c
\end{equation}
is uniformly chosen at random.  Here $P_x^\kappa = \frac{1}{2}(\eye +(-1)^x Z_{(\kappa)})$ acts non-trivially on the bit labeled $\kappa$ by projecting onto $\ket{x}$.  Then the decision problem {\scshape Three-body Projector Ising Hamiltonian} determines if:
\begin{enumerate}
    \item[(1)] ${\mathcal H}$ has at least one zero eigenvalue or otherwise if 
    \item[(2)] all eigenvalues of ${\mathcal H}$ are at least unity. 
\end{enumerate}
\end{proposition}

Proposition \ref{thm:3bph} is the first strong connection between complexity science and physics we will begin to make.  The problem is closely related to decision 3-{\sf SAT}---in fact, there is a  bijection between these problems.\footnote{Decision 3-{\sf SAT} canonical and first {\sf NP}-problem from Cook \cite{Cook1971}---Sometimes called the Cook--Levin Theorem as similar results were independently published by Leonid Levin as [Universal search problems]. Problems of Information Transmission (in Russian) {\bf 9}(3): 115–116 (1973). Translated into English by Trakhtenbrot \cite{Trakhtenbrot1984}.} We will leave the proof of 
{\sf NP}-completeness in Theorem \ref{thm:3bph} to \S~\ref{sec:cpt}. 


\begin{proposition}[{\scshape Min Three-body Projector Ising Hamiltonian} is {\sf NP}-hard]\label{thm:m3bph}
The minimization of non-negative 
\begin{equation}
   \min_{x\in\{0,1\}^n} \bra{x}{\mathcal H}\ket{x} = \min_{x\in\{0,1\}^n} \sum \bra{x}P^{\alpha \beta \gamma}_{a b c}\ket{x}
\end{equation}
as defined in Theorem \ref{thm:3bph} is {\sf NP}-hard.  Here the condition on the upper bound is lifted as {\sf NP}-hard problems need not be in {\sf NP} unless they are {\sf NP}-complete. 
\end{proposition} 

The problem, {\scshape Min Three-body Projector Ising Hamiltonian} from Proposition \ref{thm:m3bph} is readily reduced to {\scshape Max 3-{\sf SAT}} in which the objective is to violate the fewest clauses.  

The proofs of Proposition \ref{thm:3bph} and \ref{thm:m3bph} follow from known results and will wait until \S~\ref{sec:cpt}.  We will first focus on the calculus of reduction of three-body terms to two-body terms. Such techniques are often called, classical gadgets, in relation to the non-perturbative case using gadget Hamiltonians \cite{B08, spinlogic2, Cao_2015}---see \S~\ref{chap:gadgets}. 

\section{Two-body reductions} 

We study penalty functions as they describe the energy levels of a physical spin system.  Spins as stated, are binary units and each physical configuration of spins is assigned a real number representing the energy of the spin configuration.  Later on we will study natural physical processes that cause a system to evolve towards the lowest energy configuration.  As a first step towards a physical realization, we will consider here the process of embedding higher order ($ZZ\cdots Z$) interactions into two-body interactions by a process that adds ancillary spins \cite{B08, spinlogic2}.

\begin{table}[htbp]
    \centering
    \renewcommand{\arraystretch}{1.5}
  \begin{tabularx}{0.955\textwidth}{|>{\raggedright\arraybackslash}p{.3\textwidth}|>{\raggedright\arraybackslash}p{.6\textwidth}|}
  \hline
    Pseudo Boolean form truncated past quadratic order &

  $\begin{aligned}[t]
  f&=a_0+a_1x_1+a_2x_2+\dots \\
    &\qquad+a_{12}x_1x_2 + a_{13}x_1x_3+\dots \\
     &=a_0 + \sum_ia_ix_i + \sum_{i<j} a_{ij}x_ix_j 
  \end{aligned}$\\ 
    \hline
    Polarity transform of $f$ & $
    \begin{aligned}[t]
    \tilde{f} &= \omega_0+\omega_1s_1+\omega_2s_2+\dots \\
      &\qquad+\omega_{12}s_1s_2 + \omega_{13}s_1s_3+\dots \\
      &= \omega_0 + \sum_i\omega_is_i + \sum_{i<j} \omega_{ij}s_is_j  
      \end{aligned}$\\ 
    \hline
    Matrix presentation mapping the values of $f$ to the spectrum of a diagonal matrix &

    $\begin{aligned}[t]
    {\mathcal H}_f &= a_0\mathds{1} + a_1P_{\overline{x}_1}+ a_2P_{\overline{x}_2}+\dots\\
    &\qquad+a_{12}P_{\overline{x}_1} \otimes P_{\overline{x}_2} + a_{13}P_{\overline{x}_1} \otimes P_{\overline{x}_3}+\dots \\
     &=a_0\mathds{1} + \sum_i a_i P_{\overline{x}_i} + \sum_{i<j} a_{ij} P_{\overline{x}_i} \otimes P_{\overline{x}_j}
     \end{aligned}$\\
    
    \hline
    Direct Ising realization of $ \tilde{f} $ & 
     $\begin{aligned}[t]
     {\mathcal H}_{\tilde{f}} &= \omega_0 \mathds{1}  + \omega_1Z_1 + \omega_2Z_2 + \dots\\
    &\qquad+\omega_{12}Z_1Z_2+\omega_{13}Z_1Z_3+\dots \\
    &=\omega_0 \mathds{1} + \sum_i \omega_i Z_i + \sum_{i<j} \omega_{ij}Z_iZ_j
     \end{aligned}$\\
    \hline
    \multicolumn{2}{|c|}{$f$ and $\tilde{f}$ are related by an affine change of variables $s_i=1-2x_i$} \\
    \multicolumn{2}{|c|}{and ${\mathcal H}_f$ and ${\mathcal H}_{\tilde{f}}$ are related as $P_x = \frac{1}{2}(1+(-1)^xZ)$} \\
    \hline
  \end{tabularx}

    \caption{Constructing penalty Hamiltonians from Boolean (or pseudo-Boolean) functions. }
    \label{tab:my_label}
\end{table}

Physical systems with rare exception implement local (one body) and two-body interactions.  Two body interactions are terms of the form $J_{ij}Z_iZ_j$.   We have so far implemented penalty functions using operators $P_0$, $P_1$.  To translate these directly into physical interactions, we have to return to their defining relations and express them in terms of $Z$ operators.  In doing such, the following problem illustrates that our penalty function \eqref{eqn:Hand3} requires three-body terms to implement.  The penalty function \eqref{eqn:Hand3} must be brought into two-body form \cite{B08,spinlogic2} by the addition of slack or ancilary spins.  

\begin{example}\label{eqn:andZ}
Let us express $h_{\wedge}$ from \eqref{eqn:Hand3} over the basis $\big\{ \mathds{1}, Z, \otimes \big\}$. First we remark that transitioning between these two basis will not increase locality.  
We consider the penalty function over Boolean variables $x_1, x_2, x_3$ and arrive at a penalty function for {\sf AND} as 
\begin{equation}
f_\wedge = \Delta (x_3+x_1 x_2 - 2 x_1x_2x_3). 
\end{equation}
Spin variable $s_i$ and Boolean variable $x_i$ are related by the affine transformation $s_i = 1-2x_i$. Substitution of $\frac{1}{2}(1-s_i)$ into $f_\wedge$ changes from Boolean to spin variables, resulting in 
\begin{equation}
\frac{\Delta}{4}(2-s_3-s_1s_3-s_2s_3+s_1s_2s_3). 
\end{equation}
To embed this into matrix form, we replace $s_i \mapsto Z_i$ and send multiplication $\cdot \mapsto \otimes$. 
\end{example}



\subsection{Karnaugh map codomain extension}

The Karnaugh map (KM or K-map) is a method of simplifying Boolean algebra expressions. Maurice Karnaugh introduced it in 1953 as a refinement of Edward Veitch's 1952 Veitch chart, which was a rediscovery of Allan Marquand's 1881 logical diagram a.k.a.~Marquand diagrams.  

The Karnaugh map technique appears as a powerful tool to minimize Boolean expressions of type 
\begin{equation}
    \{0,1\}^n \rightarrow \{0,1\}. 
\end{equation}
In the work \cite{B08}, the Karnaugh map technique was extended to fields of characteristic zero.  The property which enables the Karnaugh map method to be applies is the Boolean lattice structure. In \cite{B08} the Karnaugh map tool was applied to specific reductions related to quardratic Ising Hamiltonian interactions (e.g.~emulating $ZZZ$ using only two-body $ZZ$ and $Z$ terms).  

To develop the Karnaugh map approach, we will explain in detail how the positive-semidefinite \AND{} penalty Hamiltonian, $H_\wedge$, is derived.  Let $\cal L$ be the null space of $H_\wedge$ and let all higher eigenspaces be given as $\cal L^\perp$.  The penalty Hamiltonian has a null space, ${\cal L}$, spanned by the vectors 
\begin{equation}
    \{\ket{x_1 x_2}\ket{z_\star}|z_\star=x_1\wedge x_2, \forall x_1,x_2\in\{0,1\}\}. 
\end{equation}
Denote $\delta$ as an energy penalty applied to any vector component in ${\cal L^\perp}$.  Our goal is to develop a Hamiltonian that adds a penalty of at least $\delta$ to any vector that does not satisfy the truth table of the $\AND$ gate---that is, we want to add an energy penalty to any vector with a component that lies in ${\cal L^\perp}$.

In order to make the penalty quadratic, one first constructs the Karnaugh map
illustrated in Figure~\ref{fig:quad} (c) for the case $x_1\wedge x_2=z_\star$.  This is done by examining Table~\ref{table:NPH}.  
\begin{table}[t]
\centering
{\begin{tabular}{ccc||c||c}
  $x_1$ & $x_2$ & $z_\star $ & $z_\star = x_1\wedge x_2$  & $H_\wedge(x_1,x_2,z_\star)$ \\\hline
  0 & 0 & 0 & $\bra{000}H_\wedge\ket{000}=0$  & 0 \\
  0 & 0 & 1 & $\bra{001}H_\wedge\ket{001}\geq\delta$ & $3 \delta$\\
  0 & 1 & 0 & $\bra{010}H_\wedge\ket{010}=0$ & 0 \\
  0 & 1 & 1 & $\bra{011}H_\wedge\ket{011}\geq\delta$ & $\delta$\\
  1 & 0 & 0 & $\bra{100}H_\wedge\ket{100}=0$ & 0 \\
  1 & 0 & 1 & $\bra{101}H_\wedge\ket{101}\geq\delta$ & $\delta$ \\
  1 & 1 & 0 & $\bra{110}H_\wedge\ket{110}\geq\delta$ & $\delta$ \\
  1 & 1 & 1 & $\bra{111}H_\wedge\ket{111}=0$ & 0 \\
\end{tabular}}\caption{Left column: possible assignments of the variables $x_1$, $x_2$ and $z_\star$.  Center column: illustrates the variable assignments that must receive an energy penalty $\geq\delta$.  Right column: truth table for $H_\wedge(x_1,x_2,z_\star )=3z_\star +x_1\wedge x_2-2z_\star \wedge x_1-2z_\star \wedge x_2$, which has a null space ${\cal L}\in\spn\{\ket{x_1 x_2}\ket{z_\star}|z_\star=x_1\wedge x_2,\forall x_1,x_2\in\{0,1\}\}$.}\label{table:NPH}%
\end{table}

In the right most column, all possible assignments for the variables $x_1$, $x_2$ and $z_\star$ are shown.  The Karnaugh map is constructed by examining the second column.  Whenever the variable $z_\star$ is not equal to the \AND{} of the variables $x_1$ and $x_2$, a penalty of at least $\delta$ must be applied, which ensures that vectors in the ground space satisfy $\ket{x_1}\ket{x_2}\ket{x_1\wedge x_2}$.  Any vector that must receive an energy penalty of $\delta$ is depicted in the Karnaugh map with a dot ($\cdot$).

\begin{figure}[t]
    \centering
    \includegraphics[width=13cm]{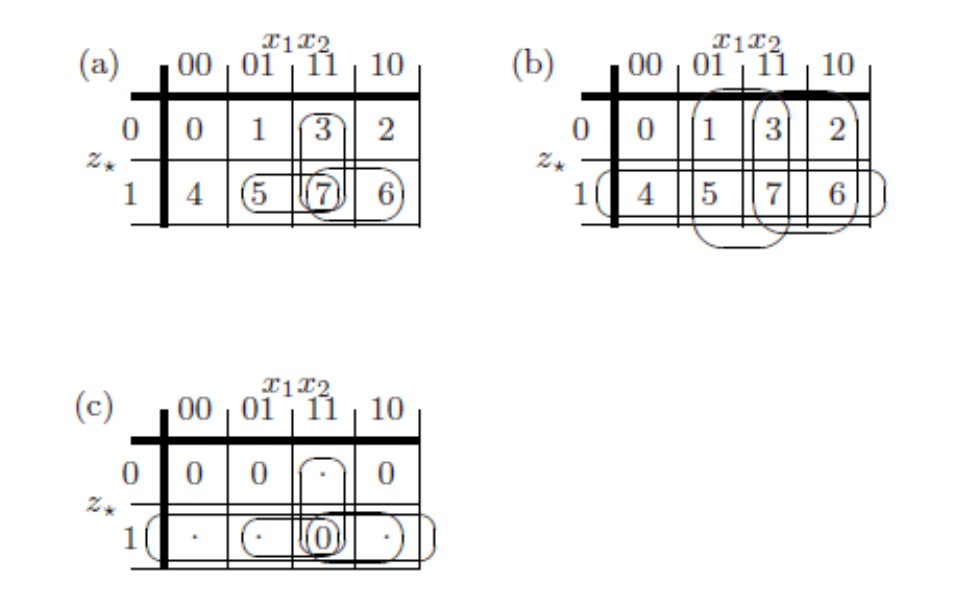}
\caption{Karnaugh maps: (a) 2-local (positive polarity) interactions circled (e.g.~$q_1  x_1x_2+q_2  x_1\wedge z_\star +q_3  x_2\wedge f^\star $). (b) Linear (positive polarity) terms circled (e.g.~$l_1  x_1+l_2  x_2+l_3  x_3$).  The interactions in cubes (a) and (b) form a basis for the space of realizable ($\leq$ 3 qubit) positive-definite logical gadget Hamiltonians expressible as:
$H(x_1,x_2,z_\star )=k_0+l_1  x_1+l_2  x_2+l_3  x_3+q_1  x_1\wedge x_2+q_2  x_1\wedge z_\star +q_3  x_2\wedge z_\star $, where $\forall i$, $k_0, l_i, q_i \geq 0$. (c) A Karnaugh map illustrating (with ovals) the linear and quadratic terms needed to set the null space of the Hamiltonian (\ref{eqn:Hwedge}) to be in $\spn\{\ket{x_1x_2}\ket{y_\star}|y_\star=x_1\wedge x_2, \forall x_1,x_2\in\{0,1\}\}$.}
\label{fig:quad}
\end{figure}

Begin by noticing that any vector associated with cube number $4$ must receive an
energy penalty, so the 1-local field corresponding to the qubit with label $z_\star $ must be at least $\delta$ --- adding the term $p_1  z_\star$ to the Hamiltonian, with the constraint $p_1 \geq\delta$.  Cube $3$ must also receive an energy penalty of at
least $\delta$, adding the term $p_2   x_1\wedge  x_2$ to the
Hamiltonian $H_\wedge$.  With both penalties applied, vectors corresponding to cube $7$ must be brought back to the null space --- accomplished by subtracting the quadratic energy
rewards $r_1   z_\star\wedge x_1$ and $r_2   z_\star\wedge x_2$ from $H_\wedge$. A system
of equations for the Hamiltonian 
\begin{equation}\label{eqn:g1}
H_\wedge(x_1,x_2,z_\star)=p_1   z_\star  + p_2   x_1\wedge x_2 - r_2   z_\star\wedge  x_1 - r_2  z_\star\wedge x_1
\end{equation}
can be solved to set the rewards ($r$'s) and the penalties ($p$'s). This system is derived from the
fact that the term $x_1x_2x_3$, corresponding to cube $7$, must have zero energy: $0=p_1+p_2-r_1-r_2$ and is
subject to the conditions that $p_1,p_2 \geq\delta$ and $|r_1+r_2|>p_1$.  For convenience, let $\delta = 1$ and then determine values for the coefficients in (\ref{eqn:g1}) and thus derive the 2-body Hamiltonian (for \AND{}):
\begin{equation}\label{eqn:Hwedge}
H_{\wedge}(x_1,x_2,z_\star)=3z_\star +x_1\wedge x_2-2z_\star\wedge x_1-2z_\star\wedge x_2.
\end{equation}
We hence or otherwise establish the following theorem. 
\begin{theorem}
The following penalty functions embed the logical product $-x_1x_2x_3$ into their lowest energy sector as:  
\begin{equation}
-x_1x_2x_3=\min_{z\in\{0,1\}}z_\star(2-x_1-x_2-x_3), 
\end{equation}
and  
\begin{equation}
-x_1x_2x_3=\min_{z\in\{0,1\}}z_\star(-x_1+x_2+x_3)-x_1x_2-x_1x_3+x_1. 
\end{equation}
\end{theorem}

\begin{example}
Let us now develop a two-body penalty function that performs the {\sf COPY} operation. In other words, let us develop the penalty function such that the low-energy subspace is in 
$$\spn\{\ket{000}, \ket{111}\}.$$ 
We will consider three spins and place pairwise equality penalties 
$$f_=(x_i, x_j) = x_i(1-x_j) + (1-x_i)x_j$$
connecting the spins with a triangular interaction graph. Assigning the value $\ket{a}$ to any spin and traversing the triangle minimizes this penalty when all spins take the same value $\ket{a}$.  Labeling the spins $i, j$ and $k$ the penalty function reduces as 
\begin{equation}
f_{\text{copy}} = 2 x_3 + 2 x_1+2 x_2-2 x_2 x_3 -2 x_2 x_1-2 x_3 x_1. 
\end{equation}
\end{example} 

In addition to embedding Boolean logic into the ground state spin system, one can consider various other quadratic binary optimisation problems.  

\begin{example}(Penalty function with constraints). 
A real valued equality constraint of the form 
$$
f(x) = c
$$ 
for $x\in \mathbb R^n$, $c\in \mathbb R$ can be converted to an unconstrained optimization problem using the following penalty function
$$ 
P(x) = A(f(x)-c)^2
$$
where $P(x) = 0$ when the constraint is reached; $P(x)>0$ when the constraint is not reached and $A\in \mathbb{R}$ is a scaling factor. Such problems are readily mapped to the Ising model by considering a Boolean embedding of $x$. 
\end{example}

\begin{example}(Number partitioning).\label{ex:partitioning}
Given a set of $N$ positive numbers $S = {n_1, \dots, n_N }$, is there a partition of this set of numbers into two disjoint subsets $R$ and $S - R$, such that the sum of the elements in both sets is the same?

Let $n_i~(i = 1,\dots, N = |S|)$ describe the numbers in set $S$. It can be shown that  
\begin{equation}
\mathcal{H}= \biggl(\sum_i n_i s_i\biggr)^2 \geq 0
\end{equation}
vanishes if and only if such a disjoint partition exists.  Here $s_i$ is a spin variable $\in \pm 1$. 
\end{example}

The following applies to number partitioning (Example \ref{ex:partitioning}). 

\begin{remark}($\mathbb{Z}_2$ symmetry). 
Consider the tunable two-body Ising Hamiltonian \index{Ising Model} acting on $n$ spins as 
\begin{equation}
\mathcal{H} = \sum_{i< j}J_{ij}Z_iZ_j
\end{equation}
Using the identity that $XZX= -Z$ or otherwise, for $\tilde{X}=\bigotimes_{l=0}^n X_l$ it can be shown that 
\begin{equation}
\left[\tilde{X}, \mathcal{H}\right] = 0 
\end{equation}
and hence one can establish that the definition of $\ket{0}$, $\ket{1}$ is entirely arbitrary with respect to $\mathcal{H}$. 
\end{remark}

\begin{proposition}[Biamonte 2008, \cite{B08}]
Out of the 16 possible functions of 2-input and 1-output variable, it can be proven that only two are not realizable using 3-spins.  These are the 2-local penalty Hamiltonians for $\XOR{}$ ($\oplus$) and \emph{equivalence} ($\odot$)~\footnote{Where \emph{exclusive} \OR{} ($\XOR{}$) is given as $f_\oplus(x_1,x_2)\bydef x_1\oplus x_2=\bar x_1x_2\vee x_1\bar x_2= x_1+x_2-2x_1\wedge x_2$, and \emph{equivalence} as
$f_\odot(x_1,x_2)\bydef x_1\odot x_2=\bar x_1\bar x_2\vee x_1 x_2= 1-x_1-x_2+2x_1\wedge x_2$.}, 
which are each possible to realize by adding a single mediator qubit. This was proven by contradiction in \cite{B08}. 
\end{proposition}

\chapter{The Structure of Quantum vs Probabilistic Computation} \label{chap:qvspc}

To understand quantum computation, we will recall and contrast quantum mechanics with stochastic mechanics.  Our development partly follows a book coauthored with John Baez \cite{2012arXiv1209.3632B} as well as other works \cite{faccin2013degree, De_Domenico_2016} which among other results were partially reviewed in \cite{2017arXiv170208459B}.  Here we offer a more computational focus.  

\section{Defining Mechanics} 

Understanding quantum computation involves being able to contrast quantum computation from other models.  We might then compare standard deterministic or classical bits (c-bits), versus stochastic or probabilistic bits (p-bits), versus quantum bits (q-bits or qubits).  We will begin by mentioning a summary in Table \ref{tab:cvspvsqbits}. 

Our starting place is to describe and contrast stochastic versus quantum mechanics and then point out some of the basic implications the similarities and differences imply when considering walks on graphs.

\begin{table}[h]
    \centering
    \begin{tabular}{m{0.15\textwidth} m{0.15\textwidth} m{0.30\textwidth} m{0.30\textwidth}}
    & bits & probabilistic bits & qubits 
    \end{tabular}
    \noindent
    \begin{tabular}{m{0.15\textwidth} |m{0.15\textwidth}| m{0.30\textwidth}| m{0.30\textwidth}}
    \hline
    state (single unit) & bit $\in \lbrace 0,1 \rbrace $ & real vector \newline $a, b \in \mathbb{R}_{+}$ \hfill $a+b=1$ \newline $\Vec{p} = a\Vec{0} + b\Vec{1}$ or $a \ket{0} + b\ket{1}$  & complex vector \newline $\alpha , \beta \in \mathbb{C}$ \hfill $\abs{\alpha}^{2} + \abs{\beta}^{2} = 1$ \newline 
    $\Vec{\psi} = \alpha\Vec{0} + \beta\Vec{1}$ or $\alpha \ket{0} + \beta \ket{1}$
    \\
    \hline
    
    state (multi-unit) & bitstring \newline $x \in \lbrace 0,1 \rbrace^{n}$ & prob.distribution (stochastic vector) \newline $\Vec{p} = \sum_{x\in\{0, 1\}^n} a_x \ket{x} \in [\mathbb{R}^2_+]^{\otimes n}$ & wavefunction (complex vector) \newline $\Vec{\psi} =\sum_{x\in\{0, 1\}^n} \alpha_x \ket{x} \in [\mathbb{C}^2]^{\otimes n}$
    \\
    \hline
    
    operations & Boolean logic & stochastic matrices \newline $\sum_{j} \mathcal{U}_{ij} =1, $ $\mathcal{U}_{ij}\geq 0$  & unitary matrices \newline $\mathcal{U}^{\dagger}\mathcal{U} = \bf{1}$ 
    \\
    \hline
    
    component ops & Boolean gates & tensor product of matrices & tensor product of matrices
    \\
    \hline
    
    \end{tabular}
    \caption{Summary of deterministic, probabilistic and quantum bits. We use the standard notation that $\mathbb{R}^2_+$ denotes the two-dimensional real vector space with non-negative entries. Likewise, $\mathbb{C}^2$ is the two-dimensional complex vector space.  The space of $n$ pbits,  $n$ qubits are respectively given by the tensor product of spaces, $[\mathbb{R}^2_+]^{\otimes n}$ and $[\mathbb{C}^2]^{\otimes n}$. }
    \label{tab:cvspvsqbits}
    
\end{table}

\subsection{Stochastic time evolution}

The stochastic master equation gives the time evolution of a state
\begin{equation}
\frac{d}{dt}\ket{\psi (t)} =-{\mathcal H}\ket{\psi (t)}
\end{equation}
with solution $\ket{\psi (t)}= e^{-t{\mathcal H}}\ket{\psi (0)}$
which can be checked as
\begin{equation}
\frac{d}{dt}e^{-t{\mathcal H}}\ket{\psi(0)} = -{\mathcal H} e^{-t{\mathcal H}} \ket{\psi(0)} = -{\mathcal H} \ket{\psi(t)}. 
\end{equation}
\subsection{Quantum time evolution}
The quantum Schr{\"o}dinger's equation \index{Schr{\"o}dinger's equation} gives the time evolution of a state in quantum mechanics
\begin{equation}
\frac{d}{dt}\ket{\psi (t)} =-\imath {\mathcal H}\ket{\psi (t)}
\end{equation}
with solution $\ket{\psi (t)}= e^{-\imath t {\mathcal H}}\ket{\psi (0)}$
which can be checked as
\begin{equation}
\frac{d}{dt}e^{-\imath t{\mathcal H}}\ket{\psi(0)} = -\imath {\mathcal H} e^{-\imath t{\mathcal H}} \ket{\psi(0)} = -\imath  {\mathcal H} \ket{\psi(t)}. 
\end{equation}
The operator $\mathcal{H}$ is called the Hamiltonian. Its properties depend whether we are working in a stochastic or quantum system.

\begin{example}(Exact Solution of the T.I.S.E.) 
If Hermitian ${\mathcal A}^2 = \eye$ then 
\begin{equation}
e^{-\imath \theta {\mathcal A}} = \eye \cos (\theta) - \imath {\mathcal A} \sin (\theta).  
\end{equation}
\end{example}

\begin{example}(Exact Solution of the T.I.S.E.)
Let ${\mathcal P}^2 = {\mathcal P}$ then   
\begin{equation}
e^{-\imath \theta {\mathcal P}} = \eye+{\mathcal P}(e^{-\imath \theta} - 1)
\end{equation} 
follows by the series expansion 
\begin{equation}\label{eqn:expexpan}
e^{\mathcal A} = \sum_{k=0}^{k=\infty} \frac{{\mathcal A}^k}{k!}. 
\end{equation}
Note that for Hermitian ${\mathcal A} \in {\mathcal L}({\mathbb{C}}^d)$, the Cayley–Hamilton theorem implies that \eqref{eqn:expexpan} can be re-expressed as a new polynomial in ${\mathcal A}$ of degree equal  to the dimension of  the column space  (equivalently  row space).  The characteristic equation
\begin{equation}
    p(\lambda) = \det (\lambda \eye - {\mathcal A}),  
\end{equation}
yields the polynomial $p(\lambda)$.  This polynomial evaluated at ${\mathcal A}$ vanishes identically, thereby bounding the upper limit $k\leq d$ in \eqref{eqn:expexpan}. 
\end{example}

\subsection{Stochastic Hamiltonians}\label{subsec:phsm}

We will now consider the operators which induce time evolution in stochastic mechanics.  These arise in the literature under several names including {\it transition rate matrix},  {\it intensity matrix}, {\it infinitesimal generator matrix} or a {\it stochastic Hamiltonian}.  

\begin{definition}[Stochastic Hamiltonian]
${\mathcal H}$ is infinitesimal stochastic for time evolution in stochastic mechanics given by $e^{-t{\mathcal H}}$ to send stochastic states to stochastic states when: 
\begin{enumerate}
\item its rows sum to zero \eqref{columns-zero}.
\item its off diagonal entries are real and non-positive \eqref{nonpositive}.
\begin{equation}
\label{columns-zero}
\sum_j {\mathcal H}_{ij} = 0
\end{equation}
\begin{equation}
\label{nonpositive}
i \neq j \Rightarrow {\mathcal H}_{ij} \leq 0. 
\end{equation}
\end{enumerate}
\end{definition}

\begin{remark}
The term {\it infinitesimal stochastic} was used in \cite{2012arXiv1209.3632B} while {\it intensity matrix} more commonly appears in the literature to describe the generator of a one dimensional stochastic semigroup. 
\end{remark}

\begin{definition}
A {\it semigroup} is an algebraic structure consisting of a set together with an associative binary $\cdot$ pairing.  Hence, a semigroup relaxes the inverse requirement of a group and the identiy requirement.  A unital semigroup is a semigroup with an identity. 
\end{definition}

\begin{remark}
In our case, a one dimensional (semi)group has a product structure for non-negative $s$, such that $U(s) \cdot U(s') = U(s+s')$ where $U(0) = \eye$ and $U$ is the $s$-dependent exponential image of a (infinitesimal stochastic/Hermitian) generator. 
\end{remark}

\subsection{Quantum Hamiltonians}

We will now recall the standard definition of an operator that generates time-evolution in quantum mechanics.  These operators are called, quantum Hamiltonians (or typically just Hamiltonians), Hermitian or self-adjoint maps. 

\begin{definition}[Quantum Hamiltonian]
${\mathcal H}$ is a quantum Hamiltonian when ${\mathcal H}={\mathcal H}^{\dagger}$).  Then time evolution given by $e^{-\imath t{\mathcal H}}$ sends quantum states to quantum states, ${\mathcal H} = {\mathcal H}^{\dagger}$. 
\end{definition}

\begin{remark}[Basis dependence]
The Hamiltonian in quantum mechanics can be written as  
\begin{equation}
    ({\mathcal H}^{\dagger})_{qr} = \overline{{\mathcal H}}_{rq}
\end{equation} 
or as
$$
{\mathcal H}^\dagger = \sum_l \overline{\lambda}_l (\ketbra{l}{l})^{\dagger} = \sum_{l} \lambda_l \ketbra{l}{l}. 
$$
\end{remark}

\begin{remark}
The eigenvalues of ${\mathcal H}$ take only real values.
\end{remark}

A linear map sending quantum states to quantum states is called an isometry (Definition~\ref{def:isometry}). Isometries represent the general collection of deterministic operations on quantum states.  A subclass of isometries are unitary maps. 
\begin{equation}
\mathcal U_t = e^{-\imath t{\mathcal H}}
\end{equation}
Such maps are characterized by
\begin{equation}
\mathcal U^{\dagger} \mathcal U=\eye. 
\end{equation}

\begin{remark}
With slight abuse of notation we use $\mathcal U$ to denote both quantum and stochastic propagators.  
\end{remark}

A linear map that sends stochastic states to stochastic states is the following {\it stochastic operator} or {\it Markov matrix}. 
\begin{equation}
\mathcal U_t = e^{-t{\mathcal H}}
\end{equation}
Such maps are characterized as
\begin{equation}
 \sum_{j}\mathcal U_{ij} = 1
 \end{equation}

\begin{equation}
 \mathcal U_{ij} \geq 0 
\end{equation}

\begin{definition}
A {\it bijection} is an invertible one-to-one function between elements of two sets.
\end{definition}

\begin{example}
The {\sf COPY} operation is bijective. Let $C(x) = (x, x)$ where $C^{-1}(a, b) = \varnothing~(0)$ and $C^{-1}(a, a) = a$. Then $C^{-1}$ is a left and right inverse, making $C$ bijective as $C^{-1}(C(x))=C^{-1}(x, x) = x$. 
\end{example}

\begin{definition}\label{def:isometry}
An {\it isometry} is a norm preserving bijective map on state(s).  
\end{definition}

\begin{example}[Quantum Isometry]
Consider the {\sf COPY} ($C$) operation defined on one basis (which then does not violate no-cloning \cite{Wootters1982}).  We have $C(\ket{\psi}) = \ket{\psi} \otimes \ket{\psi}$.  Clearly for $|\braket{\psi}{\psi}|^2=1$, $\ket{\psi} \otimes \ket{\psi}$ is normalized.  The inverse of $C$ is defined element wise.  
\end{example}

\subsection{Observables}
In quantum mechanics an observable is given by a self-adjoint (Hermitian) matrix ${\mathcal O}$ and the expected value of ${\mathcal O}$ relative to quantum state $\ket{\varphi}$ is
\begin{equation}\label{eqn:expected}
 \bra{\varphi} {\mathcal O}\ket{\varphi} = \sum_{ij} \overline{\varphi}_i {\mathcal O}_{ij}\varphi_j, 
\end{equation}
where $\ket{\varphi} = \sum_i \varphi_i \ket{i}$. 

\begin{remark}
The L.H.S.~of \eqref{eqn:expected} is alternatively and equivalently written as $\langle \varphi, {\mathcal O}(\varphi) \rangle$, borrowing functional analysis style notation still common in certain circles. 
\end{remark}

In stochastic mechanics, an observable ${\mathcal O}$ takes value ${\mathcal O}_i$ for each configuration $i$ and the expected value of ${\mathcal O}$ in the stochastic state $\varphi$ (equivalently $\ket{\varphi}$) is 
\begin{equation}
{\mathcal O}(\varphi)=\sum_i {\mathcal O}_i \varphi_i
\end{equation}

\begin{definition}\label{def:simplegraph}
A {\it simple graph} contains no-self loops (e.g.~no single node is both the source and the sink of the same edge). 
\end{definition}

\begin{remark}
Some authors alternatively define a {\it simple graph} to be an unweighted, undirected graph containing no graph loops or multiple edges.  
\end{remark}

\begin{definition}\label{symmetricgraph}
A {\it symmetric graph} has only unweighted and undirected edges. Or equivalently, for all nodes $a$ and $b$ in graph $G$, there is an edge from $b$ to $a$ for each edge from $a$ to $b$. 
\end{definition}

\begin{definition}[Adjacency matrix]\label{remark:basis}
Given a simple and symmetric graph $G$, we pick a set of {\it labels} $S$ with cardinality equal to the number of nodes of $G$.  Each ordering of $S$, induces a basis which lifts naturally to a basis to represent $G$ by an adjacency matrix $A$. 
\end{definition}

Let us then develop an extended example for the stochastic case. 
Let $A$ be the adjacency matrix of the simple, symmetric graph as follows. 
\begin{center}
\begin{circuitikz} 
\draw [fill] (-2,2) circle [radius=0.1];
\draw (-1.3, 2) node [black]  {} ;
\draw [thick] (-1,1) 
      to (-2,2) node [black, below=7] {};
      
\draw (-1.7, 1.3 )  node [black] {$1$};  
\draw (-1.7, 2.4 )  node [black] {$\ket{0}$}; 
      
\draw [fill] (-1,1) circle [radius=0.1] node [black,right=6]  {$\ket{1}$} ;
\end{circuitikz}
\end{center}
$A$ is then Pauli $X$ in this case. To form the graph Laplacian from $A$ we define $D_{ii} = \sum_j A_{ji}$ where $\forall$ $i \neq j$ $D_{ji}=0$. And so we arrive at \eqref{eqn:da2.18}. 
\begin{equation}\label{eqn:da2.18}
D =  \eye,  \hspace{20pt}  {\mathcal H} = D-A = \eye-X.
\end{equation}
From a standard calculation, ${\mathcal H}$ yields the following eigensystem \eqref{eqn:eigs1}.
\begin{equation}\label{eqn:eigs1}
\begin{split}
& \lambda_{0} = 0 \hspace{20pt} 
   \ket{\lambda_{0}} = \ket{+}\\
& \lambda_{1} = 2 \hspace{20pt} 
   \ket{\lambda_{1}} = \ket{-}
\end{split}
\end{equation}
Let us then assume that the initial state is \eqref{eqn:init2.20}. 
\begin{equation}\label{eqn:init2.20}
\ket{\varphi(0)} = \ket{0}
\end{equation}
In the eigenbasis \eqref{eqn:init2.20} becomes \eqref{eqn:basischange2.21}. 
\begin{equation}\label{eqn:basischange2.21}
\ket{0} = \frac{1}{\sqrt[]{2}} (\ket{\lambda_0}+\ket{\lambda_1})
\end{equation}
and our time dependent propagator in the Laplacian eigenbasis \eqref{eqn:da2.18} is
\begin{equation}
e^{-t{\mathcal H}} = \sum_j e^{-t\lambda_i}\ketbra{\lambda_i}{\lambda_i}. 
\end{equation}
The time evolution of a state \eqref{eqn:init2.20} starting at node $\ket{0}$ is given as \eqref{eqn:stochprop2.23}. 
\begin{equation}\label{eqn:stochprop2.23}
\begin{split}
e^{-t{\mathcal H}} \ket{0} & = \frac{1}{\sqrt{2}} (e^{-t\lambda_0}\ket{\lambda_0} +  e^{-t\lambda_1}\ket{\lambda_1})\\[6pt]
&= \frac{1}{2}(1+e^{-t\lambda_1})\ket{0} + \frac{1}{2}(1-e^{-t\lambda_1})\ket{1}
\end{split}
\end{equation}
Figure \ref{fig:stochprob2.2} illustrates the time dependence of \eqref{eqn:stochprop2.23}.  The initial state ($\ket{0}$) decays with probability 
$$
\bra{0}e^{-t{\mathcal H}}\ket{0} = \frac{(1+e^{-t\lambda_1})}{2}, 
$$
while the probability of measuring state $\ket{1}$ increases with probability 
$$ 
\bra{1}e^{-t{\mathcal H}}\ket{1} = \frac{(1-e^{-t\lambda_1})}{2}, 
$$ 
for $\lambda_1 = 2$. The long time behavior is given as follows. 
\begin{equation}
    \lim_{t \to \infty}\ket{\phi(t)} = \frac{1}{2}\ket{0} + \frac{1}{2}\ket{1}
\end{equation}


\begin{figure}[htbp!]
    \centering
    \includegraphics[width=11cm]{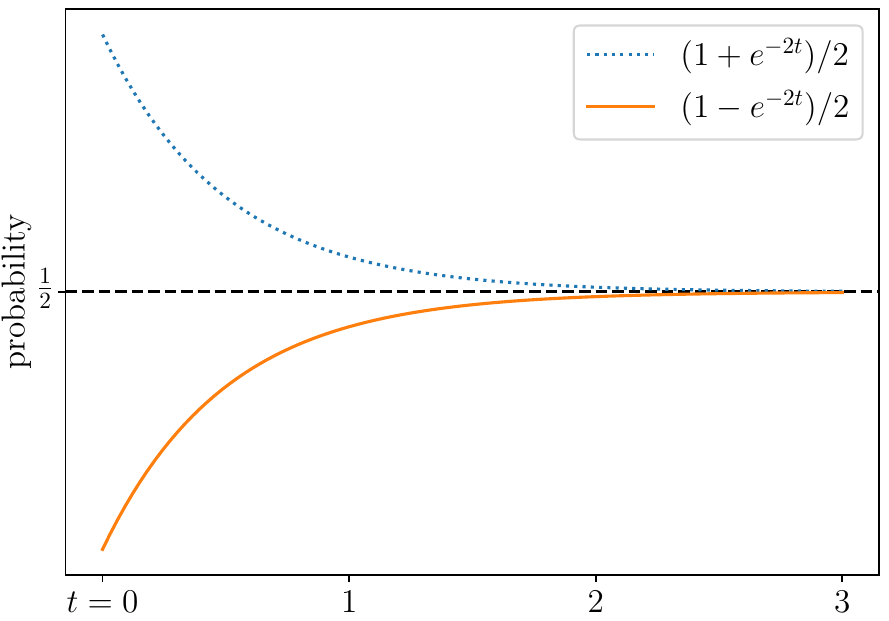}
    \caption{Stochastic probability versus time generated by the graph Laplacian \eqref{eqn:da2.18} with initial state $\ket{0}$.  The final state is the equal probabilistic mixture of $\ket{0}$ and  $\ket{1}$.}
    \label{fig:stochprob2.2}
\end{figure}

    
    
    
    

Let us repeat the analysis for the case of a quantum particle on the same graph yet with a complexified form of the generator from \eqref{eqn:da2.18}. 
The time-dependent quantum particle of a particle walking on this same graph is given as \eqref{eqn:quantumprop2.25} with initial state \eqref{eqn:quantuminit2.26}. 

\begin{align}\label{eqn:quantumprop2.25}
e^{-\imath t{\mathcal H}}& = \sum_i e^{-\imath t\lambda_i}\ketbra{\lambda_i}{\lambda_i}
\\
\label{eqn:quantuminit2.26}
\ket{0}& = \frac{1}{\sqrt[]{2}} \ket{\lambda_0} + \frac{1}{\sqrt[]{2}} \ket{\lambda_1} 
\end{align}

The dynamics readily simplify to the following form \eqref{eqn:simpleqev2.27}, from which we arrive at the time dependent probability of measuring $\ket{0}$ as \eqref{eqn:quantumprobzero2.28}. 

\begin{equation}\label{eqn:simpleqev2.27}
e^{-\imath t{\mathcal H}}\ket{0} = \frac{1}{\sqrt[]{2}} \ket{\lambda_0} + \frac{e^{-\imath 2t}}{\sqrt[]{2}} \ket{\lambda_1}
\end{equation}

\begin{equation}\label{eqn:quantumprobzero2.28} 
\begin{split}
\text{Prob}(\ket{0})&= \frac{1}{2}\left[ \frac{1}{\sqrt[]{2}} (1 + e^{-\imath 2t}) \right]
\left[ \frac{1}{\sqrt[]{2}} (1 + e^{\imath 2t}) \right]\\
&=\frac{1}{4}(2+e^{-\imath 2t} + e^{\imath 2t})\\
&=\frac{1+\cos{2}t}{2}
\end{split}
\end{equation}


\begin{figure}[htbp]
    \centering
    \includegraphics[width=11cm]{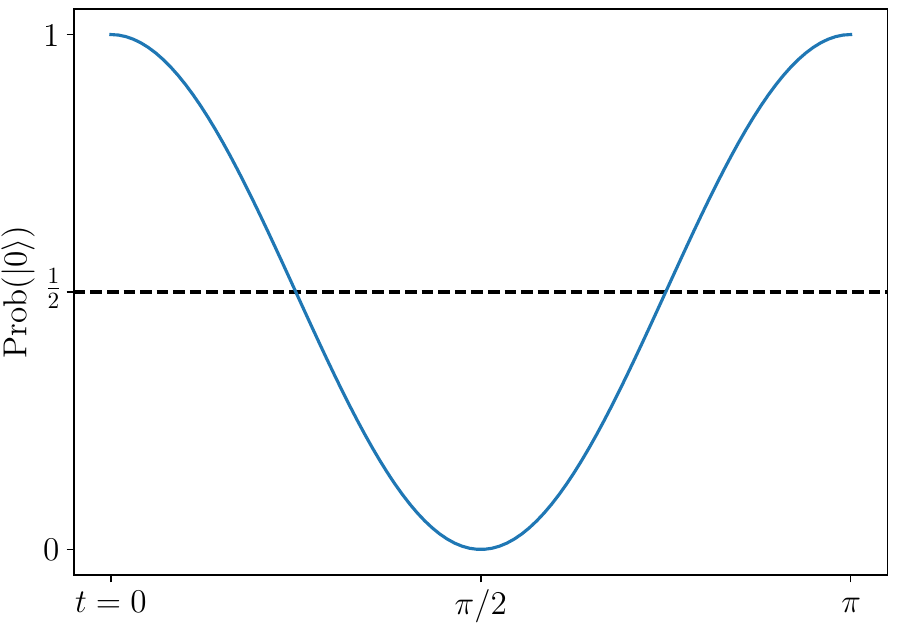}
    \caption{Quantum probability versus time generated by the graph Laplacian \eqref{eqn:da2.18} with initial state $\ket{0}$ exhibits strong oscillations.}
    \label{fig:prob-quantum}
\end{figure}

Figure \ref{fig:prob-quantum} illustrates that probability for quantum evolution of a closed systems exhibits oscillations. These oscillations are central to the function of a gate model quantum processor. 

\begin{example}[Local and composite rotation gates]\label{ex:lcrg}
In the gate model of quantum computation, various techniques exist to create quantum logic gates \cite{NC}. For example, a Z rotation is evolved as
\begin{equation}
\begin{split}
e^{-\imath \theta Z} &= 
e^{-\imath \theta} \ketbra{0}{0} +
e^{\imath \theta} \ketbra{1}{1} \\ 
&=e^{-\imath \theta}(\ketbra{0}{0} + e^{2\imath \theta}\ketbra{1}{1})\\ 
&\sim \ketbra{0}{0} + e^{\imath 2\theta} \ketbra{1}{1}. 
\end{split}
\end{equation}
It is left to the reader to consider the following gate generated by non-commuting terms. 
\begin{equation}
\text{\sf H}=\frac{1}{\sqrt[]{2}}(X+Z) \index{Quantum Gates}
\end{equation}
and to find $t$ s.t.~$e^{-\imath t \text{\sf H}} \sim  \text{\sf H}$ (up to a unit modulus complex number) realizes the Hadamard gate.  The gate can also be created by a product of unitary gates \cite{NC}. 
\end{example}

\begin{remark}[Quantum Circuits]
We will now begin to utilize an evident and widely used graphical depiction of quantum gates.  This will be further developed and connected to the theory of tensor networks in \S~\ref{chap:tensornets}.  Here we will adapt an intuitive approach and focus solely on the evident applications at hand.  
\end{remark}

\begin{definition}[Quantum Circuit---Sketch]
A {\it quantum circuit} also called {\it quantum computational network} is a model for a computation representing a sequence of unitary operators as a sequence of quantum gates on an $n$-qubit register. The identity gates are given as wires while the gates appear as boxes on the wires in which they act non-trivially upon. Time goes from right to left across the page herein (to match how equations are written). 
\end{definition}

It is typical to consider a tunable Hamiltonian that will be fully controllable and used to sequence quantum gates.  The standard Hamiltonian assumed herein follows in Definition \ref{def:IMTF}. 

\begin{definition}[Tunable Ising Model with Transverse Field]\label{def:IMTF}
The tunable Ising model allows the application of any gate formed from the Hamiltonian 
\begin{equation}
    {\mathcal H} = \sum_{i,j}J_{ij}Z_iZ_j + \sum_i h_i Z_i + \sum_i \delta_i X_i
\end{equation}
where typically gates are fundamentally assumed to act on one or two qubits at a time. Furthermore, it is typical to consider each gate as being generated by commuting Hamiltonian terms (see Example \ref{ex:lcrg}). 
\end{definition}

\begin{definition}[Single Qubit Gate]
A single qubit gate is given by the following quantum gate $U_t$. 
\begin{equation*}
\Qcircuit @C=1em @R=1em {
\lstick{U_t = e^{-\imath t{\mathcal H}} \longrightarrow} & \gate{U_t} & \qw
}
\end{equation*}
\end{definition}

As an extended example, let us create a {\sf CN}\index{Quantum Gates} gate acting as\\
\begin{center}
\begin{circuitikz}[scale=2]
\draw (-2,0) circle [radius=0.002] ;
\draw [fill] (2,0) circle [radius=0.1] ;
\draw (1,0) -- (3,0);
\draw (3.4, 0) node [black] {$\ket{a}$};
\draw (2,0) -- (2,-1);
\draw (2,-0.8) circle [radius=0.2] ;
\draw (1, -0.8) -- (3,-0.8);
\draw (3.4, -0.8) node [black] {$\ket{b}$};
\draw (0.6, 0) node [black] {$\ket{a}$};
\draw (0.25, -0.8) node [black] {$\ket{a\oplus b}$};
\draw (2, -1.4) node [black] {$
\longleftarrow$ time};
\end{circuitikz}
\end{center}

One readily expands the {\sf CN} gate as
\begin{equation}\label{eqn:cnfact}
    \text{{\sf CN}}=(P_1 \otimes X + P_0 \otimes \eye). 
\end{equation}
From the fact that $\text{{\sf CN}}^\dagger = \text{{\sf CN}}$ and $\text{{\sf CN}}^2=\eye$, it can be established that
\begin{equation}
\exists \hspace{5pt} \phi \hspace{5pt} | \hspace{5pt} e^{-\imath \phi \text{{\sf CN}}}=  \cos(\phi)\cdot \eye -\imath \sin(\phi)\cdot {\sf CN}\sim \text{{\sf CN}}.
\end{equation}
Rewriting \eqref{eqn:cnfact} as 
\begin{equation}
\text{{\sf CN}}=(\eye \otimes H )(P_1 \otimes Z + P_0 \otimes \eye)(\eye \otimes H^\dagger),  
\end{equation}
the two-body term, $P_1 \otimes Z + P_0 \otimes \eye$, can be realized with a Hamiltonian of the form
\begin{equation}\label{lb}
    \alpha Z_1 + \beta Z_2 + \gamma Z_1 Z_2
\end{equation}
as per Definition \ref{def:IMTF}. 

\begin{remark}
$P_1\otimes U + P_0\otimes \eye$ is unitary if and only if $U$ is unitary.  
\end{remark}

\begin{remark}
One readily solves for $\alpha, \beta, \gamma$ such that \eqref{lb} becomes $P_1 \otimes Z + P_0 \otimes \eye.$ Here we drop terms proportional to $\eye$ since $e^{-\imath \alpha \eye} = e^{-\imath \alpha}\eye \sim \eye$ is a global phase (see Definition \ref{def:gauge}).
\end{remark}

\begin{definition}[Unit and Scalar Gauge]\label{def:gauge}
In quantum theory, a global phase is undetectable.  Hence it is common to consider an equivalency class where $\ket{\psi}$ and $e^{\imath \phi}\ket{\psi}$ are equivalent. This is called working in the {\it unit gauge}: sometimes we work in the {\it scalar gauge}, $\mathbb{C}_{/\{0\}}$ (the field $\mathbb{C}$ exclude $0$). This amounts to mapping numbers picked up during calculation as 
$$
\mathbb{C}_{/\{0\}}\rightarrow 1. 
$$ 
(See also Definition \ref{def:global}). 
\end{definition}

Then we arrive at the sequence\\
\begin{center}
\begin{circuitikz}[scale=1.5] 
\draw  (-0.3,0) rectangle (1.3, 2.5);
\draw (0.5,1.2)node{$e^{-\imath t P_{1}\otimes Z}$};
\draw (-3, 2) -- (-0.3, 2);
\draw (-1, 0.5) -- (-0.3, 0.5);
\draw (-3, 0.5) -- (-2, 0.5);
\draw  (-2, 0) rectangle (-1, 1);
\draw (-1.5,0.5)node{\large \sf H};
\draw (1.3, 2) -- (4, 2);
\draw (1.3, 0.5) -- (2, 0.5);
\draw (3, 0.5) -- (4, 0.5);
\draw  (2,0) rectangle (3, 1);
\draw (2.5,0.5)node{\large \sf H};
\draw (5, 1.2)  node [black] {$\sim$};

\draw [fill] (7,2) circle [radius=0.1] ;
\draw (6.2,2) -- (7.8,2);
\draw (7,2) -- (7,0);
\draw (7,0.2) circle [radius=0.2] ;
\draw (6.2,0.2) -- (7.8,0.2);
\end{circuitikz}
\end{center}

\begin{proposition}
One can find $t$ over the reals such that
$$ 
\text{{\sf CN}}\sim (\eye \otimes H)\cdot (e^{-\imath tP_{1}\otimes Z})\cdot (\eye \otimes H^\dagger). 
$$ 
\end{proposition}

\begin{theorem}[{\sf CN}s plus local rotations is universal]
{\sf CN}, together with local unitary rotations is a universal gate set for quantum computation \cite{NC}.  \index{Quantum Gates}
\end{theorem}

\begin{remark}
\S~\ref{sec:variational} makes use of a result by Shi which established the following.  A gate comprising the controlled not (a.k.a.~Feynman gate) plus any one-qubit gate whose square does not preserve the computational basis is universal for quantum computation \cite{SHI02}. 
\end{remark}

\begin{remark}
\S~\ref{sec:qma} makes use of a result by Bernstein and Vazirani which showed that arbitrary quantum circuits may be simulated using real-valued gates operating on real-valued wave functions~\cite{BV97}.
\end{remark}

\begin{table}[htbp!]
\begin{center}
\renewcommand{\arraystretch}{1.5}
\resizebox{0.75\textwidth}{!}{\begin{minipage}{\textwidth}
\begin{tabular}{ p{3cm}|p{5cm}|p{5cm} }
         & {\bf quantum mechanics} & {\bf stochastic mechanics} \\\hline 
  {\bf state}  
& vector $ \psi \in \mathbb{C}^n$ with
$$ \sum_i |\psi_i|^2 = 1 $$
  &
vector $\psi \in \mathbb{R}^n$ with $$ \sum_i \psi_i = 1 $$ 
and we typically insist that,  $$ \psi_i \ge 0 $$ 
\\\hline 
{\bf observable} & $n \times n$ matrix ${\mathcal O}$ with 
$$  {\mathcal O}^\dagger = {\mathcal O}$$
where $({\mathcal O}^\dagger)_{i j} \bydef \overline{{\mathcal O}}_{j i}$
  & vector ${\mathcal O} \in \mathbb{R}^n$  \\\hline 
{\bf expected value} &
$$ \langle \psi |{\mathcal O}| \psi \rangle \bydef \sum_{i,j} \overline{\psi}_i {\mathcal O}_{i j} \psi_j $$
&  $$ \langle {\mathcal O} \psi \rangle \bydef \sum_i {\mathcal O}_i \psi_i $$ 
\\\hline
  {\bf symmetry}\break (linear map sending states to states) & unitary $n \times n$ matrix: $$ U U^\dagger = U^\dagger U = \eye $$
 & stochastic $n \times n$ matrix: $$ \sum_j U_{i j} = 1 , \quad U_{i j} \ge 0 $$
\\\hline  
  {\bf symmetry \break generator} & self-adjoint $n \times n$ matrix: $${\mathcal H}={\mathcal H}^\dagger$$ 
 & infinitesimal stochastic $n \times n$ \break matrix:
$$ \sum_j {\mathcal H}_{i j}=0 , \quad  i\neq j \;$$ 
\newline $$\Rightarrow \; {\mathcal H}_{i j} \le 0 $$ 
\\\hline
{\bf symmetries from symmetry \break generators} &
$$ \mathcal{U}(t) = \exp(-\imath t{\mathcal H}) $$ &
$$ \mathcal{U}(t) = \exp(-t{\mathcal H}) $$
\\\hline
{\bf equation of \hfill \break motion} & $$\imath \frac{d}{dt} \psi(t) = {\mathcal H} \psi(t)$$ with solution $$\psi(t) = \exp(-\imath t{\mathcal H})\psi(0)$$ & $$\frac{d}{dt} \psi(t) = -{\mathcal H} \psi(t)$$ with solution $$\psi(t) = \exp(-t{\mathcal H})\psi(0)$$ 
 \end{tabular} 
\end{minipage} }
\caption{Summary of quantum versus stochastic mechanics reproduced from \cite{2012arXiv1209.3632B}.}
\end{center}
\end{table}

\clearpage 
\section{Walks on Graphs: quantum vs stochastic}

We previously considered as an extended example, showing how a simple graph $G$ can be used to define both stochastic and quantum walks. We will extend our previous example then move towards a more general theory. This segment follows partially a blog post on Azimuth by Tomi Johnson and edited by several of us.\footnote{Quantum Network Theory (Part 1), Azimuth \url{https://johncarlosbaez.wordpress.com/2013/08/05/quantum-network-theory-part-1/}.} The focus of the article was joint work appearing in \cite{faccin2013degree}. 

As per Definition \ref{def:simplegraph}, an example of a simple graph appears in Figure~\ref{fig:simpleG}.  To avoid complications, let’s stick to simple graphs with a finite number $n$ of nodes. Let’s also assume you can get from every node to every other node via some combination of edges i.e.~the graph is connected.  So finally we note that for the graph in Figure~\ref{fig:simpleG}, there is at most one edge between any two nodes, there are no edges from a node to itself and all edges are undirected. 

\begin{figure}
\centering
\tikzset{every picture/.style={line width=0.75pt}} 

\begin{tikzpicture}[x=0.75pt,y=0.75pt,yscale=-1,xscale=1]

\draw  [color={rgb, 255:red, 155; green, 155; blue, 155 }  ,draw opacity=1 ][fill={rgb, 255:red, 0; green, 0; blue, 0 }  ,fill opacity=1 ][line width=3.75]  (166,40) .. controls (166,26.19) and (177.19,15) .. (191,15) .. controls (204.81,15) and (216,26.19) .. (216,40) .. controls (216,53.81) and (204.81,65) .. (191,65) .. controls (177.19,65) and (166,53.81) .. (166,40) -- cycle ;
\draw  [color={rgb, 255:red, 155; green, 155; blue, 155 }  ,draw opacity=1 ][fill={rgb, 255:red, 0; green, 0; blue, 0 }  ,fill opacity=1 ][line width=3.75]  (304,156.75) .. controls (304,145.29) and (313.29,136) .. (324.75,136) .. controls (336.21,136) and (345.5,145.29) .. (345.5,156.75) .. controls (345.5,168.21) and (336.21,177.5) .. (324.75,177.5) .. controls (313.29,177.5) and (304,168.21) .. (304,156.75) -- cycle ;
\draw  [color={rgb, 255:red, 155; green, 155; blue, 155 }  ,draw opacity=1 ][fill={rgb, 255:red, 0; green, 0; blue, 0 }  ,fill opacity=1 ][line width=3.75]  (303,50) .. controls (303,32.33) and (317.33,18) .. (335,18) .. controls (352.67,18) and (367,32.33) .. (367,50) .. controls (367,67.67) and (352.67,82) .. (335,82) .. controls (317.33,82) and (303,67.67) .. (303,50) -- cycle ;
\draw  [color={rgb, 255:red, 155; green, 155; blue, 155 }  ,draw opacity=1 ][fill={rgb, 255:red, 0; green, 0; blue, 0 }  ,fill opacity=1 ][line width=3.75]  (182,181.75) .. controls (182,166.42) and (194.42,154) .. (209.75,154) .. controls (225.08,154) and (237.5,166.42) .. (237.5,181.75) .. controls (237.5,197.08) and (225.08,209.5) .. (209.75,209.5) .. controls (194.42,209.5) and (182,197.08) .. (182,181.75) -- cycle ;
\draw  [color={rgb, 255:red, 155; green, 155; blue, 155 }  ,draw opacity=1 ][fill={rgb, 255:red, 0; green, 0; blue, 0 }  ,fill opacity=1 ][line width=3.75]  (264,250.5) .. controls (264,237.52) and (274.52,227) .. (287.5,227) .. controls (300.48,227) and (311,237.52) .. (311,250.5) .. controls (311,263.48) and (300.48,274) .. (287.5,274) .. controls (274.52,274) and (264,263.48) .. (264,250.5) -- cycle ;
\draw [color={rgb, 255:red, 155; green, 155; blue, 155 }  ,draw opacity=0.6 ][fill={rgb, 255:red, 155; green, 155; blue, 155 }  ,fill opacity=0.67 ][line width=4.5]    (191,65) -- (204.5,155) ;

\draw [color={rgb, 255:red, 155; green, 155; blue, 155 }  ,draw opacity=0.6 ][fill={rgb, 255:red, 155; green, 155; blue, 155 }  ,fill opacity=0.67 ][line width=4.5]    (216,40) -- (303,50) ;

\draw [color={rgb, 255:red, 155; green, 155; blue, 155 }  ,draw opacity=0.6 ][fill={rgb, 255:red, 155; green, 155; blue, 155 }  ,fill opacity=0.67 ][line width=4.5]    (335,82) -- (324.75,136) ;

\draw [color={rgb, 255:red, 155; green, 155; blue, 155 }  ,draw opacity=0.6 ][fill={rgb, 255:red, 155; green, 155; blue, 155 }  ,fill opacity=0.67 ][line width=4.5]    (304,156.75) -- (235.5,172) ;

\draw [color={rgb, 255:red, 155; green, 155; blue, 155 }  ,draw opacity=0.6 ][fill={rgb, 255:red, 155; green, 155; blue, 155 }  ,fill opacity=0.67 ][line width=4.5]    (231.5,199) -- (271.5,235) ;

\draw [color={rgb, 255:red, 155; green, 155; blue, 155 }  ,draw opacity=0.6 ][fill={rgb, 255:red, 155; green, 155; blue, 155 }  ,fill opacity=0.67 ][line width=4.5]    (318.5,176) -- (302.5,232) ;

\draw (287.5,250.5) node  [font=\Large,color={rgb, 255:red, 255; green, 255; blue, 255 }  ,opacity=1 ] [align=left] {$\bf{|5\rangle} $};
\draw (324.75,156.75) node  [font=\Large,color={rgb, 255:red, 255; green, 255; blue, 255 }  ,opacity=1 ] [align=left] {$\bf{|3\rangle} $};
\draw (209.75,181.75) node  [font=\Large,color={rgb, 255:red, 255; green, 255; blue, 255 }  ,opacity=1 ] [align=left] {$\bf{|4\rangle} $};
\draw (335,50) node  [font=\Large,color={rgb, 255:red, 255; green, 255; blue, 255 }  ,opacity=1 ] [align=left] {$\bf{|2\rangle} $};
\draw (191,40) node  [font=\Large] [align=left] {$\textcolor[rgb]{1,1,1}{\bf{|1\rangle} }$};

\end{tikzpicture}
\caption{Example of a simple graph, used herein to compare quantum versus stochastic walks.}  
\label{fig:simpleG}
  \end{figure}
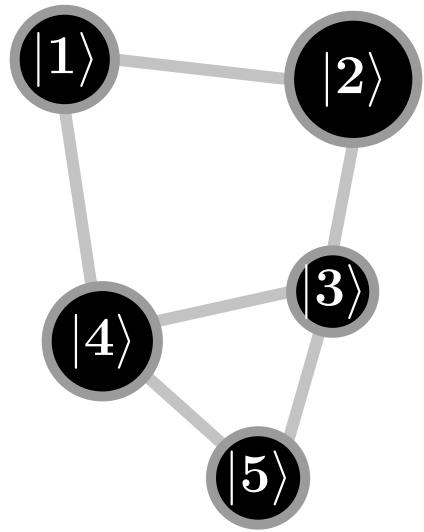  

In this particular example, the graph represents a network of $n = 5$ nodes, where nodes 3 and 4 have degree (number of edges) 3, and nodes 1, 2 and 5 have degree 2.

\begin{definition}
The {\it degree of a node} in a network (a.k.a.~graph) is the number of connections this node has to other nodes while the {\it degree distribution} is the probability distribution of these degrees over the entire network. Standard concepts in network theory appear in the book \cite{9780199206650}. 
\end{definition}

As mentioned in Definition \ref{remark:basis}, every simple graph defines a matrix $A$, called the \textit{adjacency matrix}. For a network with $n$ nodes, this matrix is of size $n \times n$, and each element $A_{i j}$ is unity if there is an edge between nodes $i$ and $j$, and zero otherwise.  

\begin{remark}
We will fix and use the basis defined in accordance with Figure~\ref{fig:simpleG} throughout this section.  
\end{remark}

Using the natural ordering per Figure~\ref{fig:simpleG}, the adjacency matrix is

$$ \left( \begin{matrix} 0 & 1 & 0 & 1 & 0 \\ 1 & 0 & 1 & 0 & 0 \\ 0 & 1 & 0 & 1 & 1 \\ 1 & 0 & 1 & 0 & 1 \\ 0 & 0 & 1 & 1 & 0 \end{matrix} \right) $$

By this construction, every adjacency matrix is symmetric (iff $A =A^\top$) and further, because each $A$ is real, it is self-adjoint (iff $A=A^\dagger$). This is an attractive features as a self-adjoint matrix generates a continuous-time \textbf{quantum walk}. Informally, a quantum walk is an evolution arising from a quantum walker moving on a network.

The state of a quantum walk is represented by a size $n$ complex column vector $\psi$ (written equivalently $\ket{\psi}$). Each element $\langle j | \psi \rangle$, is the so-called amplitude associated with node $i$ and the probability of the walker being found on that node (if measured) is the modulus of the amplitude squared $|\langle j | \psi \rangle|^2$. Here $i$ is the standard basis vector with a single non-zero $i$-th entry equal to unity, and $\langle u , v \rangle = u^\dagger v$ remains the usual inner vector product.

A quantum walk evolves in time according to the Schr{\"o}dinger's equation \eqref{schr_eq} under some Hamiltonian, $\mathcal{H}$. If the initial state is $\psi(0)$ then the solution to the T.I.S.E.~is again written as
\begin{equation*}
   \psi(t) = \exp(- \imath t \mathcal{H}) \psi(0). 
\end{equation*}

\begin{remark}
Recall that the probabilities $\{ | \langle j | \psi (t) \rangle |^2 \}_i$ are guaranteed to be correctly normalized when the Hamiltonian $\mathcal{H}$ is self-adjoint.
\end{remark}

There are several common matrices defined by graphs. Perhaps the most familiar is the graph Laplacian. The Laplacian $\mathcal L$ is the $n \times n$ matrix $\mathcal {L} = D - A$, where the degree matrix $D$ is an $n \times n$ diagonal matrix with elements given by the degrees

$$ D_{i i}=\sum_{j} A_{i j}  $$.

For the graph drawn in Figure~\ref{fig:simpleG}, the degree matrix and Laplacian are:

$$ \left( \begin{matrix} 2 & 0 & 0 & 0 & 0 \\ 0 & 2 & 0 & 0 & 0 \\ 0 & 0 & 3 & 0 & 0 \\ 0 & 0 & 0 & 3 & 0 \\ 0 & 0 & 0 & 0 & 2 \end{matrix} \right)  \qquad \mathrm{and} \qquad \left( \begin{matrix} 2 & -1 & 0 & -1 & 0 \\ -1 & 2 & -1 & 0 & 0 \\ 0 & -1 & 3 & -1 & -1 \\ -1 & 0 & -1 & 3 & -1 \\ 0 & 0 & -1 & -1 & 2 \end{matrix} \right) $$.

The Laplacian is self-adjoint and generates a quantum walk. Recall from \S~\ref{subsec:phsm} this matrix also has another property; it is infinitesimal stochastic. Recall that this means that its off diagonal elements are non-positive and its columns sum to zero. This is interesting because an infinitesimal stochastic matrix generates a continuous-time stochastic walk. Such walks have been studied extensively in the literature, including work by myself and others, partially reviewed in \cite{2017arXiv170208459B} and appearing as research in \cite{faccin2013degree, Faccin_2014} as well as \cite{Zimbors2013, Lu_2016} which considered time-symmetry breaking in quantum walks.

Recall that a stochastic walk is an evolution arising from a stochastic walker moving on a network. A state of a stochastic walk is represented by a size $n$ non-negative column vector $\psi$. Each element $\langle j | \psi \rangle$ of this vector is the probability of the walker being found on node $i$.  A stochastic walk evolves in time according to the master equation
\begin{equation}
    \frac{d}{d t} \psi(t)= - {\mathcal L} \psi(t). 
\end{equation}
where $H$ is called the stochastic Hamiltonian. If the initial state is $\psi(0)$ then the solution is written as 
\begin{equation*}
   \psi(t) = \exp(- t {\mathcal L}) \psi(0). 
\end{equation*}
The probabilities $\{ \langle j| \psi (t) \rangle \}_i$ are guaranteed to be non-negative and correctly normalized when the stochastic Hamiltonian ${\mathcal L}$ is infinitesimal stochastic.

\subsection{Normalized Laplacians}
To analyze the important uniform escape model we need to go beyond the class of generators that produce both quantum and stochastic walks. Further, we have to determine a related quantum walk. We'll see below that both tasks are achieved by considering the normalized Laplacians: one generating the uniform escape stochastic walk and the other a related quantum walk.

The two normalized Laplacians (studied in e.g.~\cite{faccin2013degree} as well as many other places) are as follows. 

\begin{definition}[Asymmetric Normalized Laplacian]
The asymmetric normalized Laplacian is given as 
$$
S = \mathcal{L} D^{-1}
$$
which generates the uniform escape stochastic walk by $S$. 
\end{definition}

\begin{definition}[Symmetric Normalized Laplacian]
The symmetric normalized Laplacian is given as  
$$
Q = D^{-1/2} \mathcal{L} D^{-1/2}
$$
which generates a unitary quantum walk $Q$.  Proposition \ref{prop:ssg} and \ref{prop:ssgc} will illuminate exactly why this choice was made: so as to exploit a relationship between eigenvectors between asymmetric and symmetric normalized Laplacians.    
\end{definition}

For the graph drawn in Figure~\ref{fig:simpleG}, the asymmetric normalized Laplacian $S$ is given as \eqref{eqn:anl}. 

\begin{equation}\label{eqn:anl}
    \left( \begin{matrix} 1 & -1/2 & 0 & -1/3 & 0 \\ -1/2 & 1 & -1/3 & 0 & 0 \\ 0 & -1/2 & 1 & -1/3 & -1/2 \\ -1/2 & 0 & -1/3 & 1 & -1/2 \\ 0 & 0 & -1/3 & -1/3 & 1 \end{matrix} \right)
\end{equation}

\begin{remark}
The identical diagonal elements indicate that the total rates of leaving each node are identical, and the equality within each column of the other non-zero elements indicates that the walker is equally likely to hop to any node connected to its current node. This is, after all, how the uniform escape model takes its name.  
\end{remark}

For the same graph drawn in Figure~\ref{fig:simpleG}, the symmetric normalized Laplacian $Q$ is given as \eqref{eqn:snl}. 

\begin{equation}\label{eqn:snl}
    \left( \begin{matrix} 1 & -1/2 & 0 & -1/\sqrt{6}  & 0 \\ -1/2 & 1 & -1/\sqrt{6}  & 0 & 0 \\ 0 & -1/\sqrt{6}  & 1 & -1/3 & -1/\sqrt{6}  \\ -1/\sqrt{6} & 0 & -1/3 & 1 & -1/\sqrt{6}  \\ 0 & 0 & -1/\sqrt{6}  & -1/\sqrt{6}  & 1 \end{matrix} \right)
\end{equation}

The diagonal elements being identical in the quantum case indicates that all the nodes are of equal energy.  Walks not obeying this are said to exhibit {\it disorder}. 

Returning to the graph from Figure~\ref{fig:simpleG}, we have consequently defined two matrices: one $S$ that generates a stochastic walk, and one $Q$ that generates a quantum walk. The natural question to ask is whether these walks are related. This question was studied extensively in \cite{faccin2013degree}. Before considering aspects of this relationship, consider the following definition. 

\begin{definition}[Similar Matricies]
Two $n \times n$ matrices ($A$ and $B$) are called {\it similar} if there exists an invertible $n \times n$ matrix $P$ such that 
\begin{equation}
    B = P^{-1} A P.
\end{equation}
\end{definition}

Underpinning a relationship between the quantum and stochastic generators is the known property that: 
\begin{proposition}[Spectral Similarity of Generators]\label{prop:ssg}
 $S$ and $Q$ are similar by the transformation \eqref{eqn:simu}. 
\begin{equation}\label{eqn:simu}
   S = D^{1/2} Q D^{-1/2}
\end{equation}
\end{proposition}

From Proposition \ref{prop:ssg} we deduce the following. 
\begin{proposition}[Quantum spectral similarity---with Faccin et al., Phys.~Rev.~X 3, 041007 (2013) \cite{faccin2013degree}]\label{prop:ssgc}
Consider, 
\begin{equation}
   \begin{aligned} 
   Q \ket{\phi_k} &= \epsilon_k \ket{\phi_k},  \\  
   (D^{1/2}  Q D^{-1/2} ) (D^{1/2} \ket{\phi_k} ) &= \epsilon_k ( D^{1/2} \ket{\phi_k} ), \\ 
   S \ket{\pi_k} &= \epsilon_k \ket{\pi_k}.
   \end{aligned} 
\end{equation}
Hence, any eigenvector $\ket{\phi_k}$ of $Q$ associated to eigenvalue $\epsilon_k$ implies that $\ket{\pi_k} \propto D^{1/2} \ket{\phi_k}$ is an eigenvector of $S$ associated to the same eigenvalue.   
\end{proposition}

One also establishes the converse. Any eigenvector $\ket{\pi_k}$ of $S$ implies an eigenvector  $\ket{\phi_k} \propto D^{-1/2} \ket{\pi_k}$ of $Q$ associated to the same eigenvalue $\epsilon_k$. And because $Q$ is self-adjoint, the symmetric normalized Laplacian can be decomposed as 
\begin{equation*}
   Q = \sum_k \epsilon_k \Phi_k,
\end{equation*}
where $\epsilon_k$ is real and $\Phi_k=\ketbra{\phi_k}{\phi_k}$ are orthogonal. Multiplying from the left by $ D^{1/2}$ and the right by $ D^{-1/2}$ results in a similar decomposition for $S$
\begin{equation*}
   S = \sum_k \epsilon_k \Pi_k,
\end{equation*}
with orthogonal projectors $\Pi_k  = D^{1/2} \Phi_k D^{-1/2}$.

We now have all the ingredients necessary to study the walks generated by the normalized Laplacians and the relationship between them.

\begin{definition}[Summary see Figure~\ref{fig:similarity}---with Faccin et al., Phys.~Rev.~X 3, 041007 (2013) \cite{faccin2013degree}]
$G$ is a simple undirected graph.  Labeling the nodes of $G$ lifts to specify:
\begin{enumerate}
    \item $A$ the adjacency matrix (generator of a quantum walk). 
    \item $D$ the diagonal matrix of the degrees. 
    \item ${\mathcal L}$ the symmetric Laplacian (generator of stochastic and quantum walks), which when normalized by $D$ returns both: 
    \item[3.1] $S$ the generator of the uniform escape stochastic walk and
    \item[3.2] $Q$ the quantum walk generator to which ${\mathcal L}$ is similar. 
\end{enumerate}
\end{definition}

\subsection{Stochastic walk}

The uniform escape stochastic walk generated by $S$ has a practically useful stationary state. A stationary state of a stochastic walk is stationary (invariant) with respect to changes in parameter playing the role of time. From the master equation $\frac{d}{d t} \ket{\psi(t)} = -S \ket{\psi(t)}$ the stationary state must be an eigenvector $\ket{\pi_0^i}$ of $S$ with eigenvalue $\epsilon_0 = 0$. A pair of well-known theorems hold:

\begin{theorem}[Uniqueness of the Stationary State \cite{2012arXiv1209.3632B}]\label{theorem:uss}
There always exists a unique (up to multiplication by a positive number) positive eigenvector $\ket{\pi_0}$ (abbreviated as $\ket{\pi}$) of $S$ with eigenvalue $\epsilon_0 = 0$, i.e., a unique stationary state $\ket{\pi}$.
\end{theorem}

The following theorem is well known, see e.g.~\cite{2012arXiv1209.3632B}. 

\begin{theorem}[Global attractor of the uniform escape model]
Consider the uniform escape stochastic walk generated by $S$. 
Regardless of the initial state $\ket{\psi(0)}$, the stationary state $\ket{\pi}$ is obtained in the long-time limit: $\lim_{t \rightarrow \infty} \ket{\psi(t)} = \ket{\pi}$.
\end{theorem}

\begin{definition}[All 1's Vector]
Let boldface $\ket{\boldsymbol{+}} = \sum \ket{i}$ denote the unnormalized all ones vector, where non-boldface $\sqrt{2}\ket{+} = \ket{0} + \ket{1}$ is the familiar eigenstate of the $X$ operator. 
\end{definition}

To determine this unique stationary state (Theorem \ref{theorem:uss}), consider the Laplacian $\mathcal{L}$, which is both infinitesimal stochastic and symmetric. Among other things, this means the rows of $\mathcal{L}$ sum to zero
\begin{equation*}
    \sum_j \mathcal{L}_{i j} = 0 
\end{equation*}
which means the all ones vector $\ket{\boldsymbol{+}}$ is an eigenvector of $\mathcal{L}$ with zero eigenvalue
\begin{equation*}
   \mathcal{L} \ket{\boldsymbol{+}} = 0 
\end{equation*}
Inserting the identity $\eye = D^{-1} D$ into this equation we recover that  
$D \ket{\boldsymbol{+}}$ 
is a zero eigenvector of $S$
\begin{equation*}
   \mathcal{L} \ket{\boldsymbol{+}} =  (\mathcal{L} D^{-1}) (D \ket{\boldsymbol{+}}) = S (D \ket{\boldsymbol{+}}) = 0 
\end{equation*}
Therefore one must normalize $\ket{\boldsymbol{+}}$ to recover the long-time stationary state of the walk, viz., 
\begin{equation*}
   \ket{\pi} = \frac{D }{\sum_i D_{i i}}\ket{\boldsymbol{+}}. 
\end{equation*}

\begin{proposition}[Long-Time Probability---with Faccin et al., Phys.~Rev.~X 3, 041007 (2013) \cite{faccin2013degree}]\label{prop:ltp}
For each element $\langle j| \pi \rangle$ of this state (where $j$ indexes nodes), the long-time probability of finding a walker at node $i$, is proportional to the degree $D_{i i}$ of node $i$. 
\end{proposition}

\begin{remark}
We can validate Proposition \ref{prop:ltp} for the graph from Figure~\ref{fig:simpleG}, where $\ket{\pi}$ is as follows. 
$$ \left( \begin{matrix} 1/6 \\ 1/6 \\ 1/4 \\ 1/4 \\ 1/6 \end{matrix} \right) $$
This hence implies that the long-time probability $1/6$ for nodes $1$, $2$ and $5$, and $1/4$ for nodes $3$ and $4$.
\end{remark}

\subsection{Quantum walks}  

Any quantum process in finite dimensions can be viewed as a spin-free quantum walk of a single particle on the graph given by the support of the generator.  For example, if the time generator is unitary $U(t)$ then if the entry $\bra{q}U(t)\ket{r}$ is non zero, then the corresponding graph edge $r$, $q$ is present. Little remains known.  Such walks arise in exciton transport studies and quantum walks represent a universal model of quantum computation \cite{Childs_2009}. 

Let us now consider the quantum walk generated by $Q$.  Below we recover the analytical results \cite{faccin2013degree} obtained by exploiting the operator similarity of $S$ and $Q$.

In contrast to stochastic walks, for a quantum walk every eigenvector $\ket{\phi_i^k}$ of $Q$ is a stationary state of the quantum walk. The stationary state $\ket{\phi_0}$ is of particular interest, both physically and mathematically. Physically, since eigenvectors of  $Q$ correspond to states of well-defined energy equal to the associated eigenvalue, $\ket{\phi_0}$ is the state of lowest energy $\epsilon_0 = 0$, hence the name ground state (ground eigenvalue/ground energy).

Mathematically, the relationship between eigenvectors implied by the similarity of $S$ and $Q$ means 
\begin{equation}
    \ket{\phi_0} \propto D^{-1/2} \ket{\pi} \propto  D^{1/2} \ket{\boldsymbol{+}}. 
\end{equation} 
Therefore, for a system in its ground state, the probability of measurement resulting in a particle at node $i$ is given by 
\begin{equation}\label{eqn:probgs}
    | \langle j | \phi_0 \rangle |^2 \propto | \langle j | D^{1/2} |+\rangle |^2 = D_{i i}.
\end{equation}

The probability \eqref{eqn:probgs} shows proportionality with a nodes degree and is therefore identical to $\langle j | \pi \rangle$: where $\ket{\pi}$ is the stationary state  of the stochastic walk (Proposition \ref{prop:ltp}). Hence, a zero energy quantum walk $Q$ leads to the same distribution over the nodes as the long-time limit of the uniform escape stochastic walk $S$ \cite{faccin2013degree}.

Hence, the work \cite{faccin2013degree} determined that the standard notion of degree distribution plays a role in quantum walks. But what if the walker starts in some other initial state? We therefore must ask, ``is there is some quantum walk analogue of the unique long-time state of a stochastic walk?''

Clearly, the quantum walk generally does not converge to a stationary state. Yet there is a probability distribution that can be thought to characterize the quantum walk in the same way as the long-time state characterizes the stochastic walk. This is the long-time average probability vector $P$.

Provided a complete lack of knowledge as to the time that had passed since the beginning of a quantum walk, the best estimate for the probability of measuring the walker to be at node $i$ would be the long-time average probability \eqref{eqn:ltap}. 
\begin{equation}\label{eqn:ltap}
   \displaystyle{ \langle j | P \rangle = \lim_{T \rightarrow \infty} \frac{1}{T} \int_0^T | \psi_i (t) |^2 d t } 
\end{equation}

Equation \ref{eqn:ltap} will be simplified. We begin by inserting the decomposition $Q= \sum_k \epsilon_k \Phi_k $ into $\ket{\psi(t)}  = e^{-\imath Q t} \ket{\psi(0)}$ to get $ \ket{\psi(t)}  = \sum_k e^{-\imath\epsilon_k t} \Phi_k \ket{\psi(0)}$ and thus,
\begin{equation*}
    \langle j | P \rangle = \lim_{T \rightarrow \infty} \frac{1}{T} \int_0^T \biggl| \sum_k e^{-\imath \epsilon_k t} \langle j| \Phi_k |\psi (0) \rangle \biggr|^2 d t .  
\end{equation*}
Due to the integral over all time, the interferences between terms corresponding to different eigenvalues average to zero, leaving
\begin{equation}\label{eqn:ltap2.43}
   \langle j | P \rangle = \sum_k | \langle j| \Phi_k | \psi(0) \rangle |^2. 
\end{equation}
The long-time average probability \eqref{eqn:ltap2.43} is then the sum of terms contributed by the projections of the initial state onto each eigenspace.

Hence we arrive at a distribution \eqref{eqn:ltap2.43} that characterizes a quantum walk for a general initial state. Our approach of better understanding the long-time average probability is through the term $| \langle j| \Phi_0| \psi (0) \rangle |^2 $ associated with the zero energy eigenspace, since we have subsequently charachterized this space.

For example, we know the zero energy eigenspace is one-dimensional and spanned by the eigenvector $\phi_0$. This means that the projector is just the usual outer product
\begin{equation}
   \Phi_0 = | \phi_0 \rangle \langle \phi_0 | ,
\end{equation}
where we have normalized $\ket{\phi_0}$ according to the inner product $\langle \phi_0| \phi_0\rangle = 1$. The zero eigenspace contribution to the long-time average probability then breaks into the product of two probabilities, as  
\begin{equation}\label{eqn:2probprod}
 \begin{split}
   | \langle j| \Phi_0|\psi (0) \rangle |^2 &= | \langle j| \phi_0\rangle \langle \phi_0|  \psi (0) \rangle |^2   = | \langle j| \phi_0\rangle |^2 | \langle \phi_0 | \psi (0) \rangle |^2   ={}\\
   &=\langle j |  \pi \rangle | \langle \phi_0 |  \psi (0) \rangle |^2 .
   \end{split}
\end{equation}

The first probability $\langle j | \pi \rangle$ in the product \eqref{eqn:2probprod} corresponds to finding a quantum state in the zero energy eigenspace at node $i$ (as we found above).  The second probability $| \langle \phi_0| \psi (0)\rangle |^2$ from \eqref{eqn:2probprod} relates to being in this eigenspace to begin with (which for the one dimensional zero energy eigenspace means just the inner product with the ground state). 

This turns out to be enough to say something interesting about the long-time average probability for all states. Sketching the results from \cite{faccin2013degree}, we have illustrated how we can break the long-time probability vector $P$ into a sum of two normalized probability vectors

\begin{equation}\label{eqn:twoterms2.46}
   P = (1- \eta) \ket{\pi} + \eta \ket{\Omega}, 
\end{equation}
the first $\ket{\pi}$ is the degree dependent stochastic stationary state associated with the zero energy eigenspace and the second is associated with the higher energy eigenspaces $\ket{\Omega}$, with 
\begin{equation}
   \langle j | \Omega \rangle = \dfrac{ \displaystyle\sum_{k\neq 0} | \langle j| \Phi_k|  \psi (0) \rangle |^2  }{ \eta}.
\end{equation}

The weight of each term in \eqref{eqn:twoterms2.46} is governed by the parameter
\begin{equation*}
   \eta =  1 - | \langle \phi_0| \psi (0)\rangle |^2. 
\end{equation*}

One could consider $\eta$ as unity minus the probability of the walker being in the zero energy eigenspace, or equivalently the probability of the walker being outside the zero energy eigenspace. So even though we don't know anything about $\ket{\Omega}$ we know its importance is controlled by a parameter $\eta$ that governs how close the long-time average distribution $P$ of the quantum walk is to the corresponding stochastic stationary distribution $\ket{\pi}$. Can we say anything physical about when $\eta$ is big or small?

As the eigenvalues of $Q$ have a physical interpretation in terms of energy, the answer is yes. The quantity $\eta$ is the probability of being outside the zero energy state. Call the next lowest eigenvalue $\Delta = \min_{k \neq 0} \epsilon_k$ the energy gap. If the quantum walk is not in the zero energy eigenspace then it must be in an eigenspace of energy greater or equal to $\Delta$. Therefore the expected energy $E$ of the quantum walker must bound $\eta$,  $E \ge \eta \Delta$. A quantum walk as such with low energy is hence similar to a stochastic walk in the long-time limit (we already knew this exactly in the zero energy limit).

\subsection{Perron’s theorem}

As we have seen in our previous extended example, stochastic systems evolve to their low eigenvalue eigenstates.  As we have cared to contrast quantum versus stochastic evolution by comparison of dynamics subject to the same generator, it is worth mentioning the following.  So called stoquastic Hamiltonian's are often considered in complexity theory as examples having elements of quantum mechanics and classical theory of stochastic matrices \cite{BDOT06, Hormozi_2017, Bravyi2010, Jordan2010}. Indeed, stoquastic Hamiltonians, those for which all off-diagonal matrix elements in the standard basis are real and non-positive (see Definition \ref{def:stoquastic}, commonly describe physical systems. 

\begin{remark}
This subsegment follows partially a blog post on Azimuth written by the author and edited by John Baez.\footnote{Network Theory (Part 20), Azimuth \url{https://johncarlosbaez.wordpress.com/2012/08/06/network-theory-part-20/}.} The focus of the post resulted in the book~\cite{2012arXiv1209.3632B} and influenced other studies including \cite{faccin2013degree}. 
\end{remark}

\begin{definition}[Stoquastic Hamiltonian \cite{BDOT06, Jordan2010, Bravyi2010}]\label{def:stoquastic}
Hamiltonians where all the off-diagonal elements in the standard basis are real and non-positive are called {\it stoquastic}. 
\end{definition}

More formally:

\begin{definition}[Stoquastic Hamiltonian \cite{BDOT06, Jordan2010, Bravyi2010}]
The Hamiltonian, 
\begin{equation}
    {\mathcal H} = \sum_i {\mathcal H}_i 
\end{equation}
is stoquastic if in the standard local basis ${\mathscr B}$ the terms ${\mathcal H}_i$ all have matrix entries that are non-positive,
\begin{equation}
    \bra{l}{\mathcal H}\ket{k} \leq 0~~~\forall~ l, k\in {\mathscr B} ~|~ l \neq k. 
\end{equation}
\end{definition}

Here we will establish some of the core elements to recover several known results about stoquastic Hamiltonians and the comparison of quantum versus stochastic walks on graphs, generated by stoquastic Hamiltonians. To begin with, we note that a simple graph can consist of many separate graphs, called components. 

\begin{definition}[Connected Simple Graph]
A simple graph is connected if it is nonempty and there is a path of edges connecting any vertex to any other.
\end{definition}

In quantum mechanics, one often considers observables that have positive expected values:
\begin{equation}
    \langle \psi|{\mathcal O} |\psi \rangle > 0 
\end{equation}
for every quantum state $\psi \in \mathbb{C}^n$. These are called {\it positive definite}. But in stochastic mechanics one often things about matrices that are positive in a more naive sense:

\begin{definition}
An $n \times n$ real matrix $T$ is positive if all its entries are positive:
\begin{equation}
   T_{i j} > 0  
\end{equation}
for all $1 \le i, j \le n$.
\end{definition}

\begin{definition}
A vector $\psi \in \mathbb{R}^n$ is positive if all its components are positive:
\begin{equation}
    \psi_i > 0
\end{equation}
for all $1 \leq i \leq n$.
\end{definition}

\begin{remark}[Nonnegative]
One will also define nonnegative matrices and vectors in the same way, replacing $> 0$ by $\geq 0$. A good example of a nonnegative vector is a stochastic state.
\end{remark}

In 1907, Perron proved the following fundamental result about positive matrices \cite{Perron1907}. 

\begin{theorem}[Perron’s Theorem \cite{Perron1907}]\label{theorem:perron} 
Given a positive square matrix $T$, there is a positive real number $r$, called the Perron–Frobenius eigenvalue of $T$, such that $r$ is an eigenvalue of $T$ and any other eigenvalue $\lambda$ of $T$ has $|\lambda| < r$. Moreover, there is a positive vector $\psi \in \mathbb{R}^n$ with $T \psi = r \psi$. Any other vector with this property is a scalar multiple of $\psi$. Furthermore, any nonnegative vector that is an eigenvector of $T$ must be a scalar multiple of $\psi$.
\end{theorem}

\begin{remark}
Hence, if $T$ is positive, it has a unique eigenvalue with the largest absolute value. This eigenvalue is positive. Up to a constant factor, it has an unique eigenvector. We can choose this eigenvector to be positive. And then, up to a constant factor, it’s the only nonnegative eigenvector of $T$.
\end{remark}

\begin{definition}[Strongly Connected Graph]
A directed graph is {\it strongly connected} if there is a directed path of edges going from any vertex to any other vertex.
\end{definition}

\begin{definition}[Irreducible Matrix]
A nonnegative square matrix $T$ is irreducible if its graph $G_T$ is strongly connected.
\end{definition}

\begin{theorem}[Perron--Frobenius Theorem \cite{Perron1907,2012arXiv1209.3632B}]\label{theorem:pft}
Given an irreducible nonnegative square matrix $T$, there is a positive real number $r$, called the {\it Perron–-Frobenius eigenvalue} of $T$, such that $r$ is an eigenvalue of $T$ and any other eigenvalue $\lambda$ of $T$ has $|\lambda| \leq r$. Moreover, there is a positive vector $\psi \in \mathbb{R}^n$ with $T\psi = r \psi$. Any other vector with this property is a scalar multiple of $\psi$. Furthermore, any nonnegative vector that is an eigenvector of $T$ must be a scalar multiple of $\psi$.
\end{theorem}

\begin{example}
The only conclusion of Theorem \ref{theorem:pft} that is weaker than Theorem \ref{theorem:perron} is that there may be other eigenvalues with $|\lambda| = r$. For example, the Pauli-$X$ matrix is irreducible and nonnegative.  Its Perron--Frobenius eigenvalue is $1$, but it also has $-1$ as an eigenvalue. In general, Perron--Frobenius theory says quite a lot about the other eigenvalues on the circle $|\lambda| = r$.
\end{example}

\begin{definition}[Dirichlet Operator]
${\mathcal L}$ is a {\it Dirichlet operator} if it’s both self-adjoint and infinitesimal stochastic
\end{definition}

With these theorems (\ref{theorem:pft} and \ref{theorem:perron}), we arrive at several consequences. 

\begin{proposition}[Baez-Biamonte 2018 \cite{2012arXiv1209.3632B}]\label{prop:ido}
Let ${\mathcal L}$ be an irreducible Dirichlet operator with n eigenstates. In stochastic mechanics, there is only one valid state that is an eigenvector of ${\mathcal L}$: the unique so-called {\it Perron--Frobenius state}. The other $n-1$ eigenvectors are forbidden states of a stochastic system: the stochastic system is either in the Perron–Frobenius state, or in a superposition of at least two eigensvectors. In quantum mechanics, all $n$ eigenstates of ${\mathcal L}$ are valid states.
\end{proposition}


\begin{proof}
To establish Proposition \ref{prop:ido}, we note first that, the matrix $H$ is rarely nonnegative: its off-diagonal entries will always be nonnegative while its diagonal entries can be negative.  This can be fixed by a simple {\it energy shift} as, 
\begin{equation}\label{eqn:thc}
    T = H + c \eye 
\end{equation}
where $c > 0$ is chosen from $H$. 
\end{proof}

\begin{remark}
The matrix $T$ from \eqref{eqn:thc} has 
\begin{enumerate}
    \item the same eigenvectors as $H$, 
    \item off-diagonal matrix entries which are the same as those of $H$, so $T_{i j}$ is nonzero for $i \ne j$ exactly when the graph we started with has an edge from $i$ to $j$. 
\end{enumerate}
So, for $i \ne j$, the graph $G_T$ will have an directed edge going from $i$ to $j$ precisely when our original graph had an edge from $i$ to $j$. And that means that if our original graph was connected, $G_T$ will be strongly connected. Thus, by definition, $T$ is irreducible. 
\end{remark}

Since $T$ is nonnegative and irreducible, the Perron–Frobenius theorem implies the following.  

\begin{proposition}[Baez-Biamonte 2018 \cite{2012arXiv1209.3632B}]
Suppose $H$ is the Dirichlet operator coming from a connected finite simple graph with edges labelled by positive numbers. Then the eigenvalues of $H$ are real. Let $\lambda$ be the largest eigenvalue. Then there is a positive vector $\psi \in \mathbb{R}^n$ with $H\psi = \lambda \psi$. Any other vector with this property is a scalar multiple of $\psi$. Furthermore, any nonnegative vector that is an eigenvector of H must be a scalar multiple of $\psi$.
\end{proposition}

\begin{proof}
The eigenvalues of $H$ are real since $H$ is self-adjoint. Notice that if $r$ is the Perron–Frobenius eigenvalue of $T = H + c \eye$ and
\begin{equation}
    T \psi = r \psi
\end{equation}
then 
\begin{equation}
    H \psi = (r - c)\psi
\end{equation}
By the Perron–Frobenius theorem the number $r$ is positive, and it has the largest absolute value of any eigenvalue of $T$. Thanks to the subtraction, the eigenvalue $r - c$ may not have the largest absolute value of any eigenvalue of $H$. It is, however, the largest eigenvalue of $H$, so we take this as our $\lambda$. The rest follows from the Perron–Frobenius theorem. 
\end{proof}

\begin{definition}
A Dirichlet operator is irreducible if it comes from a connected finite simple graph with edges labelled by positive numbers.
\end{definition}

Proof of the following is left to the reader, see \cite{2012arXiv1209.3632B}. 

\begin{theorem}
Suppose $H$ is an irreducible Dirichlet operator. Then $H$ has zero as its largest real eigenvalue. There is a positive vector $\psi \in \mathbb{R}^n$ with $H\psi = 0$. Any other vector with this property is a scalar multiple of $\psi$. Furthermore, any nonnegative vector that is an eigenvector of $H$ must be a scalar multiple of $\psi$.
\end{theorem}

\section{Subadditivity of entropy of stochastic generators}\label{sec:subadd}

\begin{definition}
A simple undirected graph with edges weighted by real numbers gives rise to a {\it generalized symmetric adjacency matrix}.  For edges labeled $l$ and $m$ weighted by $w \in \mathbb{R}$, the $l$-$m^{\textup th}$ entry of the corresponding adjacency matrix is $w$.
\end{definition} 

\begin{definition}
A {\it generalized Laplacian} arises as 
\begin{equation}
    \mathcal{L} = \mathcal{D} - \mathcal{A} 
\end{equation}
where $A$ is a generalized symmetric adjacency matrix and $D$ stores on its diagonal entries the sums of the corresponding rows of $A$. 
\end{definition}

\begin{theorem}[Biamonte-DeDomenico 2016~\cite{De_Domenico_2016}]\label{thm:gl}
	Given two generalized Laplacians and their sum  $\mathcal{L}_C=\mathcal{L}_A+\mathcal{L}_B$, and corresponding Gibbs state density matrices $ \mathcal{\rho_{C}} = e^{\beta(\mathcal{L_{A}+L_{B}})}/\mathcal{Z} $, the von Neumann entropy $S(\rho) = Tr\{\rho \ln_2 \rho\}$ is subadditive as, 
	\begin{equation}
	S(\mathbf{\rho_{C}}) \leq S(\mathbf{\rho_{A}}) + S(\mathbf{\rho_{B}}). 
	\end{equation}
\end{theorem}

\begin{remark}
We adopt the notation that $ S(\rho_A) \equiv S_A $, $ S(\rho_B) \equiv S_B $, etc.
\end{remark}
	
To prove Theorem \ref{thm:gl} we will establish and then combine three propositions. Proof of the following proposition is well known. 

\begin{proposition}
The Kullback-Leibler divergence satisfies, 
\begin{equation}
    \mathcal{D}_{1}(\boldsymbol{\rho}||\boldsymbol{\sigma})=\text{Tr}[{\boldsymbol{\rho}(\ln_{2}\boldsymbol{\rho}-\ln_{2}\boldsymbol{\sigma})]} \geq 0.
\end{equation}		
\end{proposition}



\begin{proposition}[Biamonte-DeDomenico 2016~\cite{De_Domenico_2016}]
The following quantity is non-negative as
		    \begin{eqnarray}
			\label{comp2}
			\Tr{[\mathcal{L}_{X}\boldsymbol{\rho}_{X}]} \geq 0. 
			\end{eqnarray}
\end{proposition}

\begin{proof}
From the Cholesky decomposition of $\mathbf{L}_A$, $ \mathcal{L}_B $, $ \mathcal{L}_C $, $\mathbf{\rho}_A$, $ \mathbf{\rho}_B $, and $ \mathbf{\rho}_C $ we obtain,
\begin{eqnarray}
			\label{comp2}
			\Tr{[\mathcal{L}_{X}\boldsymbol{\rho}_{X}]} &=& \Tr{[(\mathcal{D}\mathcal{D}^{\dagger}) (\mathcal{Q}\mathcal{Q}^{\dagger}])} \nonumber\\
			&=&\Tr{[(Q^{\dagger}D)(Q^{\dagger}D)^{\dagger}]}\geq 0.
\end{eqnarray}
\end{proof}

\begin{proposition}[Biamonte-DeDomenico 2016~\cite{De_Domenico_2016}]
The following quantity is non-negative as
   \begin{equation}\label{comp3}
			\ln_{2} Z_X \geq 0. 
   \end{equation}
\end{proposition}

\begin{proof}
We have that $ Z_X=\sum\limits_{i=1}^{N}e^{-\beta\lambda_{i}(\mathcal{L}_{X})} $. Using $\lambda_{1}(\mathcal{L}_{X})=0$, since $ \mathcal{L}_{X} $ is a generalized Laplacian (or Perron-Frobenius theorem), we get,  
		\begin{eqnarray}
			Z_X=1+\sum\limits_{i=2}^{N}e^{-\beta\lambda_{i}(\mathcal{L}_{X})}\geq 1.
		\end{eqnarray}
This leads to Equation~\eqref{comp3}.
\end{proof}

\begin{proof}[Proof of Theorem \ref{thm:gl}---Biamonte-DeDomenico 2016~\cite{De_Domenico_2016}]
Adding together the three propositions and including versions corresponding to both $ \mathbf{\rho}_A $ and $ \mathbf{\rho}_B $ for the first two propositions, we arrive at,
		\begin{eqnarray}
			\mathcal{D}(\boldsymbol{\rho}_{C} ||\boldsymbol{\rho}_{A}) + \mathcal{D}(\boldsymbol{\rho}_{C} ||\boldsymbol{\rho}_{B}) + \nonumber\\
			+ \beta\Tr{[\mathcal{L}_{A}\boldsymbol{\rho}_{A}]} + \beta\Tr{[\mathcal{L}_{B}\boldsymbol{\rho}_{B}]} + \ln_{2} Z_C \geq 0.
		\end{eqnarray}
Expansion yields, 
	 	\begin{eqnarray}
	 		-S_C + \beta\Tr{[\mathcal{L}_{A}\boldsymbol{\rho}_{C}]} +  \ln_{2}Z_{A} \nonumber\\
	 		-S_C + \beta\Tr{[\mathcal{L}_{B}\boldsymbol{\rho}_{C}]} +  \ln_{2}Z_{B} + \nonumber\\
	 		\beta \Tr{[\mathcal{L}_{A}\boldsymbol{\rho}_{A}]} + \beta \Tr{[\mathcal{L}_{B}\boldsymbol{\rho}_{B}]} + \ln_{2} Z_C\geq0.
	 	\end{eqnarray}
We further use $ \mathcal{L}_C=\mathcal{L}_A+\mathcal{L}_B $ and von Neumann entropy corresponding to a Gibbs state density matrix $ S(\boldsymbol{\rho}_{X})=  \beta \Tr{[\mathcal{L}_{X}\boldsymbol{\rho}_{X}]} +  \ln_{2}Z_{X} $, to obtain,
    \begin{eqnarray}
			S_{A} + S_{B} -2S_{C} + \ln_{2} Z_C + \beta\Tr{[\mathcal{L}_C\boldsymbol{\rho}_C]}\geq 0.\nonumber
	\end{eqnarray}
From the expansion of $S_C$, we recover Theorem \ref{thm:gl}.
\end{proof}

\section{Google page rank---a ground eigenvector problem}  

Evidently nodes in a complex network have different roles and their influence on system
dynamics varies widely depending on their topological characteristics. (this is particularly relevant in walks that break time reversal symmetry, as introduced in \cite{Zimbors2013} and experimentally probed in \cite{Lu_2016} and characterized in \cite{biamonte2017topological}).

One of the simpler (and widely applied) characteristics in network analysis~\cite{albert2002statistical} is the degree centrality, defined as the number of edges incident on that node. Many real world networks have been found to follow a widely heterogeneous distribution of degree values~\cite{albert2002statistical}.

Despite the complexity of the linking pattern, the degree distribution of a network directly (and in some regimes dominantly) affects the associated dynamics. In fact, it can be shown
that the probability of finding a memoryless random walker at a given node of a symmetric network in the stationary low-energy state, is just proportional to the degree of such a node~\cite{noh2004random}.

We will now summarize our findings to explain Figure \ref{fig:similarity}. 

In~\cite{faccin2013degree} together with coauthors, we consider the relationship between the stochastic and the quantum version of such processes, with the ultimate goal of shedding light on the meaning of degree centrality in the case of quantum networks.

We considered a stochastic evolution governed by the Laplacian matrix $\mathcal{L}_{S}=\mathcal LD^{-1}$, the stochastic generator that characterizes classical random walk dynamics and leads to an occupation probability proportional to node degree.

In the quantum version, a Hermitian generator is required and we proposed the symmetric Laplacian matrix ${\mathcal L}_{Q}=D^{-\frac 12}\mathcal L D^{-\frac 12}$, generating a valid quantum walk that, however, does not lead to a stationary state, making difficult a direct comparison between classical and quantum versions of the dynamics.
A common and useful workaround to this issue was to average the occupation probability over time, as in~\eqref{eqn:ltap2.43}. 

The generators of the two dynamics are spectral similar (see Figure~\ref{fig:similarity}) and hence share the same eigenvalues, while the eigenvectors are related by the transformation $\ket{\phi_i^C} = D^{-\frac 12}\ket{\phi_i^Q}$.
As a consequence, if the system is in the ground state the average probability to find the walker on a node will be the same as in the classic case, which will depend solely on the degree of each node.
For the cases in which the system is not in the ground state, it is possible to define a quantity
\begin{equation*}
  \varepsilon = 1 - \bra{\phi_0^Q} \rho_0 \ket{\phi_0^Q},
\end{equation*}
describing how far from the classical case the probability distribution of the quantum walker will be. In the case of uniformly distributed initial state $\rho_0$, this provides a measure for the heterogeneity of the degree distribution of a quantum network \cite{faccin2013degree}.

\begin{remark}[Directed Networks]
Many networks of practical significance are not symmetric, but directed.  How can these networks be encoded into quantum systems when considering the constraint that Hamiltonians must be Hermitian? (One approach is to consider directing transport by breaking time-reversal symmetry \cite{Zimbors2013}). 
\end{remark}

The non-symmetric adjacency matrix representing the directed connectivity of the World Wide Web, a.k.a.~the Google matrix $G$, satisfies the Perron-Frobenius theorem \cite{2012arXiv1209.3632B} and hence there is a maximal eigenvalue corresponding to an eigenvector of
positive entries $G\ket{p} = \ket{p}$. The eigenvector $\ket{p}$ corresponds to the limiting distribution of occupation probabilities of a random web surfer---it represents a unique attractor for the dynamics independently of the initial state. The vector $\ket{p}$ is known as the Page-Rank. A dumping factor is often included in the computation to assure the Perron-Frobenius theorem is satisfied. 
 
Several recent studies embed $G$ into a quantum system and consider quantum versions of Google's Page-Rank \cite{QuantumPageRank,garnerone2012pagerank, paparo2014google}. The work~\cite{Garnerone2012google} relied on an adiabatic quantum algorithm to compute the Page-Rank of a given directed network, whereas Burillo et al.~\cite{sanchez2012quantum} rely on a mixture of unitary and dissipative evolution to define a ranking that converges faster than classical PageRank.

The page-ranking vector $\ket{p}$ is an eigenvector of $\eye-G$ corresponding to the zero eigenvalue (the lowest). This fact leads to a definition of a Hermitian operator which can play the role of a Hamiltonian, defined as:
\begin{equation}
  \mathcal{H} = (\eye -G)^\dagger(\eye -G),
\end{equation}
though non-local (if $G$ is $k$-local then $h$ is $k^2$-local), its ground state represents the target Page-Rank which could be found by adiabatic quantum annealing into the ground state.  Using a quantum computer to accelerate the calculation of various network properties has been considered widely, see e.g.~the survey~\cite{2017arXiv170208459B}.

\begin{figure}[p]
  \centering
  \includegraphics[width=0.75\textwidth]{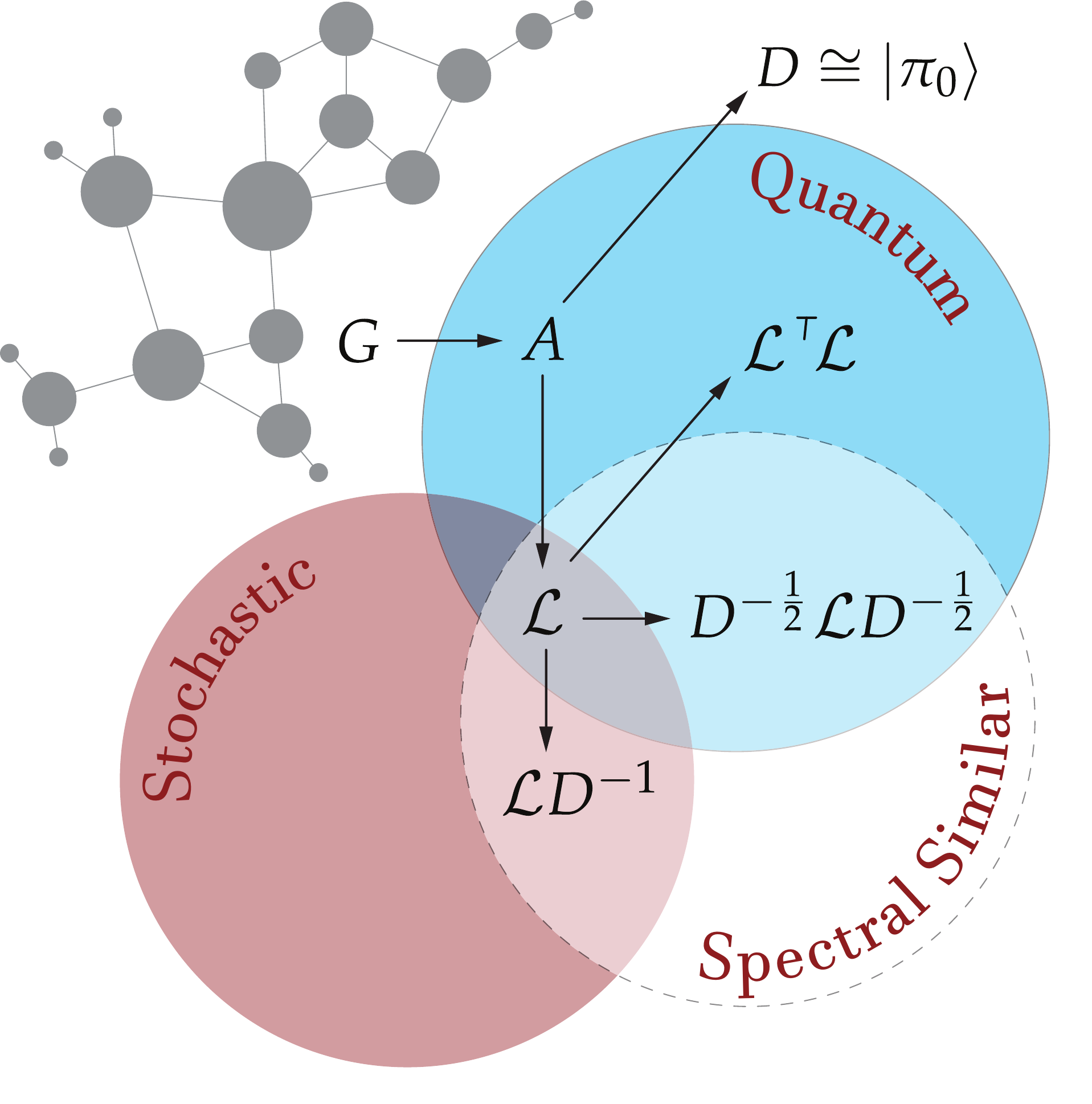}
  \caption{Known mappings between quantum and stochastic generators. Here $G=(V,E)$ is a graph with adjacency matrix $A$, $D$ is a diagonal matrix of node degrees.  This yields the graph Laplacian   $\mathcal{L}=D-A$, and hence, the stochastic walk generator $\mathcal{L}_{S}=\mathcal{L}D^{-1}$, from this a similarity transform results in $\mathcal{L}_{Q}=D^{-\frac{1}{2}}\mathcal{L}D^{-\frac{1}{2}}$, which generates a valid quantum walk and exhibits several interesting connections to the classical case.   The mapping $\mathcal{L}\longrightarrow \mathcal{L}^\top \mathcal{L}$ preserves the lowest $0$ energy ground state, opening the door for adiabatic quantum annealing which solves computational problems by evolving a system into its ground state. (Figure from \cite{2017arXiv170208459B} which was modified from \cite{faccin2013degree}). }
\label{fig:similarity}
\end{figure}

\section{Kitaev's quantum phase estimation algorithm}\label{sec:kqpea}
\index{Quantum Algorithms}\index{Phase Estimation}

Quantum algorithms have a long and rich history, with many seminal  developments.  Early milestones include the Deutsch-Jozsa algorithm~\cite{DJ1992}, the Bernstein–Vazirani algorithm~\cite{Bernstein1997}, Simon's algorithm~\cite{Simon1997} as well as others~\cite{Cleve1998}. 

The {\it quantum phase estimation algorithm} (a.k.a.~{\it quantum eigenvalue estimation algorithm}), estimates the eigenvalue of an eigenvector of a unitary operator. Loosely speaking, given a $U$ and a quantum state $\ket{\psi}$ such that $U\ket{\psi} = e^{\imath \cdot \lambda}\ket{\psi}$, the algorithm estimates $\lambda$ with high probability within additive error $\varepsilon$, using ${\mathcal O}(1/\varepsilon)$ controlled-$U$ gates.  Phase estimation is a central component (subroutine) to many other quantum algorithms, such as Shor's algorithm \cite{NC}. This algorithm has appeared in my own reseach, including \cite{WBA11, sim11, Wang2015, Lanyon2010}. 

Let 
\begin{equation}
\mathcal{H}\in \mathscr{L}(\mathbb{C}^{\otimes d}_2) \hspace{30pt} \Phi \in \mathbb{C}^{\otimes d}_2
\end{equation}
with
\begin{equation}
\mathcal{H}=\mathcal{H}^\dagger \hspace{30pt} \mathcal{H} \varphi_i = \lambda_i\varphi_i  
\end{equation}
so $\mathcal{H}$ is a finite dimensional Hamiltonian, then\\
\begin{equation}
\label{Phi}
\Phi = \sum_i e^{- \imath \lambda_i} \braket{\varphi_i}{\phi}\ket{\varphi_i} 
\end{equation} 
is a solution to Schr{\"o}dinger's equation
\begin{equation}
\label{schr_eq}
\imath \frac{d}{dt}\Phi=\mathcal{H}\Phi. 
\end{equation}
In general, eigenvectors evolve under the T.I.S.E.~as 
\begin{equation}
\varphi_i (t) = e^{-\imath t\lambda_i}\varphi_i(0). 
\end{equation}

\begin{definition} (The global phase)\label{def:global}
We say a global phase is not detectable as $|\braket{\phi}{\varphi}|^2$ is invariant under $\varphi \mapsto e^{-\imath \alpha} \varphi$, $\psi \mapsto e^{-\imath \alpha} \psi$ $  \forall \alpha, \beta \in \mathbb{R}$. Some authors will write 
for $\varphi \in \mathbb{C}^{\otimes d}_2$ and for $c \in \mathbb{C}/_{\{0\}}$,
$c\varphi \sim \varphi$ (spoken, ``is reducible to'' or hence ``equivalent'') and call this a {\it ray in Hilbert space}.  Mathematically this is a pairing $\mathbb{C}/_{\{0\}} \times \mathbb{C}^{\otimes d}_2$ which induces an equivalence $\mathbb{C}/_{\{0\}} \times \mathbb{C}^{\otimes d}_2 \sim \mathbb{C}^{\otimes d}_2$ as  
\begin{equation}
\frac{|\braket{c\phi}{d\varphi}|^2}{\braket{c\phi}{c\phi}\braket{d\varphi}{d\varphi}} =  \frac{\braket{c\phi}{d\varphi}\braket{d\varphi}{c\phi}}{\braket{c\phi}{c\phi}\braket{d\varphi}{d\varphi}} = |\braket{\phi}{\varphi}|^2
\end{equation}
for unit vectors $\varphi$ and $\phi$ where $ \braket{c\phi}{c\phi} = \overline{c}c$. Here we work in the unit normalized gauge, fixing probability such that $\braket{\varphi}{\varphi}=1$ which is indeed invariant under $\varphi \mapsto e^{-\imath\alpha} \varphi$. We will write $e^{-\imath \alpha} \varphi \sim \varphi$. 
(See also Definition \ref{def:gauge}). 
\end{definition}

\begin{proposition}
Equation \eqref{Phi} satisfies \eqref{schr_eq}. 
\end{proposition}

We indeed can't measure a global phase, but we can measure a relative one. Consider the Hamiltonian \eqref{eqn:relativeH} (see also \cite{sim11}). 

\begin{equation}\label{eqn:relativeH}
\mathcal{H} \rightarrow \ketbra{1}{1}\otimes \mathcal{H} 
\hspace{20pt} \textnormal{or} \hspace{20pt} \mathcal{H} \rightarrow P_1\otimes \mathcal{H} 
\end{equation}

Hence from \eqref{eqn:relativeH} one partitions an enlarged (double) eigenspace, into a zero subspace 
\begin{equation}
\hspace{10pt} \ker \{\ketbra{1}{1}\otimes \mathcal{H}\} = \spn \{\ket{0}\otimes \ket{q}\}
\end{equation}\\
and the eigenspace of the original $\mathcal{H}$ is the same with eigenvectors $\otimes$-concatenated by $\ket{1}$
\begin{equation}
(\ketbra{1}{1}\otimes \mathcal{H})\ket{1}\otimes \ket{\varphi_i} = \lambda_i \ket{1} \otimes \ket{\varphi_i}
\end{equation}
\noindent
We consider the T.I.S.E.
\begin{equation}
e^{-\imath  t\ketbra{1}{1}\otimes \mathcal{H}} \underbrace{\left(\frac{1}{\sqrt[]{2}}\left(\ket{0} + \ket{1}\right)\otimes \ket{\varphi_i}\right)}_{\text{initial state}}
\end{equation}
by linearity
\begin{equation}
\begin{split}
=& \frac{1}{\sqrt[]{2}} \left( \ket{0, \varphi_i} + e^{-\imath \lambda_i t} \ket{1, \varphi_i}\right) \\ =& \frac{1}{\sqrt[]{2}} ( \ket{0} + e^{-\imath \lambda_i t} \ket{1})\ket{\varphi_i} 
\end{split}
\end{equation}\\
We change to the $X$-basis with the Hadamard gate $H = \frac{1}{\sqrt{2}}\left(X + Z\right)$ and then measure in $Z$.
\begin{equation}
\begin{split}
  \frac{1}{\sqrt{2}} H( \ket{0} + e^{-\imath \lambda_i t} \ket{1})
= \frac{1}{2}( \ket{0} - e^{-\imath \lambda_i t} \ket{1})
+ \frac{1}{2}( \ket{1} + e^{-\imath \lambda_i t} \ket{0})
\end{split}
\end{equation}

The probabilities for measuring $\ket{0}$ and $\ket{1}$, respectively are (see Figure \ref{fig:logo}) 
\begin{equation}
\begin{split}
\text{Prob}(\ketbra{0}{0}) =& \frac{1}{4}(1+e^{-\imath  \lambda_i t})(1+e^{\imath \lambda_i t})\\ =& \frac{1}{4}(1+e^{-\imath \lambda_i t} + e^{\imath \lambda_i t} + 1)\\ =& \frac{1}{2}(1+\cos(\lambda_i t))
\end{split}
\end{equation}
\noindent and 
\begin{equation}
\text{Prob}(\ketbra{1}{1}) =\frac{1}{2}(1-\cos(\lambda_i t)). 
\end{equation}
By sampling from the original wavefunction and measuring $\bra{0}$ we can recover $\lambda_i$---see also Figure \ref{fig:logo}. The described procedure was the quantum phase estimation algorithm \cite{KSV02} related to how it appeared in \cite{sim11}. 




\begin{figure}[p]
  \centering 
  \includegraphics[width=0.7\textwidth]{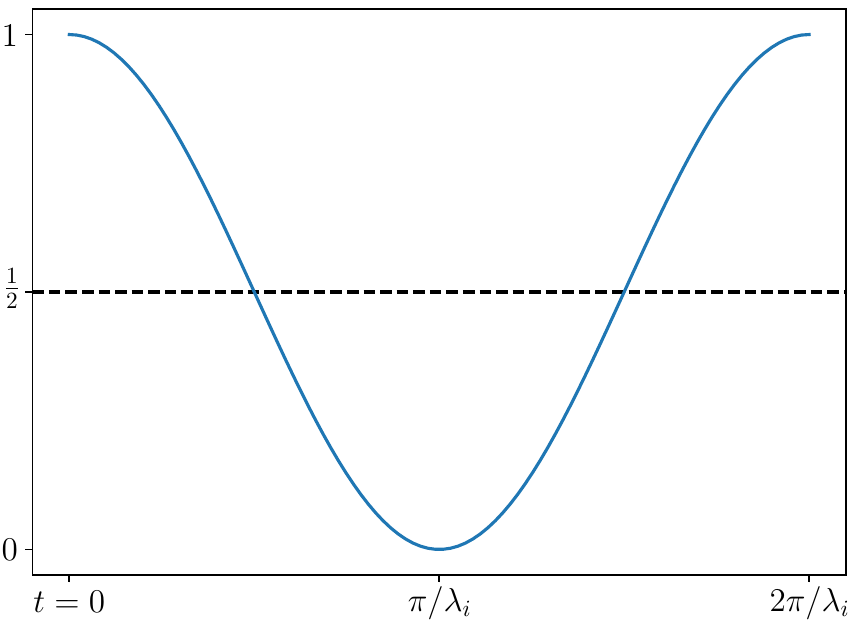}
  \caption{A plot of $\frac{1}{2} + \frac{1}{2}\cos \lambda_i t$. The plane wave $\cos \lambda_i t$ has angular frequency $\lambda_i$ with period $\frac{2\pi}{\lambda_i}$.}
\label{fig:logo}
\end{figure}

From phase estimation we can obtain the eigenvalue of a given eigenvector. The following questions now follow
\begin{enumerate}
\item How do we simulate $\mathcal{H}$ using elementary gates?
\item How do we prepare $\ket{\psi_i}$?
\end{enumerate}

For the first question, consider $ \mathcal{H} = A + B$ with $\left[ A,B \right] \neq  0 $. It can be shown that it is possible to approximate the evolution of $\mathcal{H}$.  Let us recall the following. 

\begin{definition}(Lie Product Formula). For Hermitian matrices ${\mathcal A}$ and ${\mathcal B}$,  
\begin{equation}
e^{{\mathcal A} + {\mathcal B}}=\lim_{k\rightarrow \infty} \left(e^{{\mathcal A}/k}e^{{\mathcal B}/k} \right)^k.
\end{equation}
\end{definition}

\begin{definition}[Suzuki--Trotter Expansion]
The second order operator Taylor expansion establishes that 
\begin{equation}
e^{({\mathcal A}+{\mathcal B})\delta} = e^{\delta {\mathcal A}}e^{\delta {\mathcal B}} + O(\delta^2). 
\end{equation}
\end{definition}

Similarly to the product formula, 
It is often typical that one of the terms ${\mathcal A}$ or ${\mathcal B}$ is considered as a diagonal operator in the standard computational $Z$ basis.  And that the non-diagonal second term is typically 1- or 2-body.  

We will first simulate ${\mathcal A}$. Let ${\mathcal A}$ be diagonal in $\mathscr{L}(\mathbb{C}^{\otimes d}_2)$. This term could embed e.g.~the $3$-{\sf SAT} problem \index{Boolean Satisfiability} and hence may contain 3-body interactions. To develop an example to simulate 3-body interactions consider
\begin{equation}\label{eqn:zzz}
    \alpha Z\otimes Z \otimes Z. 
\end{equation}

\begin{proposition}
The basis vector $\ket{a,b,c}$ for $a,b,c \in \mathbb{B}$ evolves under the propagator $e^{-\imath t \alpha Z_1Z_2Z_3}$ as 
\begin{equation}
e^{\imath \alpha (-1)^{a\odot b\odot c}}\ket{a,b,c}
\end{equation}
where $\odot$ denotes the operation of negated exclusive {\sf OR}, called logical equvlience ($a\odot b\odot c=1\oplus a \oplus b \oplus c$). 
\end{proposition}

\begin{proposition}
The following circuit simulates the 3-body interaction \eqref{eqn:zzz} as $e^{-\imath t \alpha Z_1Z_2Z_3}$. 
\begin{center}
\begin{circuitikz}[scale=1.5] 
\draw [fill] (7,2) circle [radius=0.1] ;
\draw (6,2) -- (8,2);
\draw (7,2) -- (7,0.8);
\draw (7,1) circle [radius=0.2] ;
\draw (6,1) -- (8,1);
\draw (6,0) -- (8,0);
\draw (8,1) -- (8,-0.2);
\draw (8,0) circle [radius=0.2] ;
\draw [fill] (8,1) circle [radius=0.1] ;
\draw (8,2) -- (9,2);
\draw (8,1) -- (9,1);
\draw (8,0) -- (8.50,0);
\draw  (8.50,0.50) rectangle (9.50,-0.5);
\draw (9, 0)  node [black] {$e^{-\imath\alpha Z}$};
\draw (9.50,0) -- (10,0);
\draw (9,2) -- (10,2);
\draw (9,1) -- (10,1);
\draw [fill] (10,1) circle 
[radius=0.1] ;
\draw (10,1) -- (10,-0.2);
\draw (10,0) circle [radius=0.2] ;
\draw (10,2) -- (11,2);
\draw (10,1) -- (11,1);
\draw (10,0) -- (11,0);
\draw (11,2) -- (11,0.8);
\draw [fill] (11,2) circle 
[radius=0.1] ;
\draw (11,1) circle [radius=0.2] ;
\draw (11,2) -- (12,2);
\draw (11,1) -- (12,1);
\draw (11,0) -- (12,0);
\draw (5.7, 2)  node [black] {};
\draw (5.7, 1)  node [black] {};
\draw (5.7, 0)  node [black] {};
\end{circuitikz}
\end{center}
\end{proposition}

Now we will focus on simulating ${\mathcal B}$, which we will assume is a sum over 2-body terms. For example terms such as $X_i X_j$. In a circuit this would be implemented as

\begin{center}
\begin{circuitikz}[scale=1.5]
\draw (0,1.8) -- (0.5,1.8);
\draw (0,0) -- (0.5,0);
\draw  (0.5,2.3) rectangle (1.5,1.3);
\draw (1, 1.8)  node [black] {$\sf H$};
\draw  (0.5,0.5) rectangle (1.5,-0.5);
\draw (1, 0)  node [black] {$\sf H$};
\draw (1.5,1.8) -- (2,1.8);
\draw (1.5,0) -- (2,0);
\draw [fill] (2,1.8) circle 
[radius=0.1] ;
\draw (2,1.8) -- (2,-0.2);
\draw (2,0) circle [radius=0.2] ;
\draw (2,1.8) -- (4,1.8);
\draw (2,0) -- (2.5,0);
\draw  (2.55,0.5) rectangle (3.5,-0.5);
\draw (3, 0)  node [black] {$e^{-\imath\alpha Z}$};
\draw (3.5,0) -- (4,0);
\draw [fill] (4,1.8) circle 
[radius=0.1] ;
\draw (4,1.8) -- (4,-0.2);
\draw (4,0) circle [radius=0.2] ;
\draw (4,1.8) -- (4.5,1.8);
\draw (4,0) -- (4.5,0);
\draw  (4.5,2.3) rectangle (5.5,1.3);
\draw (5, 1.8)  node [black] {$\sf H$};
\draw  (4.5,0.5) rectangle (5.5,-0.5);
\draw (5, 0)  node [black] {$\sf H$};
\end{circuitikz}
\end{center}
We have hence implemented the evolution for a 2-body term in the $X$ basis.  In general we note that $\alpha X_i X_j = \alpha (HZ_iH)\otimes (HZ_jH)$ and since we used $Z_iZ_j$ interactions to realize the {\sf CN} gate, we would typically implement $X_1 X_2$ terms without the need for {\sf CN} gates. 

\begin{remark}
Readers might complain that we realized a $ZZ$ interaction using two {\sf CN} gates. This method does scale up to realize many-body interactions and it is sometimes the case that direct access to the Hamiltonian (presumably to realize a controlled-$Z$ gate) is not always possible. 
\end{remark}

\begin{proposition}
It can be shown that,  
\begin{equation}
e^{-\imath t\sum J_{ij}X_iX_j}= H^{\otimes n}e^{-\imath t\sum J_{ij}Z_iZ_j}H^{\otimes n}
\end{equation}
for $H$ the Hadamard gate.  Likewise, $Z$ terms are typically enough to be transformed into other Hamiltonians using local rotations, see for instance work on simulating electronic structure Hamiltonians \cite{WBA11}. 
\end{proposition}

\section{Finding the ground state of a graph on a quantum processor}

As we introduced early in \S~\ref{chap:progGS}, the ground state problem was considered here in several forms.  
\begin{enumerate}
    \item Stochastic systems as described here, evolve to their lowest-eigenvalue state(s). 
    \item Page-rank is a ground state sampling problem, which (as explained) can be preformed on a quantum processor (such as an adiabatic processor or ground state computer). 
    \item Phase-estimation can be used to determine eigenvalues of operators on a quantum processor.  
\end{enumerate}
To prepare a canidate lowest-energy state $\ket{\psi}$ on a quantum processor, a classical iterative optimization method can be uses (as explored next in \S~\ref{chap:varintro}). If $\ket{\psi}$ is the ground state of the Hamiltonian then we consider the following optimization problem (a variational upper bound). 

\begin{equation}\label{eqn:varmin2.68}
\min_{\varphi} \bra{\varphi}{\mathcal H}\ket{\varphi} \geq \bra{\psi} {\mathcal H} \ket{\psi}
\end{equation}

Note that a solution to this optimization problem corresponds to a ground state of the Hamiltonian ${\mathcal H}$. Several approaches will be presented to solve optimization problems given in the form \eqref{eqn:varmin2.68}.




\chapter{Tensor Networks and Parameterised Quantum Circuits}
\label{chap:tensornets}

Tensor networks in physics are often traced back to a 1971 paper by Penrose \cite{Pe71} but actually date back further, appearing in various forms in the work of Cayley. \index{Tensor Networks}  Such network diagrams appear in digital circuit theory, and form the foundations of quantum computing---starting with the work of Feynman and others in the 1980's \cite{Feynman1986} and further extended by Deutsch in his `quantum computational network model'~\cite{Deutsch73}.   

Most of the modern interest in tensor networks stems from their use in numerical algorithms employing tensor contraction.  There are many reviews and surveys devoted to tensor network algorithms---see~\cite{2014AnPhy.349..117O, Vidal2010, MPSreview08, TNSreview09, 2011AnPhy.326...96S, 2010arXiv1006.0675S, Schollw, 2014EPJB...87..280O,2013arXiv1308.3318E, 2011JSP...145..891E, 2016arXiv160303039B, MAL-059, Pervishko_2019, MAL-067,2017arXiv170809213R}, as well as, {\it Tensor Networks in a Nutshell}, which I wrote with Ville Bergholm \cite{biamonte2017tensor} as well as my book, {\it Quantum Tensor Networks: a pathway to modern diagrammatic reasoning}~\cite{biamonte2019lectures}. 

The Feynman gate is a (if not the) central building block for quantum information processing tasks.  It has been used extensively in diagrammatic reasoning in categorical quantum mechanics \cite{CD, redgreen}. Let's recall the building blocks used in the tensor network construction of the Feynman gate.  Here we will review and introduce tensor networks by considering this gate. 

\begin{figure}[p]\label{fig:F2-presentation}
\begin{center}
\includegraphics[width=.9\textwidth]{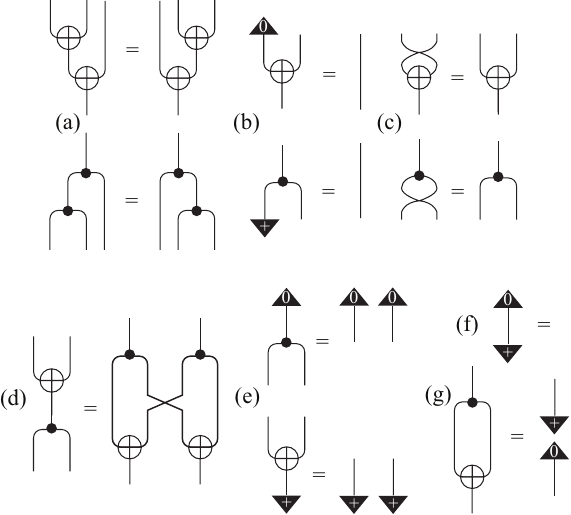}
\end{center} 
\caption{Lafont's 2003 model of circuits included the bialgebra (d) and Hopf (g) relations between the building blocks needed to form a controlled-{\sf NOT} gate---see also ZX calculus \cite{redgreen}.  [Redrawn from \cite{Lafont03towardsan} as it appeared in \cite{2011AIPA....1d2172B}].  [(a) associativity; (b) gate unit laws; (c) symmetry; (e) copy laws; (f) unit scalar given as a blank on the page.]  }
\end{figure}

As related to diagramatic reasoning in quantum theory \cite{CD, redgreen}, an algebraic theory of logic gates \cite{Lafont03towardsan} was cast into the setting of tensor network states---and used at the crossroads of condensed matter theory and quantum computation \cite{2011AIPA....1d2172B, 2011JPhA...44x5304B, 2013JPhA...46U5301B}.  We adapted these and other tools~\cite{CD, redgreen} and discovered efficient tensor network descriptions of finite Abelian lattice gauge theories \cite{2012JPhA...45a5309D}.  These tools also lead to the discovery of a wide class of efficiently contractable tensor networks, representing counting problems \cite{2015JSP...160.1389B}.  The methods (in part) trace some of their roots back to categorical quantum mechanics in work~\cite{CD, redgreen} surrounding the so called, ZX-calculus (stabilizer tensor networks).

\section{Clifford gates}
\index{Quantum Gates}
\begin{definition}
The collection of Clifford gates is generated with the following. 
\begin{enumerate}
    \item[(a)] The controlled-{\sf NOT} gate, {\sf CN}. 
    \item[(b)] The Hadamard gate $H = \frac{1}{\sqrt{2}}(X+Z)$. 
    \item[(c)] The phase gate 
$P = \ket{0}\bra{0} + \imath \ket{1}\bra{1}$. 
\item[(d)] The Pauli gates generating the Pauli algebra on a single qubits with 16 elements.
\end{enumerate}
The gates (a-d) are presented graphically in Figure \ref{fig:cGates}.     
\end{definition}

\begin{figure}[h]
    \centering
    \includegraphics[width=.7\textwidth]{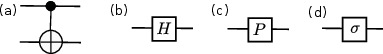}
    \caption{Clifford generators drawn as quantum gates.}
    \label{fig:cGates}
\end{figure}

\begin{remark}[Single qubit Clifford group]
The standard properties of single qubit gates follow. 
\begin{enumerate}
 \item[(i)] $HXH=Z$ and $HZH = X$
 \item[(ii)] $PXP^\dagger = Y$ and $PYP^\dagger = Z = P^2$
\end{enumerate}
These gates above generate the single qubit Clifford group.   
\end{remark}

\begin{remark}[Properties of Clifford circuits]
We now list elementary properties of Clifford circuits.  
 \begin{enumerate}
  \item[(i)] Clifford circuits generate the Clifford group.
  \item[(ii)] Let $P_n$ be the collection of $16^n$ $n$-letter words with $\otimes$ as concatenation generated from the alphabet 
  \begin{equation}
      \{\pm\eye, \pm X,\pm Y,\pm Z, \pm \imath\eye, \pm \imath X,\pm \imath Y,\pm \imath Z\}. 
  \end{equation}
All operators $g$ in the Clifford group acts as an involution when $P_n$ is conjugated by $g$, that is 
\begin{equation}
    g P_n g^\dagger = P_n.
\end{equation}
 \item[(iii)] Note, alternative notation to write $P_1$ could be 
 \begin{equation}
     P_1 = \{\pm 1, \pm \imath\}\{\eye, X, Y, Z\}.
 \end{equation}
 \item[(iv)] Defining properties of the Pauli matrices: 
 \begin{equation}
     X^2=Y^2=Z^2=\eye = -\imath XYZ.  
 \end{equation}
  \end{enumerate}
\end{remark}

\section{Tensor network building blocks}
A universal model of computation can be expressed in terms of networks (i.e.~circuits built from gates)~\cite{NC, 2011JPhA...44x5304B}.  The first gate to consider copies binary inputs ($0$ and $1$) like this 
\begin{subequations}
\begin{align}
        0 &\rightarrow 0,0\label{eqn:copy}\\
        1 &\rightarrow 1,1\label{eqn:copy2}
\end{align}
\end{subequations}
In the diagrammatic tensor network language, the {\sf COPY}-gate is 
\begin{center}
\includegraphics[width=0.1\textwidth]{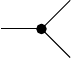}
\end{center} 
and graphically, equation \eqref{eqn:copy} and \eqref{eqn:copy2} become 
\begin{center}
\includegraphics[width=0.6\textwidth]{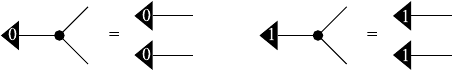}
\end{center} 
The next gate preforms the exclusive OR operation ({\sf XOR}).  Given two binary inputs (say $a$ and $b$), the output ($a\oplus b = a + b - 2 ab$) is $1$ iff exactly a single input is $1$ (that is, addition modulo 2).  The gate is drawn as 
\begin{center}
\includegraphics[width=0.1\textwidth]{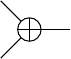}
\end{center} 
The {\sf XOR} \index{Logic Gates} gate allows one to realize any linear Boolean function.  Let $f$ be a function from $n$-long bit strings $x_1 x_2\dots x_n$ to single bits $\in \{0,1\}$. Then $f(x_1,x_2,\dots,x_n)$ is linear over $\oplus$ if it can be written as 
\begin{equation}\label{eqn:linear}
f = c_1 x_1 \oplus c_2 x_2 \oplus \cdots \oplus c_{n-1} x_{n-1} \oplus c_n x_n
\end{equation}
where ${\bf c} \bydef (c_1, c_2, \dots ,c_{n-1}, c_n)$ is any $n$-long Boolean string.  Hence, there are $2^n$ linear Boolean functions and note that negation is not allowed.  When negation is allowed a constant $c_0\oplus f$ is added (mod 2) to equation \ref{eqn:linear} and for $c_0=1$, the function is called affine. In other words, negation is equivalent to allowing constant $1$ as 
\begin{equation}\label{eqn:constants}
\includegraphics[width=0.28\textwidth]{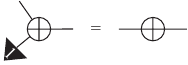}
\end{equation}
which sends Boolean variable $x$ to $1-x$.  Using the polarity representation of $f$, 
\begin{equation}
\hat f ({\bf x}) = (-1)^{f({\bf x})}
\end{equation}
we note that linear Boolean functions index the columns of the $n$-fold tensor product of $2\times 2$ Hadamard matrices (that is, $H^{\otimes n}$ where the $i$--$j$th entry of each $2\times 2$ is $\sqrt{2}H_{ij} \bydef (-1)^{i\cdot j}$). Importantly (where equality is up to a scalar), 
\begin{equation}\label{eqn:xorH}
\includegraphics[width=0.3\textwidth]{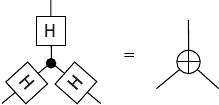}
\end{equation} 
up to isometry as there could be an omitted scale factor depending on conventions \cite{redgreen}.  By equation \ref{eqn:xorH} one can think of {\sf XOR} as being a copy operation in another basis.  We send binary $0$ to $\ket{0}\bydef(1,0)^\top$ and $1$ to $\ket{1}\bydef(0,1)^\top$ where $\top$ is transpose. Then {\sf XOR} acts as a copy operation: 
\begin{subequations}
\begin{align}
        \ket{+} &\rightarrow \ket{+,+}\label{eqn:xcopy}\\
        \ket{-} &\rightarrow \ket{-,-}\label{eqn:xcopy2}
\end{align}
\end{subequations}
using $H^2 = \eye$, $\ket{+} \bydef H\ket{0}$ and $\ket{-} \bydef H\ket{1}$.  

Concatenating the {\sf COPY}- and {\sf XOR} gates yields the logically reversible Feynman gate \cite{NC, redgreen, biamonte2019lectures}) \index{Quantum Gates}
\begin{equation}\label{fig:feynman}
\includegraphics[width=0.3\textwidth]{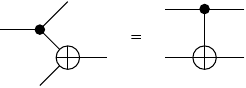}
\end{equation} 

A simplistic methodology to connect quantum circuits with indexed tensor networks starts with the definition of two tensors, in terms of components. 
\begin{center}
\includegraphics[width=0.45\textwidth]{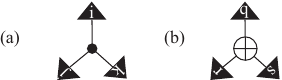}
\end{center} 
In (a) we have 
\begin{equation}
\delta^i_{~jk} = 1 - (i + j + k)+ i j + i k + j k 
\end{equation} 
where the indicies ($i$, $j$ and $k$) take values $\in \{0,1\}$. In other words, the following contractions evaluate to unity. 
\begin{center}
\includegraphics[width=0.42\textwidth]{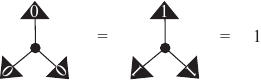}
\end{center} 
Likewise for (b) we have 
\begin{equation} 
\oplus^q_{~r s} = 1-(q+r+s)+2(qr+qs+sr)-4qrs
\end{equation} 
where the following contractions evaluate to unity (the {\sf XOR} tensor is fully symmetric, hence the three rightmost contractions are identical by wire permutation). 
\begin{center}
\includegraphics[width=0.8\textwidth]{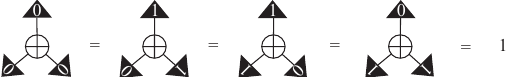}
\end{center} 
Then the Feynman gate ({\sf CN}) \index{Quantum Gates} is given as the following tensor contraction\footnote{Equation \eqref{eqn:cnot} expands to $1 - (i + j + q + r) + i j + i q + j q + i r + j r + 2 (q r - i q r - j q r)$.} 
\begin{equation}\label{eqn:cnot}
\sum_m\delta^{ij}_{~~m}\oplus_{~qr}^{m} = {\text {\sf CN}}^{ij}_{qr}
\end{equation}
where we raised an index on $\delta$.  
All quantum circuits can be broken into their building blocks and thought of as indexed tensor contractions in this way.

\subsection{Reversible logic}
A reversible computer is built using gates that implement bijective functions.  Quantum gates are unitary: hence reversible classical gates are a subclass.  Let us recall the critical implication of reversible logic.  

We will consider $n$-long bit strings in lexicographic order indexed by natural numbers $i$. So $y_0 = 00\cdots 0$, $y_2 = 00\cdots 10$ etc.~We will further consider inputs as being uniformly distributed over the $y$'s and define the change in Shannon's entropy between a circuits input and output (implementing $g$) as 
\begin{equation}\label{eqn:deltaS}
\Delta S \bydef \sum_i P\{g(y_i)\} \ln_2 P\{g(y_i)\} -  \sum_i P(y_i) \ln_2 P(y_i) 
\end{equation}
where the probability $P\{y_i\} = 2^{-n}, \forall i$ for the uniform distribution.  Equation \eqref{eqn:deltaS} vanishes identically iff $g$ is a reversible function.  

For $g$ non-reversible (a.k.a.~a non-injective surjective function), there exists at least one pair $y_i$, $y_j$ such that $g(y_j)=g(y_i)$ and hence, information is lost as the input can not be uniquely recovered from the output (so the Shannon entropy of the output distribution is strictly $<n$) and hence,  \eqref{eqn:deltaS} is non-vanishing.  The vanishing of \eqref{eqn:deltaS} is a central implication of reversible computation, and provides an abstract argument related to Landauer's principle. 

Universal classical computation can be realized with reversible logic gates.  However, using the Feynman gate is not enough since it only can be used to implement linear functions.  An additional reversible gate must be added, such as the Toffoli or Fredkin gate(s).  

\subsection*{Stabilizer tensor networks} 

Stabilizer circuits use gates from the normalizer of the $n$ qubit Pauli group---generated by the Clifford gates \cite{stabs}.  The gates include: (i) the single qubit Pauli gates $X$, $Y$ and $Z$; (ii) the Feynman gate; (iii) the Hadamard gate; (iv) the phase gate $P=\ket{0}\bra{0}+ \imath \ket{1}\bra{1}$.  
\index{Quantum Circuits} \index{Quantum Gates}

\begin{theorem}[Minimal Stabilizer Tensor Generators~\cite{B17}]\label{theorem:mstg}
The following generating tensors are sufficient to simulate any stabilizer quantum circuit: \begin{enumerate}
    \item[(a)] a vector $\ket{t}\bydef \ket{0}+\imath \ket{1}$,
    \item[(b)] the Hadamard gate and 
    \item[(c)] the {\sf XOR}- and {\sf COPY} tensors and 
    \item[(d)] a covector $\bra{+}\bydef \bra{0}+\bra{1}$.  
\end{enumerate}
\end{theorem}

\begin{proof}
We will establish this by recovering the Clifford gates (i-iv).  By linearity, the copy tensor induces a product between vector pairs, producing a third vector where the coefficients of the input vector pair are multiplied.  This allows us to recover from (a) the family $\ket{t^k}\bydef \ket{0}+ \imath^k \ket{1}$ for integer $k$ as, for instance, 
\begin{equation}
\includegraphics[width=0.3\textwidth]{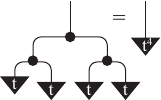}
\end{equation} 
which recovers the vector $\sqrt{2}\ket{+}$ from  \eqref{eqn:xcopy}.  Then from (d) we can recover a cup (or cap) allowing one to raise/lower indicies (\ref{eqn:tcup}.a), and importantly, (\ref{eqn:tcup}.b) illustrates that 
\begin{equation}\label{eqn:tcup}
\includegraphics[width=0.55\textwidth]{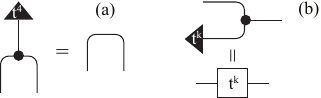}
\end{equation} 
$\ket{t^k}\bydef \ket{0}+ \imath^k \ket{1}$ lifts to a unitary operator where the Clifford gate $P$ is recovered for $k=1$---thereby establishing (iv).  Considering (\ref{eqn:tcup}.a) above, we have $t^4$ which turns into the identity operator.  The L.H.S.~of this equation (\ref{eqn:tcup}.a) then corresponds to a Bell state or costate depending on your convention.  From equation (\ref{eqn:tcup}.b) above, we create a map corresponding to $k^k$. For $k=2$ we recover the standard Pauli $Z$ matrix, then $HZH= X$ and $PXP^3 = Y$.  So we recover the Pauli gates (i).  The Feynman gate was constructed in (\ref{fig:feynman}) establishing (ii) and the Hadamard gate (iii) was assumed.   
\end{proof}

\begin{definition}
The cross on wire 
\begin{center}
    \begin{circuitikz}
    \draw (-0.5,0) -- (0.5, 0);
    \draw (-0.1,0.2) -- (0.1,-0.2);
    \draw (-0.1,-0.2) -- (0.1,0.2);
    \end{circuitikz}
\end{center}
denotes the $Z$ gate. 
\end{definition}

Before continuing on, we note that various other identities can be derived for stabilizer tensors, such as the following.\footnote{Showing that $Z_i\ket{\psi} = Z_j \ket{\psi}$ for all $i, j$ for $\ket{\psi}$ a {\sf GHZ}-state.}

\begin{center}

\begin{circuitikz}
\draw (-0.5,0) -- (0.5, 0);
\draw (-0.1,0.2) -- (0.1,-0.2);
\draw (-0.1,-0.2) -- (0.1,0.2);
\draw [fill] (0.5,0) circle [radius=0.1] ;
\draw (0.5,0) -- (1, 0.5);
\draw (0.5,0) -- (1, -0.5);

\draw (1.5, 0) node {$=$};

\draw (2,0) -- (3, 0);
\draw [fill] (3,0) circle [radius=0.1] ;
\draw (3,0) -- (3.5, 0.5);
\draw (3,0) -- (3.5, -0.5);

\path (3.3,0.3) coordinate (x);
\draw[rotate around={25:(x)}] (3.2,0.45) -- (3.4,0.15);
\draw[rotate around={25:(x)}] (3.2,0.15) -- (3.4,0.45);

\draw (4, 0) node {$=$};

\draw (4.5, 0) -- (5.5, 0);
\draw [fill] (5.5,0) circle [radius=0.1] ;
\draw (5.5,0) -- (6, 0.5);
\draw (5.5,0) -- (6, -0.5);

\path (5.8,-0.3) coordinate (xx);
\draw[rotate around={-25:(xx)}] (5.7,-0.45) -- (5.9,-0.15);
\draw[rotate around={-25:(xx)}] (5.7,-0.15) -- (5.9,-0.45);

\draw (0,1); 
\end{circuitikz}
\end{center}

\subsection{Heisenberg picture}

We consider the following unitary operator as 
\begin{equation}
U_t = e^{-\imath t \mathcal{H}}, \hspace{30pt} U_t U_{t}^{\dagger} = \eye 
\end{equation}
with $\mathcal{H}^{\dagger} = \mathcal{H}$. The time evolution of a quantum state $\ket{\varphi_0}$ is 
\begin{equation}
\ket{\varphi_t} = e^{-\imath t \mathcal{H}} \ket{\varphi_0}, \hspace{30pt} \forall t \braket{\varphi_t}{\varphi_t} =1. 
\end{equation}
This is called \textit{Schr\"odinger's picture}.

A second time evolution formalism is called \textit{Heisenberg's picture}. Suppose we have a quantum system in state $\ket{\psi}$, we apply $U$ ($U U^{\dagger} = \eye$)
\begin{equation}
UN\ket{\psi} = UNU^{\dagger}U \ket{\psi}
\end{equation}

Then the evolution of operator $N$ is given by 
\begin{equation}
\label{eq:evoN}
N \to U N U^{\dagger}
\end{equation}

\begin{enumerate}
\item We want to follow the evolution of a number of $N$'s to reconstruct the evolution of $\ket{\psi}$.
\item Evolution in \eqref{eq:evoN} is linear so we will follow a complete basis of $n \times n$ matrices.
\end{enumerate} 

We call $P$ the Pauli group. As mentioned before it contains $4 \cdot 4^n$ elements. This elements are tensor products of $X, Y, Z, \eye$ and with prefactors $\pm 1, \pm \imath$. There's a multiplicative group homomorphism 
$$MN \to UMN U^{\dagger} = (UMU^{\dagger})(UNU^{\dagger}) $$
so we can follow just a generating set of the group. A good one for the Pauli group is $\{X_1, ..., X_n, Z_1, ..., Z_n \}$.

The set of operators that leave $P$ fixed under conjugation form the normalizer called the Clifford group $C$. This group $C$ is much smaller than the unitary group on $n$-qubits, ${\mathfrak U}(2^n)$, yet contains many operations of interest. Some gates in $C$ are the Hadamard gate
$$H = \frac{1}{\sqrt{2}} (X+Z)$$
the phase gate
\begin{equation}
    P = \ket{0}\bra{0} + \imath \ket{1}\bra{1},
\end{equation} 
and the {\sf CN} gate 
$$
\mathsf{CN} \ket{i j} = \ket{i, i \oplus j}
$$
\begin{center}
\begin{circuitikz}[scale=1.5]
\draw (-2,0) circle [radius=0.002] ;
\draw [fill] (2,0) circle [radius=0.1] ;
\draw (1,0) -- (3,0);
\draw (3.4, 0) node [black] {$\ket{i}$};
\draw (2,0) -- (2,-1);
\draw (2,-0.8) circle [radius=0.2] ;
\draw (1, -0.8) -- (3,-0.8);
\draw (3.4, -0.8) node [black] {$\ket{j}$};
\draw (0.6, 0) node [black] {$\ket{i}$};
\draw (0.25, -0.8) node [black] {$\ket{i\oplus j}$};
\end{circuitikz}
\end{center}


We will derive basic circuit identities and state them graphically, as is common \cite{NC} in quantum circuits.  Here we will adopt a tensor network approach and state identities that are also tensor symmetries \cite{biamonte2019lectures}. 

\begin{proposition}[{\sf CN}$^2=\eye$] This is given graphically as follows. (See also e.g.~bialgebra law in \cite{redgreen}). 

\begin{center}
\begin{circuitikz}

\draw (1.2,1) -- (3.8,1);
\draw (1.2,0) -- (3.8,0);

\draw [fill] (2,1) circle 
[radius=0.1] ;
\draw (2,1) -- (2,-0.2);
\draw (2,0) circle [radius=0.2] ;

\draw [fill] (3,1) circle 
[radius=0.1] ;
\draw (3,1) -- (3,-0.2);
\draw (3,0) circle [radius=0.2] ;

\draw (4.5, 0.5) node {$=$};

\draw (4.9,1) -- (6.1,1);
\draw (4.9,0) -- (6.1,0);
\draw (5.5,0) circle [radius=0.2];
\draw (5.3,0) -- (5.7,0);
\draw (5.5,-0.2) -- (5.5,0.2);
\draw [fill] (5.5,1) circle [radius=0.1] ;
\draw (5.6,0.15) .. controls (5.8,0.5)  .. (5.6,0.95);
\draw (5.4,0.15) .. controls (5.2,0.5)  .. (5.4,0.95);

\draw (6.8, 0.5) node {$=$};

\draw (7.5,1) -- (9,1);
\draw (7.5,0) -- (9,0);

\draw (0,-0.5); 
\end{circuitikz}
\end{center}
\end{proposition}

\begin{proposition}[$Z^2=\eye$] This identity is given graphically as follows. 

\begin{center}
    \begin{circuitikz}

\draw (0,0) -- (1.4,0);

\draw (0.4,0.2) -- (0.6,-0.2);
\draw (0.4,-0.2) -- (0.6,0.2);
\draw (0.8,0.2) -- (1,-0.2);
\draw (0.8,-0.2) -- (1,0.2);

\draw (2, 0) node {$=$};

\draw (2.5,0) -- (3.5,0);
\end{circuitikz}  
\end{center}
\end{proposition}

\subsection{Stabilizer tensor theory}

Here we will build towards establishing Theorem \ref{theorem:gk}. 
$U$ stabilizes a quantum state $\ket{\varphi}$ iff $U \ket{\varphi} = (+1) \ket{\varphi}$. Stabilizers of $\ket{\varphi}$ form a group. Stabilizers states in terms of $n$ commuting and different operators $S$ from the Pauli algebra. The operators that generate this group are 
$$\left\{\eye, X, Y, Z, \pm \imath, \otimes \right\}  $$

\begin{example}[Single qubit stabilizers]
\
\begin{center}
  \begin{tabular}{| c | c | }
    \hline
    $X$ & $\ket{+} = \frac{\ket{0} + \ket{1}}{\sqrt{2}}$  \\ \hline
    $Y$ & $\ket{y^+} = \frac{\ket{0} + \imath \ket{1}}{\sqrt{2}}$  \\ \hline
     $Z$ & $\ket{0}$  \\ \hline
     $-X$ & $\ket{-} = \frac{\ket{0} - \ket{1}}{\sqrt{2}}$  \\ \hline
     $-Y$ & $\ket{y^{-}} = \frac{\ket{0} - \imath \ket{1}}{\sqrt{2}}$  \\ \hline
     $-Z$ &  $ \ket{1}$  \\ \hline
  \end{tabular}
\end{center}
\end{example}

\begin{example} 
The Bell state is a stabilizer state.
\begin{equation}
\sqrt{2} \ket{\Phi^{+}} = \sum_{a,b} (a \oplus \lnot b  ) \ket{a,b} = \ket{00} + \ket{11}
\end{equation}
The group that corresponds to this state is 
\begin{equation}
S = \left\{\eye \otimes \eye, X\otimes X,-Y\otimes Y,Z \otimes Z\right\}
\end{equation}
This group is of order $2^n$ and is indeed abelian.
\end{example} 





Before presenting the main theorem of this segment (Theorem \ref{theorem:gk}), we will state some helpful graphical identities. 

\begin{remark}
The {\sf COPY}-tensor has stabilizer generators $X_iX_jX_k$, $Z_i Z_j$ for $i$,  $j$, $k =1$, $2$, $3$ a qubit index. For example,  
\begin{equation}
    Z_iZ_j(\ket{000}+\ket{111})=\ket{000}+\ket{111}. 
\end{equation}
The following will present these identities graphically (Proposition \ref{prop:stabscopy}). 
\end{remark}

\begin{proposition}[Stabilizers of {\sf COPY}]\label{prop:stabscopy}
The following the the Pauli stabilizers of {\sf COPY}. 

\begin{center}

\begin{circuitikz}
\draw (-0.5,0) -- (0.5, 0);
\draw [fill] (0.5,0) circle [radius=0.1] ;
\draw (0.5,0) -- (1, 0.5);
\draw (0.5,0) -- (1, -0.5);

\draw (1.5, 0) node {$=$};

\draw (2.4,0.2) -- (2.6,-0.2);
\draw (2.4,-0.2) -- (2.6,0.2);

\draw (2,0) -- (3, 0);
\draw [fill] (3,0) circle [radius=0.1] ;
\draw (3,0) -- (3.5, 0.5);
\draw (3,0) -- (3.5, -0.5);

\path (3.3,0.3) coordinate (x);
\draw[rotate around={25:(x)}] (3.2,0.45) -- (3.4,0.15);
\draw[rotate around={25:(x)}] (3.2,0.15) -- (3.4,0.45);

\draw (4, 0) node {$=$};

\draw (4.9,0.2) -- (5.1,-0.2);
\draw (4.9,-0.2) -- (5.1,0.2);

\draw (4.5, 0) -- (5.5, 0);
\draw [fill] (5.5,0) circle [radius=0.1] ;
\draw (5.5,0) -- (6, 0.5);
\draw (5.5,0) -- (6, -0.5);

\path (5.8,-0.3) coordinate (xx);
\draw[rotate around={-25:(xx)}] (5.7,-0.45) -- (5.9,-0.15);
\draw[rotate around={-25:(xx)}] (5.7,-0.15) -- (5.9,-0.45);

\draw (6.5, 0) node {$=$};

 \draw (7, 0) -- (8, 0);
 \draw [fill] (8,0) circle [radius=0.1] ;
 \draw (8,0) -- (8.5, 0.5);
 \draw (8,0) -- (8.5, -0.5);

\path (8.3,0.3) coordinate (y);
\draw[rotate around={25:(y)}] (8.2,0.45) -- (8.4,0.15);
\draw[rotate around={25:(y)}] (8.2,0.15) -- (8.4,0.45);

 \path (8.3,-0.3) coordinate (yy);
 \draw[rotate around={-25:(yy)}] (8.2,-0.45) -- (8.4,-0.15);
 \draw[rotate around={-25:(yy)}] (8.2,-0.15) -- (8.4,-0.45);

\draw (0,1); 

\end{circuitikz}
    
\end{center}

\begin{center}

\begin{circuitikz}
\draw (0,0) circle [radius=0.2] ;
\draw (0, -0.2) -- (0, 0.2);
\draw (-0.5,0) -- (0.5, 0);
\draw [fill] (0.5,0) circle [radius=0.1] ;

\path (0.5, 0) coordinate (x);
\draw [rotate around={45:(x)}] (0.5, 0) -- (0.8, 0);
\draw [rotate around={45:(x)}] (1.2, 0) -- (1.4, 0);
\draw [rotate around={45:(x)}] (0.8, 0.2) rectangle (1.2, -0.2);
\draw (0.87, 0.35)  node [black] {$y$};
\draw [rotate around={-45:(x)}] (0.5, 0) -- (0.8, 0);
\draw [rotate around={-45:(x)}] (1.2, 0) -- (1.4, 0);
\draw [rotate around={-45:(x)}] (0.8, 0.2) rectangle (1.2, -0.2);
\draw (0.87, -0.35)  node [black] {$y$};

\draw (1.5, 0) node {$=$};

\draw (2,0) -- (2.3, 0);
\draw (2.7,0) -- (3, 0);
\draw (2.3, 0.2) rectangle (2.7, -0.2);
\draw (2.5, 0)  node [black] {$y$};
\draw [fill] (3,0) circle [radius=0.1] ;

\path (3, 0) coordinate (x2);
\draw [rotate around={45:(x2)}] (3, 0) -- (3.3, 0);
\draw [rotate around={45:(x2)}] (3.7, 0) -- (3.9, 0);
\draw [rotate around={45:(x2)}] (3.3, 0.2) rectangle (3.7, -0.2);
\draw (3.36, 0.35)  node [black] {$y$};

\draw [rotate around={-45:(x2)}] (3,0) -- (3.9, 0) (3.5,0) circle [radius=0.2] (3.5, -0.2) -- (3.5, 0.2);

\draw (4.5, 0) node {$=$};

\draw (5,0) -- (5.3, 0);
\draw (5.7,0) -- (6, 0);
\draw (5.3, 0.2) rectangle (5.7, -0.2);
\draw (5.5, 0)  node [black] {$y$};
\draw [fill] (6,0) circle [radius=0.1] ;

\path (6, 0) coordinate (x3);
\draw [rotate around={-45:(x3)}] (6, 0) -- (6.3, 0);
\draw [rotate around={-45:(x3)}] (6.7, 0) -- (6.9, 0);
\draw [rotate around={-45:(x3)}] (6.3, 0.2) rectangle (6.7, -0.2);
\draw (6.35, -0.35)  node [black] {$y$};

\draw [rotate around={45:(x3)}] (6,0) -- (6.9, 0) (6.5,0) circle [radius=0.2] (6.5, -0.2) -- (6.5, 0.2);

\end{circuitikz}
\end{center}

\begin{center}
    
\begin{circuitikz}
\draw (0,0) circle [radius=0.2] ;
\draw (0, -0.2) -- (0, 0.2);
\draw (-0.5,0) -- (0.5, 0);
\draw [fill] (0.5,0) circle [radius=0.1] ;
\draw (0.5,0) -- (1.1, 0.6);
\draw (0.5,0) -- (1.1, -0.6);

\path (0.8,-0.3) coordinate (x1);
\draw [rotate around={25:(x1)}](0.8,-0.3) circle [radius=0.2] ;
\draw [rotate around={90:(x1)}] (0.65,-0.15) -- (0.95,-0.45);

\path (0.8,0.3) coordinate (x2);
\draw [rotate around={25:(x2)}](0.8,0.3) circle [radius=0.2] ;
\draw [rotate around={90:(x2)}] (0.65,0.15) -- (0.95,0.45);

\draw (1.5, 0) node {$=$};
 \draw (2,0) -- (3, 0);
 \draw [fill] (3,0) circle [radius=0.1] ;
 \draw (3,0) -- (3.5, 0.5);
\draw (3,0) -- (3.5, -0.5);
\end{circuitikz}
\end{center}
\end{proposition}

\begin{proposition}[Gate-{\sf COPY}]
The following identity (gate-{\sf COPY}) is derived. 

\begin{center}
\begin{circuitikz}
\draw (5,0) circle [radius=0.2] ;
\draw (5, -0.2) -- (5, 0.2);
\draw (5.6,0) circle [radius=0.2] ;
\draw (5.6, -0.2) -- (5.6, 0.2);
\draw (4.5,0) -- (6, 0);
\draw [fill] (6,0) circle [radius=0.1] ;
\draw (6,0) -- (6.7, 0.6);
\draw (6,0) -- (6.7, -0.6);

\path (6.35,-0.3) coordinate (x3);
\draw [rotate around={25:(x3)}](6.35,-0.3) circle [radius=0.2] ;
\draw [rotate around={90:(x3)}] (6.2,-0.15) -- (6.5,-0.45);

\path (6.35,0.3) coordinate (x4);
\draw [rotate around={25:(x4)}](6.35,0.3) circle [radius=0.2] ;
\draw [rotate around={90:(x4)}] (6.2,0.15) -- (6.5,0.45);

 \draw (7.3, 0) node {$=$};

\draw (8,0) -- (9, 0);
\draw [fill] (9,0) circle [radius=0.1] ;
\draw (9,0) -- (9.7, 0.6);
\draw (9,0) -- (9.7, -0.6);

\path (9.35,-0.3) coordinate (x5);
\draw [rotate around={25:(x5)}](9.35,-0.3) circle [radius=0.2] ;
\draw [rotate around={90:(x5)}] (9.2,-0.15) -- (9.5,-0.45);

\path (9.35,0.3) coordinate (x6);
\draw [rotate around={25:(x6)}](9.35,0.3) circle [radius=0.2] ;
\draw [rotate around={90:(x6)}] (9.2,0.15) -- (9.5,0.45);

\draw (0,1); 
\end{circuitikz}
\end{center}
\end{proposition}

We will now state the main theorem of this segment.

\begin{remark}
The Gottesman–Knill theorem states that stabilizer circuits---circuits that only consist of gates from the normalizer of the qubit Pauli group, a.k.a.~Clifford group---can be simulated in polynomial time on a probabilistic classical computer. 
\end{remark}

Here we will construct a sequence of graphical rewrites to establish this theorem by algebraic properties of tensor contraction.  

\begin{theorem}[Graphical Proof Gottesman--Knill Theorem] \label{theorem:gk}
For $n$-qubits acted on by $L$ Clifford gates, there exists a confluent sequence of rewrites, that establishes the Gottesman--Knill theorem in ${\mathcal O}(\text{poly}(n, L))$ steps. 
\end{theorem}

\begin{remark}
More over the networks remain efficiently contractible for circuits with $\ln$(n) non-Clifford gates.
\end{remark}

Theorem \ref{theorem:gk} is proven by using the graphical tensor calculus to derive the gate identities in paper \cite{stabs}. We will establish this piece-wise, by working example.  

\begin{remark}
To prove the theorem (\ref{theorem:gk}), we need to understand how
\begin{equation}\label{eqn:systoev}
   \{Z_1, \ldots Z_n, X_1, \ldots X_n\} 
\end{equation}
evolve under conjugation by $H$, $P$, $\mathsf{CN}$. 
To simulate the evolution of, 
\eqref{eqn:systoev} there are $2n$ tensor contractions.  Each can be done graphically.  We can establish the following (The negative one above an arrow $\myarrow$ means that the mapping has an inverse).  
\begin{equation}\label{eqn:xtozbasis}
H\colon Z \myarrow X,~~~~H\colon  X\myarrow Z
\end{equation}
\begin{equation}
P\colon  X \myarrow Y, Z\myarrow Z
\end{equation}
\begin{equation}
\begin{split}
\mathsf{CN}\colon  \hspace{2pt} X\otimes \eye & \myarrow X\otimes X \\
\eye \otimes X & \myarrow \eye \otimes X \\
Z \otimes \eye & \myarrow Z \otimes\eye \\
\eye\otimes Z & \myarrow Z\otimes Z
\end{split}
\end{equation}
\begin{equation}
\begin{split}
XYX=-X \\
YZY^\dagger=-Z
\end{split}
\end{equation}
\end{remark}

\begin{proposition}[Computational and $\pm$-Basis Change] \label{prop:capmbc}

This relationship \eqref{eqn:xtozbasis} is given graphically as follows. 

\begin{center}
\begin{circuitikz}
\draw(0, 1) node{};

	\draw (0,0) -- (0.5,0);

	\draw (0.5,0.3) rectangle (1.1,-0.3);
	\draw [scale=1.1](0.75, 0)  node [black] {$\sf H$};
	\draw (1.1,0) -- (2.3,0);

	\draw (1.6, 0.2) -- (1.8, -0.2);
	\draw (1.6, -0.2) -- (1.8, 0.2);

	\draw (2.3,0.3) rectangle (2.9,-0.3);
	\draw [scale=1.1](2.4, 0)  node [black] {$\sf H$};
	\draw (2.9,0) -- (3.5,0);

\draw (4, 0) node {$=$};

	\draw (4.5,0) -- (5.5,0);
    \draw (5,0.2) -- (5,-0.2);
	\draw (5,0) circle [radius=0.2] ;
\end{circuitikz}
\end{center}

\end{proposition}

\begin{proposition}
From $H^2=\eye$ and Proposition \ref{prop:capmbc} we recover the R.H.S.~of~\eqref{eqn:xtozbasis}. 

\begin{circuitikz}

\draw(0, 1) node{};

	\draw (-1.2,0.1) node{};
	\draw (0,0) -- (0.5,0);

	\draw (0.5,0.3) rectangle (1.1,-0.3);
	\draw [scale=1.1](0.75, 0)  node [black] {$\sf H$};
	\draw (1.1,0) -- (2.3,0);

    \draw (1.7,0.2) -- (1.7,-0.2);
	\draw (1.7,0) circle [radius=0.2] ;

	\draw (2.3,0.3) rectangle (2.9,-0.3);
	\draw [scale=1.1](2.4, 0)  node [black] {$\sf H$};
	\draw (2.9,0) -- (3.5,0);

\draw (4, 0) node {$=$};

	\draw (4.5,0) -- (5,0);

	\draw (5,0.3) rectangle (5.6,-0.3);
	\draw [scale=1.1](4.85, 0)  node [black] {$\sf H$};
    \draw (5.6,0) -- (6.1,0);
    
	\draw (6.1,0.3) rectangle (6.7,-0.3);
	\draw [scale=1.1](5.85, 0)  node [black] {$\sf H$};
    \draw (6.7,0) -- (7.4,0);

	\draw (7.3, 0.2) -- (7.5, -0.2);
	\draw (7.3, -0.2) -- (7.5, 0.2);

    \draw (7.4,0) -- (8.1,0);
	\draw (8.1,0.3) rectangle (8.7,-0.3);
	\draw [scale=1.1](7.65, 0)  node [black] {$\sf H$};
   
    \draw (8.7,0) -- (9.4,0);
	\draw (9.4,0.3) rectangle (10,-0.3);
	\draw [scale=1.1](8.8, 0)  node [black] {$\sf H$};
    \draw (10,0) -- (10.6,0);
    
\draw (11.2, 0) node {$=$};    
    
    \draw (11.8,0) -- (12.8,0);

	\draw (12.2, 0.2) -- (12.4, -0.2);
	\draw (12.2, -0.2) -- (12.4, 0.2);

	\draw (12.3, 0.5) node{$Z$};
\end{circuitikz}
\end{proposition} 

\begin{remark}
The other identities follow accordingly and are encapsulated as part(s) of the following examples. 
\end{remark}

\begin{proposition}[Stabilizers Transform Covariently---see e.g.~\cite{biamonte2019lectures}]\label{prop:covarient}
Let $N$ be a stabilizer of state $\ket{o}$. Then $M N M^\dagger$ is a stabilizer of the state $(M\ket{o})$. 
\end{proposition}

The following Lemma is used as Lemma \ref{lemma:invariance} in \S~\ref{sec:variational}.  It is stated below as Lemma~\ref{lemma:normalizer} as it is relevant to Proposition \ref{prop:covarient}.   

\begin{lemma}[Clifford Normalizer] \label{lemma:normalizer} 
Let $\mathcal{C}$ be the set of all Clifford circuits on $n$ qubits, and let $\mathcal{P}$ be the set of all elements of the Pauli group on $n$ qubits.  Let $C\in\mathcal{C}$ and $P\in\mathcal{P}$ then it can be shown that 
$$
CPC^\dagger \in \mathcal{P}
$$ 
or in other words 
$$
C\left(\sigma^a_\alpha \sigma^b_\beta \cdots \sigma^c_\gamma \right)C^\dagger = \sigma^{a'}_{\alpha'} \sigma^{b'}_{\beta'} \cdots \sigma^{c'}_{\gamma'}
$$ 
and so Clifford circuits act by conjugation on tensor products of Pauli operators to produce tensor products of Pauli operators.   
\end{lemma}

To apply Lemma \ref{lemma:normalizer}, essentially, we are interested in Heisenberg evolution of $X_i$ and $Z_j$ for $i, j$ ranging over all qubit labels. Hence, we arrive at the following. Denote by $\sigma$ some Pauli observable evolving under a Clifford circuit $U$.  Then 
\begin{equation}
    U: \sigma \rightarrow \sigma' = U \sigma U^\dagger = \sigma' U U^\dagger 
\end{equation}
where from Lemma \ref{lemma:normalizer}, $\sigma'$ is a Pauli string. To prove Theorem \ref{theorem:gk}, this is exactly our strategy.  We consider the circuit $U \sigma U^\dagger$.  Then we determine (by graphical rewrites) $\sigma U^\dagger = U^\dagger \sigma'$. We show that this is always possible with the considered Clifford generators.  Hence, graphical rewrites exist to prove Theorem~\ref{theorem:gk}.

We will now consider an extended example related to Bell state stabilizers which will illustrate the building blocks needed to prove Theorem \ref{theorem:gk}.  Let us now verify by tensor contraction that ${X_1X_2,-Y_1Y_2,Z_1Z_2}$ are indeed stabilizers. The following circuit produces the Bell state $\ket{\Phi^+}$. 

\begin{center} 
	\begin{circuitikz}
	\draw (1,1) -- (2,1);

	\draw (1,0) -- (2,0);
	\draw [fill] (2,1) circle 
[radius=0.1] ;
	\draw (2,1) -- (2,-0.2);
	\draw (2,0) circle [radius=0.2] ;

	\draw (2,1) -- (2.7,1);
	\draw (2.7,1.3) rectangle (3.3,0.7);
	\draw [scale=1.1](2.73, 0.93)  node [black] {$\sf H$};
	\draw (3.3,1) -- (4,1);

	\draw (4, 1.4) -- (4, 0.6);
	\draw (4, 1.4) -- (4.5, 1);
	\draw (4, 0.6) -- (4.5, 1);
	\draw (4.18, 1) node [black] {$0$};

	\draw (2,0) -- (4,0);
	\draw (4, 0.4) -- (4, -0.4);
	\draw (4, 0.4) -- (4.5, 0);
	\draw (4, -0.4) -- (4.5, 0);
	\draw (4.18, 0) node [black] {$0$};

\draw[dotted,rounded corners=2ex] (1.5,-0.5) rectangle (3.5,1.5);
\draw (2.5, -0.8) node {$u$};

\draw (3, 2.5) node [black] {$\underset{\overbrace{\hphantom{qwertyuiopasdfg}}}{\ket{\Phi^+}} $};

\draw (4.15, 1.7) node [black] {$\underset{\overbrace{\hphantom{qwer}}}{\ket{00}} $};
\end{circuitikz}
\end{center}

Such that 
\begin{equation}
    \ket{\Phi^+} = \frac{\ket{00}+\ket{11}}{\sqrt[]{2}} = u\ket{00}. 
\end{equation}

We begin by application of Proposition \ref{prop:covarient}. First, consider stabilizers of the initial state, as follows. 

\begin{equation}
    Z_1Z_2\ket{00}=\ket{00}
\end{equation}

\begin{equation}
    Z\otimes \eye \ket{00}=\ket{00}
\end{equation}

\begin{equation}
    \eye\otimes Z\ket{00}=\ket{00}
\end{equation}

Hence, the stabilizers of $\ket{00}$ form the Abelian group \eqref{eqn:abgroup}. 
\begin{equation}\label{eqn:abgroup}
    \{\eye, Z_1, Z_2, Z_1Z_2\}
\end{equation}
Acting on the initial state $\ket{00}$ with $u$ yeilds: 
\begin{equation}
\begin{split}
u \ket{00} = \ket{\Phi^+}, \\ N\ket{\Phi^+} = \ket{\Phi^+}. 
\end{split}
\end{equation}
From Proposition \ref{prop:covarient} we know that 
\begin{equation}
\begin{split}
N\ket{00}=\ket{00},\\
u\ket{00}=uN\ket{00}=uNu^\dagger u\ket{00} = uNu^\dagger \ket{\Phi^+},\\
N\rightarrow uNu^\dagger. \\
\end{split}
\end{equation}
Hence, Proposition \ref{prop:covarient} asserts that if \eqref{eqn:abgroup} is a stabilizer of the initial state, then \eqref{eqn:abgroupcov} is a stabilizer of the final state $u\ket{00}$

\begin{equation}\label{eqn:abgroupcov}
    \{\eye, u Z_1 u^\dagger, u Z_2 u^\dagger, u Z_1Z_2 u^\dagger\}
\end{equation}

Recall Lemma \ref{lemma:normalizer}. We then must determine e.g.~$Z_1Z_2 u^\dagger = u^\dagger \sigma'$. This is done graphically as follows.  

\begin{center}
\begin{circuitikz}
\draw (0,1.5); 

\draw (-0.6, 0.5) node {$Z\otimes Z\ket{\Phi^+}:~~~~~~~~~~~~$};
\draw (0,0) -- (3,0);
\draw (0,1) -- (2,1);

\draw (0.4,0.2) -- (0.6,-0.2);
\draw (0.4,-0.2) -- (0.6,0.2);

\draw (0.4,1.2) -- (0.6,0.8);
\draw (0.4,0.8) -- (0.6,1.2);

\draw[-{Stealth[black]}] (0.8,0.1) to [bend right = 40] (1.15,0.65);

\draw [fill] (1.3,1) circle 
[radius=0.1] ;
\draw (1.3,1) -- (1.3,-0.2);
\draw (1.3,0) circle [radius=0.2] ;

\draw[-{Stealth[black]}] (2,0.1) -- (2.5,0.1);

\draw (2,1.3) rectangle (2.6,0.7);
\draw [scale=1.1](2.11, 0.93)  node [black] {$\sf H$};
\draw (2.6,1) -- (3,1);

\draw (3, 1.4) -- (3, 0.6);
\draw (3, 1.4) -- (3.5, 1);
\draw (3, 0.6) -- (3.5, 1);
\draw (3.18, 1) node [black] {$0$};

\draw (3, 0.4) -- (3, -0.4);
\draw (3, 0.4) -- (3.5, 0);
\draw (3, -0.4) -- (3.5, 0);
\draw (3.18, 0) node [black] {$0$};

\draw (4, 0.5) node {$=$};
\end{circuitikz} 
\end{center}

\begin{center}
    \begin{circuitikz}


\draw (4.5,0) -- (6.5,0);
\draw (4.5,1) -- (6.5,1);

\draw (4.7,1.2) -- (4.9,0.8);
\draw (4.7,0.8) -- (4.9,1.2);

\draw [fill] (5.3,1) circle [radius=0.1] ;
\draw (5.3,1) -- (5.3,-0.2);
\draw (5.3,0) circle [radius=0.2] ;
\draw (5.2,0.7) -- (5.4,0.4);
\draw (5.2,0.4) -- (5.4,0.7);

\draw (5.9,0.2) -- (6.1,-0.2);
\draw (5.9,-0.2) -- (6.1,0.2);

\draw (6.5,1.3) rectangle (7.1,0.7);
\draw [scale=1.1](6.18, 0.93)  node [black] {$\sf H$};

\draw (6.5,0) -- (7.5,0);
\draw (7.1,1) -- (7.5,1);

\draw (7.5, 1.4) -- (7.5, 0.6);
\draw (7.5, 1.4) -- (8, 1);
\draw (7.5, 0.6) -- (8, 1);
\draw (7.68, 1) node [black] {$0$};

\draw (7.5, 0.4) -- (7.5, -0.4);
\draw (7.5, 0.4) -- (8, 0);
\draw (7.5, -0.4) -- (8, 0);
\draw (7.68, 0) node [black] {$0$};

\draw (8.5, 0.5) node {$=$};

\draw (9,0) -- (11,0);
\draw (9,1) -- (10,1);

\draw [fill] (9.5,1) circle [radius=0.1] ;
\draw (9.5,1) -- (9.5,-0.2);
\draw (9.5,0) circle [radius=0.2] ;

\draw (10,1.3) rectangle (10.6,0.7);
\draw [scale=1.1](9.38, 0.93)  node [black] {$\sf H$};
\draw (10.6,1) -- (11,1);

\draw (11, 1.4) -- (11, 0.6);
\draw (11, 1.4) -- (11.5, 1);
\draw (11, 0.6) -- (11.5, 1);
\draw (11.18, 1) node [black] {$0$};

\draw (11, 0.4) -- (11, -0.4);
\draw (11, 0.4) -- (11.5, 0);
\draw (11, -0.4) -- (11.5, 0);
\draw (11.18, 0) node [black] {$0$};

\draw (0,-0.5); 

\end{circuitikz}
\end{center}
Which simply shows that $Z\otimes Z$ is a stabilizer of $\ket{\Phi^+}$, as expected. Consider then evolution of $Z_1$.  We arrive at the following rewrites. 

\begin{center}
\begin{circuitikz}

\draw(0, 2) node{};

	\draw (0,1) -- (1,1);
	\draw (0,0) -- (4.8,0);
    
	\draw [fill] (0.5,1) circle [radius=0.1] ;
	\draw (0.5,1) -- (0.5,-0.2);
	\draw (0.5,0) circle [radius=0.2] ;

	\draw (1,1.3) rectangle (1.6,0.7);
	\draw [scale=1.1](1.18, 0.93)  node [black] {$\sf H$};
	\draw (1.6,1) -- (2.8,1);

	\draw (2, 1.2) -- (2.2, 0.8);
	\draw (2, 0.8) -- (2.2, 1.2);

	\draw (2.8,1.3) rectangle (3.4,0.7);
	\draw [scale=1.1](2.85, 0.93)  node [black] {$\sf H$};
	\draw (3.4,1) -- (4.8,1);

	\draw [fill] (4.2,1) circle [radius=0.1] ;
	\draw (4.2,1) -- (4.2,-0.2);
	\draw (4.2,0) circle [radius=0.2] ;

\draw (5.5, 0.5) node {$=$};

	\draw (6,1) -- (8.5,1);
	\draw (6,0) -- (8.5,0);

	\draw [fill] (6.5,1) circle [radius=0.1] ;
	\draw (6.5,1) -- (6.5,-0.2);
	\draw (6.5,0) circle [radius=0.2] ;
    
    \draw (7.25,1) circle [radius=0.2] ;
    \draw (7.25,1.2) -- (7.25,0.8);

    \draw [fill] (8,1) circle [radius=0.1] ;
	\draw (8,1) -- (8,-0.2);
	\draw (8,0) circle [radius=0.2] ;

\draw (9, 0.5) node {$=$};

	\draw (10,1) -- (12.5,1);
	\draw (10,0) -- (12.5,0);
    
	\draw [fill] (10.5,1) circle [radius=0.1] ;
	\draw (10.5,1) -- (10.5,-0.2);
	\draw (10.5,0) circle [radius=0.2] ;
    
    \draw [fill] (11.3,1) circle [radius=0.1] ;
	\draw (11.3,1) -- (11.3,-0.2);
	\draw (11.3,0) circle [radius=0.2] ;
 	\draw (11.1,0.5) -- (11.5,0.5);
	\draw (11.3,0.5) circle [radius=0.2] ;
    
    \draw (12.2,1.2) -- (12.2,0.8);
	\draw (12.2,1) circle [radius=0.2] ;

\draw (13, 0.5) node {$=$};
\end{circuitikz}

\begin{circuitikz}
\draw (0,2) node{};

\draw (0,1) -- (2.5,1);
	\draw (0,0) -- (2.5,0);
    
	\draw [fill] (0.5,1) circle [radius=0.1] ;
	\draw (0.5,1) -- (0.5,-0.2);
	\draw (0.5,0) circle [radius=0.2] ;
    
    \draw [fill] (1.3,1) circle [radius=0.1] ;
	\draw (1.3,1) -- (1.3,-0.2);
	\draw (1.3,0) circle [radius=0.2] ;
    
    \draw (2.1,1.2) -- (2.1,0.8);
	\draw (2.1,1) circle [radius=0.2] ;
    
    \draw (2.1,0.2) -- (2.1,-0.2);
	\draw (2.1,0) circle [radius=0.2] ;

\draw (3, 0.5) node {$=$};

	\draw (4,1) -- (5,1);
	\draw (4,0) -- (5,0);

    \draw (4.5,1.2) -- (4.5,0.8);
	\draw (4.5,1) circle [radius=0.2] ;
    
    \draw (4.5,0.2) -- (4.5,-0.2);
	\draw (4.5,0) circle [radius=0.2] ;

\end{circuitikz}
\end{center}

Which arrives graphically at $X\otimes X$ being a stabilizer of $\ket{\Phi^+}$.  From this, Theorem \ref{theorem:gk} follows directly as the required graphical rewrites have all been established.

\chapter{Variational Quantum Search and Optimization}\label{chap:varintro}
\index{Variational Quantum Algorithms}

The topic of real-world Noisy Intermediate Scale Quantum Technology (NISQ) processors is exploding in interest. As NISQ processors operate in the presence of noise and systematic fabrication errors, novel methodologies to program and control NISQ processors are now emerging.  
Situated between a quantum simulator and a gate-model based device, a leading methodology, known as variational quantum algorithms, utilizes a classical-to-quantum feedback loop to control and tune a quantum system to produce desired output(s) \cite{2014NatCo...5E4213P, farhi2014quantum}. Interesting recent findings include the discovery of barren plateaus \cite{McClean_2018} and (the author together with coauthors) of reachability deficits \cite{2019arXiv190611259A}, a recent connection between variational algorithms and contextuality has been made in \cite{Kirby_2019}, and novel findings relating barren plateaus to circuit depth appeared in \cite{2020arXiv200100550C}. 

As variational quantum algorithms rely on minimization of some cost function, an important finding shows how to evaluate the gradients of said cost function exactly \cite{Mitarai_2018,Schuld_2019}. In addition, less general, though equally interesting work includes studying level-1 QAOA \cite{2019arXiv190507047H} and recent findings related to circuit-parameter concentrations between instance solutions \cite{2019arXiv191008187F}.

Variational quantum algorithms seek to reduce quantum state preparation requirements while necessitating measurements of individual qubits in the computational basis \cite{2014NatCo...5E4213P, farhi2014quantum}.   A sought reduction in coherence time is mediated through an iterative classical-to-quantum feedback and optimization process.  Systematic errors which map to deterministic yet unknown control parameters---such as time-variability in the application of specific Hamiltonians or poor pulse timing---can have less impact on variational algorithms, as states are prepared iteratively and varied over to minimize objective function(s).  

In the contemporary noisy intermediate-scale quantum (NISQ) enhanced technology setting \cite{2018arXiv180100862P}, variational quantum algorithms are used to prepare quantum state(s) for one of three purposes.
\begin{enumerate}
\item[(i)]  In the case of variational quantum approximate optimization ({\sf QAOA} \index{Variational Quantum Algorithms} \cite{farhi2014quantum}), a state is prepared by alternating a Hamiltonian representing a penalty function (such as the \NP{}-hard Ising embedding of \SAT) with a Hamiltonian representing local tunneling terms.  The state is measured and the resulting bit string minimizes the penalty function.  The minimization process is iterated by updating Hamiltonian application times. 

\item[(ii)] In the case of variational eigenvalue minimization (VQE \cite{2014NatCo...5E4213P}) \index{Variational Quantum Algorithms}, the state is repeatedly prepared and measured to obtain a  set of expected values which is individually calculated and collectively minimized. (applications of VQE include \cite{2018QuIP...17..223Z, 2020arXiv200500544U}) 

\item[(iii)] In the case of generative quantum enhanced machine learning \cite{2017arXiv171205304V} \index{Quantum Machiene Learning}, a quantum circuit is tuned subject to a given training dataset which is hence used to represent (non-linear) probability distributions such as $p(x) = |\psi(x)|^2/\mathcal{Z}$ or $p(x) = e^{-\mathcal{H}(x)}/\mathcal{Z}$ where $p(x)$ is the expected value of sampling bit string $x$, $\psi$ is a wave function, $\mathcal{H}$ is a Hamiltonian and $\mathcal{Z}$ is a normalization factor. 

\end{enumerate} 

In the present chapter, we will present a general framework that is expressive enough to describe any known variational algorithm.  We will however zoom in and focus on variational search and optimization, leaving a more detailed  study of universal variational quantum computation~\cite{UVQC} to Chapter \ref{sec:variational}. We begin with the idea of an approximate or random algorithm.


\section{Random vs quantum complexity} 

One-sided vs two-sided error: The answer returned by a {\it deterministic algorithm} is always expected to be correct.  That is not the case for random algorithms, including Monte Carlo algorithms.   

\begin{definition}[One-Sided Error]
 For {\it decision problems}, random algorithms are generally classified as either
 \begin{enumerate}
     \item {\scshape No}- (a.k.a.~false-) biased or 
     \item {\scshape Yes}- (a.k.a.~true-) biased, 
 \end{enumerate}
 subject to the following. 
 \begin{enumerate}
     \item A false-biased algorithm is always correct when it returns false. 
     \item A true-biased algorithm is always correct when it returns true.
 \end{enumerate}
\end{definition}

\begin{definition}[Two-Sided Error]
Probabilistic (a.k.a.~random) algorithms that have no bias are said to have {\it two-sided errors}. The answer provided by the algorithm (true or false; {\scshape Yes} or {\scshape No}) will be incorrect, or correct, with some bounded probability.
\end{definition}

\begin{example}[Randomized Primality Testing]
The Solovay–Strassen primality test is used to determine whether a given number is prime and functions as follows. 
\begin{enumerate}
    \item The algorithm always returns {\it true} for prime number inputs. 
    \item For composite inputs, the algorithm always returns {\it false} with probability at least $\frac{1}{2}$ and {\it true} with probability less than $\frac{1}{2}$. 
\end{enumerate}
Thus, {\it false} answers are certain to be correct, whereas {\it true} answers remain uncertain---this is said to be a {\it $\frac{1}{2}$-correct false-biased algorithm}.
\end{example}

Informally, a problem is in {\sf BPP} if there is a corresponding algorithm with the following properties:
\begin{enumerate}
    \item the algorithm has access to {\it coin flips} to make random decisions, 
    \item the algorithm is guaranteed to run in polynomial time (in the input size), 
    \item on any given run of the algorithm, it has a probability of at most $\frac{1}{3}$ of giving the wrong answer, whether the answer is {\scshape Yes} or {\scshape No}.
\end{enumerate}

\begin{definition}(bounded-error probabilistic polynomial time). 
{\sf BPP} is the class of decision problems solvable in polynomial time with an error probability bounded away from $\frac{1}{3}$ for all instances. 
\end{definition}

More formally as a language (decision) problem. 

\begin{definition} \label{def:bpp}
A language ${\mathcal L}$ is in {\sf BPP} if and only if there exists a computational machine ${\mathcal M}$ that runs for polynomial time on all inputs and 
\begin{enumerate}
    \item $\forall x \in {\mathcal L}$, ${\mathcal M}$ outputs $1$ with probability greater than or equal to $\frac{2}{3}$, 
    \item $\forall x \notin {\mathcal L}$, ${\mathcal M}$ outputs $1$ with probability less than or equal to $\frac{1}{3}$. 
\end{enumerate}
\end{definition}

\begin{remark}
{\sf BPP} contains {\sf P}, the class of problems solvable in polynomial time with a deterministic machine, since a deterministic machine is a special case of a probabilistic machine.
\end{remark}

In practice, an error probability of $\frac{1}{3}$ in Definition \ref{def:bpp} might not be acceptable.  The choice of $\frac{1}{3}$ in Definition \ref{def:bpp} is however, arbitrary. The probability can take any constant between 0 and $\frac{1}{2}$ (exclusive) and the set {\sf BPP} will be unchanged. 

Furthermore, the probability does not even have to be constant:
the same class of problems is defined by 
\begin{enumerate}
    \item allowing error as high as $\frac{1}{2} - n^{-c}$  or 
    \item or requiring error as small as $2^{-n^c}$,
\end{enumerate}
where $c$ is any positive constant, and $n$ is the length of input.

Hence, there is a probability of error, but if the algorithm is run many times, the chance that the majority of the runs are wrong drops off exponentially as a consequence of the {\bf Chernoff bound}. This makes it possible to create a highly accurate algorithm by merely running the algorithm several times and taking a {\it majority vote}. 

\begin{example}
If we instead defined {\sf BPP} with the restriction that the algorithm can be wrong with probability at most $2^{-100}$, this would result in the same class of problems.
\end{example}

We now arrive at {\sf BQP} \cite{KSV02},  the quantum analogue of the complexity class {\sf BPP}.

\begin{definition}
Bounded-error quantum polynomial time ({\sf BQP}) is the class of decision problems solvable by a quantum computer in polynomial time, with an error probability of at most $\frac{1}{3}$ for all instances.
\end{definition}

A decision problem is a member of {\sf BQP} if there exists a quantum algorithm that solves the decision problem with {\it high probability} and is guaranteed to run in polynomial time in the size of the input. A run of the algorithm will correctly solve the decision problem with a probability of at least $\frac{2}{3}$. More formally, {\sf BQP} is the languages associated with certain bounded-error uniform families of quantum circuits \cite{NC}. 

\begin{definition}\label{def:bqp}
 A language ${\mathcal L}$ is in {\sf BQP} if and only if there exists a polynomial-time uniform family of quantum circuits $\{Q_n:n \in \mathbb{N}\}$, such that $\forall n \in \mathbb{N}$, $Q_n$ takes $n$ qubits as input and outputs $1$ bit: 
 \begin{enumerate}
     \item $\forall x\in {\mathcal L}$, $\mathrm{Pr}(Q_{|x|}(x)=1)\geq \tfrac{2}{3}$, 
     \item $\forall x\notin {\mathcal L}$,  $\mathrm{Pr}(Q_{|x|}(x)=0)\geq  \tfrac{2}{3}$. 
 \end{enumerate}
\end{definition}

\begin{remark}
Similarly to other {\bf bounded error} probabilistic classes such as {\sf BPP}, the choice of $\frac{1}{3}$ in Definition \ref{def:bqp} is arbitrary. We can run the algorithm a constant number of times and take a majority vote to achieve any desired probability of correctness less than 1, using the aforementioned Chernoff bound. The complexity class is unchanged by allowing error as high as $\frac{1}{2} - n^{-c}$ on the one hand, or requiring error as small as $2^{-n^c}$ on the other hand, where $c$ is any positive constant, and $n$ is the length of input.
\end{remark}

\begin{definition}
 {\sf PSPACE} is the set of all decision problems that can be solved using a polynomial amount of computational space.
\end{definition}

The exact relationship of {\sf BQP} to {\sf P}, {\sf NP}, and {\sf PSPACE} remains unknown. However: 
\begin{remark}
Known class relationships include (see Figure \ref{fig:bqp}):
\begin{center}
   {\sf P} $\subseteq$ {\sf BPP} $\subseteq$ {\sf BQP} $\subseteq$ {\sf PSPACE} 
\end{center}
\end{remark}
Hence, the class of problems that can be efficiently solved by quantum computers includes all problems that can be efficiently solved by deterministic classical computers.  It however does include problems that cannot be solved by classical computers with polynomial space resources (regardless of the time taken to solve each instance). The suspected class relationships are given in Figure \ref{fig:bqp}.  


\begin{figure}[ht!]
    \centering
    \includegraphics[width=0.5\textwidth]{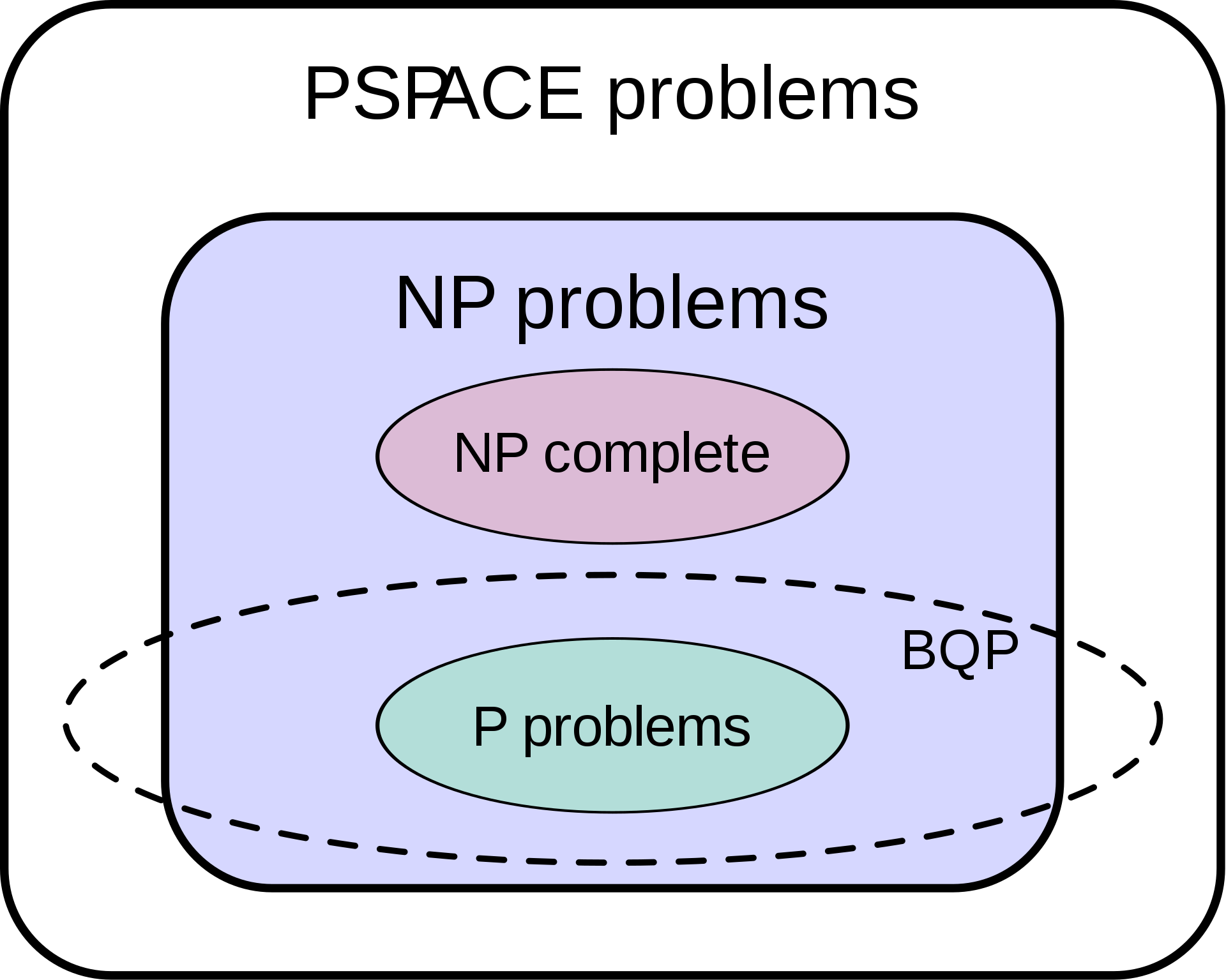}
    \caption{The suspected relationship of {\sf BQP} to other classes \cite{NC}.}
    \label{fig:bqp}
\end{figure}



\section{Variational quantum computation framework}

\begin{definition}[Variational Statespace---Biamonte \cite{UVQC}]
The variational statespace of a $p$-parameterized $n$-qubit state preparation process is the union of $\ket{\psi({\boldsymbol \theta})}$ over real assignments of ${\boldsymbol \theta}$, 
\begin{equation}
    \bigcup_{{\boldsymbol \theta} \subset \mathbb{R}^{\times p}} \{ \ket{\psi({\boldsymbol \theta})} \} \subseteq \mathbb{C}_2^{\otimes n}. 
\end{equation} 
\end{definition} 

Variational statespace examples include preparing $\ket{\psi({\boldsymbol \theta})}$ by either: 
 
 \begin{enumerate}
     \item[(1.i)]  A quantum circuit with e.g.~${\boldsymbol \theta} \in [0, 2\pi)^{\times p}$ tunable parameters as 
\begin{equation}\label{eqn:cir}
\ket{\psi({\boldsymbol \theta})} = \prod _{l=1}^pU_l \ket{0}^{\otimes n}, 
\end{equation}

where $U_l$ is adjusted by $\theta_l$ for $l=1$ to $p$. 

\item[(1.ii)] By tuning accessible time-dependent, piece-wise continious and appropriately bounded parameters (function $\theta_k(t)$ corresponding to Hermitian $A^{(k)}$) as 
\begin{equation}\label{eqn:time}
\ket{\psi} =\mathcal{T}\{ e^{-\imath \sum \theta_k(t) A^{(k)}}\}\ket{0}^{\otimes n}, 
\end{equation} 
where $\mathcal{T}$ time orders the sequence and superscript $k$ indexes the $k$th operator $A^{(k)}$.
 \end{enumerate}
  
\begin{definition}[Variational Sequence] \label{def:sequence}
A variational sequence specifies parameters to prepare a state in a variational statespace.  It can be given by defining a specific sequence of gates or by specifying control parameter values. 
\end{definition}

Contrasting (1.i) and (1.ii) above suggests an interesting connection between variational algorithms and optimal quantum control. Indeed, we will provide an objective function that can be efficiently calculated given access to a suitable quantum processor.  Minimization of this objective function produces a close 2-norm approximation to the output of a given quantum circuit.  The variational sequence minimizing the objective function is given precisely in the form (1.i Eq.~\ref{eqn:cir}).  Finding a shorter sequence to accomplish this same task \eqref{eqn:cir} by finding a variational sequence in the form \eqref{eqn:time} can be stated as optimization of a suitable objective function.  The definition of such a {\it universal} objective function, is a contribution of this work.  

\begin{definition}[Objective Function] \label{def:ofunc}
We consider an objective function as the expected value of an operator expressed with real coefficients $\mathcal{J}^{a b \cdots c}_{\alpha \beta \cdots \gamma}$ in the Pauli basis as, 
\begin{equation}\label{eqn:Ham}
  \mathcal{H} = \sum \mathcal{J}^{a b \cdots c}_{\alpha \beta \cdots \gamma} \cdot  \sigma^a_\alpha \sigma^b_\beta \cdots \sigma^c_\gamma 
\end{equation}
\index{Ising Model}
where Greek letters index Pauli matrices, Roman letters index qubits and the sum is over a subset of the $4^n$ elements in the basis, where the tensor ($\otimes$) is omitted. 
\end{definition}

We are concerned primarily with Hamiltonians where $\mathcal{J}^{a b \cdots c}_{\alpha \beta \cdots \gamma}$ is given and known to be non-vanishing for at most some $\text{poly}(n)$ terms.  This wide class includes Hamiltonians representing electronic structure \cite{2017Natur.549..242K}.  More generally, such Hamiltonians are of bounded cardinality. 

\begin{definition}[Objective Function Cardinality---Biamonte \cite{UVQC}] \label{def:cardinality}
The number of terms in the Pauli basis $\{\openone , X, Y, Z\}^{\otimes n}$ needed to express an objective function.  
\end{definition}

\begin{definition}[Bounded Objective Function---Biamonte \cite{UVQC}] \label{def:efcom}
A family of objective functions is {\it efficiently computable} when uniformly generated by calculating the expected value of an operator with $\text{poly}(n)$ bounded cardinality over 
\begin{equation} 
\begin{split}
\Omega &  \subset \{\openone , X, Y, Z\}^{\otimes n}.
 \end{split}
\end{equation}
\end{definition}

\begin{definition}[Poly-Computable Objective Function---Biamonte \cite{UVQC}]\label{def:of}
An objective function
\begin{equation}
    f\colon\ket{\phi}^{\times \mathcal{O}(\text{poly}(n))} \rightarrow \mathbb{R}_+
\end{equation}
is called poly-computable provided $\text{poly}(n)$ independent physical copies of $\ket{\phi}$ can be efficiently prepared to evaluate a bounded objective function.  
\end{definition}  

Efficiently computable objective function examples include: 
\begin{enumerate}
\item[(2.i)] Calculating the expected value of Hamiltonian's in the Pauli basis known to have bounded cardinality.  
\item[(2.ii)] Calculating the expected value of the $\mathcal{O}(\ln n)$ generated ring without inverses with $\mathcal{H}$ of bounded cardinality 
\begin{equation}
\{\mathcal{H}, \cdot, +, \mathbb{R}\}
\end{equation}  
which includes the divergence $\mathbb{E}^2 = \langle \mathcal{H}^2\rangle- \langle \mathcal{H}\rangle^2$ that vanishes if and only if the prepared state is an eigenstate of the Hamiltonian $\mathcal{H}$.   
\end{enumerate}

Acceptance (as follows) must be shown by providing a solution to the optimisation problem defined by the objective function. 

\begin{definition}[Accepting a Quantum State---Biamonte \cite{UVQC}]
An objective function $f$ {\it accepts} $\ket{\phi}$ when given $\mathcal{O}(\text{poly}~n)$ copies of $\ket{\phi}$, 
\begin{equation}
    f(\ket{\phi}^{\times \mathcal{O}(\text{poly}(n)}) = f(\ket{\phi}, \ket{\phi}, \cdots, \ket{\phi})< \Delta
\end{equation}
 evaluates strictly less than a chosen real parameter $\Delta > 0$.  
  \end{definition} 
  
\begin{remark}
Appropriately bounded objective functions (Hamiltonians) from Chapter \ref{chap:progGS} are readily seem to represent Poly-computable objective functions.  For example, consider some penalty function described by the following Hamiltonian of bounded cardinality 
\begin{equation}\label{eqn:obj}
\mathcal{H} = \sum_{i< j}J_{ij}Z_iZ_j + \sum_i h_i Z_i. \index{Ising Model} 
\end{equation}

Assume further that this Hamiltonian embeds some objective function $f({\bf x})$.  Then the quantum state $\ket{{\bf x}}$ will minimize \eqref{eqn:obj}.  This is shown by measuring \eqref{eqn:obj} in the $Z$ basis and recovering ${\bf x}$ and calculating $f({\bf x})$ classically from the description of $\mathcal{H}$. 
\end{remark}

The following theorem (\ref{thm:e2overlap}) applies rather generally to variational quantum algorithms that minimise energy by adjusting a variational state to cause an objective function to accept. Herein acceptance will imply the preparation of a quantum state, which begs to establish the following. 

\begin{theorem}[Energy to Overlap Variational Stability Lemma---Biamonte \cite{UVQC}]\label{thm:e2overlap}
Let non-negative $\mathcal{H}= \mathcal{H}^\dagger\in \mathscr{L}(\mathbb{C}_d)$ have spectral gap $\Delta$ and non-degenerate ground eigenvector $\ket{\psi}$ of eigenvalue $0$.  
Consider then a unit vector $\ket{\phi}\in \mathbb{C}_d$ such that 
\begin{equation}
\bra{\phi }\mathcal{H}\ket{\phi } < \Delta 
\end{equation} 
it follows that 
\begin{equation}
1 - \frac{\bra{\phi}\mathcal{H}\ket{\phi} }{\Delta} \leq | \braket{\phi}{\psi}|^2 \leq 1 - \frac{\bra{\phi}\mathcal{H}\ket{\phi} }{\text{Tr}\{ \mathcal{H}\}}.
\end{equation} 
\end{theorem}

\begin{proof}(Theorem \ref{thm:e2overlap}---Biamonte \cite{UVQC}).
Let $\{\ket{\psi_x}\}_{x=0}^{d-1}$ be the eigenbasis of $\mathcal{H}$.  The integer variable $x$ monotonically orders this basis by corresponding eigenvalue
\begin{equation}
\lambda_0 = 0 \leq \lambda_1 = \Delta \leq\cdots. 
\end{equation}
Consider $\ket{\phi}= \sum c_x \ket{\psi_x}$ and $P(x) = |c_x|^2$ with $c_x = \braket{\psi_x}{\phi}$.  Then 
\begin{equation}
 \bra{\phi}\mathcal{H}\ket{\phi} = \sum P(x)\lambda(x) = \sum_{\neg 0} P(x)\lambda(x)
\end{equation} 
Consider $1 = \braket{\phi}{\phi} = \sum P(x) = P(0) + \sum_{\neg 0} P(x)$.  So $1 - P(0) \geq P(x)~~ \forall x \neq 0$. Hence 
\begin{equation}
\sum_{\neg 0}P(x)\lambda(x) \leq (1-P(0))  \sum_{\neg 0}\lambda(x) 
\end{equation} 
From $\lambda(0) = 0$ we have that 
\begin{equation}
(1-P(0)) \sum\lambda(x) = (1-P(0))\cdot \text{Tr}\{\mathcal{H}\} = (1-|\braket{\psi_0}{\phi}|^2)\cdot \text{Tr}\{\mathcal{H}\}
\end{equation} 
from which the upper bound follows. 

For the lower bound, starting from the assumption $\bra{\phi}\mathcal{H}\ket{\phi} < \Delta$ we consider a function of the integers $\tilde{\lambda}(x)\cdot \Delta = \lambda(x)$ $\forall x>0$.  Then 
\begin{equation}
\bra{\phi}\mathcal{H}\ket{\phi} = \sum P(x)\lambda(x) = \sum_{\neg 0} P(x)\lambda(x) =  \Delta\cdot \sum_{\neg 0} P(x)\tilde{\lambda}(x)
\end{equation} 
hence $\sum_{\neg 0} P(x)\tilde{\lambda}(x) < 1$ as $\tilde{\lambda}(x)\geq 1$, $\forall x \geq 0$. We then have that 
\begin{equation}
\sum_{\neg 0} P(x) \leq \sum_{\neg 0} P(x)\tilde{\lambda}(x) < 1 
\end{equation} 
and so $1 - \sum_{\neg 0} P(x)\tilde{\lambda}(x)  \geq 0$.  Hence, 
\begin{equation}
1 = \braket{\phi}{\phi} = P(0) + \sum_{\neg 0} P(x) \leq P(0) + \sum_{\neg 0} P(x)\tilde{\lambda}(x) 
\end{equation} 
We establish that 
\begin{equation} 
1 \leq P(0) + \sum_{\neg 0} \frac{P(x)\lambda(x)}{\Delta} = |\braket{\psi_0}{\phi}|^2 + \frac{\bra{\phi}\mathcal{H}\ket{\phi}}{\Delta}
\end{equation} 
which leads directly to the desired lower bound.   
\end{proof} 

\section{Variational quantum search} 
\index{Variational Quantum Algorithms}
Here we will consider a two-level variant of the alternating operator ansatz following \cite{2018arXiv180509337M}.  The procedure can be readily adapted to 3-{\sf SAT} (as first developed in \cite{farhi2014quantum}).  Quantum search has been studied in the variational framework in several recent studies, including \cite{2020PhRvA.101c2346Z}. 

\begin{definition}
In this section, we let
\begin{enumerate}
    \item $n$ be the number of qubits and 
    \item $N=2^n$ be the size of the search space, 
    \item where we are searching for a particular bitstring $\bm{\omega} = \omega_1, \omega_2, \omega_3, ..., \omega_n$.
\end{enumerate}
\end{definition}

\begin{remark}[Defining \textsf{QAOA} Hamiltonians]\label{remark:qaoaops}
As in \cite{farhi2014quantum} and elsewhere, \textsf{QAOA} works by sequencing a pair of operators.  One driver (defined in the $X$-basis typically) and the other defined as the classical penalty function we seek to minimize.  
\end{remark}

To specify the quantum algorithm, as in \cite{2018arXiv180509337M} we consider a pair of rank-$1$ projectors \eqref{def:omega-op}.
\begin{align}
  \label{def:omega-op}
  P_{\bm{\omega}} &= \ketbra{\bm{\omega}}{\bm{\omega}}\\
  \label{def:plus-op}
  P_{\bm{+}} &= \ketbra{+}{+}^{\otimes n} = \ketbra{\bm{s}}{\bm{s}}
\end{align}
where 
\begin{equation}
    \ket{\bm{s}}=\frac{1}{\sqrt{N}}\sum_{\bm{x}\in \left\{0,1\right\}^n} \ket{\bm{x}}. 
\end{equation}

To find $\ket{\bm{\omega}}$, we consider a split-operator variational ansatz, formed by sequencing a pair of operators (see Remark \ref{remark:qaoaops}). These operators prepare a state $\ket{\varphi(\bm{\alpha}, \bm{\beta})}$, with vectors $\bm{\alpha} = \alpha_1, \alpha_2, ..., \alpha_p$ and $\bm{\beta}=\beta_1, \beta_2, ..., \beta_p$. We seek to minimize the orthogonal complement of the subspace for the searched string \eqref{eq:projector-w}.
\begin{align}
\label{eq:projector-w}
P_{\bm{\omega}^{\perp}} &= \eye - P_{\bm{\omega}}
\end{align}
\index{Variational Quantum Algorithms}

\begin{remark}
We sometimes call \eqref{def:plus-op} the driver Hamiltonian or diffusion operator. 
\end{remark}

The state is varied to minimize this orthogonal component \eqref{eq:minimization}.
\begin{align}
\label{eq:minimization}
\min_{\bm{\alpha},\bm{\beta}} \bra{\varphi(\bm{\alpha}, \bm{\beta})}P_{\bm{\omega}^{\perp}} \ket{\varphi(\bm{\alpha},\bm{\beta})} \geq \min_{\ket{\phi}} \bra{\phi}P_{\bm{\omega}^{\perp}}\ket{\phi}
\end{align}
To prepare the state we develop the sequence \eqref{state_prep}.
\begin{align}
\label{state_prep}
\ket{\varphi(\bm{\alpha},\bm{\beta})} = \mathcal{K}(\beta_p) \mathcal{V}(\alpha_p)\cdots\mathcal{K}(\beta_1)\mathcal{V}(\alpha_1)\ket{\bm{s}}
\end{align}
Where the operators are defined as \eqref{eq:operator-v} and \eqref{eq:operator-k}.
\begin{align}
\label{eq:operator-v}
\mathcal{V}(\alpha) &\bydef e^{\imath \alpha P_{\bm\omega}}\\
\label{eq:operator-k}
\mathcal{K}(\beta) &\bydef e^{\imath \beta P_{\bm{+}}}
\end{align}
The length of the sequence is $2p$, for integer $p$.
We consider now the following problems on which the variational algorithm will work.

\begin{example}[Standard Oracle, Variational Diffusion]\label{def:standar-oracle-variational-diffusion}
Using a computer program, one can find $p$ angles $\bm\beta = (\beta_1, ..., \beta_p)$ and fixing $\bm\alpha = (\alpha_1 = \pi, ..., \alpha_p = \pi) $ to minimize \eqref{eq:minimization} via the sequence \eqref{state_prep}, given the operators \eqref{eq:operator-v} and \eqref{eq:operator-k}.
\end{example}

In this problem (Example \ref{def:standar-oracle-variational-diffusion}) we have fixed the standard black-box oracle of Grover's algorithm and the algorithm optimizes to find the angles in the diffusion operator. We will also consider a restricted variational problem where all the diffusion operators must apply the same variational angle.

\begin{example}[Standard Oracle, Restricted Variational Diffusion---Morales, Tlyachev and Biamonte \cite{2018arXiv180509337M}]
\label{def:standard-oracle-restricted-variational-diffusion} As in Example \ref{def:standar-oracle-variational-diffusion}, using a computer program one can find $p$ angles $\bm\beta = (\beta, ..., \beta)$ and choosing $\bm\alpha = (\alpha_1 = \pi, ..., \alpha_p = \pi)$.
\end{example}

A third problem to which we will apply the variational algorithm is considering both the oracle and the diffusion angles as variational parameters. We consider in this case a phase matching condition, meaning that angles are restricted to be equal.

\begin{example}[Restricted Variational Oracle and Diffusion---Morales, Tlyachev and Biamonte \cite{2018arXiv180509337M}]
\label{def:restricted-variational-oracle-diffusion} As Example \ref{def:standar-oracle-variational-diffusion} we will use a computer program to find $2p$ angles $(\bm\alpha, \bm\beta) = (\alpha_1, ..., \alpha_p, \beta_1, ..., \beta_p)$ with the restriction  $\bm\beta = \bm\alpha = \alpha_1, \alpha_2, ..., \alpha_p$ and $\alpha_1 = \alpha_2 = ... = \alpha_p$.
\end{example}

Finally we consider variations of the oracle angles of the oracle and the diffusion operator separately.

\begin{example}[Variational Oracle and Diffusion---Morales, Tlyachev and Biamonte \cite{2018arXiv180509337M}]
\label{def:variationa-oracle-diffusion} As problem \ref{def:standar-oracle-variational-diffusion} we will use a computer program to find $2p$ angles $(\bm\alpha, \bm\beta) = (\alpha_1, ..., \alpha_p, \beta_1, ..., \beta_p)$.
\end{example}

\begin{remark}
We also call this last variation (Example \ref{def:variationa-oracle-diffusion}) a two-level split operator ansatz (following Morales, Tlyachev and Biamonte \cite{2018arXiv180509337M}). Note that the angles obtained in \eqref{eq:minimization} only minimize the selected cost function for a particular number of qubits. Once the number of qubits change, the angles obtained in the minimization do not necessarily give the highest probability to find the searched item. Also it's important to note that these angles are independent of $\bm\omega$, if we fix the number of qubits in the problem and run the algorithm with a particular set of angles, then these angles give the same probability no matter the $\bm\omega$ we are looking for. As stated earlier our objective is to see if variational algorithms are able to recover Grover's algorithm, for this we need a way of comparing both algorithms.
\end{remark}

Following Morales, Tlyachev and Biamonte \cite{2018arXiv180509337M}, to compare these variational algorithms with Grover's algorithm, consider the two-level split-operator ansatz case for $p=1$. Here we recover Grover's operators as the optimal solution for finding a particular string. To prove this, first note that there is only one pair of angles $(\alpha, \beta)$, so we consider \eqref{eq:operator-v}
and \eqref{eq:operator-k} directly. Since $\ketbra{\bm\omega}{\bm\omega}$ is a projector we can expand \eqref{eq:operator-v}. 
\begin{equation}
\begin{split}
\mathcal{V}(\alpha) &= e^{\imath \alpha \ket{\bm\omega}\bra{\bm\omega}} = \eye + (e^{\imath \alpha} - 1)\ket{\bm\omega}\bra{\bm\omega}
\\
&= \eye - (e^{\imath \widetilde{\alpha}} + 1)\ket{\bm\omega}\bra{\bm\omega}
\end{split}
\end{equation}
Where in the last step we have defined $\widetilde{\alpha} = \alpha - \pi $. Now we expand the unitary for the driver Hamiltonian \eqref{eq:driver}.
\begin{equation}
\begin{split}
\label{eq:driver}
\mathcal{K}(\beta)&=  e^{\imath \beta \ket{\bm s}\bra{\bm s}}\\
&= H^{\otimes n}(\eye + (e^{\imath \beta} - 1)\ket{\bm 0}\bra{\bm 0}) H^{\otimes n}\\
&\sim  H^{\otimes n}(-\eye + (e^{\imath \widetilde{\beta}} + 1)\ket{\bm 0}\bra{\bm 0}) H^{\otimes n}\\
&= (e^{\imath \widetilde{\beta}}+1)\ket{\bm s}\bra{\bm s} - \eye
\end{split}
\end{equation}
Where $\sim$ relates the equivalence class of operators indiscernible by a global phase, $H$ is the Hadamard gate and $\widetilde{\beta}= \beta  - \pi$. Notice that for $\widetilde{\alpha} = \widetilde{\beta} = 0$ Grover's oracle and diffusion operators are recovered.

Again following Morales, Tlyachev and Biamonte \cite{2018arXiv180509337M}, to see that the variational search includes Grover's operators for the case $p > 1$, let us impose $\alpha_1 = \alpha_2 = ... = \alpha_p$ and $\beta_1 = \beta_2 = ... = \beta_p$. In Figure~\ref{fig:oracle} and Figure~\ref{fig:diffusion} the circuits for the oracle and the diffusion operator respectively are shown.  If $i$ pairs of operators \eqref{eq:operator-v} and \eqref{eq:operator-k} are applied to the initial state $\ket{\bm{s}}$ as in \eqref{state_prep}, then we write the prepared state as \eqref{state2-1}. 
\begin{align}
\label{state2-1}
\ket{\varphi_i} &= A_i\frac{1}{\sqrt{N-1}}\sum_{x\neq \bm\omega} \ket{x} + B_i \ket{\bm\omega}
\\
\label{recursion1}
A_{i+1} &= \left(1-\frac{2}{N}\right)A_i - 2\frac{\sqrt{N-1}}{N}B_i
\\[6pt]
\label{recursion2}
B_{i+1} &= 2\frac{\sqrt{N-1}}{N}A_i - \left(1-\frac{2}{N}\right)B_i
\end{align}
We can relate the amplitudes of one step with the amplitudes of the next step with a recursion such as those that appear in \eqref{recursion1} and \eqref{recursion2}; we express the effect of the operators for the variational search over the state as a matrix \eqref{eq:matrix}.
\begin{align}
\label{eq:matrix}
\begin{pmatrix}
    1+\dfrac{a(N-1)}{N}    & -a(b+1)\dfrac{\sqrt{N-1}}{N}  \\[9pt]
    -a\dfrac{\sqrt{N-1}}{N}   & (b+1)\left(1+\dfrac{a}{N}\right) 
\end{pmatrix}
\end{align}
Here $a = e^{\imath \alpha} -1$ and $b = e^{\imath \beta} -1$. Notice that for $a = b = -2$ the same relation between amplitudes at different steps in \eqref{recursion1} and \eqref{recursion2} up to a global phase in the definition of the Grover operators is obtained. Thus, the variational search space includes Grover's original algorithm. From this matrix it is also possible to see that if the target state is changed, then the angles found through the variational algorithm will give the same probabilities.

\begin{remark}[Comparison with Grover---Morales, Tlyachev and Biamonte \cite{2018arXiv180509337M}]
Grover's search is provably optimal in the large size limit.  The success probability of Grover's algorithm goes from unity for two-qubits, decreases for three- and four-qubits and returns near unity for five-qubits then oscillates ever-so-close to unity, reaching unity in the infinite qubit limit. The variational approach employed here found an experimentally discernible improvement of 5.77\% and 3.95\% for three- and four-qubits respectively as shown in Figure \ref{fig:alpha3d} ~\cite{2018arXiv180509337M}. 
\end{remark}


\begin{table}[p] \setlength{\tabcolsep}{2.5pt}
\centering
\begin{tabular}{ c c c c } 
 \hline
 \hline
 $N$ & $100 \times (P_{\textrm{variational}} - P_{\textrm{Grover}})/P_{\textrm{Grover}}$ & step $p_{\textrm{max}}$  & angle\\ 
 \hline
 $2^3$ & 5.77\% & 2 & 2.12$^\text{rad}$  \\ 
 $2^4$ & 3.95\% & 3 & 2.19$^\text{rad}$\\ 
 $2^5$ & 0.08\% & 4 & 2.76$^\text{rad}$ \\
 $2^6$ & 0.34\% & 6 & 2.60$^\text{rad}$ \\
 \hline
 \hline
\end{tabular}
 \caption{Percentage increase between the highest probability for finding the solution after measurement obtained in Grover and the two-level variational ansatz \cite{2018arXiv180509337M}. Percent given as a function of $N = 2^n$ where $n$ is the number of qubits and at step $p_{\textrm{max}}$ on which the probability of finding the solution string is maximum. The same table is obtained for the two-level split-operator ansatz with one angle or with $2p$ angles. Both the diffusion and oracle use the same angle. (Table originally from \cite{2018arXiv180509337M}). 
 } 
 \label{table:comparison}
\end{table}

\begin{figure}[p!] 
\begin{center}
\mbox{
\Qcircuit @C=2.5em @R=3em {
& \gate{X^{1-\omega_1}} & \ctrl{3} & \gate{X^{1-\omega_1}} & \qw \\
& \gate{X^{1-\omega_2}} & \control \qw & \gate{X^{1-\omega_2}} & \qw \\
& \vdots &  & \vdots &  \\
& \gate{X^{1-\omega_n}} & \gate{R_{\alpha}} $e^{i\alpha}$ \qw & \gate{X^{1-\omega_n}} & \qw
}
}
\index{Quantum Circuits}
\caption{Oracle circuit used in Grover search and its variational incarnation.}
\label{fig:oracle}
\end{center}
\end{figure}
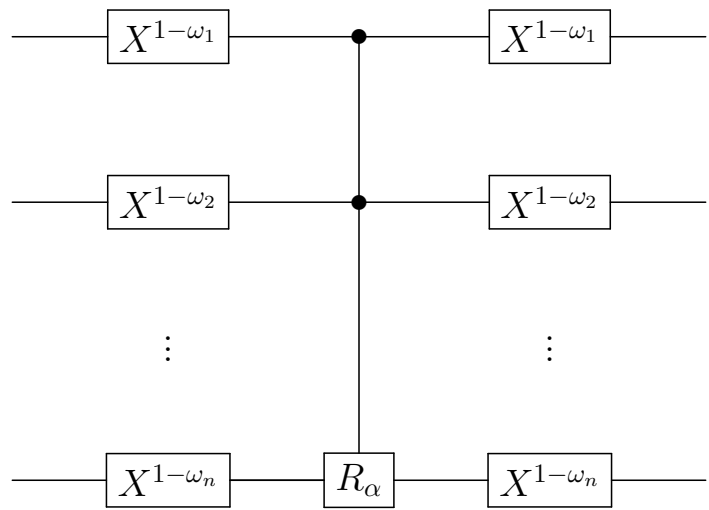

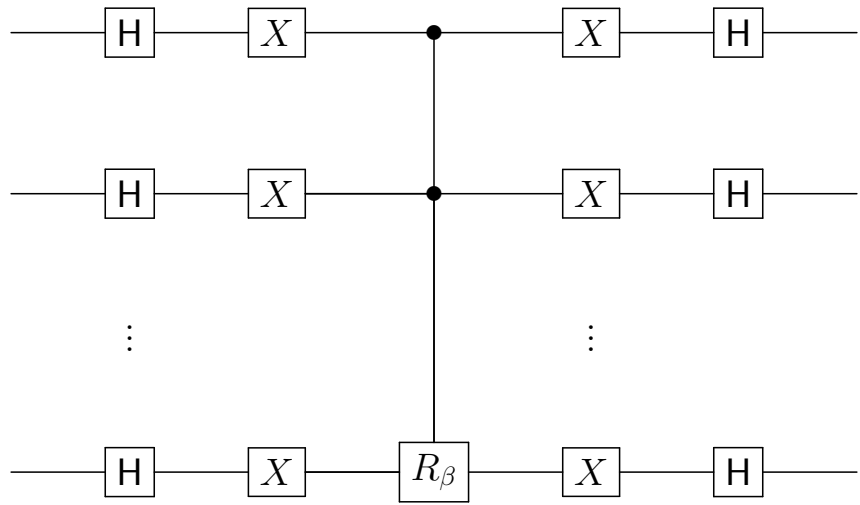
\begin{figure}[p!]
\begin{center}
\mbox{
\Qcircuit @C=2.5em @R=3em {
& \gate{\sf H} &\gate{X} & \ctrl{3} & \gate{X} &  \gate{\sf H} & \qw \\
& \gate{\sf H} &\gate{X} & \ctrl{2} \qw & \gate{X} & \gate{\sf H} & \qw \\
& \vdots & &  & \vdots &  \\
& \gate{\sf H} &\gate{X} & \gate{R_{\beta}} $e^{i\alpha}$ \qw & \gate{X} & \gate{\sf H} & \qw
}
} 
\index{Quantum Circuits}
\caption{Diffusion circuit used in Grover search and variational Grover search.}
\label{fig:diffusion}
\end{center}
\end{figure}

Realization of diffusion and oracle circuits (Figures \ref{fig:oracle} and \ref{fig:diffusion}). Oracle and diffusion operators can be rewritten via $n$-body control gates as
\begin{equation}
    \mathcal{V}(\alpha) = \bigotimes_{i=1}^{n}X_i^{1-\omega_i}\biggr(\eye_{2^n-1}\oplus e^{\imath \alpha}\biggl)\bigotimes_{i=1}^{n}X_i^{1-\omega_i}
\end{equation}
and
\begin{equation}
    \mathcal{K}(\beta) = H^{\otimes n}X^{\otimes n}\biggr(\eye_{2^n-1} \oplus e^{\imath \beta} \biggl)X^{\otimes n}H^{\otimes n}
\end{equation}
and therefore can be realized using $O(n^2)$ basic gates \cite{Barenco1995}. Here operator $\eye_{2^n-1}$ is the $(2^n-1)\times(2^n-1)$ identity matrix. (See also the gate realizations in \cite{3qbit-experiment} which can be readily bootstrapped to realize the circuits in Figures \ref{fig:oracle} and \ref{fig:diffusion}).

\begin{figure}[ht!]
\begin{subfigure}[b]{0.49\textwidth}
        \includegraphics[width=\textwidth]{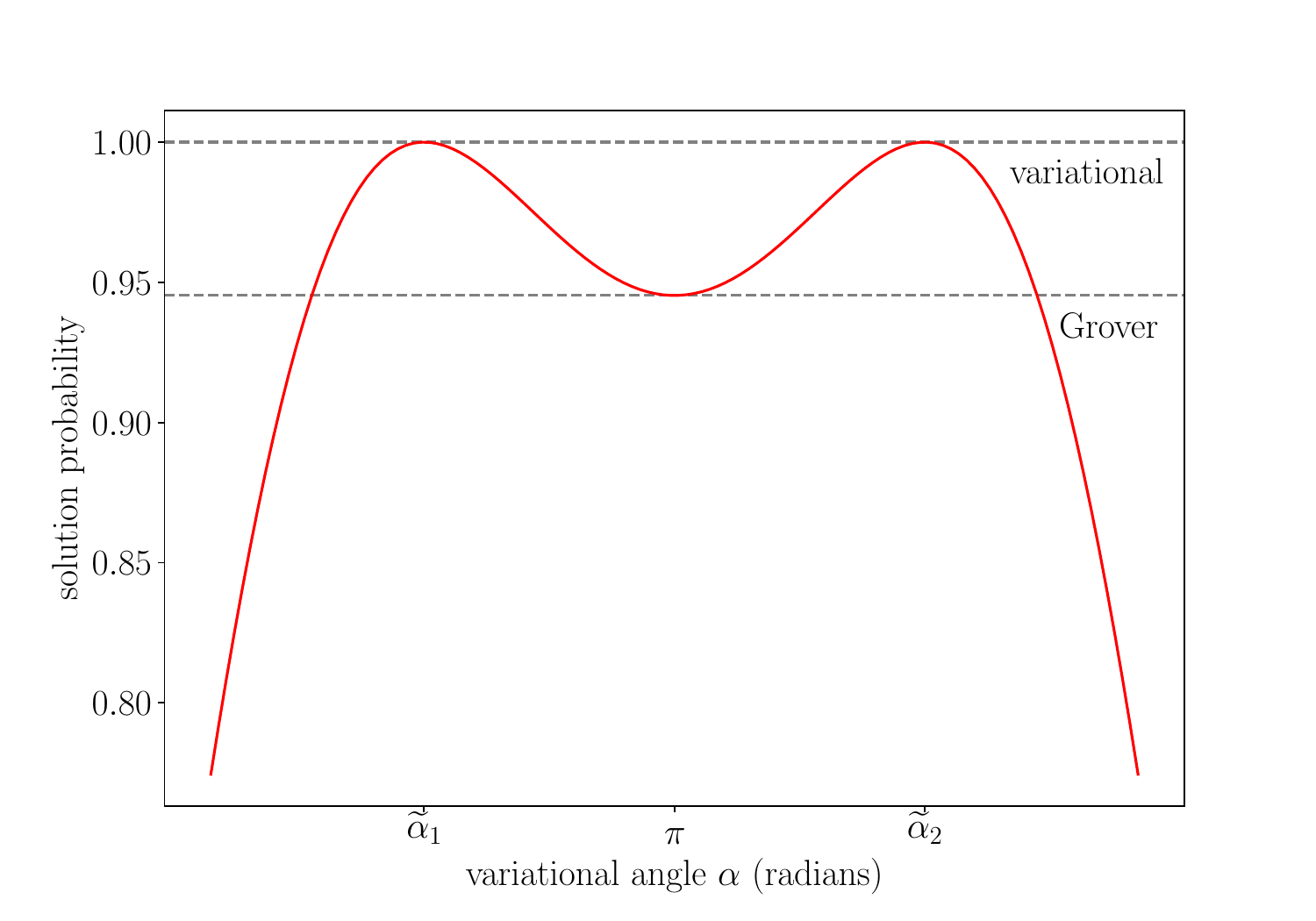}
        \caption{}
    \end{subfigure}
    \begin{subfigure}[b]{0.49\textwidth}              
        \includegraphics[width=\textwidth]{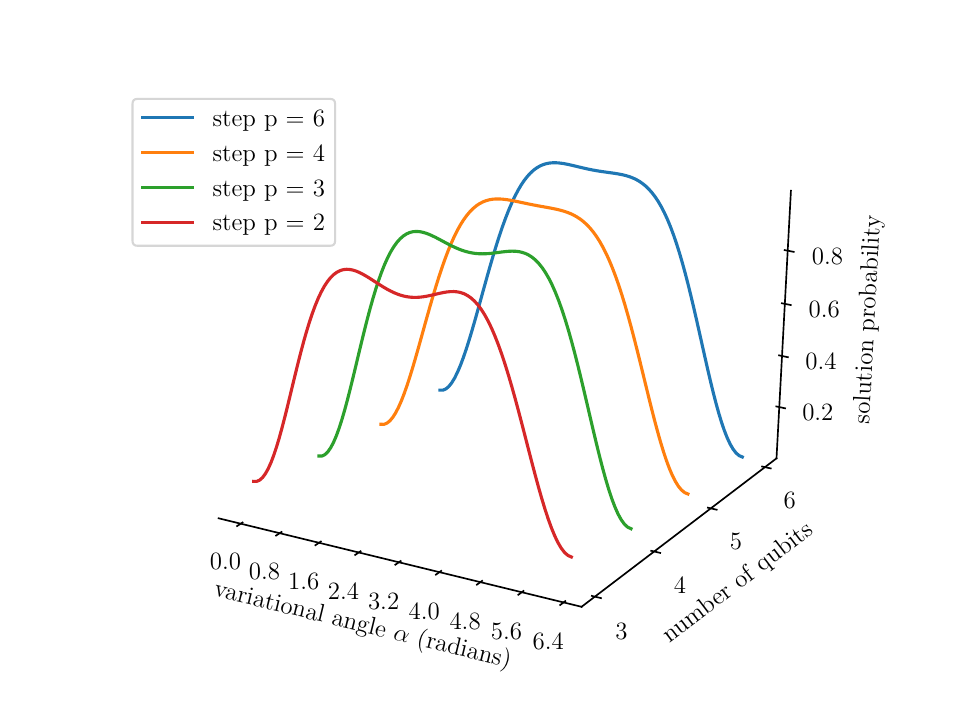}
        \caption{}
    \end{subfigure}
    \caption{(color online) (left) Grover's algorithm takes a saddle point between two hills.  Variational search recovers the hill peaks.  Note that the valley becomes increasingly less pronounced past four qubits, providing negligible range for improvement. (right) Probability as a function of the variational angle for the 3 qubit case. Grover's algorithm is recovered in the case $\alpha = \pi$, the variational algorithm obtains angles $\widetilde{\alpha}_{1} = 2.12^{\text{rad}}$  and $\widetilde{\alpha}_2 = 2\pi - \widetilde{\alpha}_{1}$. (Figure origonally from \cite{2018arXiv180509337M}).  }
    \label{fig:alpha3d}
\end{figure}

\section{On Kitaev's $k$-controlled $U$ factorization} 
\index{Universal Quantum Gates}

We have employed $k$-controlled gates when considering {\sf QAOA} (see Figure \ref{fig:oracle} and \ref{fig:diffusion}). A classical problem of importance is how to realize these gates as a sequence of two-body gates.  Here we were inspired by work appearing in the book \cite{KSV02} and with some small but notable improvements present a competitive method to realize such gates.  We quantify the scaling of the method exactly.  

\begin{remark}
The decomposition of arbitrary $n$-qubit gates into universal gate sets has been studied extensively \cite{Barenco1995,Vidal2004}.  Here we develop just a decomposition for $k$-controlled unitary gates.  Our method is inspired by a Toffoli and related factorization(s) appearing in the book \cite{KSV02}---see Figure \ref{fig:decomposition-2}. 
\end{remark}

\begin{remark}
The results of the method are competitive with the best contemporary findings.  See for example the results in \cite{Yu2013} which argue that a Toffoli gate requires a total of five two-qubit gates. 
\end{remark}


We will begin by recalling some basic facts about the Pauli group which will later will be used to define a new factorization for $k$-controlled $X$ gates.

\begin{remark}[Pauli group algebra]
The complete properties of the Hermitian Pauli (group-) algebra can be derived from the following identity:
\begin{equation}
    \sigma_i \sigma_j = \eye \delta_{ij} + \imath \epsilon^{ijk}\sigma_k
\end{equation}
Where $\sigma_i$ corresponds to an element in the Pauli group algebra. 
In particular it is easy to see that for $i\neq j$
\begin{equation}\label{eq:anticom}
    \{\sigma_i, \sigma_j\} = 0, 
\end{equation}
where $\{\cdot , \cdot\}$ is the anti-commutator.
\end{remark}


\begin{definition}[Group commutator]\label{def:taction}
Consider unitaries $U, K\in  \mathscr{L}\big(A\big)$. Then the group commutators of $U$ and $K$  is defined as 
\begin{equation}
    [U, K] = U K U^\dagger K^\dagger. 
\end{equation}
\end{definition} 

\begin{lemma} \label{lem:conj}
For $i\neq j$ and $t \in \mathbb{R}$
\begin{equation}\label{eq:conj}
    \sigma_i e^{-\imath t \sigma_j}\sigma_i = e^{\imath t \sigma_j}. 
\end{equation}
\end{lemma}

\begin{proof}
Consider that
\begin{equation}
    \sigma_i e^{-\imath t \sigma_j}\sigma_i = \cos (t) \eye - \imath \sin (t) \sigma_i \sigma_j \sigma_i.  
\end{equation}
Using the anti-commutator property of \eqref{eq:anticom}, we establish that $\sigma_i \sigma_j \sigma_i  = -\sigma_j$. Hence, 
\begin{equation}
    \sigma_i e^{-\imath t \sigma_j}\sigma_i = \cos (t) \eye + \imath \sin (t) \sigma_j = e^{\imath t \sigma_j}. 
\end{equation}
\end{proof}

\begin{lemma}\label{lem:dag} 
For $i\neq j$ and using \eqref{eq:conj}
\begin{equation}
    [\sigma_i, e^{-\imath t \sigma_j}] = \sigma_i e^{-\imath t \sigma_j}\sigma_i e^{\imath t \sigma_j} = e^{\imath 2 t \sigma_j}. 
\end{equation}
\end{lemma}

\begin{proof}
From Lemma \ref{lem:conj}, we establish that
\begin{equation}
   \sigma_i e^{-\imath t \sigma_j}\sigma_i e^{\imath t \sigma_j} = e^{\imath  t \sigma_j} e^{\imath  t \sigma_j}   =e^{\imath 2 t \sigma_j}. 
\end{equation}
\end{proof}

Now that we have established the necessary facts of the Pauli group algebra, we define the notation that we will use for the decomposition.

\begin{definition}[Controlled gates]\label{def:ctrl-notation}
Denote by $U^i_j$ the gate $U$ acting on qubit $j$ and controlled by qubit $i$.  Specifically let $X$ be the {\sf NOT} gate and $V$ be its square-root. Then $X^i_j$ is the {\sf CN} gate and $V^i_j$ is the controlled-$V$ gate (it is also true that $V^i_j = \sqrt{X}^i_j$). We will denote qubit with label $j$ to initially be in state $\ket{j}$. If the control qubits are given by indices $i_1$, ..., $i_k$ and the target qubit is $j$, then we denote this gate as $U^{i_1 i_2 \cdots i_k}_j$.
\end{definition}

\begin{remark}[Unwanted phase factor]
One needs to keep in mind that $X, V \notin ~{\bf SU}(2)$. Indeed, $\det X = -1, \ \det V = \imath$. Exponentials as $e^{\imath t \sigma_i}$ all belong to ${\bf SU}(2)$ and cannot emulate $X$ or its roots without adding an extra multiplier. For this reason, there will be trailing phase multipliers in what follows.  Kitaev's gate choice also had this feature \cite{KSV02}.
\end{remark}

Before fully defining our decomposition, we start by showing a specific case, the Toffoli gate (Figure \ref{fig:decomposition-2}) and the $k$-controlled $X$ gate (Figure \ref{fig:decomposition-k}).

\begin{lemma}[Toffoli gate]\label{lem:toffoli}
The Toffoli gate, or using notation of Definition \ref{def:ctrl-notation}, the gate $X^{a_1 a_2}_b$,  has a factorization into four control gates as 
\begin{equation}
    X^{a_1 a_2}_b = V^{a_1}_b Z^{a_2}_b (V^\dagger)^{a_1}_b Z^{a_2}_b.
\end{equation}
The circuit for the decomposition of the Toffoli gate is shown in Figure \ref{fig:decomposition-2}. Note that from the previous decomposition, we can also give a factorization for the gate $Z^{a_1 a_2}_b$ by multiplying from and left and right by a Hadamard gate acting on qubit $b$ as  
\begin{equation}
Z^{a_1 a_2}_b = \sqrt{Z}^{a_1}_b X^{a_2}_b  (\sqrt{Z}^\dagger)^{a_1}_b X^{a_2}_b.  
\end{equation}
\end{lemma}

\begin{proof}
All gates in the factorization act in part on the same qubit $b$. We just need to check the resulting operator over qubit $b$ given the different combinations for $a_1,a_2 \in \{0,1\}^2$. It is clear that for $a_1,a_2 \in \{00,01,10\}$ the result is an identity operation on qubit $b$. To prove that $VZV^\dagger Z = X$, consider Lemma \ref{lem:dag}. The result follows by replacing $\sigma_i = Z$, $\sigma_i = X$ and $t = \pi/4$.
\end{proof}

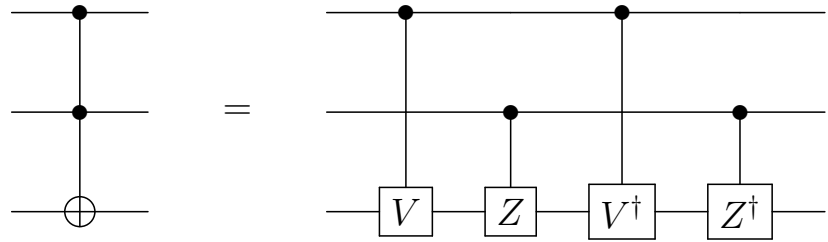
\begin{figure}[ht!]
\begin{center}
\mbox{
\Qcircuit @C=1.4em @R=1.2em @!R {
& \ctrl{1} & \qw & & & \ctrl{2} & \qw & \ctrl{2} & \qw & \qw \\
& \ctrl{1} & \qw & \push{\rule{.3em}{0em}=\rule{.3em}{0em}} & &  \qw &  \ctrl{1} & \qw & \ctrl{1} & \qw \\
& \targ & \qw    & & & \gate{V} &  \gate{Z} & \gate{V^{\dagger}} & \gate{Z^{\dagger}} & \qw
}
}
\caption{Kitaev decomposition for 2-controlled (Toffoli) gate. The structure of the decomposition appeared in \cite{KSV02}, where the gates $(V, Z)$ are chosen here specifically as they satisfy the group commutator relation. [Note that the decomposition introduces a factor of $\imath$ in front of the controlled $X$. This requires slight modification of algorithms in applications.]
}
\label{fig:decomposition-2}
\end{center}
\end{figure}

\begin{lemma}\label{lem:general decomposition X}
For a k-controlled $X^{a_1,a_2,...,a_k}_b$ gate, we have the following decomposition. Assume first that $k$ is even, thus $k=2q$ with $q$ a natural number. Then consider the following factorization \eqref{eqn:factor1}.
\begin{equation}\label{eqn:factor1}
    X^{a_1...a_{2q}}_b = V^{a_1...a_q}_b Z^{a_{q+1}...a_{2q}}_b (V^\dagger)^{a_1...a_q}_b Z^{a_{q+1}...a_{2q}}_b 
\end{equation}
If $k$ is odd, then $k = 2q+1$. In this case we consider the following factorization \eqref{eqn:factor2}.
\begin{equation}\label{eqn:factor2}
    X^{a_1...a_{2q+1}}_b = V^{a_1...a_q}_b Z^{a_{q+1}...a_{2q+1}}_b (V^\dagger)^{a_1...a_q}_b Z^{a_{q+1}...a_{2q+1}}_b 
\end{equation}

The circuits for $3$- and $k$-controlled $X$ gate are shown in Figures \ref{fig:decomposition-3} and \ref{fig:decomposition-k}.

\begin{figure}[ht!]
\begin{center}
\mbox{
\Qcircuit @C=1.4em @R=1.2em @!R {
& \ctrl{1} & \qw & & & \ctrl{1} & \qw & \ctrl{1} & \qw & \qw \\
& \ctrl{1} & \qw & \push{\rule{.3em}{0em}=\rule{.3em}{0em}} & &  \ctrl{2} &  \qw & \ctrl{2} & \qw & \qw \\
& \ctrl{1} & \qw    & & & \qw &  \ctrl{1} & \qw & \ctrl{1} & \qw \\
& \targ & \qw    & & & \gate{V} &  \gate{Z} & \gate{V^{\dagger}} & \gate{Z^{\dagger}} & \qw
}
}
\caption{Decomposition for 3-controlled gate. [Note that the decomposition introduces a factor of $\imath$ in front of the controlled $X$. This requires slight modification of algorithms in applications.]
}
\label{fig:decomposition-3}
\end{center}
\end{figure}
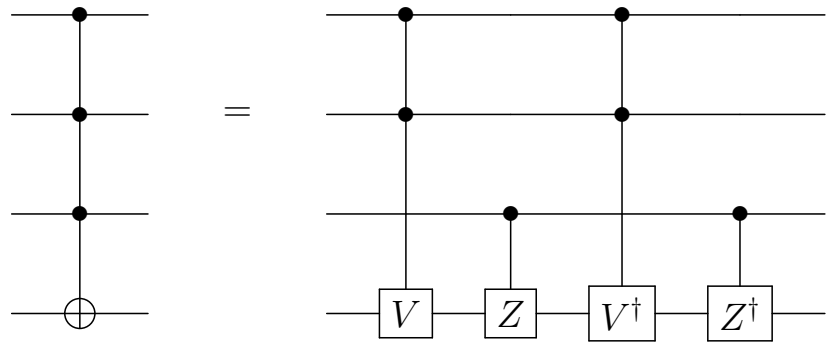

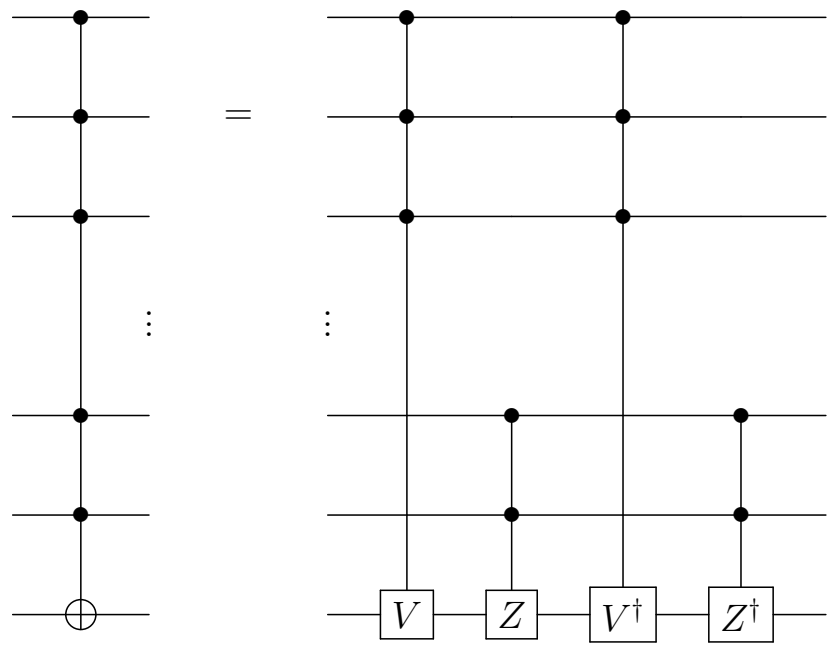
\begin{figure}[ht!]
\begin{center}
\mbox{
\Qcircuit @C=1.4em @R=1.2em @!R {
& \ctrl{1} & \qw    & & & \ctrl{1} & \qw & \ctrl{1} & \qw & \qw \\
& \ctrl{1} & \qw &  \push{\rule{.3em}{0em}=\rule{.3em}{0em}} & &  \ctrl{1} &  \qw & \ctrl{1} & \qw & \qw \\
& \ctrl{2} & \qw    & & & \ctrl{4} &  \qw & \ctrl{4} & \qw & \qw \\
&  & \vdots & & \vdots & \\
& \ctrl{1} & \qw         & & & \qw & \ctrl{1} & \qw & \ctrl{1} & \qw\\
& \ctrl{1} & \qw         & & & \qw & \ctrl{1} & \qw & \ctrl{1} & \qw\\
& \targ & \qw            & & & \gate{V} &  \gate{Z} & \gate{V^{\dagger}} & \gate{Z^{\dagger}} & \qw
}
}
\caption{Decomposition for $k$-controlled gate. [Note that the decomposition introduces a factor of $\imath$, yielding a controlled $\imath X$. This requires slight modification of algorithms in applications.]
}
\label{fig:decomposition-k}
\end{center}
\end{figure}

\end{lemma}
\begin{proof}
The proof follows in the same way as in Lemma \ref{lem:toffoli} but considering all possible bitstrings $a_1 ... a_k \in \{0,1\}^k$.
\end{proof}

Now we show the decomposition proposed for general $k$-controlled rotations in the $X$ and $Z$ axis (up to of course a controlled phase factor $\imath$). Note that having defined a decomposition for this rotations, we can define a general decomposition for any $k$-controlled unitary.

\begin{lemma}\label{lem:general decomposition rotation}

For a k-controlled $X$ rotation acting on qubit $b$, $R_x(\theta)^{a_1,a_2,...,a_k}_b$, a factorization for $k$ even follows as \eqref{eqn:decomp1} where $k=2q$.
\begin{equation}\label{eqn:decomp1}
    R_x(\theta)^{a_1..a_{2q}}_b = R_x(\theta/2)^{a_1...a_q}_b Z^{a_{q+1}...a_{2q}}_b R_x(-\theta/2)^{a_1...a_q}_b Z^{a_{q+1}...a_{2q}}_b   
\end{equation}
For the $k$ odd case we have the following \eqref{eqn:decomp2}.
\begin{equation}\label{eqn:decomp2}
    R_x(\theta)^{a_1...a_{2q+1}}_b = R_x(\theta/2)^{a_1...a_q}_b Z^{a_{q+1}...a_{2q+1}}_b R_x(-\theta/2)^{a_1...a_q}_b Z^{a_{q+1}...a_{2q+1}}_b  
\end{equation}
\end{lemma}
\begin{proof}
As in Lemma \ref{lem:general decomposition X}, we check for all possible bitstrings $a_1 ... a_k \in \{0,1\}^k$. The important case is when $a_1...a_k = 111...1$. Lemma \ref{lem:dag} proves that
$$R_x(\theta) = R_x(\theta/2) Z R_x(-\theta/2) Z $$
which verifies the given decomposition.
\end{proof}

\begin{remark}
From these decompositions \eqref{eqn:decomp1} and \eqref{eqn:decomp2}, we can also obtain similar results for $R_z(\phi)^{a_1,a_2,...,a_k}_b$.
\end{remark}

\begin{lemma}
The number of $1$-controlled gates needed to implement a $k$-controlled $X$ gate is given by the function $g(k)$ such that  $g(1)=1$ and defined recursively \eqref{eq:recursion_formula} or equivalently 
\eqref{eq:recursion_formula2}.
\begin{equation} \label{eq:recursion_formula}
g(k)=\begin{cases}4g(k/2)&k \text{ even}\\ 2g(\lfloor k/2 \rfloor) + 2g(\lfloor k/2 \rfloor +1)&k \text{ odd}\end{cases}
\end{equation}
\begin{equation}\label{eq:recursion_formula2}
    g(k) = 2g(\lfloor k/2 \rfloor) + 2g(\lceil k/2 \rceil)
\end{equation}
\end{lemma}

\begin{proof}
We just need to prove this by induction. For the initial condition we have that the $2$-controlled gates are counted correctly. Assume the property holds for $1$, $2$, \ldots, $k-1$ and let us prove it for $k$. We want to prove that $g(k)$ counts the number of $2$-controlled gates needed to implement the $k$-controlled $X$ gate. By Lemma \ref{lem:general decomposition rotation} we arrive at the decomposition for $k$ even \eqref{eqn:decomp1} where $k=2q$.
\begin{equation}\label{eqn:decomp1b}
    R_x(\theta)^{a_1...a_{2q}}_b = R_x(\theta/2)^{a_1...a_q}_b Z^{a_{q+1}...a_{2q}}_b R_x(-\theta/2)^{a_1...a_q}_b Z^{a_{q+1}...a_{2q}}_b  
\end{equation}
For $k$ odd, we arrive at the following \eqref{eqn:decomp2}.
\begin{equation}\label{eqn:decomp2b}
    R_x(\theta)^{a_1...a_{2q+1}}_b = R_x(\theta/2)^{a_1...a_q}_b Z^{a_{q+1}...a_{2q+1}}_b R_x(-\theta/2)^{a_1...a_q}_b Z^{a_{q+1}...a_{2q+1}}_b 
\end{equation}
For $k$ even,  $g(k) = 4 g(k/2)$. Comparing with the decomposition from Lemma \ref{lem:general decomposition rotation}, we see that it counts correctly two-body gates. The same reasoning applies for $k$ odd.
\end{proof}

\begin{lemma}[Solution of recursion \eqref{eq:recursion_formula2}]
The solution for the recurrence of $g(k)$, defined in  \eqref{eq:recursion_formula} has the form 
\begin{equation}
f(k) = 3k2^{\lfloor log_2 k \rfloor} - 2^{1+2\lfloor log_2 k \rfloor}. 
\end{equation} 
\end{lemma}

\begin{proof}
Note that if $k=2^p$ then $g(k) = f(k) = k^2$. Now we show that when $2^p \leq k \leq 2^{p+1}$ we have $g(k) = f(k) = 3k2^p - 2^{1+2p}$. 

To prove this for $g$, we proceed again by strong induction. The base case is trivial. Assume then that $k$ is even. Then, 
\begin{equation}
    g(k) = 4g(k/2) = 4 [3(k/2)^{p-1} - 2^{1+2(p-1)}] = 3k2^{p} - 2^{1+2p}. 
\end{equation}
When $k$ is odd, we have
\begin{equation} \label{eqn:rec1}
\begin{split}
g(k) &=  2 g(\lfloor l/2 \rfloor) + 2 g(\lfloor k/2 \rfloor +1)  \\
 & = [3(k-1)2^{p-1} - 2^{2p}]+  [3(k+1)2^{p-1} - 2^{2p}] =  3k2^{p} - 2^{1+2p}. 
\end{split}
\end{equation}
Which establishes the property for $g$. When $2^p \leq k \leq 2^{p+1}$, then
$$f(k) = 3k2^{p} - 2^{1+2p} $$
and thus $g(k)=f(k)$.
\end{proof}

\begin{lemma}[Lower bound for scaling of $g$]
The function $g(k)$ defined as in \eqref{eq:recursion_formula} has an asymptotic lower bound that scales as $\Omega (k^2)$.
\end{lemma}

\begin{proof}
Let $g(k)$ be defined as \eqref{eq:recursion_formula}. We proceed by strong induction. Note first that $1^2 \leq g(1)$ and $2^2 \leq g(2)$. Assume that the property holds for $g(k-1), g(k-2), \cdots, g(1)$. Consider the cases $k$ even and $k$ odd separately.
For $k$ even:  
\begin{equation}
    g(k)= 4g(k/2) \geq 4 (k/2)^2 = k^2. 
\end{equation}
For $k$ odd: 
\begin{equation}
    g(k) \geq 2 (\lfloor k/2 \rfloor)^2 + 2 (\lfloor k/2\rfloor +1)^2 \geq 2 (\lfloor k/2 \rfloor)^2 + k^2 \geq l^2. 
\end{equation}
Since for every $k>0$, $g(k) \geq k^2$ we conclude $g(k)=\Omega (k^2)$.
\end{proof}

\begin{lemma}[Upper bound for scaling of $g$]
The function $g(k)$ defined as in \eqref{eq:recursion_formula} has an asymptotic upper bound that scales as $\mathcal{O}(k^2)$.
\end{lemma}

\begin{proof}
We use the solution of the recursion $f(k) = 3k2^{\lfloor log_2 k \rfloor} - 2^{1+2\lfloor log_2 k \rfloor}$ to bound the scaling for the number of gates. Notice that
\begin{equation}
    f(k) = 3k2^{\lfloor log_2 k \rfloor} - 2^{1+2\lfloor log_2 k \rfloor} \leq  3k2^{log_2 k } - 2^{1+2( log_2 k - 1) } = \frac{5}{2}k^2. 
\end{equation}
Thus $f(k)= \mathcal{O}(k^2)$.
\end{proof}

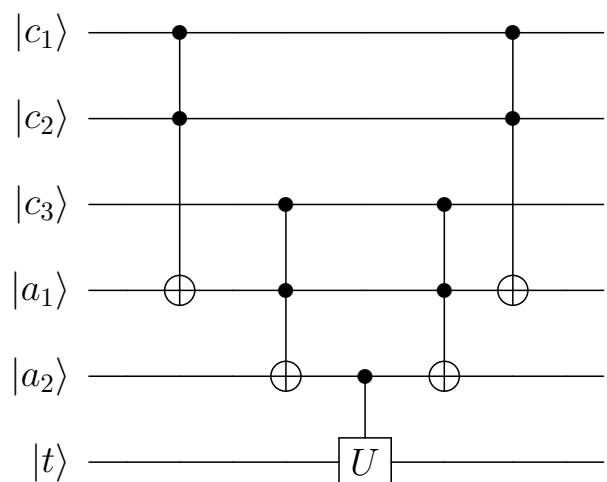
\begin{figure}[ht!]
\begin{center}
\mbox{
\Qcircuit @C=1em @R=1em @!R {
 \lstick{\ket{c_1}} & \qw & \ctrl{1} & \qw & \qw & \qw & \qw & \ctrl{1}  & \qw& \qw \\
 \lstick{\ket{c_2}} & \qw & \ctrl{2} & \qw & \qw & \qw & \qw & \ctrl{2} & \qw & \qw\\
 \lstick{\ket{c_3}} & \qw & \qw & \qw &  \ctrl{1} & \qw & \ctrl{1} & \qw & \qw & \qw \\
 \lstick{\ket{a_1}} & \qw & \targ & \qw & \ctrl{1} & \qw & \ctrl{1} & \targ & \qw & \qw \\
 \lstick{\ket{a_2}} & \qw & \qw & \qw & \targ & \ctrl{1} & \targ & \qw & \qw & \qw\\
 \lstick{\ket{t}}   & \qw & \qw & \qw & \qw & \gate{U} & \qw & \qw & \qw & \qw
}
}
\caption{Standard decomposition~\cite{NC} for 3-{\sf CN} gate which uses ancillary qubits labeled with states $\ket{a_1}$ and $\ket{a_3}$.}
\label{fig:NandC}
\end{center}
\end{figure}

\section{QAOA vs optimization by adiabatic and quantum annealing}\index{Physical Computing} 
Here we will briefly discuss another form of optimization, known now as quantum annealing and closely related to adiabatic quantum optimization.  Several variational algorithms were inspired by algorithms that have their roots in the adiabatic model of quantum computation.   \index{Quantum Algorithms}

\begin{remark}
We must recall that the complete dynamics of a physical system is given by the Schr{\"o}dinger equation 
\begin{equation}
    \imath\hbar \frac{d}{dt}\ket{\psi(t)}= \mathcal{H}(t)\ket{\psi(t)},
\end{equation}
where $\mathcal{H}(t)$ is the Hamiltonian of the system and $\ket{\psi(t)}$ is the state of the system at time $t$. The evolution of the Hamiltonian induces the unitary operator $U(t)$. In the case of a time independent Hamiltonian, the unitary evolution can be expressed as $U(t)=\exp{-\imath t\mathcal{H}}$ for time-independent $\mathcal{H}$.
\end{remark}


Informally, the adiabatic theorem states that when the ground state of a system evolves under a time-dependent Hamiltonian which is varied slowly enough, the system tends to remain in its lowest energy state. In other words, starting in an easy to prepare ground state of Hamiltonian $\mathcal{H}_{0}$, we slowly transform $\mathcal{H}_{0}$ to some final Hamiltonian $\mathcal{H}_{f}$.  In doing so, provided the adiabatic theorem is satisfied, ensure that we evolve to the ground state of $\mathcal{H}_{f}$. Here $\mathcal{H}_{f}$ can embed a problem instance, solved by e.g.~finding the ground state. 

\begin{remark}
In the framework of adiabatic quantum computation, the quantum mechanical system is
\begin{enumerate}
    \item Described by the Hamiltonian $\mathcal{H}_{0}$ at time $t=0$.
    \item Then the system is slowly evolved into the final Hamiltonian $\mathcal{H}_{f}$.
    \item   We set the final Hamiltonian $\mathcal{H}_{f}$ such that finding the ground state of $\mathcal{H}_{f}$ is equivalent to a minimization problem of a function 
\begin{equation}
    f : \left\{ 0,1\right\}^{n}\rightarrow \mathbb{R}. 
\end{equation}
\end{enumerate}
As is typical, the step (3.) above is used to give an explicit form of the final Hamiltonian as 
\begin{equation}
    \mathcal{H}_{f}=\sum_{z\in \left\{ 0,1\right\}^{n} } f(z) \ket{z}\bra{z}, 
\end{equation}
for a binary string $z$. The choice of the initial Hamiltonian $\mathcal{H}_{0}$ is independent of the solution of the problem and will be such that it is not diagonal in the computational basis. One example is to choose $\mathcal{H}_{0}$ to be diagonal in the Hadamard basis, 
\begin{equation}
    \mathcal{H}_{0}=\sum_{z\in \left\{ 0,1\right\}^{n} } h(z) \ket{\hat{z}}\bra{\hat{z}},
\end{equation}        
where $\ket{\hat{z}}$ are the state of $\ket{z}$ in Hadamard basis and $h(0^{n})=0$ and $h(z) \geq 1$ for all $z\neq 0^{n}$.
The time dependent Hamiltonian $\mathcal{H}(s)$ can be defined as
\begin{equation}\label{eqn:hofs}
    \mathcal{H}(s)=(1-s)\mathcal{H}_{0} + s\mathcal{H}_{f},\ \
    s\in \left[0,1\right].
\end{equation}
When $s$ varies from $0 \longrightarrow 1$ the Hamiltonian changes from $\mathcal{H}_{0}\rightarrow \mathcal{H}_{f}$. 
By the adiabatic theorem, the initial ground state, $\ket{\psi_{0}}$ is mapped to the global minimum of the function $f$. 
\end{remark}

\subsection{Approximating adiabatic evolution}

\begin{remark}
The state preparation takes inspiration from the quantum adiabatic algorithm, where a system is initialized in an easy to prepare ground state of a local Hamiltonian $\mathcal{H}_{x}= \sum_{i} \sigma_{x}^{(i)},$ the driver Hamiltonian, which  is  then  slowly  transformed  to  the problem  Hamiltonian, $\mathcal{V}$ \cite{kadowaki1998quantum}. As we will see, Trotterization of this procedure gives a long {\sf QAOA} sequence. 
\end{remark}

\begin{remark}
One can discretize the continuous time evolution of $\mathcal{H}_{s}$ from \eqref{eqn:hofs} by a quantum circuit \cite{van_Dam_2001}.
\end{remark}

\begin{remark}
The approximation is established in two steps following \cite{van_Dam_2001}.
\begin{enumerate}
    \item Discretize the evolution from $\mathcal{H}_0$ to
$\mathcal{H}_f$ by a finite sequence of Hamiltonians
$\mathcal{H}'_1$, $\mathcal{H}'_2,\cdots$ that 
gives rise to the same overall behavior.
\item Show how at any instant the combined 
Hamiltonian $\mathcal{H}'_j = (1-s)\mathcal{H}_0 + s \mathcal{H}_f$ is 
approximated by interleaving two simple 
unitary transformations.
\end{enumerate}
\end{remark}

\begin{remark}[$\ell_2$ induced operator norm]
To express the error of our approximation, we 
use the $\ell_2$ induced operator norm $\norm{\circ}_2$:
\begin{eqnarray*}
\norm{M}_2 & \bydef & \max_{\norm{x}_2=1}{\norm{M x}_2}.
\end{eqnarray*}
\end{remark}

\begin{remark}
The next lemma compares two Hamiltonians $\mathcal{H}(t)$ and $\mathcal{H}'(t)$ and their respective 
unitary transformations $U(T)$ and $U'(T)$.
\end{remark}

\begin{lemma}[Approximation Error \cite{van_Dam_2001}]\label{lemma:aevan}
Let $\mathcal{H}(t)$ and $\mathcal{H}'(t)$ be two time-dependent 
Hamiltonians for $0\leq t \leq T$, and 
let $U(T)$ and $U'(T)$ be the respective 
unitary evolutions that they induce. 
If the difference between the Hamiltonians is 
limited by $\norm{\mathcal{H}(t)-\mathcal{H}'(t)}\leq \delta$ for every $t$, 
then the distance between the induced transformations  
is bounded by $\norm{U(T)-U'(T)}\leq \sqrt{2T\delta}$.
\end{lemma}

\begin{proof}(Lemma \ref{lemma:aevan} \cite{van_Dam_2001}). 
Let $\psi(t)$ and $\psi'(t)$ be the two state trajectories  
of the two Hamiltonians $\mathcal{H}$ and $\mathcal{H}'$ with initially
$\psi(0)=\psi'(0)$. 
Then, for the inner product between the two states
(with initially $\braket{\psi'(0)}{\psi(0)}=1$), 
we have
\begin{eqnarray*}
\frac{d}{dt}
{\braket{\psi'(t)}{\psi(t)}} & = & 
-\imath \bra{\psi'(t)}(\mathcal{H}(t)-\mathcal{H}'(t))\ket{\psi(t)}.
\end{eqnarray*}  
Because at any moment $t$ we have 
$\norm{\ket{\psi(t)}}=\norm{\ket{\psi'(t)}}=1$ 
and $\norm{\mathcal{H}(t)-\mathcal{H}'(t)}\leq \delta$, we see 
that at $t=T$ the lower bound
$|\braket{\psi'(T)}{\psi(T)}|\geq 1-T\delta$
holds.
This confirms that for every vector $\psi$ we have
$\norm{U(T)\ket{\psi}-U'(T)\ket{\psi}}\leq \sqrt{2T\delta}$.
\end{proof}

\begin{remark}
Lemma \ref{lemma:aevan} tells us how we might deviate from the ideal 
Hamiltonian $\mathcal{H}(t)\bydef(1-\frac{t}{T})\mathcal{H}_0 + \frac{t}{T}\mathcal{H}_f$, 
without introducing too much error to the induced evolution. 
\end{remark}

An approximation of the evolution, $\mathcal{H}_{0}\rightarrow \mathcal{H}_{f}$ is given by a sequence of Hamiltonians,
\begin{equation}
  \mathcal{H}^{'}_{1}, \mathcal{H}^{'}_{2},\cdots, \mathcal{H}^{'}_{r},\cdots  
\end{equation}
where 
\begin{equation}
    \mathcal{H}^{'}_{r}=(1-r\Delta s)\mathcal{H}_{0} + r\Delta s \mathcal{H}_{f}
\end{equation}
and the parameter $s$ is slightly varied from $0$ to $1$. In the limit of $\Delta s \rightarrow 0$, the sequence becomes infinite. 
The unitary evolution then becomes 
\begin{equation}
U'(s=1)=    \cdots\exp{-\imath \mathcal{H}'_{r}}\cdots\exp{-\imath \mathcal{H}'_{1}} = \cdots U'_{r}\cdots U'_{1}
\end{equation} 
A second approximation, using the Campbell-Baker-Hausdorff theorem gives, 
\begin{equation}
\label{eq:6}
\begin{split}
    U^{'}_{r} &= \exp{-\imath (1-r\Delta s)\mathcal{H}_{0}}\cdot\exp{-\imath r \Delta s\mathcal{H}_{f}} \\
    &=  W(1-r\Delta s)\cdot V(r\Delta s)
\end{split}
\end{equation}
where $W(x)=\exp{-\imath x \mathcal{H}_0}$, and $V(x)=\exp{-\imath x \mathcal{H}_f}$.
The  evolution reduces to a sequence of unitary operations,
\begin{align}\label{eqn:U}
U'(s=1) &=\cdots U'_{r}U'_{r-1}\cdots U'_{2}U'_{1}
={}\\
&=(\cdots V\cdot V\cdot V\cdot V\cdot V\cdots)(\cdots V\cdot V\cdot W\cdot V\cdot V\cdots)\cdots\notag\\
& \quad (\cdots V\cdot  W\cdot V\cdot W\cdot V\cdots)\cdots(\cdots W\cdot V\cdot W\cdot V\cdot W\cdots)\cdots\notag\\
&\quad(\cdots W\cdot V\cdot W\cdot W\cdots)(\cdots W\cdot W\cdot W\cdot W\cdot W\cdots)
\notag
\end{align}
acting on the ground state of $\mathcal{H}_{0}$. It is clearly seen that the frequency of occurrence of $W$ is increasing but the frequency of occurrence of $V$ is decreasing. In the limit of $\Delta s \rightarrow 0$, the sequence in equation (\ref{eqn:U}) becomes infinite long length with difference frequency of $W$'s and $V$'s. For a finite $\Delta s$, the sequence is not infinite and becomes an approximation of the adiabatic evolution. For a mathematically focused survey on the annealing processes (quantum and stochastic) see e.g.~\cite{Morita2008}.   

\section{Computational phase transitions}\label{sec:cpt}

Here in \S~\ref{sec:cpt} we largely follow \cite{2019arXiv190610705P} and cover some joint work I've recently done with Philathong, Akshay, and Zacharov. 

As information is necessarily represented in physical media, the processing, storage and manipulation of information is governed by the laws of physics. Indeed, the theory of computation is intertwined with the laws governing physical processes \cite{deutsch1985quantum}.  Many physical systems and physical processes can be made to represent, and solve computational problem instances.  Viewed another way, many variants of naturally occurring processes (such as protein folding) have been shown to represent computationally significant problems, such as \NP-hard optimization problems.  But how long does it take for a physical process to solve problem instances?  How can difficult problem instances be generated? 

The physical Church-Turning thesis \cite{church1936unsolvable, turing1937computable} asserts that a universal classical computer can simulate any physical process and vise versa (outside of quantum mechanical processes).  It does not propose the algorithm, yet asserts its existence.  One might wrongly suspect that undecidable problems can be embedded into physical systems: attempts at this fail, i.e.~due to instabilities. What about the \P~vs.~\NP~problem?  If no physical process existed to solve \NP-complete problems in polynomial time, then by the physical Church-Turning thesis, no algorithm would exist either.  Hence if the laws of physics ruled out such a scenario, this would imply that \P~$\neq$~\NP.  

This distinction between polynomial and exponential resources is a course graining that computational complexity theory is based around. We do not know if a physical process can be made to solve \NP-complete problems in polynomial time or not. However, it is asserted that computational phase transitions are a  feature of \NP-complete problems---although specifics of the transition have yet to be formulated (proven) rigorously. We will turn to the theory of computational phase transitions to understand how physics responds to changes in the complexity landscape across the algorithmic phase transition.

This algorithmic phase transition occurs where randomly generated problem instances are thought to be difficult \cite{CrawfordA93, friedgut1999sharp, selman1996critical}.  It is observed by the fact that computer algorithms experience a slowdown around this transition point.  For example, let us consider the familiar problems of $2$- (and $3$)-satisfiability\index{Boolean Satisfiability} (detailed in the next section: the phase transition provably exists for $2$-{\sf SAT} and is only known to be inside a window for $3$-{\sf SAT}).  If we let the number of variables be $N$ and uniformly generate $M$ random clauses over these $N$ variables, computer algorithms appear to slow down at a certain clause to variable ratio (critical clause density of the order parameter $\alpha = M/N$). 

What about physical systems that bootstrap physics to naturally solve problems instances? Instances of these problems can be embedded in the lowest energy configuration of physical models \index{Ising Model} \cite{lucas2014ising, B08, spinlogic2}. Hence, building such a physical system and cooling (annealing) this system can enable a process which solves such problems \index{Physical Computing}  \cite{kirkpatrick1983optimization, utsunomiya2011mapping, inagaki2016coherent, pierangeli2019large, marandi2014network, nixon2013observing, berloff2017realizing, dung2017variable, kalinin2018global}.  For example, a system settling into its low-energy configuration can be programmed such that this low energy configuration represents the solution to $2$- (and $3$)-{\sf SAT} instances.  

We found that the algorithmic phase transition has a statistical signature in Gibbs' states of problem Hamiltonians generated randomly across the algorithmic phase transition.  This was confirmed by exact calculations of 26 binary units (spins) on a mid-scale supercomputer.  Physical observation of the effect is hence within reach of near term and possibly even existing physical computing hardware: such as Ising machines \cite{inagaki2016coherent, pierangeli2019large}, annealers \cite{kirkpatrick1983optimization} and quantum enhanced annealers \index{Physical Computing} \cite{ johnson2011quantum, barends2016digitized, harris2010experimental,harris2018phase,king2018observation}.   

Since the algorithmic phase transition takes places where randomly generated problem instances are thought to contain difficult instances, this discovery provides a practical benchmark for contemporary physics based processors. The problem of finding hard instances is prominent to such emerging technologies.  We propose a physical experiment to witness the algorithmic phase transition signature and to benchmark contemporary physics based processors.

Propositional satisfiability ({\sf SAT}) is the problem of determining the satisfiability of sentences in propositional logic. If $k$ is the number of literals in each clause the problem is called $k$-{\sf SAT}. Determining the satisfiability of a formula in conjunctive normal form where each clause is limited to at most $k$ literals is \NP-complete \index{Boolean Satisfiability} (Cook~\cite{Cook1971}).  We recall from \S~\ref{chap:progGS} that a decision problem $c$ is \NP-complete iff: (i) $c$ is in \NP and (ii) every problem in \NP~is reducible to c in polynomial time. We recall Theorem \ref{thm:3bph}, and state the following. 

\begin{theorem}\label{thm:nppji}
{\scshape Three-body Projector Ising Hamiltonian} is~{\sf NP}-complete. 
\end{theorem}

We will proceed informally and complete a construction proof of Theorem \ref{thm:nppji}. As one recalls from \S~\ref{chap:progGS}, in \NP~problems a candidate solution can be checked in polynomial time.

$k$-satisfiability problems can be reduced---generally called Karp reduction---to 3-satisfiability (3-{\sf SAT}) when formulas in conjunctive normal form (CNF) are considered with each clause containing at least 3 literals.

\begin{remark}(The spectrum of Hamiltonian {\sf SAT}) 
Given a 3-{\sf SAT} formula $f$, it is readily established that the 3-body Ising projector form of $f$ has the following eigenvalue problem. 
\begin{equation}
    {\mathcal H}_f \ket{x} = f(x)\ket{x}
\end{equation}
where $f(x)$ counts the number of clauses violated by bit string $x$. 
\end{remark}

\begin{proof}({\scshape Three-body Projector Ising Hamiltonian} is~{\sf NP}-complete---Theorem \ref{thm:nppji}) 
We readily see that every 3-{\sf SAT} instance can be embedded into an instance of {\scshape Three-body Projector Ising Hamiltonian}.  Each candidate solution is evaluated in time proportional to the number of Hamiltonian terms. Moreover, a~{\sf Yes} instance is shown by a witness $\ket{y}$ such that 
\begin{equation}
    \bra{y}{\mathcal H}_f \ket{y}=0. 
\end{equation}
Hence the problem is in {\NP}.  We recall the classical result that 3-{\sf SAT} is Karp reducible ( polynomial-time reduction) to general circuit {\sf SAT}, which is the cannonical \NP-complete problem.  Hence, we establish the completeness result in Theorem~\ref{thm:nppji}. 
\end{proof}

Given the  clause density defined as 
$\alpha \bydef M/N$  where  $M$ is the number of  clauses and $N$ is the number of  variables, we will generate random {\sf SAT} instances.  
As random {\sf SAT} instances are generated increasing $\alpha$, a sharp transition occurs.  The likely-hood of a random instance being satisfiable goes from near unity to near zero across an increasingly small domain. 
Crawford and Auton (1993)~\cite{CrawfordA93} empirically located the 3-{\sf SAT} transition at a clause/variable ratio around $4.27$. 

\begin{figure}[htbp]
  \centering 
  \includegraphics[width=0.7\textwidth]{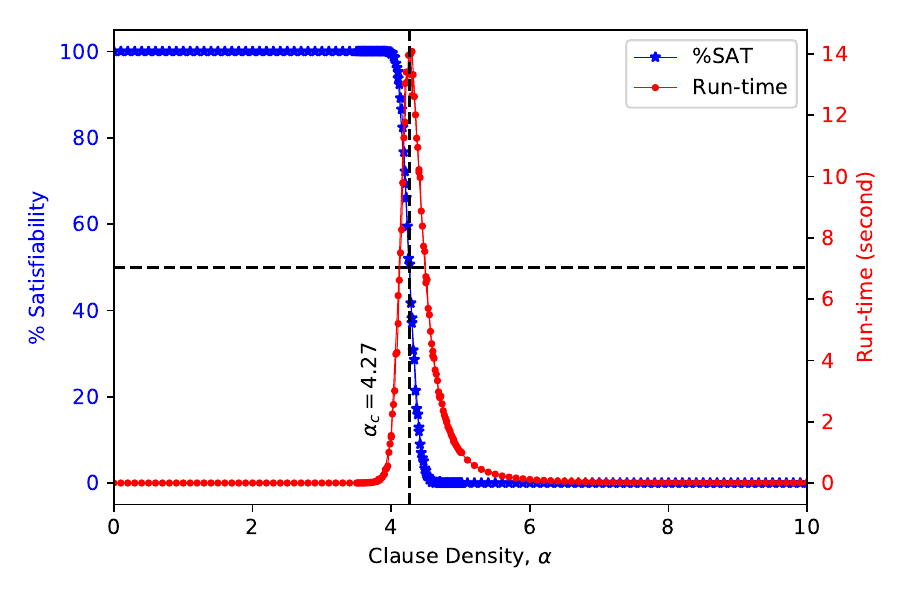}
  \caption{Percent of satisfiable instances (left axis) and run-time (right axis) versus clause density, $\alpha$. We randomly generated $1,000$ $3$-SAT instances with $300$ variables with observed $\alpha_{c} \approx 4.27$. Figure taken from \cite{2019arXiv190610705P}. \label{fig:3satphase} }
\label{fig:PhaseTransitionSAT}
\end{figure}

\begin{example}
As a worked example, let us consider a value assignment of $x_1, x_2, x_3, x_4$ such that $f(x_1, x_2, x_3, x_4) = 1$ and hence (i) show that $f$ in \eqref{eqn:satex} is satisfiable.  
\begin{equation}\label{eqn:satex}
f = (x_1 \vee \overline{x}_3 \vee x_4)\wedge (\overline{x}_2 \vee x_3\vee \overline{x}_4) \wedge (\overline{x}_1\vee x_2 \vee x_3)
\end{equation}
From binary search (by hand) we find the assignment $x_1=0, x_2=0, x_3 =0$ and $x_4$ is free.  The clause density is $0.75$ and the number of nodes in the expansion graph (e.g.~a Boolean decision diagram) is negligible.  
\end{example}

\subsection{Thermal states at the SAT phase transition}

Physical systems at thermal equilibrium are often approximated by a Gibbs state 
\begin{equation}
\rho_{\beta} \bydef \frac{e^{-\beta h}}{\langle e^{-\beta h} \rangle} = \frac{1}{\mathcal{Z}} \sum_{l} e^{-\beta \lambda_l} \ketbra{l}{l} = \dfrac{1}{\sum_l e^{-\beta \lambda_l}} \sum_l e^{-\beta \lambda_l} \ketbra{l}{l}. 
\end{equation} 
Where $\lambda_{i}$ are the eigenvalues. Note that the computational basis forms the energy basis of the Hamiltonian, hence $\ketbra{l}{l}$ is diagonal. We will soon be concerned with the case that the Hamiltonian $\mathcal{H} = \sum_i C_i$ is a sum over projectors onto 3-SAT clauses and the partition function is $\mathcal{Z}=\sum_l e^{-\beta \lambda_l}$.

Consider the monotonic list,
\begin{eqnarray*}
\lambda_0 \leq \lambda_1 \leq \lambda_2 \leq  \ldots \leq \lambda_r
\end{eqnarray*}
and assume that $\lambda_0$ is possibly degenerate. Let $\varphi^{i}_0$ be an eigenvector with eigenvalue $\lambda_0$. Let $i$ index all $d$ potentially degenerate eigenvalues satisfying
\begin{equation}
\mathcal{H} \ket{\varphi^{i}_0} = \lambda_0 \ket{\varphi^i_0}. 
\end{equation}
We are interested in the probability of being in the subspace 
\begin{equation}
  \spn \{ \ket{ \varphi^i_0}  | \mathcal{H} \ket{\varphi^{i}_0} = \lambda_0 \ket{\varphi^i_0} \}.   
\end{equation}
We let $i$ = 1 to $d$ for $d$ the degeneracy count. We are interested in the quantity
\begin{equation} \label{eqn:plotted}
\begin{split}
\frac{1}{\mathcal{Z}} & \sum^d_{i=1} \bra{\varphi^i_0} \biggl[\sum_l e^{-\beta \lambda_l} \ketbra{l}{l} \biggr] \ket{\varphi^i_0}  = \\
& =\frac{1}{\mathcal{Z}}\sum_{i,l} \braket{\varphi^i_0}{l} \braket{l}{\varphi^i_0}e^{-\beta \lambda_l} 
= \frac{1}{\mathcal{Z}} \sum_{i=1}^d e^{-\lambda_0 \beta} = \frac{\displaystyle\sum_{i=1}^d e^{-\beta \lambda_0}}{\displaystyle\sum_{l} e^{-\beta \lambda_l}}.
\end{split}
\end{equation}

\begin{proposition}(Eigenvalues of thermal states). 
It can be shown that 
 \begin{equation}
\rho_\beta = \frac{e^{-\beta \sum_j Z_j}}{\text{tr}\{e^{-\beta \sum_j Z_j}\}}=
\bigotimes_{j} \frac{1}{2}\sum_{b\in\{0,1\}}\left(1-(-1)^b \tanh (\beta)\right) \ket{b}_j\!\bra{b}_j. 
\end{equation}
We then derive a formula for the eigenvalues of a state in terms of the Boolean variables in the general bit string $\ket{x_1, x_2, \dots, x_l}$, viz.,
\begin{equation} \label{eqn:thermalZ}
\bra{x_1, x_2, \dots, x_l}\rho_\beta \ket{x_1, x_2, \dots, x_l} = \prod_{j=1}^l \frac{1}{2}\left(1-(-1)^{x_j}\tanh(\beta)\right). 
\end{equation}
If we were to sample $\rho_\beta$, the probability of measuring $\ket{x_1, x_2, \dots, x_l}$ is given precisely by \eqref{eqn:thermalZ}. 
\end{proposition}

We want to investigate the properties of a thermal system, for fixed $\beta$ playing the role of an inverse temperature.  We further want to investigate this across the phase transition.  We can intuitively expect that around the phase transition, a thermal system would exhibit a decreased occupancy in its ground state.  

Let $\ket{i}$ denote (possibly degenerate) lowest eigenstates of $\mathcal{H}$.  We label these possibly degenerate states by letting $i$ range from $1$ up to $d$.  Call $\lambda_\text{min}$ the lowest eigenvalue of $\mathcal{H}$, then we have \eqref{eqn:lambdamin}.  

We are concerned with the quantity \eqref{eqn:plowest}, giving occupancy in the low-energy subspace for a system at equilibrium for fixed finite inverse temperature $\beta$. We call this quantity $p\left(\lambda_\text{min}, \beta\right)$, 
\begin{equation}\label{eqn:plowest}
p\left(\lambda_\text{min}, \beta\right) = \frac{1}{\mathcal{Z}}\sum_{i=1}^d 
\bra{i} e^{-\beta \mathcal{H}} \ket{i} 
\end{equation}
\noindent where 
\begin{equation}\label{eqn:lambdamin}
\forall i \in\{1, ..., d\}, ~\bra{i} \mathcal{H} \ket{i}. 
\end{equation}

\begin{proposition}\label{prop:zerotlim}
It can be shown that 
\begin{equation}
\lim_{\beta \rightarrow \infty}
\frac{1}{\mathcal{Z}}\sum_{i=1}^d\bra{i} e^{-\beta \mathcal{H}} \ket{i} = 1
\end{equation}
and hence one can establish that in the zero temperature limit ($\beta$ is inverse temperature), sampling a thermal system can solve {\sf SAT} instances with probability one.  
\end{proposition}

\begin{remark}[On determining a threshold temperature to reveal the algorithmic phase transition]
From Figure \ref{fig:3sattherm} the ground state occupancy of thermal states can reveal an algorithmic phase transition signature.  When the inverse temperature $\beta$ is relatively small, the easy-hard transition around $\alpha = 4.27$ is missing (Proposition \ref{prop:zerotlim} establishes no transition in the zero temperature limit). As pointed out to the author by Vladimir Korepin, a threshold temperature to reveal the algorithmic phase transition appears to be lacking.   
\end{remark}

\begin{figure}[htbp]
  \centering 
  \includegraphics[width=0.8\textwidth]{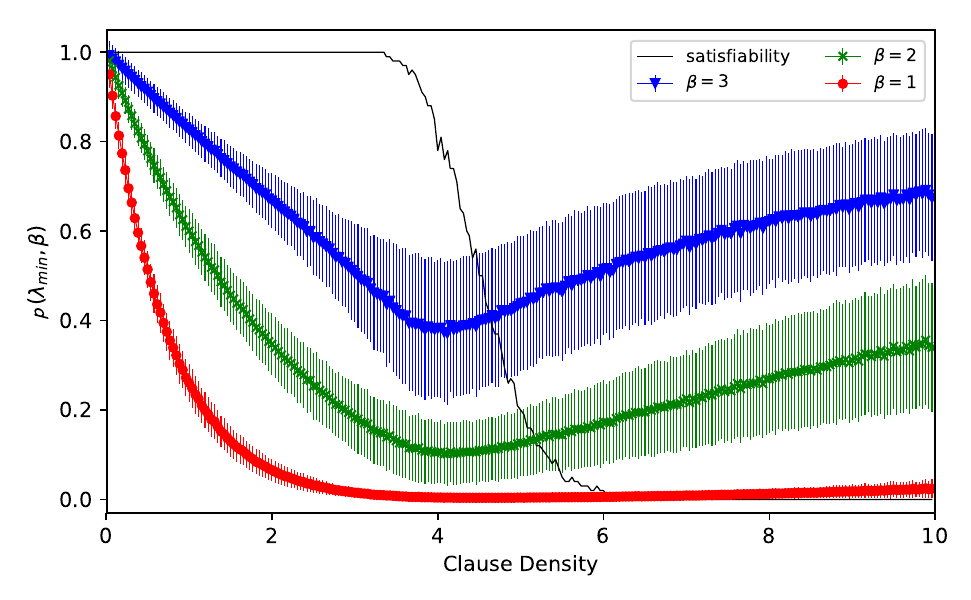}
  \caption{ Ground state occupancy of the thermal states corresponding to Hamiltonians embedding 3-SAT instances across the algorithmic phase transition for 26 spins with $\beta = 1, 2, 3$. Figure taken from \cite{2019arXiv190610705P}. }
\label{fig:3sattherm}
\end{figure}


\section{Low-depth quantum circuits}

The theory and application of low-depth quantum circuits has grown dramatically in recent years.  In the absence of error correction, NISQ era quantum computation is focused on quantum circuits that are short enough and with gate fidelity high enough that these short quantum circuits can be executed without quantum error correction, as in the recent quantum supremacy experiments \cite{Arute2019, Zhong1460}. Herein lies the heart of the variational model: by adjusting parameters in an otherwise fixed quantum circuit, low-depth noisy quantum circuits are pushed to their ultimate use case. We will now upper bound the possible bipartite entanglement that such circuits can possibly generate \cite{UVQC}. 

\subsection{A combinatorial quantum circuit area law} 

Our construction of universal variational quantum computation has not considered whether a restricted form of ansatz is capable of universal quantum computation at some arbitrary depth as Lloyd~\cite{2018arXiv181211075L} and others~\cite{morales2019universality} have.  Instead, the objective function to be minimised is defined in terms of the unitary gates arising in the target circuit to be simulated.  What ansatz states are then required to simulate a given target circuit?

This question appears to be difficult and not much is currently known.  In the case of {\sf QAOA} it was recently shown by myself and coauthors that the ability of an ansatz to approximate the ground state energy of a satisfiability instance worsens with increasing problem density (the ratio of constraints to variables) \cite{2019arXiv190611259A}. These related results however do not immediately apply to our interests here.  

Towards our goals, we have included a derivation showing that reasonable depth circuits can saturate bipartite entanglement---the depth of these circuits scales with the number of qubits and also depends on the interaction geometry present in a given quantum processor.  This result establishes a first relationship between an objective circuit (to be simulated) and a given ansatz state.  Consider the following.  

An ebit is a unit of entanglement contained in a maximally entangled two-qubit (Bell) state.  A quantum state with $q$ ebits of entanglement (quantified by any entanglement measure) contains the same amount of entanglement (in that measure) as $q$ Bell states.  
\begin{lemma}[Combinatorial Quantum Circuit Area Law---Biamonte \cite{UVQC}]\label{lemma:cqcal}
 Let $c$ be the depth of 2-qubit controlled gates in the $n$-qubit hardware-efficient ansatz.  Then the maximum possible number of ebits across any bipartition is 
    $$ E_b = \min \{ \left \lfloor{n/2}\right \rfloor, c \}.$$ 
 \end{lemma} 
In a low-depth circuit, the underlying geometry of the processor heavily dictates $c$ above.  For example, for a line of qubits and for a ring, the minimal $c$ required to possibly maximise $E_b$ is $\sim n/2$ and $\sim n/4$ respectively. 
However, in the case of a grid, the minimal depth scales as $\sim \sqrt{n}/2$.  

Hence, if we wish to simulate a quantum algorithm described by a low-depth circuit, having access to a grid architecture could provide an intrinsic advantage.  Specifically, our combinatorial quantum circuit area law establishes that  an objective circuit generating $k< \left \lfloor{n/2}\right \rfloor$ ebits across every bipartition, must be simulated by an ansatz of at least minimal required circuit depth $\sim \sqrt{k}$ on a grid.  

While this does establish a preliminary relationship, the general case remains unclear at the time of writing.  For example, given a quantum circuit with application time $t^\star$ which outputs $\ket{\psi}$, what is the minimal $t(\epsilon)\leq t^\star$ for a control sequence \eqref{eqn:time} to provide an $\epsilon$ close 2-norm approximation to $\ket{\psi}$? 

Here we will consider some properties of the quantum states that are accessible in NISQ era quantum information processing.  Here we provide a bound for the minimal depth circuit (generated from the so called, hardware efficient Ansatz as used in recent experiments \cite{2017Natur.549..242K}) to possibly saturate bipartite entanglement on any bipartition.  Understanding the computational power of these circuits represents a central open question in the field of quantum computation today.  This appendix seeks to quantify contemporary capacities.  We are currently not able to express the success probability as a function of the circuit depth required for an objective function to accept.

\begin{definition}[Interaction graph]    
Consider the Hamiltonian 
$$ \mathcal{H} = \sum_{ij} J_{ij} A_i A_j + \sum_i b_i B_i  + \sum_i c_i C_i. $$
The support matrix $S$ of $J_{ij}$ is defined to have the entries
 $$ S_{ij} =  (J_{ij})^0 $$
 and is called the \textit{interaction graph of $\mathcal{H}$}---a symmetric adjacency matrix. (Here we assume that $a^0$ (written alternatively as $(a)^0$) vanishes for real $a < \varepsilon$ and $a^0$ goes to unity for $a > \varepsilon$ for some finite (small) real $\varepsilon$ cutoff.) 
\end{definition}

\begin{remark} 
The interaction graph induces a space-time quantum circuit defined by a \textit{tiling} on the multiplex from $S$.
\end{remark}

\begin{definition} 
 A \textit{Tiling} is a gate sequence acting on a multiplex network induced by $S$.  
 \end{definition} 
 
 \begin{remark}     An active edge (node) will specify if neighboring edges (nodes) can be active. As a general rule, non-commuting terms must be active on different layers.  
     \begin{enumerate}
        \item Qubits connect network layers by time propagation.
        \item Nodes of $S$ in each layer can be acted on by local gates.
        \item Edges of $S$ in each layer can be acted on by two-qubit gates. 
    \end{enumerate}
 \end{remark} 
   
 As an example, consider the following multiplex network.    
 \begin{figure}[htb!]
 
    \begin{center}
        \minipage{0.5\textwidth}
        \includegraphics[width=\linewidth]{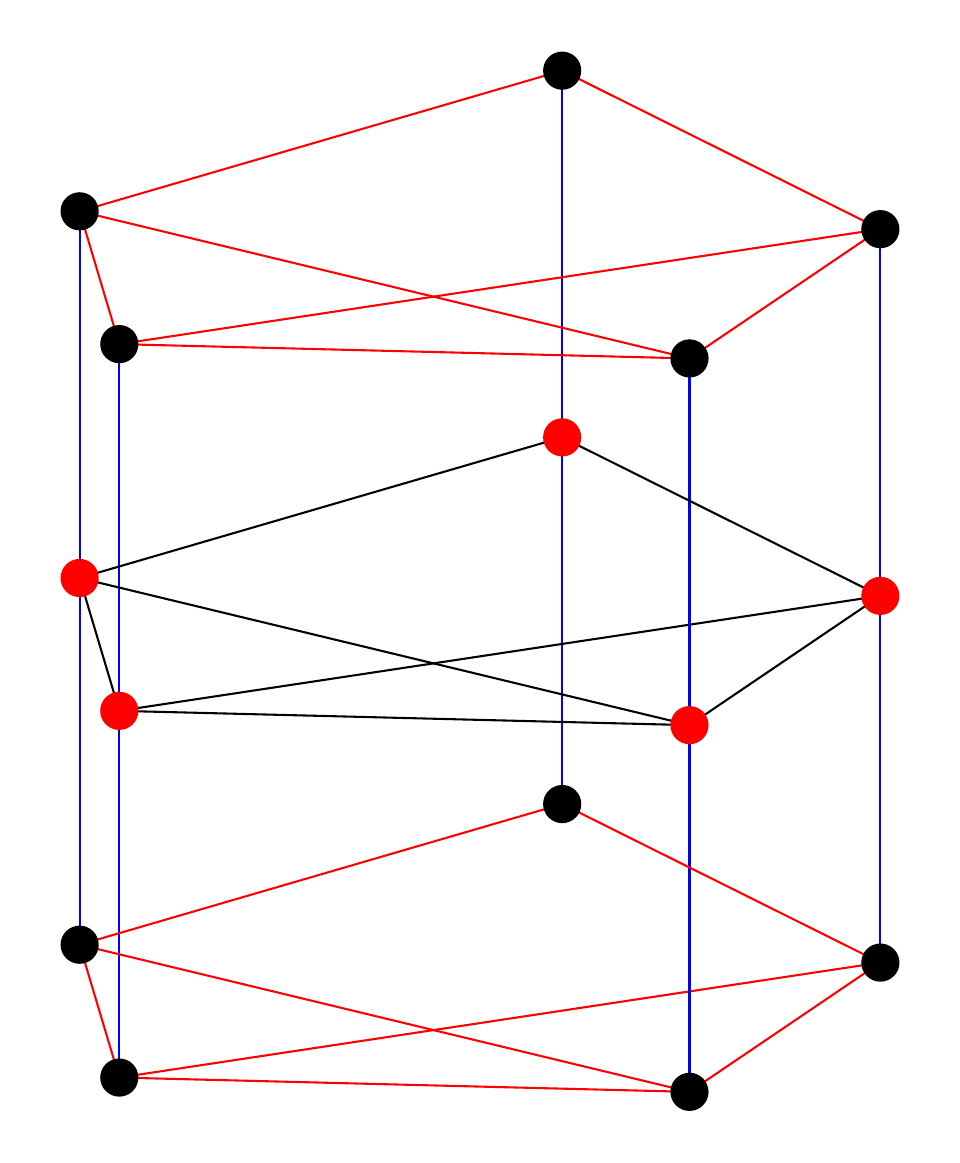}
        \endminipage
        \minipage{0.01\textwidth}
        \begin{tikzpicture}[scale=5, axis/.style={very thick, ->, >=stealth'}]
            \draw[axis] (0,-0.1) -- (0,1.1) node(yline)[above] {$t$};
        \end{tikzpicture}
        \endminipage
    \end{center}
    \captionof{figure}{Five qubits evolve as time goes up on the page.  At $t=0$, commuting interaction terms (red edges) are applied between the qubits.  At $t=1$ local gates (red nodes) are applied to each of the qubits. }
    \label{fig:5qe}
 \end{figure}
    
    In the bottom layer of the multiplex graph in Figure \ref{fig:5qe}, a sequence of red edges correspond to the application of two-body gates.  In the next layer, the red (active) notes correspond to local rotation gates being applied to all the qubits. Finally, in the top layer the commuting two-body gates are applied.  The blue vertical edges represent time passing (going up on the page).  The hardware-efficient ansatz is exactly such a tiling.  
    
In terms of example circuits, we typically will be given a short repetitive sequence.  For example, here we consider qubits interacting on a ring.  A layer of local gates followed by a layer of two-body gates is as follows.     
    
    \begin{equation*}
        \Qcircuit @C=.75em @R=.75em {
            & \gate{U} & \qw & \qw & \ctrl{1}   & \qw        & \qw & \qw    & \qw & \qw        & \gate{R_Y} & \qw \\
            & \gate{U} & \qw & \qw & \gate{R_Y} & \ctrl{1}   & \qw & \qw    & \qw & \qw        & \qw        & \qw \\
            & \gate{U} & \qw & \qw & \qw        & \gate{R_Y} & \qw & \qw    & \qw & \ctrl{2}   & \qw        & \qw \\
            &          &     &     &            &            &     & \ddots &     &            &            &     \\
            & \gate{U} & \qw & \qw & \qw        & \qw        & \qw & \qw    & \qw & \gate{R_Y} & \ctrl{-4}  & \qw \\
        }
    \end{equation*}
    
    \bigskip
    
Where the local gates, $U$ are arbitrary and the two-body gates are controlled $Y$ rotations.  What is the maximal bipartite entanglement that such a circuit can generate when acting on a product state? 

To understand this question, let us consider a pure $n$-qubit state $\ket{\psi}$. 

\begin{definition} 
Bipartite rank is the Schmidt number (the number of non-zero singular values) across any reduced bipartite density state from $\ket{\psi}$ (i.e.~$\lfloor n/2 \rfloor$ qubits).  
\end{definition} 

\begin{remark} 
Rank provides an upper-bound on the bipartite entanglement that a quantum state can support---a rank-$k$ state has at most $\log_2(k)$ ebits of entanglement. 
\end{remark} 

\begin{definition}
An ebit is a unit of entanglement contained in a maximally entangled two-qubit (Bell) state. 
\end{definition} 

\begin{remark} 
A quantum state with $q$ ebits of entanglement (quantified by any entanglement measure) contains the same amount of entanglement (in that measure) as $q$ Bell states. 
\end{remark} 

\begin{remark} 
If a task requires $r$ ebits, it can be done with $r$ or more Bell states, but not with fewer.  Maximally entangled states in $\mathbb{C}^d\otimes \mathbb{C}^d$ have $\log_2(d)$ ebits of entanglement. 
\end{remark} 

Now we arrive at what we call a {\it quantum circuit combinatorial area law}.  It is the minimal depth circuit that possibly could saturate the bipartite entanglement with respect to any bipartition. 

\begin{lemma}[Biamonte \cite{UVQC}]
 Let $c$ be the depth of 2-qubit controlled gates in the $n$-qubit hardware-efficient ansatz.  Then the maximum possible number of ebits across any bipartition is 
    $$ \min \{ \left \lfloor{n/2}\right \rfloor, c \}.$$ 
 \end{lemma} 
 
    
        
    
 \begin{remark}[Combinatorial Quantum Circuit Area Law] 
    Minimal possible $c$ saturating specific graphs are given in Table \ref{tab:cqcal}. 
 \end{remark}

\begin{table}[h t!]
  \centering
  \begin{tabular}{p{2.2cm}|c|c|c}
     & line & ring & grid \\
    \hline    \hline
    interaction geometry
    &
    \begin{minipage}{.25\textwidth}
      \includegraphics[width=1.\linewidth]{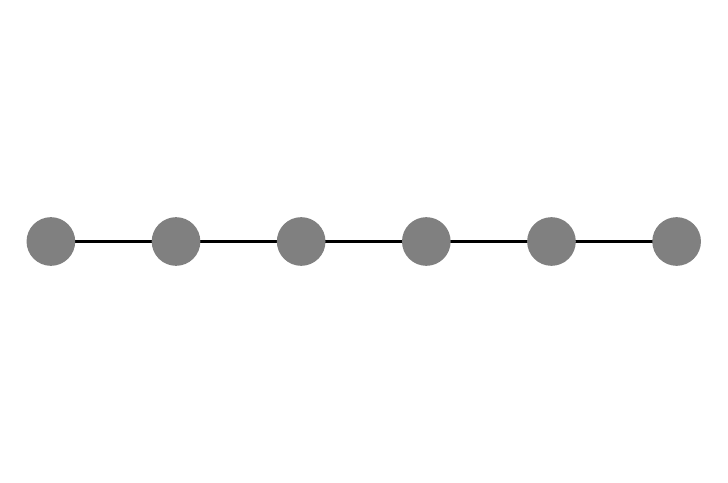}
    \end{minipage}
    &
    \begin{minipage}{.25\textwidth}
      \includegraphics[width=1.\linewidth]{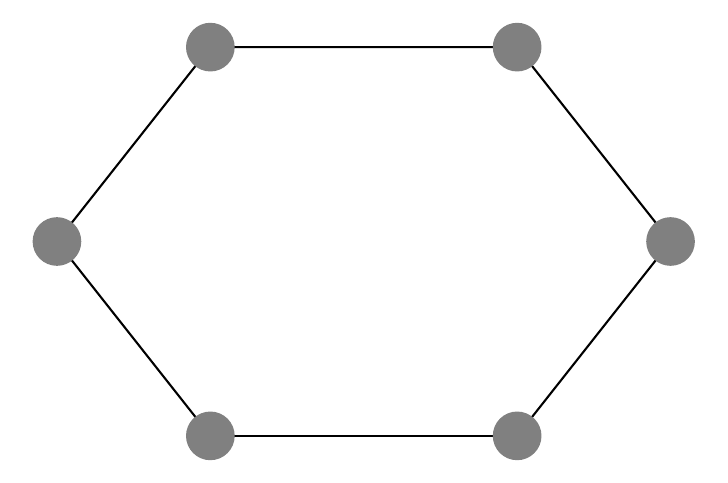}
    \end{minipage}
    &
    \begin{minipage}{.25\textwidth}
      \includegraphics[width=1.\linewidth]{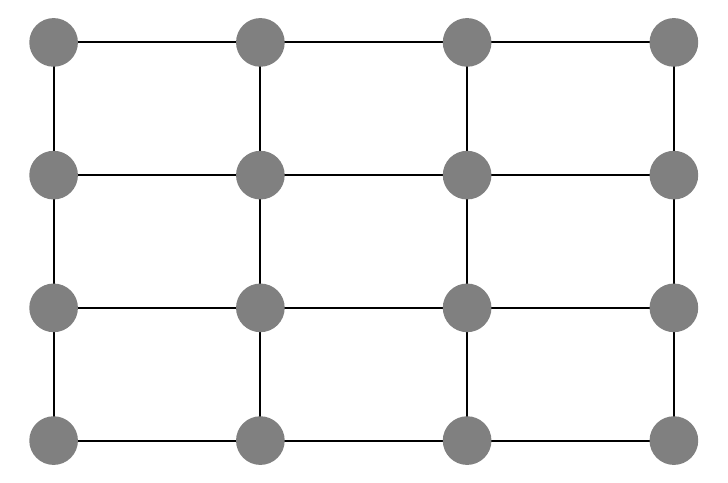}
    \end{minipage}
    \\
    \hline
    saturating depth & $c \sim n/2$ & $c \sim n/4$ & $c \sim \sqrt{n}/2$

  \end{tabular}
  \caption{Minimal possible circuit depth $c$ possibly saturating bipartite entanglement with respect to different interaction geometries.}\label{tab:cqcal}
\end{table}

\section{Reachability deficits}

In this section we explore a fundamental limitation discovered in \cite{2019arXiv190611259A} with coauthors.  The {\sf QAOA} algorithms that exploit variational state preparation to find approximate ground states of Hamiltonians encoding combinatorial optimization problems. 

In particular we consider the Quantum Aproximate Optimization Algorithm or {\sf QAOA}~\cite{farhi2014quantum}. As a means to study the performance of {\sf QAOA}, we turn to constraint satifiability---a tool with a successful history which was covered in detail in \S~\ref{chap:progGS} (see particularly \S~\ref{sec:cpt}).  Such problems are expressed in terms of $N$ variables and $M$ clauses (or constraints). 

\begin{remark}
Recall from \S~\ref{sec:cpt}, the density of such problem instances is the clause to variable ratio, the clause density $\alpha = M/N$. $k$-{\sf SAT} clauses are randomly generated to form instances by uniformly selecting unique $k$-tupels from the union of a variable set (cardinality $n>k$) and its element wise negation.  We consider random instances of the {\sf NP}-complete decision problem, $3$-{\sf SAT}. 
\end{remark}

{\sf QAOA} aims to approximate solutions to optimization version of this problem.  Here we consider MAX-3-{\sf SAT} which is {\sf NP}-Hard for exact solutions and {\sf APX}-complete for approximations beyond a certain ratio \cite{haastad2001some}. In these settings, the algorithm's limiting performance exhibits strong dependence on the problem density. As discovered in \cite{2019arXiv190611259A}. 

Let \textit{N} be the number of variables in a {\sf SAT} instance with $M$ clauses. MAX-{\sf SAT} solutions are the Boolean strings $\bm{\omega} = \omega_1, \omega_2, \omega_3, ..., \omega_N$ that violate the least number of clauses. Techniques discussed in \S~\ref{chap:progGS} map {\sf SAT} instances into Hamiltonians,  

\begin{equation}\label{3sathamiltonian}
    \mathcal{H}_{\text{SAT}} = \sum_{l=1}^{M}\mathcal{P}(l),
\end{equation}

\noindent where $\mathcal{P}(l)$ are rank-1 projectors acting on the $l^{th}$ clause. It is easy to verify that the ground state energy of $\mathcal{H}_{\text{SAT}}$ is representative of the minimum number of violated clauses. 

The variational state generated by {\sf QAOA} can be described as running a \textit{p}-depth quantum circuit on the state $\ket{+}^{\otimes{n}}$,
\begin{equation}\label{qaoa equation}
    \ket{\psi(\boldsymbol{\gamma},\boldsymbol{\beta})}=\prod_{i=1}^{p} \mathcal{U}(\gamma_{i},\beta_{i})\ket{+}^{\otimes{n}},
    \end{equation} 
    where
\begin{equation}\label{driverandproblem}
    \mathcal{U}(\gamma_{k},\beta_{k})=\exp{-\imath \beta_{k}\mathcal{H}_{x}} \cdot \exp{-\imath \gamma_{k}\mathcal{V}}.
\end{equation}

In order to approximate solutions of MAX-{\sf SAT}, {\sf QAOA} with standard settings, $\mathcal{H}_x=\sum_{i}\sigma_{x}^{(i)}$ and $\mathcal{V} = \mathcal{H}_{\text{SAT}}$, is used to calculate the energy approximation $E_{g}^{\text{QAOA}}$, where

\begin{equation}\label{qaoaoptimization}
    E_{g}^{\text{QAOA}} =  \min_{\boldsymbol{\gamma},\boldsymbol{\beta}} \bra{\psi(\boldsymbol{\gamma},\boldsymbol{\beta})}\mathcal{H}_{\text{SAT}}\ket{\psi(\boldsymbol{\gamma},\boldsymbol{\beta})}.
\end{equation}

We numerically studied $f = E_{g}^{\text{QAOA}} - \min(\mathcal{V})$ as a function of clause density $\alpha$, for a \textit{p}-depth {\sf QAOA} circuit on randomly generated $3$-{\sf SAT} instances (see Figure~\ref{fig:qaoareachability}). Although increased depth versions achieve better approximations, the limiting performance exhibits a non-trivial dependence on the problem density. Based on this finding we formulate the following:   
 
\begin{definition}
 Let $\ket{\psi}$, be the ansatz states generated from a \textit{p}--depth {\sf QAOA} circuit as shown in \eqref{qaoa equation}. Then 
\begin{equation}\label{eqn:reachabilitydef}
    f = \min_{\psi \subset \mathcal{H}} \bra{\psi}\mathcal{V}\ket{\psi} - \min_{\phi \in \mathcal{H}} \bra{\phi}\mathcal{V}\ket{\phi},
\end{equation}
characterises the limiting performance of {\sf QAOA}. The R.H.S.~of equation \eqref{eqn:reachabilitydef} can be expressed as a function, $f(p,\alpha,n)$.
\end{definition}

\begin{proposition}[Reachability Deficit \cite{2019arXiv190611259A}]
For $p \in \mathbb{N}$ and fixed problem size, $\exists$ $ \alpha > \alpha_c$ such that $f$ from \eqref{eqn:reachabilitydef} is non-vanishing. This is a reachability deficit.
\end{proposition}

\begin{figure}[ht!]
  \centering 
  \includegraphics[width=0.67\textwidth]{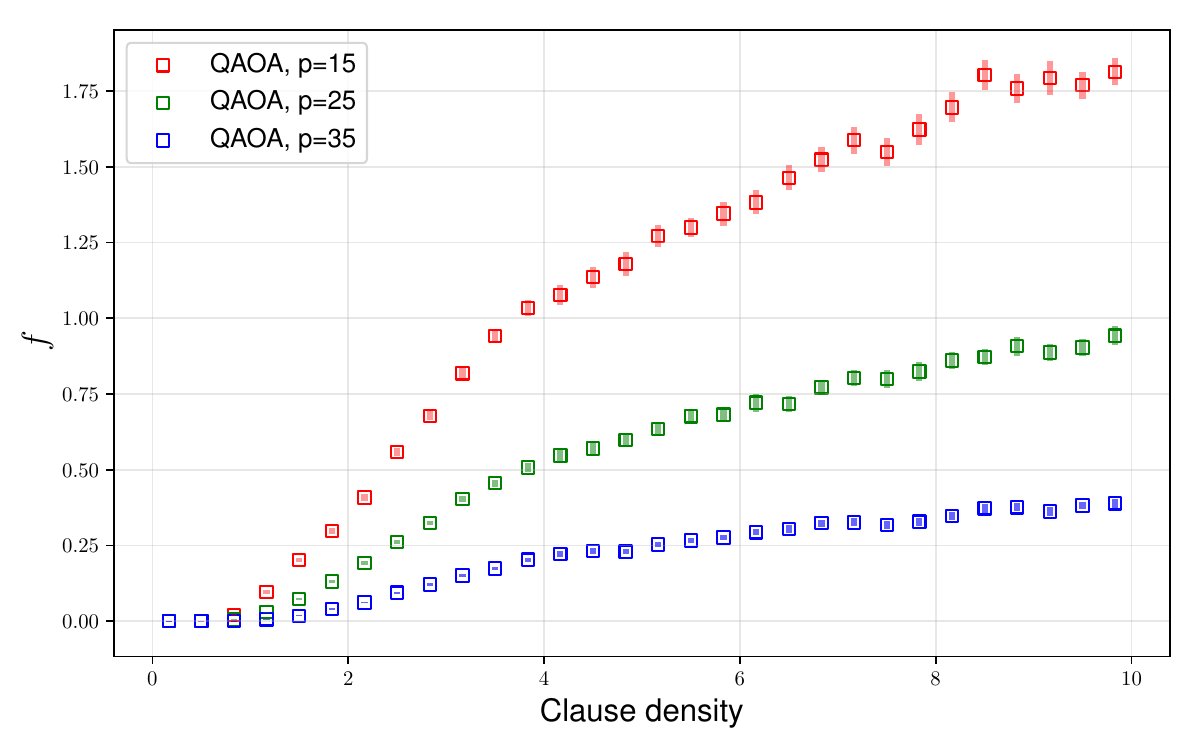}
  \caption{$f = E_{g}^{\text{QAOA}} - \min(\mathcal{H}_{SAT})$ vs clause density for $3$-{\sf SAT} for differing {\sf QAOA} depths. Squares show the average value obtained over 100 randomly generated instances for $n=6$ with error bars indicating the standard error of mean. Figure taken from \cite{2019arXiv190611259A}.  }
\label{fig:qaoareachability}
\end{figure}
\vspace{10pt}
\noindent


\chapter[Variational Quantum Computation]{Universal variational quantum computation} 
\label{sec:variational}

In Chapter \ref{chap:varintro} we considered variational quantum search and optimization.  Here we address a different problem.   We wish to simulate the output of an $L$ gate quantum circuit acting on the $n$-qubit product state $\ket{0}^{\otimes n}$.  We have access to $p$ appropriately bounded and tunable parameters to prepare and vary over a family of quantum states. All coefficients herein are assumed to be accurate to not more than $\text{poly}(n)$ decimal places.  We will define an objective function that when minimized will produce a state close to the desired quantum circuit output. We will provide a~solution to the minimization problem.  

\section{Notions of quantum computational universality} 
There are different notions of {\it computational universality} in the literature.  A strong notion is algebraic, wherein a system is called universal if its generating Lie algebra is proven to span ${ \mathfrak{su}}(2^n)$ for $n$ qubits.  Here we call this {\it controllability}.  An alternative notion is that of {\it computational universality}, in which a system is proven to emulate a universal set of quantum gates---which implies directly that the system has access to any polynomial time quantum algorithm on $n$-qubits (the power of quantum algorithms in the class \BQP{}). Evidently these two notions can be related (or even interrelated): by proving that a controllable system can efficiently simulate a universal gate set, a controllable system becomes computationally universal.  It is conversely anticipated by the strong Church-Turing-Deutsch principle~\cite{deutsch1985quantum}, that a universal system can be made to simulate any controllable system. 

The present chapter follows the work I did in \cite{UVQC}, where we assume access to control sequences which can create quantum gates such as \cite{PhysRevLett.90.247901, PhysRevA.70.032314, SHI02}.  Given a quantum circuit of $L$ gates, preparing a state 
$$
\ket{\psi} = \prod_{l=1}^L U_l \ket{0}^{\otimes n}
$$
for unitary gates $U_l$, we construct a universal objective function that is minimised by $\ket{\psi}$.  The objective function is engineered to have certain desirable properties.  Importantly, a gapped Hamiltonian and minimisation past some fixed tolerance, ensures sufficient overlap with the desired output state $\ket{\psi}$.  Recent work of interest by Lloyd considered an alternative form of universality, which is independent of the objective function being minimised \cite{2018arXiv181211075L}.  

Specifically, in the case of {\sf QAOA} the original goal of the algorithm was to alternate and target- and a driver-Hamiltonian to evolve the system close to the target Hamiltonian's ground state---thereby solving an optimization problem instance. Lloyd showed that alternating a driver and target Hamiltonian can also be used to perform universal quantum computation: the times for which the Hamiltonians are applied can be programmed to give computationally universal dynamics \cite{2018arXiv181211075L}.  In related work \cite{morales2019universality}, myself with two coauthors extended Lloyd's {\sf QAOA} universality result \cite{2018arXiv181211075L}.   

In yet another approach towards computational universality, Hamiltonian minimization has long been considered as a means towards universal quantum computation---in the setting of adiabatic quantum computation \cite{2004quant.ph..5098A, BL08}. In that case however, such mappings adiabatically prepare quantum states close to quantum circuit outputs.  Importantly, unlike in ground state quantum computation, in universal variational quantum computation we need not simulate Hamiltonians explicitly.  We instead expand the Hamiltonians in the Pauli basis, and evaluate expected values of these operators term-wise (with tolerance $\sim\epsilon$ for some $\sim\epsilon^{-2}$ measurements---see Hoeffding's inequality~\cite{doi:10.1080/01621459.1963.10500830}).  Measurement results are then paired with their appropriate coefficients and the objective function is calculated classically.  Hence, we translate universal quantum computation to that of (i) state preparation (ii) followed by measurement in a single basis where (iii) the quantum circuit being simulated is used to seed an optimizer to ideally reduce coherence time. 

After introducing variational quantum computation as it applies to our setting, we construct an objective function (called a telescoping construction).  The number of expected values has no dependence on Clifford gates appearing in the simulated circuit and is efficient for circuits with $\mathcal{O}(\text{poly} \ln n)$ non-Clifford gates, making it amenable for near term demonstrations.   We then modify the Feynman--Kitaev clock construction and prove that universal variational quantum computation is possible by minimising $\mathcal{O}(L^2)$ expected values while introducing not more than $\mathcal{O}(\ln L)$ slack qubits, for a quantum circuit partitioned into $L$ Hermitian blocks.  

We conclude by considering how the universal model of variational quantum computation can be utilised in practice.  In particular, the given gate sequence prepares a state which will minimise the objective function.  In practice, we think of this as providing a starting point for a classical optimizer. Given a $T$ gate sequence, we consider the first $L \leq T$ gates.  This $L$ gate circuit represents an optimal control problem where the starting point is the control sequence to prepare the $L$ gates.  The goal is to modify this control sequence (shorten it) using a variational feedback loop.  We iterate this scenario, increasing $L$ up to $T$.  Hence, the universality results proven here also represent a means towards optimal control which does not suffer from the exponential overheads of classically simulating quantum systems.

\section{Maximizing projection onto a circuit} 

We will now explicitly construct an elementary Hermitian penalty function that is non-negative, with a non-degenerate lowest ($0$) eigenstate.  Minimisation of this penalty function prepares the output of a quantum circuit.  We state this in Lemma \ref{thm:tele}. 

\begin{theorem}[Telescoping Theorem---Biamonte \cite{UVQC}] \label{thm:tele}
Consider $\prod_l U_l \ket{0}^{\otimes n}$ a $L$-gate quantum circuit preparing state $\ket{\psi}$ on $n$-qubits and containing not more than $\mathcal{O}(\text{poly} \ln n)$ non-Clifford gates. Then there exists a Hamiltonian $\mathcal{H}\geq0$ on $n$-qubits with $\text{poly}(L, n)$ cardinality, a $(L, n)$-independent gap $\Delta$ and non-degenerate ground eigenvector $\ket{\phi}\propto\prod_l U_l \ket{0}^{\otimes n}$.  In particular, a variational sequence exists causing the Hamiltonian to accept $\ket{\phi}$ viz., $0\leq \bra{\phi}\mathcal{H}\ket{\phi}  < \Delta$ then Theorem \ref{thm:e2overlap} implies stability (Theorem \ref{thm:e2overlap}).  
\end{theorem} 

To prove Theorem \ref{thm:tele}, first we show existence of the penalty function. Construct Hermitian $\mathcal{H} \in \mathscr{L}(\mathbb C_2^{\otimes n})$ with $\mathcal{H}\geq 0$ such that there exists a non-degenerate  $\ket{\psi}\in \mathbb C_2^{\otimes n}$  with the property that $\mathcal{H}\ket{\psi}=0$.   We will view the Hamiltonian $\mathcal{H}$ as a penalty function preparing the initial state and restrict $\mathcal{H}$ to have bounded cardinality ($\text{poly}(n)$ non-vanishing terms in the Pauli-basis).  Define $P_\phi$ as a sum of projectors onto product states, i.e. 
\begin{equation}\label{eqn:proj}
P_\phi = \sum_{i=1}^n \ket{1}\bra{1}^{(i)} = \frac{n}{2}\left(\openone - \frac{1}{n}\sum_{i=1}^n   Z^{(i)} \right) 
\end{equation} 
and consider \eqref{eqn:proj} as the initial Hamiltonian, preparing state $\ket{0}^{\otimes n}$.  
We will act on \eqref{eqn:proj} with a sequence of gates $ \prod_{l=1}^L U_l$ corresponding to the circuit being simulated as 
\begin{equation}\label{eqn:isoaffine}
h(k) = \left(\prod_{l=1}^{k\leq L} U_l \right)P_\phi \left(\prod_{l=1}^{k\leq L} U_l\right)^\dagger \geq 0
\end{equation}
which preserves the spectrum. From the properties of $P_\phi$ it hence follows that $h(k)$ is non-negative and non-degenerate $\forall k \leq L$.  We now consider the action of the gates \eqref{eqn:isoaffine} on \eqref{eqn:proj}. 

At $k=0$ from \eqref{eqn:proj} there are $n$ expected values to be minimized plus a global energy shift that will play a multiplicative role as the circuit depth increases. To consider $k=1$ we first expand a universal gate set expressed in the linear extension of the Pauli basis.  

Interestingly, the coefficients $\mathcal{J}^{a b \dots c}_{\alpha \beta \dots \gamma}$ of the gates will not serve as direct input(s) to the quantum hardware; these coefficients play a direct role in the classical  step where the coefficients weight the sum to be minimized.  Let us then consider single qubit gates, in general form viz.,  
\begin{equation}
e^{-\imath {\bf{a}.\boldsymbol {\sigma}} \theta} = \openone \cos(\theta) - \imath {\bf{a}.\boldsymbol {\sigma}} \sin(\theta) 
\end{equation} 
where $\bf a$ is a unit vector and ${\bf{a}.\boldsymbol {\sigma}}  = \sum_{i=1}^3 a_i\sigma_i$.  So each single qubit gate increases the number of expected values by a factor of at most $4^2$.  At first glance, this appears prohibitive yet there are two factors to consider.  The first is the following Lemma (\ref{lemma:invariance}). 

\begin{lemma}[Clifford Gate Cardinality Invariance---Biamonte \cite{UVQC}] \label{lemma:invariance} 
Let $\mathcal{C}$ be the set of all Clifford circuits on $n$ qubits, and let $\mathcal{P}$ be the set of all elements of the Pauli group on $n$ qubits.  Let $C\in\mathcal{C}$ and $P\in\mathcal{P}$ then it can be shown that 
$$
CPC^\dagger \in \mathcal{P}
$$ 
or in other words 
$$
C\left(\sigma^a_\alpha \sigma^b_\beta \cdots \sigma^c_\gamma \right)C^\dagger = \sigma^{a'}_{\alpha'} \sigma^{b'}_{\beta'} \cdots \sigma^{c'}_{\gamma'}
$$ 
and so Clifford circuits act by conjugation on tensor products of Pauli operators to produce tensor products of Pauli operators.   
\end{lemma}

For some $U$ a Clifford gate, Lemma \ref{lemma:invariance} shows that the cardinality is invariant.  Non-Clifford gates increase the cardinality by factors $\mathcal{O}(e^n)$ and so must be logarithmically bounded from above.  Hence, telescopes bound the number of expected values by restricting to circuit's with 
$$
k\sim\mathcal{O}(\text{poly} \ln n)
$$ 
general single qubit gates. Clifford gates do however modify the locality of terms appearing in the expected values---this is highly prohibitive in adiabatic quantum computation yet arises here as local measurements. 

A final argument supporting the utility of telescopes is that the initial state is restricted primarily by the initial Hamiltonian having only a polynomial number of non-vanishing coefficients in the Pauli basis.  In practice---using today's hardware---it should be possible to prepare an $\epsilon$-close 2-norm approximation to any product state 
$$
\bigotimes_{k=1}^n \left(\cos \theta_k \ket{0}+e^{\imath \phi_k}\sin \theta_k \ket{1}\right)
$$
which is realised by modifying the projectors in \eqref{eqn:proj} with a product of single qubit maps $\bigotimes_{k=1}^n U_k$.  Other more complicated states are also possible.  

Finally to finish the proof of Lemma \ref{thm:tele}, the variational sequence is given by the description of the gate sequence itself.  Hence a state can be prepared causing the Hamiltonian to accept and stability applies (Lemma \ref{thm:e2overlap}). 

To explore telescopes in practice, let us then explicitly consider the quantum algorithm for state overlap (a.k.a., {\it swap test} see e.g.~\cite{2018NJPh...20k3022C}).  This algorithm has an analogous structure to phase estimation, a universal quantum primitive which  form the backbone of error-corrected quantum algorithms.  

\begin{example} 
We are given two $d$-qubit states $\ket{\rho}$ and $\ket{\tau}$ which will be non-degenerate and minimal eigenvalue states of some initial Hamiltonian(s) on $n+1$ qubits 
\begin{equation}
h(0) \ket{+, \rho, \tau} = 0 
\end{equation}
corresponding to the minimization of $\text{poly}(n/2)+1$ expected values where the first qubit (superscript 1 below) adds one term and is measured in the $X$-basis. The controlled swap gate takes the form 
\begin{equation}
[U_\text{swap}]^1_m=\frac{1}{2}\left(\openone^1 +Z^1\right)\otimes \openone^m + \frac{1}{2}\left(\openone^1 - Z^1\right)\otimes \mathcal{S}^m
\end{equation} 
where $m=(i,j)$ indexes a qubit pair and the exchange operator of a pair of qubit states is $\mathcal{S}=\openone + \boldsymbol {\sigma}.\boldsymbol {\sigma}$.  For the case of $d=1$ we arrive at the simplest (3-qubit) experimental demonstration.  At the minimum ($=0$), the expected value of the first qubit being in logical zero is $\frac{1}{2}+\frac{1}{2}|\braket{\rho}{\tau}|^2$.  The final Hadamard gate on the control qubit is considered in the measurement step. 
\end{example} 

Telescopes offer some versatility yet fail to directly prove universality in their own right.  The crux lies in the fact that we are only allowed some polynomial in $\ln n$ non-Clifford gates (which opens an avenue for classical simulation, see \cite{Bravyi_2016, Bravyi_2019}).  Interestingly however, we considered the initial Hamiltonian in \eqref{eqn:proj} as a specific sum over projectors.  We instead could bound the cardinality by some polynomial in $n$.  Such a construction will now be established.  It retroactively proves universality of telescopes and in fact uses telescopes in its construction.  However true this might be, the spirit is indeed lost.  Telescopes are a tool which gives some handle on what we can do without adding additional slack qubits. The universal construction then follows. 

\section{Maximizing projection onto the history state}

We will now prove the following theorem (\ref{thm:history}) which establishes universality of the variational model of quantum computation. 

\begin{theorem}[Universal Objective Function---Biamonte \cite{UVQC}]\label{thm:history}
Consider a quantum circuit of $L$ gates on $n$-qubits producing state $\prod_l U_l \ket{0}^{\otimes n}$.  Then there exists an objective function (Hamiltonian, $\mathcal{H}$) with non-degenerate ground state, cardinality $\mathcal{O}(L^2)$ and spectral gap $\Delta\geq \mathcal{O}(L^{-2})$ acting on $n+\mathcal{O}(\ln L)$ qubits such that acceptance implies efficient preparation of the state $\prod_l U_l \ket{0}^{\otimes n}$.  Moreover, a variational sequence exists causing the objective function to accept.  
\end{theorem} 

To construct an objective function satisfying Theorem \ref{thm:history}, we modify the Feynman-Kitaev clock construction \cite{Fey82, KSV02}.  Coincidentally (and tangential to our objectives here), this construction is also used in certain definitions of the complexity class quantum-Merlin-Arthur (\QMA{}), the quantum analog of \NP{}, through the \QMA{}-complete problem k-\LH{}~\cite{KSV02}.  

Feynman developed a time-independent Hamiltonian that induces unitary dynamics to simulate a sequence of gates \cite{Fey82}. Consider 
\begin{equation} \label{eqn:hpropfey}
\begin{split}
\tilde{\mathcal{H}}_t &= U_t\otimes \ket{t}\bra{t-1} + U_t^\dagger \otimes \ket{t-1}\bra{t} \\
 \tilde{\mathcal{H}}_{\text{prop}} & = \sum_{t=1}^L \tilde{\mathcal{H}}_t
 \end{split}
\end{equation}
where the Hamiltonian \eqref{eqn:hpropfey} acts on a clock register (right) with orthogonal clock states $0$ to $L$ and an initial state $\ket{\xi}$ (left).  Observation of the clock in state $\ket{L}$ after some time $s=s_\star$ produces 
\begin{equation}
\openone \otimes \bra{L} e^{-\imath \cdot s \cdot \tilde{\mathcal{H}}_{\text{prop}}} \ket{\xi}\otimes \ket{0} = U_L \cdots U_1 \ket{\xi}. 
\end{equation} 

The Hamiltonian $\mathcal{H}_{\text{prop}}$ in \eqref{eqn:hpropfey} can be modified as \eqref{eqn:hprop2} so as to have the history state \eqref{eqn:hist} as its ground state 
\begin{equation} \label{eqn:hprop2}
- U_t\otimes \ket{t}\bra{t-1} - U_t^\dagger \otimes \ket{t-1}\bra{t} + \ket{t}\bra{t} + \ket{t-1}\bra{t-1}
  = 2\cdot \mathcal{H}_t \geq 0 
\end{equation}
where $\mathcal{H}_t$ is a projector.  Then $\mathcal{H}_{\text{prop}} = \sum_{t=1}^L \mathcal{H}_t$ has the history state as its ground state as 
\begin{equation}\label{eqn:hist}
\ket{\psi_{\text{hist}}} = \frac{1}{\sqrt{L+1}} \sum_{t=0}^L U_t\cdots U_1\ket{\xi}\otimes \ket{t} 
\end{equation}
for any input state $\ket{\xi}$ where $0 = \bra{\psi_{\text{hist}}}\mathcal{H}_{\text{prop}} \ket{\psi_{\text{hist}}}$.  This forms the building blocks of our objective function.  We will hence establish Theorem \ref{thm:history} by a series of lemma.  

\begin{lemma}[Lifting the degeneracy] \label{lemma:degen}
Adding the tensor product of a projector with a telescope 
\begin{equation}
\mathcal{H}_{\text{in}} = V\left( \sum_{i = 1}^n P_1^{(i)} \right)V^\dagger \otimes P_0
\end{equation} 
lifts the degeneracy of the ground space of $\mathcal{H}_{\text{prop}}$ and the history state with fixed input as 
\begin{equation}
\frac{1}{\sqrt{L+1}} \sum_{t=0}^L \prod_{l = 1}^t U_l(V\ket{0}^{\otimes n}) \otimes \ket{t} 
\end{equation}
is the non-degenerate ground state of $J\cdot \mathcal{H}_{\text{in}} + K\cdot  \mathcal{H}_{\text{prop}}$ for real $J, K>0$.  
\end{lemma} 

\begin{proof}[Lemma \ref{lemma:degen}]
The lowest energy subspace of $\mathcal{H}_{\text{prop}}$ is spanned by $\ket{\psi_{\text{hist}}}$ which has degeneracy given by freedom in choosing any input state $\ket{\xi}$.  To fix the input, consider a tensor product with a telescope 
\begin{equation}
\mathcal{H}_{\text{in}} = V\left( \sum_{i = 1}^n P_1^{(i)} \right)V^\dagger \otimes P_0
\end{equation} 
for $P_1 = \ket{1}\bra{1} = \openone - P_0$ acting on the qubits labeled in the superscript $(i)$ and $P_0$ on the clock space (right).  It is readily seen that $\mathcal{H}_{\text{in}}$ has unit gap and 
\begin{equation} 
\ker \{\mathcal{H}_{\text{in}}\} = \text{span}\{ \ket{\zeta}\otimes \ket{c}, V\ket{0}^{\otimes n}\otimes \ket{0} | 0<c \in \mathbb{N}_+ \leq L, \ket{\zeta}\in \mathbb{C}_2^{\otimes n}\} 
\end{equation} 
Now for positive $J$, $K$ 
\begin{equation}
\arg\min \{J\cdot \mathcal{H}_{\text{in}} + K\cdot  \mathcal{H}_{\text{prop}} \} \propto \frac{1}{\sqrt{L+1}} \sum_{t=0}^L \prod_{l = 1}^t U_l(V\ket{0}^{\otimes n}) \otimes \ket{t} 
\end{equation} 
\end{proof}

\begin{lemma}[Existence of a gap---Biamonte \cite{UVQC}]\label{lemma:gap}
For appropriate non-negative $J$ and $K$, the operator $J\cdot \mathcal{H}_{\text{in}} + K\cdot  \mathcal{H}_{\text{prop}}$ is gapped with a non-degenerate ground state and hence, Theorem \ref{thm:e2overlap} applies with 
\begin{equation}
\Delta \geq \max\left\{ J,\,  \frac{K \pi^2}{2(L+1)^2} \right\}. 
\end{equation} 
\end{lemma} 

\begin{proof}[Lemma \ref{lemma:gap}] 
$\mathcal{H}_{\text{prop}}$ is diagonalized by the following unitary transform (see the technique in \cite{KSV02}) 
\begin{equation} 
W = \sum_{t=0}^L U_t\cdots U_1\otimes \ket{t}\bra{t} 
\end{equation} 
then $W\mathcal{H}_{\text{prop}}W^\dagger$ acts as identity on the register space (left) and induces a quantum walk on a 1D line on the clock space (right). Hence the eigenvalues are known to be $\lambda_k = 1-\cos\left(\frac{\pi k}{1+L}\right)$ for integer $0\leq k \leq L$. 
From the standard inequality, $1-\cos(x)\leq x^2/2$, we find that $\mathcal{H}_{\text{prop}}$ has a gap lower bounded as 
\begin{equation} 
\lambda_0 = 0 \leq \frac{\pi^2}{2(L+1)^2} \leq \lambda_1
\end{equation}  

From Weyl's inequalities, it follow that $J\cdot \mathcal{H}_{\text{in}} + K\cdot  \mathcal{H}_{\text{prop}}$ is gapped as
\begin{align}
\lambda_0 = 0 &< \max\{ \lambda_1(J\cdot \mathcal{H}_{\text{in}}), \lambda_1(K\cdot  \mathcal{H}_{\text{prop}}) \}\\
&\leq \lambda_1(J\cdot \mathcal{H}_{\text{in}} + K\cdot  \mathcal{H}_{\text{prop}}) \\
&\leq \min\{ \lambda_{n-1}(J\cdot \mathcal{H}_{\text{in}}), \lambda_{n-1}(K\cdot  \mathcal{H}_{\text{prop}}) \}
\end{align}
with a non-degenerate ground state and hence, Theorem \ref{thm:e2overlap} applies with 
\begin{equation}
\Delta \geq \max\left\{ J,\,  \frac{K \pi^2}{2(L+1)^2} \right\}
\end{equation} 
\end{proof} 

\begin{lemma}[$\mathcal{H}_{\text{prop}}$ admits a log space embedding---Biamonte \cite{UVQC}] \label{lemma:logem}
The clock space of $\mathcal{H}_{\text{prop}}$ embeds into $\mathcal{O}(\ln L)$ slack qubits, leaving the ground space of $J\cdot \mathcal{H}_{\text{in}} + K\cdot  \mathcal{H}_{\text{prop}}$ and the gap invariant. 
\end{lemma}

\begin{proof}[Lemma \ref{lemma:logem}---Biamonte \cite{UVQC}] 
An $L$-gate circuit requires at most $k = \lceil \ln_2 L \rceil$ clock qubits. Consider a projector $P$ onto the orthogonal compliment of a basis state given by bit string ${\bf x} = x_1 x_2 \dots x_k$.  Then 
\begin{equation} \label{eqn:clockpro}
P_{\bf x} = \ket{\bar{\bf x}}\bra{\bar{\bf x}} = \bigotimes_{i=1}^{\lceil \ln_2 L \rceil} \frac{1}{2}\left(\openone + (-1)^{x_i} Z_i\right) 
\end{equation} 
where $\bar{\bf x}$ is the bitwise logical compliment of ${\bf x}$.   
\end{proof} 

\begin{lemma}[Existence and Acceptance---Biamonte \cite{UVQC}] \label{lemma:uob}
The objective function $J\cdot \mathcal{H}_{\text{in}} + K\cdot  \mathcal{H}_{\text{prop}}$ satisfies Theorem \ref{thm:history}. 
The gate sequence $\prod_l U_l \ket{0}^{\otimes n}$ is accepted by the objective function from Lemma \ref{lemma:uob} thereby satisfying Theorem \ref{thm:history}. 
\end{lemma}

\begin{proof}
We will add $M$ identity gates to boost the probability of the desired circuit output state $\ket{\phi } = \Pi _{l=1}^LU_l \ket{0}^{\otimes n}$. The telescoping construction, we have that 
\begin{equation}\label{eqn:application}
1 - \frac{\bra{\phi}\mathcal{H}\ket{\phi} }{\max\{ J,  \frac{K \pi^2}{2(L+1)^2} \}} \leq | \braket{\phi}{\psi_{\text{hist}}}|^2 = \frac{1}{1+\frac{L+1}{M}}
\end{equation} 
whenever $\bra{\phi }\mathcal{H}\ket{\phi } < \max\{ J,  \frac{K \pi^2}{2(L+1)^2} \}$.  For large enough $M>L$, the right hand side of \eqref{eqn:application} approaches unity, implying acceptance. 

The term \eqref{eqn:clockpro} contributes $L$ terms and hence so does each of the four terms in $\mathcal{H}_t$ from \eqref{eqn:hprop2}.  Hence, the entire sum contributes $3\cdot L^2$ expected values, where we assume $U=U^\dagger$ and that $L$ is upper bounded by some family of circuits requiring $O(\text{poly}~ n)$ gates.  The input penalty $\mathcal{H}_{\text{in}}$ contributes $n$ terms and for an $L$-gate circuit on $n$-qubits we arrive at a total of $\mathcal{O}(\text{poly}~ L^2)$ expected values and $\mathcal{O}(\lceil \ln_2 L \rceil)$ slack qubits.  Adding identity gates to the circuit can boost output probabilities, causing the objective function to accept for a state prepared by the given quantum circuit.  
\end{proof} 

We are faced with considering self-inverse gates.  Such gates ($U$) have a spectrum $\text{Spec}(U)\subseteq\{\pm1\}$, are bijective to idempotent projectors ($P^2=P=P^\dagger$), viz. ${U = \openone - 2P}$ and if $V$ is a self-inverse quantum gate, so is the unitary conjugate $\tilde{V}= G V G^\dagger$ under arbitrary $G$.  Shi showed that a set comprising the controlled not gate (a.k.a.~Feynman gate) plus any one-qubit gate whose square does not preserve the computational basis is universal \cite{SHI02}.  Consider Hermitian 
\begin{equation}
R(\theta) = X\cdot \sin(\theta) + Z \cdot \cos(\theta), 
\end{equation} 
then 
\begin{equation}\label{eqn:RR}
e^{\imath \theta Y} = R(\pi / 2) \cdot R(\theta). 
\end{equation} 
Hence, a unitary $Y$ rotation is recovered by a product of two Hermitian operators.  A unitary $X$ rotation is likewise recovered by the composition \eqref{eqn:RR} when considering Hermitian $Y\cdot \sin(\theta) - Z \cdot \cos(\theta)$.  The universality of self-inverse gates is then established, with constant overhead.  Hence and to conclude, the method introduces not more than $\mathcal{O}(L^2)$ expected values while requiring not more than $\mathcal{O}(\ln L)$ slack qubits, for an $L$ gate quantum circuit.

\section{Discussion} 

We have established that variational methods can approximate any quantum state produced by a sequence of quantum gates and hence that variational quantum computation admits a universal model.  It appears evident that this method will yield shorter control sequences compared to the control sequence of the original quantum circuit---that is the entire point.  Indeed, the control sequence implementing the gate sequence being simulated serves as an upper-bound showing that a sequence exists to minimize the expected values.  These expected values are the fleeting resource which must be simultaneously minimized to find a shorter control sequence which prepares the desired output state of a given quantum circuit.    

Although error correction would allow the circuit model to replace methods developed here, the techniques we develop in universal variational quantum computation should augment possibilities in the NISQ setting, particularly with the advent of error suppression techniques  \cite{2016NJPh...18b3023M, PhysRevX.7.021050}.  Importantly, variational quantum computation forms a universal model in its own right and is not (in principle) limited in application scope.  

An interesting feature of the model of universal variational quantum computation is how many-body Hamiltonian terms are realized as part of the measurement process.  This is in contrast with leading alternative models of universal quantum computation.  

In the gate model, many-body interactions must be simulated by sequences of two-body gates.  The adiabatic model applies perturbative gadgets to approximate many-body interactions with two-body interactions  \cite{2004quant.ph..5098A, BL08}.    The variational model of universal quantum computation simulates many body interactions by local measurements.  Moreover the coefficients weighting many-body terms need not be implemented into the hardware directly; this weight is compensated for in the classical-iteration process which in turns controls the quantum state being produced.   
 
Many states cause a considered objective function to accept.  Hence the presented model is somewhat inherently agnostic to how the states are prepared. This enables  experimentalists to now vary the accessible control parameters to minimize an external and iteratively calculated objective function.  Though the absolute limitations of this approach in the absence of error correction are not known, a realizable method of error suppression could get us closer to implementation of the traditional text-book quantum algorithms such as Shor's factorisation algorithm.

\chapter{Gadget Hamiltonian Constructions in Ground State Computation}\label{chap:gadgets}

We considered in detail in \S~\ref{chap:qvspc} how to use gates to realize many-body Hamiltonian terms.  Here we consider a different situation.  One which does not assume access to pulses (unitary gates).  In this setting, which is relevant to a host of physical systems, we must turn to alternative methods to realize $k$-body interactions. This section follows \cite{BL08, Cao_2015}. 

Although adiabatic quantum computation is known to be a universal model of quantum computation \cite{2004quant.ph..5098A, OT06, BL08, CL08} and hence, in principle equivalent to the circuit model, the mappings between an adiabatic process and an arbitrary quantum circuit require significant overhead. 
Currently the approaches to universal adiabatic quantum computation require implementing multiple higher order and non-commuting interactions by means of perturbative gadgets \cite{BL08}.

Early work by Kitaev \emph{et al}.\ \cite{KSV02} established that an otherwise arbitrary Hamiltonian restricted to have at most $5$-body interactions has a ground state energy problem which is complete for the quantum analog of the complexity class \textsc{NP} (\textsc{{\sf QMA}-complete}).  Reducing the locality of the Hamiltonians from 5-body down to 2-body remained an open problem for a number of years.  In their 2004 proof that \textsc{2-local Hamiltonian} is \textsc{{\sf QMA}-Complete}, Kempe, Kitaev and Regev formalized the idea of a perturbative gadget, which finally accomplished this task \cite{KKR06}. Oliveira and Terhal further reduced the problem, showing completeness when otherwise arbitrary 2-body Hamiltonians were restricted to act on a square lattice \cite{OT06}. The form of the simplest \textsc{{\sf QMA}-complete} Hamiltonian is reduced to physically relevant models in \cite{BL08} (see also \cite{CM13}), e.g. 
\begin{equation}
\mathcal{H} = \sum_i h_i Z_i + \sum_{i<j} J_{ij}Z_iZ_j + \sum_{i<j}K_{ij}X_iX_j.  
\end{equation}

Although this model contains only physically accessible terms, programming problems into a universal adiabatic quantum computer \cite{BL08} or an adiabatic quantum simulator \cite{sim11} involves several types of $k$-body interactions (for bounded $k$).
To reduce from $k$-body interactions to 2-body is accomplished through the application of gadgets. Hamiltonian gadgets were introduced as theorem-proving tools in the context of quantum complexity theory yet their experimental realization currently offers the only path towards universal adiabatic quantum computation. In terms of experimental constraints, an important parameter in the construction of these gadgets is a large spectral gap introduced into the slack space as part of a penalty Hamiltonian.


\section{Hamiltonian complexity}\label{sec:qma}

What is the simplest Hamiltonians that allow universal adiabatic quantum computation?  For this we turn to the complexity class quantum-Merlin-Arthur (\QMA{}), the quantum analog of \NP{}, and consider the \QMA{}-complete problem k-\LH{}~\cite{KSV02}.  One solves $k$-\LH{} by determining if there exists an eigenstate with energy above a given value or below another---with a promise that one of these situations is the case---when the system has at most k-local interactions.  A \YES{} instance is shown by providing a witness eigenstate with energy below the lowest promised value.

\begin{definition}
 The $k$-local Hamiltonian problem: The input is a $k$-local Hamiltonian acting on n qubits, which is the sum of poly many Hermitian matrices that act on only $k$ qubits. The input also contains two numbers $a< b \in [0,1]$, such that $\frac{1}{b-a}=O(n^c)$ for some constant $c$. The problem is to determine whether the smallest eigenvalue of this Hamiltonian is less than $a$ or greater than $b$, promised that one of these is the case.
\end{definition}
    
\begin{remark}
The $k$-local Hamiltonian is {\bf QMA}-complete for $k \geq 2$.  The minimisation of $k$-local Hamiltonians is {\bf QMA}-hard for $k \geq 2$.
\end{remark}

The problem 5-\LH{} was shown to be \QMA{}-complete by Kitaev~\cite{KSV02}.  To accomplish this, Kitaev modified the autonomous quantum computer proposed by Feynman~\cite{Fey82}.  This modification later inspired a proof of the polynomial equivalence between quantum circuits and adiabatic evolutions by Aharonov et al.~\cite{2004quant.ph..5098A}.  Kempe, Kitaev and Regev subsequently proved \QMA{}-completeness of $2$-\LH{}~\cite{KKR06}.  Oliveira and Terhal then showed that universality remains even when the $2$-local Hamiltonians act on particles in a subgraph of the {\sc 2D} square lattice~\cite{OT06}. Any \QMA{}-complete Hamiltonian may realize universal adiabatic quantum computation, and so these results are also of interest for the implementation of quantum computation.

Since $1$-\LH{} is efficiently solvable, an open question is to determine which combinations of $2$-local interactions allow one to build \QMA{}-complete Hamiltonians.  Furthermore, the problem of finding the minimum set of interactions required to build a universal adiabatic quantum computer is of practical, as well as theoretical, interest: every type of $2$-local interaction requires a separate type of physical interaction.  To address this question we prove the following theorems:

\begin{theorem}[ZZXX Hamiltonian is \QMA{}-complete---Biamonte \& Love \cite{BL08}]\label{theorem:zzxx}
\emph{The decision problem \SH{} is \QMA{}-complete, with the ZZXX Hamiltonian given as: }
\begin{equation}\label{eqn:zzxx}
    {\mathcal H}_{\text{ZZXX}} = \sum_{i}h_i Z_i+\sum_{i} l_i X_i+{} 
\sum_{i,j}J_{ij}Z_i Z_j+\sum_{i,j}K_{ij}X_iX_j.
\end{equation}
\end{theorem}

\begin{theorem}[ZX Hamiltonian is \QMA{}-complete---Biamonte \& Love \cite{BL08}]\label{theorem:zx}
\emph{The decision problem \AH{} is \QMA{}-complete, with the ZX Hamiltonian given as:}
\begin{equation}\label{eqn:zx}
    {\mathcal H}
    _{\text{ZX}}=\sum_{i}h_i Z_i+\sum_i l_iX_i+{}
\sum_{i<j}J_{ij}Z_i X_j+\sum_{i<j}K_{ij}X_iZ_j.
\end{equation}
\end{theorem}

\section{{\scshape Local Hamiltonian} is {\sf QMA}-complete}

The translation from quantum circuits to adiabatic evolutions began when Kitaev~\cite{KSV02} replaced the time-dependence of gate model quantum algorithms with spatial degrees of freedom using the non-degenerate ground state of a positive semidefinite Hamiltonian:
\begin{equation}\label{eqn:htot}
0=\mathcal{H}\ket{\psi_{\text{hist}}} = 
(\mathcal{H}_{\text{in}}+ \mathcal{H}_{\clock}+ \mathcal{H}_{\text{clockinit}}+ \mathcal{H}_{\text{prop}})\ket{\psi_{\text{hist}}}.
\end{equation}
To describe this, let $T$ be the number of gates in the quantum circuit with gate sequence $U_T\cdots U_2U_1$ and let $n$ be the number of logical qubits acted on by the circuit.  Denote the circuit's classical input by $\ket{x}$ and its output by $\ket{\psi_{\text{out}}}$.  The {\em history state} representing the circuit's entire time evolution is:
\begin{align}
\ket{\psi_{\text{hist}}}&= \frac{1}{\sqrt{T+1}}\biggl[\ket{x}\otimes\ket{0}^{\otimes T}+U_1\ket{x}\otimes\ket{1}\ket{0}^{\otimes T-1}+{}\nonumber\\
&\qquad+U_2U_1\ket{x}\otimes\ket{11}\ket{0}^{\otimes T-2}+{}\nonumber\\
&\qquad+\ldots+{}\\
&\qquad+ U_T\cdots
U_2U_1\ket{x}\otimes\ket{1}^{\otimes T}\biggr],\nonumber
\end{align}
where we have indexed distinct time steps by a $T$ qubit unary clock.  In the following, tensor product symbols separate operators acting on logical qubits (left) and clock qubits (right).

$\mathcal{H}_{\text{in}}$ acts on all $n$ logical qubits and the first clock qubit.  By annihilating time-zero clock states coupled with classical input $x$, $\mathcal{H}_{\text{in}}$ ensures that valid input state ($\ket{x}\otimes\ket{0...0}$) is in the low energy eigenspace:
\begin{equation}\label{eqn:Hin}
\mathcal{H}_{\text{in}}=\sum_{i=1}^n (\eye-\ket{x_i}\bra{x_i})\otimes \ket{0}\bra{0}_1
=\left(\frac{1}{4}\right)\sum_{i=1}^n(\eye -(-1)^{x_i}Z_i)\otimes(\eye+Z_1).
\end{equation}

$\mathcal{H}_{\clock}$ is an operator on clock qubits ensuring that valid unary clock states $\ket{00...0}$, $\ket{10..0}$, $\ket{110..0}$ etc., span the low energy eigenspace:
\begin{equation}\label{eqn:Hclock}
\mathcal{H}_{\clock}=\sum_{t=1}^{T-1}\ketbra{01}{01}_{(t,t+1)}
=\frac{1}{4}\left[(T-1)\eye+Z_1-Z_T-\sum_{t=1}^{T-1}Z_tZ_{(t+1)}\right],
\end{equation}
where the superscript $(t,t+1)$ indicates the clock qubits acted on by the projection. This Hamiltonian has a simple physical interpretation as a line of ferromagnetically coupled spins with twisted boundary conditions, so that the ground state is spanned by all states with a single domain wall. The term $\mathcal{H}_{\text{clockint}}$ applies a penalty $\ket{1}\bra{1}_{t=1}$ to the first qubit to ensure that the clock is in state $\ket{0}^{\otimes T}-$ at time $t=0$.

$\mathcal{H}_{\text{prop}}$ acts both on logical and clock qubits. It ensures that the ground state is the history state corresponding to the given circuit. $\mathcal{H}_{\text{prop}}$ is a sum of $T$ terms, $\mathcal{H}_{\text{prop}} =  \sum_{t=1}^T \mathcal{H}_{{\text{prop}},t}$, where each term checks that the propagation from time $t-1$ to $t$ is correct.  For $2 \leq t \leq T-1$, $\mathcal{H}_{{\text{prop}},t}$ is defined as:
\begin{equation}\label{eqn:Hprop1}
\mathcal{H}_{{\text{prop}},t}\bydef\eye \otimes \ketbra{t-1}{t-1} - U_t\otimes\ketbra{t}{t-1}
- U_t^\dag \otimes \ketbra{t-1}{t} + \eye \otimes\ketbra{t}{t},
\end{equation}
where operators ${\ketbra{t}{t-1}=\ketbra{110}{100}_{(t-1,t,t+1)}}$ etc., act on clock qubits $t-1$, $t$, and $t+1$ and where the operator $U_t$ is the $t^{th}$ gate in the circuit. For the boundary cases ($t=1,T$), one writes $\mathcal{H}_{{\text{prop}},t}$ by omitting a clock qubit ($t-1$ and $t+1$ respectively).

We have now explained all the terms in the Hamiltonian from~\eqref{eqn:htot}---a key building block used to prove the \QMA{}-completeness of $5$-\LH{}~\cite{KSV02}.  The construction reviewed in the present section was also used in a proof of the polynomial equivalence between quantum circuits and adiabatic evolutions~\cite{2004quant.ph..5098A}.  Which physical systems can implement the Hamiltonian model of computation from~\eqref{eqn:htot}?  Ideally, we wish to find a simple Hamiltonian that is in principle realizable using current, or near-future technology.  The ground states of many physical systems are real-valued, such as the ground states of the Hamiltonians from~\eqref{eqn:zzxx} and~\eqref{eqn:zx}.  So a logical first step in our program is to show the \QMA{}-completeness of general real-valued local Hamiltonians.

\subsection{{\scshape Real Hamiltonian} is {\sf QMA}-complete}

\begin{remark}[Real Hamiltonians]
We call Hamiltonian's expressed in the real subset of the Pauli basis, {\it real Hamiltonians}.  That is, qubit Hamiltonians that contain no tensor product terms with odd numbers of $Y$ operator(s). The corresponding ground state energy problem is called {\scshape Real Hamiltonian}. 
\end{remark}

\begin{lemma}
The ground state energy decision problem {\scshape Real Hamiltonian} is {\sf QMA}-complete. 
\end{lemma}

One can show that \RFLH{} is already \QMA{}-complete---leaving the proofs in~\cite{KSV02} otherwise intact and changing only the gates used in the circuits.  $\mathcal{H}_{\text{in}}$ from~\eqref{eqn:Hin} and $\mathcal{H}_{\text{clock}}$ from~\eqref{eqn:Hclock} are already real-valued and at most $2$-local.  Now consider the terms in $\mathcal{H}_{\text{prop}}$ from~\eqref{eqn:Hprop1} for the case of {\em self-inverse} elementary gates $U_t=U_t^\dagger$:
\begin{equation}\label{eqn:Hprop}
\mathcal{H}_{{\text{prop}},t}=\frac{\eye}{4}(\eye-Z_{(t-1)})(\eye+Z_{(t+1)})
-\frac{U}{4}(\eye-Z_{(t-1)})X_t(\eye+Z_{(t+1)})
\end{equation}
For the boundary cases ($t=1,T$), define:
\begin{eqnarray}\label{eqn:H1T}
\mathcal{H}_{\text{prop,1}} &=& \frac{1}{2}(\eye+Z_2)-U_1\otimes\frac{1}{2}(X_1+X_1Z_2)\\\nonumber
\mathcal{H}_{{\text{prop}},T} &=& \frac{1}{2}\left(\eye-Z_{(T-1)}\right)-U_T\otimes\frac{1}{2}\left(X_T-Z_{(T-1)}X_T\right).
\end{eqnarray}
The terms from~\eqref{eqn:Hprop} and~\eqref{eqn:H1T} acting on the clock space are already real-valued and at most $3$-local.  As an explicit example of the gates $U_t$, let us define a universal real-valued and self-inverse 2-qubit gate as in \eqref{eqn:realrotationgate}. 
\begin{equation}\label{eqn:realrotationgate}
R_{ij}(\phi)=\frac{1}{2}(\eye+Z_i)+\frac{1}{2}(\eye-Z_i)\otimes\left(\sin(\phi)X_i+ \cos(\phi)Z_j\right).
\end{equation}
The gate sequence $R_{ij}(\phi)R_{ij}(\pi/2)$ recovers the universal gate from~\cite{RG02}. This is a continuous set of elementary gates parameterized by the angle $\phi$. Discrete sets of self inverse gates which are universal are also readily constructed. For example, Shi showed that a set comprising the {\sf CNOT} plus any one-qubit gate whose square does not preserve the computational basis is universal \cite{SHI02}. We immediately see that a universal set of self-inverse gates cannot contain only the {\sf CNOT} and a single one-qubit gate. However, the set $\{\text{\sf CNOT}, X,\cos\psi X + \sin\psi Z\}$ is universal for any {\em single} value of $\psi$ which is not a multiple of $\pi/4$.

A reduction from {\sc 5-local} to {\sc 2-local Hamiltonian} was accomplished by the use of \emph{gadgets} that reduced $3$-local Hamiltonian terms to $2$-local terms~\cite{KKR06}.  From the results in~\cite{KKR06} (see also~\cite{OT06}) and the \QMA{}-completeness of \RFLH{}, it now follows that \RH{} is \QMA{}-complete and universal for adiabatic quantum computation. We note that the real product $Y_i\otimes Y_j$, or tensor powers thereof, are not necessary in any part of our construction, and so Hamiltonians composed of the following pairwise products of real-valued Pauli matrices are \QMA{}-complete and universal for adiabatic quantum computation:
\begin{eqnarray}\label{eqn:resubset}
&&\{{\eye\otimes \eye},{\eye}\otimes X, \eye\otimes Z,X\otimes\eye, \\\nonumber
&&~~Z\otimes\eye,X\otimes Z,Z\otimes X, X\otimes X,Z\otimes Z\}.
\end{eqnarray}

To prove our Theorems (\ref{theorem:zzxx}) and (\ref{theorem:zx}), we will next show that one can approximate all the terms from~\eqref{eqn:resubset} using either the ZX or ZZXX Hamiltonians---the Hamiltonians from~\eqref{eqn:zzxx} and~\eqref{eqn:zx} respectively.  We do this using perturbation theory~\cite{KKR06,OT06} to construct gadget Hamiltonians that approximate the operators $Z_i X_j$ and $X_j Z_i$ with terms from the ZZXX Hamiltonian as well as the operators $Z_iZ_j$ and $X_iX_j$ with terms from the ZX Hamiltonian.

\begin{remark}[Experimental Gadget Realizations~\cite{menke2019automated, 2017npjQI...3...21C,2019APS..MARA42011M}]
While recent progress in the experimental implementation of adiabatic quantum processors \cite{2006cond.mat..8253H,Boixo2012,BCM+13} suggests the ability to {perform} sophisticated adiabatic quantum computing experiments, the perturbative gadgets require very large values of $\Delta$. This places high demands on experimental control precision by requiring that devices enforce very large couplings between slack qubits while still being able to resolve couplings from the original problem---even though those fields may be orders of magnitude smaller than $\Delta$. Accordingly, if perturbative gadgets are to be used, it is necessary to find gadgets which can efficiently approximate their target Hamiltonians with significantly lower values of $\Delta$. See for example the work in~\cite{menke2019automated, 2017npjQI...3...21C,2019APS..MARA42011M}.
\end{remark}

\begin{example}(Two-Qubit Experimental Proposal: Towards Adiabatic Universality). 
Consider a quantum state promised to be of the following form where $f(x)$ is unknown
\begin{equation}
    \ket{\psi_{\rm in}}\propto \sum_{x\in\{0,1\}^n}(-1)^{f(x)}\ket{x}. 
\end{equation}
Here $n$-variable Boolean function $f(x)$ is promised to be either \emph{constant} or \emph{balanced}~\footnote{A balanced Boolean function outputs $1$ (and $0$) for exactly half of all input strings, while a constant function always outputs the same value $1$ or $0$.}.  The universal building blocks to demonstrate ground state quantum computation by adiabatic evolution can be demonstrated by preforming the Deutsch-Jozsa algorithm on two-qubits.  This demonstration is closely related to other two-qubit variants, such as phase estimation---see \S~\ref{sec:kqpea}. 

The Deutsch-Jozsa algorithm decides the promise, given $\ket{\psi_{\rm in}}$, by means of a Fourier transform over the group $\mathbb{Z}_2$.  Demonstrating this Hamiltonian would demonstrate all the universal building blocks from \cite{BL08}. 

Consider the function $f(x_1)$ yielding state
\begin{equation}
    \ket{\psi_{\rm in}}\propto(-1)^{f(0)}\ket{0} + (-1)^{f(1)}\ket{1}. 
\end{equation}
In the following circuit, denote the outcome of observable $\ket{1}\bra{1}$ by $m$, where $m$ decides the Deutsch-Jozsa promise:\vspace{-0.2in}
\begin{equation*}
~~~~~~~~~~~~~~~~~~~~~~~~~~~~~~\Qcircuit @C=.5em @R=.3em {
                       & & &         &\mbox{$= m$}\\
\lstick{(-1)^{f(0)}\ket{0} + (-1)^{f(1)}\ket{1}}& \qw & \gate{H} & \qw & \meter}
\end{equation*}
If $m=0$ ($m=1$), $f(x_1)$ is constant (balanced).  Now the circuit's history,
\begin{equation}
    \ket{\psi_{\rm hist}} \propto \ket{\psi_{\rm in}}\otimes\ket{0} + H\cdot \ket{\psi_{\rm in}}\otimes\ket{1}, 
\end{equation}
be the ground state of ${\mathcal H}_{\rm in}+{\mathcal H}_{\rm prop}$, where:
\begin{eqnarray*}
{\mathcal H}_{\rm in}&=& \frac{1}{4}\left(\openone+(-1)^{f(0) +f(1)}X_{l}\right)\otimes\left(\openone+ Z_c \right),\\
{\mathcal H}_{\rm prop} &=& \frac{1}{2}\left(\openone - \frac{1}{\sqrt{2}} Z_l\otimes X_c -\frac{1}{\sqrt{2}} X_l\otimes\ X_c\right).
\end{eqnarray*}
${\mathcal H}_{\rm in}$ plays the role of an oracle Hamiltonian provided to us without knowledge of the function. The computation requires a 2-qubit effective subspace (logical qubit $l$ and clock qubit $c$) combined with dual ancillary qubits implementing respective ZX or ZZXX gadgets.  The adiabatic path Hamiltonian,
$$
{\mathcal H}(s) = (1-s) {\mathcal H}_{\rm in} + s ({\mathcal H}_{\rm in}+{\mathcal H}_{\rm prop})
$$
for monotonic $s\in[0,1]$, has $\ket{\psi_{\rm hist}}$ as the $s=1$ ground state.  The experiment is repeated until clock qubit $c$ is measured in $\ket{1}$---thereby projectively measuring qubit $l$, such that $m = H \cdot \ket{\psi_{\rm in}}$.

\end{example}

\begin{example}(Experimental Proposal: Realization of Ground State Quantum Gates).  Realization of universal ground state quantum computation requires the realization of quantum gates implemented as Hamiltonian operators.  In this direction, realization of the gates appearing in \eqref{eqn:realrotationgate}---see also \eqref{eqn:realrotationgate}---would represent a step towards a significant milestone. 

A 4-qubit quantum Fourier transform (QFT) suffices for many applications of recursive phase estimation of the ground state energy of molecules. The QFT circuit requires both controlled and single qubit $R_k \bydef \ket{0}\bra{0} + e^{2\pi \imath/2^k}\ket{1}\bra{1}$ gates as well as single qubit Hadamard gates.  Here we outline a real-valued mapping of the QFT.

In what follows, normalization constants are often omitted.  Consider the state of $n$-qubits $\ket{\psi} = \sum_{x\in\{0,1\}^n}\alpha_x\ket{x}$. Let $\alpha_x = a_x+b_x \imath$ and define a real-valued wave function $\ket{\tilde\psi}$ to represent $\ket{\psi}$ using an extra qubit that indexes the real and imaginary parts of the wavefunction:
\begin{equation}
    \ket{\tilde \psi}=\sum_{x\in\{0,1\}^n}a_x\ket{x}\otimes\ket{0}+\sum_{x\in\{0,1\}^n}b_x\ket{x}\otimes\ket{1}. 
\end{equation}
For example, an arbitrary pure single-qubit state
\begin{equation}
    \ket{\psi}=\cos(\theta)\ket{0}+e^{\imath\phi}\sin(\theta)\ket{1}
\end{equation}
is written as:
\begin{equation}
    \ket{\tilde\psi}=\cos(\theta)\ket{0}\otimes\ket{0}+\cos(\phi)\sin(\theta)\ket{1}\otimes\ket{0}+\sin(\phi)\sin(\theta)\ket{1}\otimes\ket{1}. 
\end{equation}
One now replaces each complex gate $U$ operating on $k$ qubits by its {\it real valued version}, denoted as $\tilde{U}$ and operating on $k+1$ qubits. Let $\mathfrak{Re}\{U\}$ ($\mathfrak{Im}\{U\}$) denote the real (imaginary) part of the operator $U$ and define the real-valued version of any complex gate as:
\begin{equation}
    \tilde U = \mathfrak{Re}\{U\}\otimes \openone + \mathfrak{Im}\{U\}\otimes \mathfrak{Re}\{-\imath \sigma_y\}. 
\end{equation}
The Hadamard gates needed in the QFT are real-valued.  The real-valued version of the $R_k$ gate now acts on an extra qubit:
\begin{equation}
    \tilde R_k = \ket{0}\bra{0}\otimes\openone+ \cos(2\pi /2^k)\ket{1}\bra{1}\otimes \openone + \sin(2\pi /2^k)\ket{1}\bra{1}\otimes(\ket{1}\bra{0}-\ket{0}\bra{1})
\end{equation}
and can be constructed from the gate sequence $R_{ij}(2\pi /2^k)R_{ij}(\pi/2)$.  Similarly, the required controlled $\tilde R_k$ gates are concatenations of 3-qubit rotation gates:
\begin{equation}
\begin{split}
    R_{ijk}(\phi)& = \frac{1}{4}( 3\cdot \openone + Z_i + Z_j - Z_i\otimes Z_j)+{}\\
    &\qquad+\frac{1}{4}( \openone -Z_i)\otimes( \openone -Z_j)\otimes(\sin(\phi)Z_k   +\cos(\phi)Z_k). 
\end{split}
\end{equation}
To construct the QFT, add an extra qubit initialized in $\ket{0}$ and replace each gate with its real-valued version.  Write the Hamiltonian describing the 2-local terms from~\eqref{eqn:Hin} and~\eqref{eqn:Hclock} as well as the 6-local terms from~\eqref{eqn:Hprop}.  The terms in~\eqref{eqn:Hprop} can be reduced to 4-local without perturbation theory~\cite{Nagaj_2007} (see also~\cite{KR03}).  Parallel application of the gadget in~\cite{KKR06} followed by application of ZX (ZZXX) gadgets completes the reduction.

\end{example}


\begin{remark}[Gadget Overview]
A perturbative gadget consists of 
\begin{enumerate}
    \item a slack system acted on by Hamiltonian $\mathcal{H}$, characterized by the spectral gap $\Delta$ between its ground state subspace and excited state subspace,
    \item and a perturbation $V$ which acts on both the slack and the system.
\end{enumerate}
\end{remark}

\begin{remark}[Perturbation Overview]
$V$ perturbs the ground state subspace of $\mathcal{H}$ such that the perturbed low-lying spectrum of the gadget Hamiltonian $\mathcal{\widetilde H}=\mathcal{H}+V$ captures the spectrum of the target Hamiltonian, $\mathcal{H}_\text{targ}$, up to error $\epsilon$. 
\end{remark}

\begin{remark}[Reduction Gadgets~\cite{KKR06}]
The purpose of a gadget is dependent on the form of the desired target Hamiltonian $\mathcal{H}_\text{targ}$. For example, if the target Hamiltonian is $k$-local with $k\ge 3$ while the gadget Hamiltonian is 2-local, the gadget serves as a tool for reducing locality. 
\end{remark}

\begin{remark}[Creation Gadgets~\cite{BL08}]
Also if the target Hamiltonian involves interactions that are difficult to implement experimentally and the gadget Hamiltonian contains only interactions that are physically accessible, the gadget becomes a generator of physically inaccesible terms from accessible ones. For example the gadget which we introduce in \S~\ref{sec:yy} emulates $YY$ interactions from $ZZ$ and $XX$ interactions so might fall into this use category.
\end{remark}

\begin{remark}[Exact, Non-Perturbative Gadgets~\cite{B08}]
Exact, non-perturbative, gadgets~\cite{B08} are commonly used for embedding classical optimization problems into adiabatic quantum computations. These optimization algorithms require diagonal many-body Hamiltonians where one needs a Hamiltonian of form $\mathcal{H}=\eye-|s\rangle\langle{s}|$ with $s$ being a binary string. 
\end{remark}

Apart from the physical relevance to quantum computation, gadgets have been central to many results in quantum complexity theory \cite{BDLT08, BL08, BDOT06, CM13}.  Hamiltonian gadgets were employed to help characterize the complexity of density functional theory \cite{Schuch09} and are required components in current proposals related to error correction on an adiabatic quantum computer \cite{Ganti2013} and the adiabatic and ground state quantum simulator \cite{sim11}.

\begin{remark}[Kempe, Kitaev and Regev \cite{KKR06}]
The first use of perturbative gadgets \cite{KKR06} relied on a 2-body gadget Hamiltonian to simulate a 3-body Hamiltonian of the form $\mathcal{H}_\text{targ} = \mathcal{H}_\text{else}+\alpha \cdot A\otimes B \otimes C$ with three auxiliary spins in the slack space. Here $\mathcal{H}_\text{else}$ is an arbitrary Hamiltonian that does not operate on the auxiliary spins. Further, $A$, $B$ and $C$ are unit-norm operators and $\alpha$ is the desired coupling. For such a system, it is shown that it suffices to construct $V$ with $\|V\|<\Delta/2$ to guarantee that the perturbative self-energy expansion approximates $\mathcal{H}_\text{targ}$ up to error $\epsilon$ \cite{OT06,KKR06,BDLT08}. Because the gadget Hamiltonian is constructed such that in the perturbative expansion (with respect to the low energy subspace), only virtual excitations that flip all 3 slack bits would have non-trivial contributions in the $1^\text{st}$ through $3$rd order terms. 
\end{remark}

\begin{remark}[Jordan and Farhi \cite{JF08}]
In \cite{JF08} Jordan and Farhi generalized the construction in \cite{KKR06} to a general $k$-body to $2$-body reduction using a perturbative expansion that appears to date to 1958 due to Bloch.\footnote{{Sur la th\'{e}orie des perturbations des \'{e}tats li\'{e}s}, {\em Nuclear Physics}, 6:329--347, 1958} They showed that one can approximate the low-energy subspace of a Hamiltonian containing $r$ distinct $k$-local terms using a $2$-local Hamiltonian. 
\end{remark}

\begin{remark}[Oliveira and Terhal \cite{OT06}]
Two important gadgets were introduced by Oliveira and Terhal \cite{OT06} in their proof that \textsc{2-local Hamiltonian on square lattice} is \textsc{{\sf QMA}-Complete}. In particular, they introduced  an alternative 3- to 2-body gadget which uses only one additional spin for each 3-body term as well as a ``subdivision gadget'' that reduces a $k$-body term to a $(\lceil k/2 \rceil+1)$-body term using only one additional spin \cite{OT06}.  
\end{remark}

These gadgets (listed in the three remarks above), which we improved in the work~\cite{Cao_2015}, find their use as the de facto standard whenever the use of gadgets is necessitated. For instance, the gadgets from \cite{OT06} were used by Bravyi, DiVincenzo, Loss and Terhal \cite{BDLT08} to show that one can combine the use of subdivision and 3- to 2-body gadgets to recursively reduce a $k$-body Hamiltonian to $2$-body, which is useful for simulating quantum many-body Hamiltonians.

\section{Exact ZZZ-gadget from Z, ZZ} 

Non-perturbative, exact or classical gadgets, as introduced in \cite{B08}, will be considered along with their limitations and interrelations with perturbative gadgets.  This section employs heavily the techniques we developed from \S~\ref{chap:progGS}. Similar ideas originating from \cite{B08} and \cite{spinlogic2} have been experimentally demonstrated in \cite{melanson2019tunable}---see also \cite{Schndorf2019}. 

Suppose we need to simulate a diagonal target Hamiltonian
\begin{equation}
    \mathcal{H}_\text{targ}=\mathcal{H}_\text{else}+\alpha\bigotimes_{i=1}^k\sigma_i. 
\end{equation}
Here all operators $\sigma_i$ share the same basis and $\mathcal{H}_\text{else}$ commutes with  $\bigotimes_{i=1}^k\sigma_i$. 

\begin{remark}[Perturbative Approach to ZZZ]
The perturbative approach requires a $5$th order gadget to simulate a 3-body term $Z_1Z_2Z_3$.  The gap however scales as $\Delta=\Theta(\epsilon^{-5})$ renders it challenging to realize using current experimental resources. 
\end{remark}

\begin{remark}[Penalty Function Approach to ZZZ]
Unlike their perturbative counterparts, non-perturbative gadgets do not require a large penalty on the slack spins, which render them more realistic to implement experimentally \cite{B08, spinlogic2, BOA13}.  
\end{remark}

\begin{remark}
The idea of a non-perturbative 3-body reduction gadget comes originally from two insights which relate the spectra of diagonal Hamiltonians to Boolean algebra \cite{B08}. 
\begin{enumerate}
    \item The first insight is that the spectrum of $\mathcal{H}_\text{targ}=Z_1Z_2Z_3$ has a unique correspondence with the truth table of the logic operation
    \begin{equation}
        f(s_1,s_2,s_3)=s_1\oplus s_2\oplus s_3,~~~~s_i\in\{0,1\}. 
    \end{equation}
    The ground state subspace of $\mathcal{H}_\text{targ}$ corresponds to all states with $s_1\oplus s_2\oplus s_3=0$ and the first excited subspace corresponds to  $s_1\oplus s_2\oplus s_3=1$. 
    \item The second insight is that the Boolean function {\sf XOR} $s=s_1\oplus s_2$ cannot be mapped to the ground state subspace of a 3-spin diagonal Hamiltonian using only 2-body interactions while {\sf AND} $s=s_1\wedge s_2$ can be implemented with 2-body interactions \cite{B08,spinlogic2}. 
\end{enumerate}
Thus, we see that if we can express $f(s_1,s_2,s_3)$ as a function of only $1$- or $2$-variable {\sf AND} clauses with auxiliary variables, then the spectrum of $Z_1Z_2Z_3$ can be directly mapped to a diagonal Hamiltonian consisting of only $2$-body interactions.
\end{remark}

\begin{remark}[$k$-body exact gadget]
This idea can be generalized to simulate a $k$-body diagonal target term $\bigotimes_{i=1}^k Z_i$ \cite{B08,BOA13}. Since the spectrum of the $k$-body term can be mapped to Boolean expression
\begin{equation}
    f(s_1,\cdots,s_k)=\bigoplus_{i=1}^ks_i, 
\end{equation}
we can introduce $k-3$ auxiliary variables $y_1$, $y_2$, $\cdots$ $y_s$ such that $y_1=s_1\oplus s_2$ and 
\begin{equation}
    y_j=y_{j-1}\oplus s_{j+1},~~~~\text{with}~~j=2,\dots, k-3. 
\end{equation}
Then $f(s_1,\cdots,s_k)$ becomes a sum of 3-variable {\sf XOR} terms up to $k-3$ constraints. Each of the 3-variable {\sf XOR}s can then be reduced to 2-body using gadget described above. 
\end{remark}

\begin{example}
Non-perturbative gadgets \cite{B08} do not require large gaps and they capture the target term exactly. Given these advantages, we ask if these gadgets can replace perturbative gadgets. To answer this, consider two non-commuting Hamiltonians $\mathcal{H}_Z=Z_1Z_2Z_3$ and $\mathcal{H}_X=X_1X_2X_3$. Introduce slack spins $w$ and $w'$ for reducing $\mathcal{H}_Z+\mathcal{H}_X$ to 2-body respectively. For each eigenstate $|j\rangle$ such that
\begin{equation}
    (\mathcal{H}_Z+\mathcal{H}_X)|j\rangle=\beta_j|j\rangle, 
\end{equation}
let $\mathcal{\widetilde{H}}_Z$ and $\mathcal{\widetilde{H}}_X$ be the non-perturbative gadgets generated for $\mathcal{H}_Z$ and $\mathcal{H}_X$ respectively, we are faced with the question whether the following is true for all $j$:
\begin{equation}\label{eq:min_XXX_ZZZ}
\underset{\psi,\phi}{\min}\langle\psi|\langle\phi|\langle{j}|(\mathcal{\widetilde{H}}_Z+\mathcal{\widetilde{H}}_X)|j\rangle|\phi\rangle|\psi\rangle=\beta_j
\end{equation}
where $|\phi\rangle$ and $|\psi\rangle$ denote the states of slack qubits $w$ and $w'$ respectively. We claim that~\eqref{eq:min_XXX_ZZZ} is impossible to hold for all $j$. This would violate the gadget theorem because $\mathcal{H}_Z$ and $\mathcal{H}_X$ induces transitions in each other's slack spaces that cause the perturbation series to no longer converge. However, the condition of the gadget theorem is only sufficient. In order to show that non-perturbative gadgets cannot simulate a sum of non-commuting target terms, a counter-example is easy to find. Therefore, for a general $k$-body target Hamiltonian, one still needs perturbative methods for simulating it using 2-body interactions.
\end{example}

\section{Perturbation theory} 

\begin{definition}
In our notation the spin-1/2 Pauli operators will be represented as $\{X,Y,Z\}$ with subscript indicating which spin-1/2 particle (qubit) it acts on. For example $X_2$ is a Pauli operator $X=|0\rangle\langle{1}|+|1\rangle\langle{0}|$ acting on the qubit labelled as $2$.
\end{definition}

\begin{remark}
In the literature there are different formulations of the perturbation theory that are adopted when constructing and analyzing the gadgets. This adds to the challenge faced in comparing the physical resources required among the various proposed constructions. For example, Jordan and Farhi \cite{JF08} use a formulation due to Bloch, while Bravyi et al.\ use a formulation based on the Schrieffer-Wolff transformation \cite{BDLT08}. Here we employ the formulation used in \cite{KKR06,OT06}. For a small survey on various formulations of perturbation theory, refer to \cite{BDL11}.
\end{remark}

A gadget Hamiltonian 
\begin{equation}
    \mathcal{\tilde{H}}=\mathcal{H}+V
\end{equation}
consists of a penalty Hamiltonian $\mathcal{H}$, which applies an energy gap onto an slack space, and a perturbation $V$. To explain in further detail how the low-lying sector of the gadget Hamiltonian $\mathcal{\tilde{H}}$ approximates the entire spectrum of a certain target Hamiltonian $\mathcal{H}_\text{targ}$ with  error $\epsilon$, we set up the following notations:
\begin{enumerate}
    \item Let $\lambda_j$ and $|\psi_j\rangle$ be the $j^\text{th}$ eigenvalue and eigenvector of $\mathcal{H}$ and similarly define $\tilde\lambda_j$ and $|\tilde\psi_j\rangle$  as those of $\mathcal{\tilde{H}}$, assuming all the eigenvalues are labelled in a weakly increasing order ($\lambda_1\le\lambda_2\le\cdots$, same for $\tilde{\lambda}_j$). 
    
    \item Using a cutoff value $\lambda_*$, let
    \begin{equation}
        \mathscr{L}_-=\text{span}\{|\psi_j\rangle|\forall j:\lambda_j\le\lambda_*\}
    \end{equation}
    be the low energy subspace and
    \begin{equation}
        \mathscr{L}_+=\text{span}\{|\psi_j\rangle|\forall j:\lambda_j>\lambda_*\}
    \end{equation}
    be the high energy subspace.
    
    \item Let ${\Pi_-}$ and ${\Pi_+}$ be the orthogonal projectors onto the subspaces $\mathscr{L}_-$ and $\mathscr{L}_+$ respectively. 
    
    \item  For an operator $O$ we define the partitions of $O$ into the subspaces as $O_-={\Pi_-}O{\Pi_-}$, $O_+={\Pi_+}O{\Pi_+}$, $O_{-+}={\Pi_-}O{\Pi_+}$ and $O_{+-}={\Pi_+}O{\Pi_-}$.
\end{enumerate}

With the definitions above, one can turn to perturbation theory to approximate $\mathcal{\tilde{H}_-}$ using $\mathcal{H}$ and $V$. We now consider the operator-valued resolvent
\begin{equation}
    \tilde{G}(z)=(z\eye-\mathcal{\tilde H})^{-1}. 
\end{equation}
Similarly one would define
\begin{equation}
    G(z)=(z\eye-\mathcal{H})^{-1}. 
\end{equation}
Note that
\begin{equation}
    \tilde{G}^{-1}(z)-G^{-1}(z)=-V
\end{equation}
so that this allows an expansion in powers of $V$ as 
\begin{multline}\label{eq:G_expand}
\tilde{G}=(G^{-1}-V)^{-1}=G(\eye-VG)^{-1}=\\
=G+GVG+GVGVG+GVGVGVG+\cdots.  
\end{multline}
It is then standard to define the self-energy
\begin{equation}
    \Sigma_-(z)=z\eye-({\tilde G}_-(z))^{-1}. 
\end{equation}
The self-energy is important because the spectrum of $\Sigma_-(z)$ gives an approximation to the spectrum of $\mathcal{\tilde{H}_-}$ since by definition
\begin{equation}
    \mathcal{\tilde{H}_-}=z\eye-{\Pi_-}(\tilde{G}^{-1}(z)){\Pi_-}
\end{equation}
while
\begin{equation}
  \Sigma_-(z)=z\eye-({\Pi_-}\tilde G(z){\Pi_-})^{-1}.   
\end{equation}

\begin{remark}
As is explained by Oliveira and Terhal \cite{OT06}, loosely speaking, if $\Sigma_-(z)$ is roughly constant in some range of $z$ (defined below in Theorem \ref{th:perturbation}) then $\Sigma_-(z)$ is playing the role of $\mathcal{\tilde{H}_-}$.
\end{remark}

The gadget theorem was formalized in \cite{KKR06} and improved in \cite{OT06} where the following theorem is proven (as in \cite{OT06} we state the case where $\mathcal{H}$ has zero as its lowest eigenvalue and a spectral gap of $\Delta$. 

\begin{definition}
We use operator norm $\|\cdot\|$ which is defined as
\begin{equation}
    \|M\|\equiv\max_{|\psi\rangle\in\mathscr{M}}|\langle\psi|M|\psi\rangle|
\end{equation}
for an operator $M$ acting on a vectors in the Hilbert space $\mathscr{M}$. 
\end{definition}

\begin{theorem}[Gadget Theorem \cite{KKR06,OT06}]\label{th:perturbation}
 Let
 $$
 \|V\|\le\Delta/2
 $$ 
 where $\Delta$ is the spectral gap of $\mathcal{H}$ and let the low and high spectrum of $\mathcal{H}$ be separated by a cutoff
 $$
 \lambda_*=\Delta/2. 
 $$
 Now let there be an effective Hamiltonian $\mathcal{H}_\text{eff}$ with a spectrum contained in $[a,b]$. If for some real constant $\epsilon>0$ and
 $$
 \forall z\in[a-\epsilon,b+\epsilon]
 $$ 
 with 
 $$
 a<b<\Delta/2-\epsilon, 
 $$
 the self-energy $\Sigma_-(z)$ has the property that
 $$
 \|\Sigma_-(z)-\mathcal{H}_\text{eff}\|\leq\epsilon, 
 $$ 
 then each eigenvalue $\tilde\lambda_j$ of $\mathcal{\tilde{H}_-}$ differs to the $j^\text{th}$ eigenvalue of $\mathcal{H}_\text{eff}$, $\lambda_j$, by at most $\epsilon$. In other words
 $$
 |\tilde{\lambda}_j - \lambda_j|\leq\epsilon, ~~~\forall j.
 $$
\end{theorem}

To apply Theorem \ref{th:perturbation}, a series expansion for $\Sigma_-(z)$ is truncated at low order for which $\mathcal{H}_\text{eff}$ is approximated. The 2-body terms in $\mathcal{H}$ and $V$ by construction can give rise to higher order terms in $\mathcal{H}_\text{eff}$. For this reason it is possible to engineer $\mathcal{H}_\text{eff}$ from $\Sigma_-(z)$ to approximate $\mathcal{H}_\text{targ}$ up to error $\epsilon$ in the range of $z$ considered in Theorem \ref{th:perturbation} by introducing auxiliary spins and a suitable selection of 2-body $\mathcal{H}$ and $V$. Using the series expansion of $\tilde{G}$ in~\eqref{eq:G_expand}, the self-energy
\begin{equation}
    \Sigma_-(z)=z\eye-\tilde{G}_-^{-1}(z)
\end{equation}
can be expanded as (for further details see \cite{KKR06})
\begin{equation}\label{eq:selfenergy}
\Sigma_-(z)={H}_-+V_-+V_{-+}G_+(z)V_{+-}+V_{-+}G_+(z)V_+G_+(z)V_{+-}+\cdots.
\end{equation}
The terms of $2^\text{nd}$ order and higher in this expansion give rise to the effective many-body interactions.

\section{Improved subdivision gadget} 

\begin{figure}
\begin{center}
\makebox[1.6cm][l]{ }\includegraphics[scale=0.15]{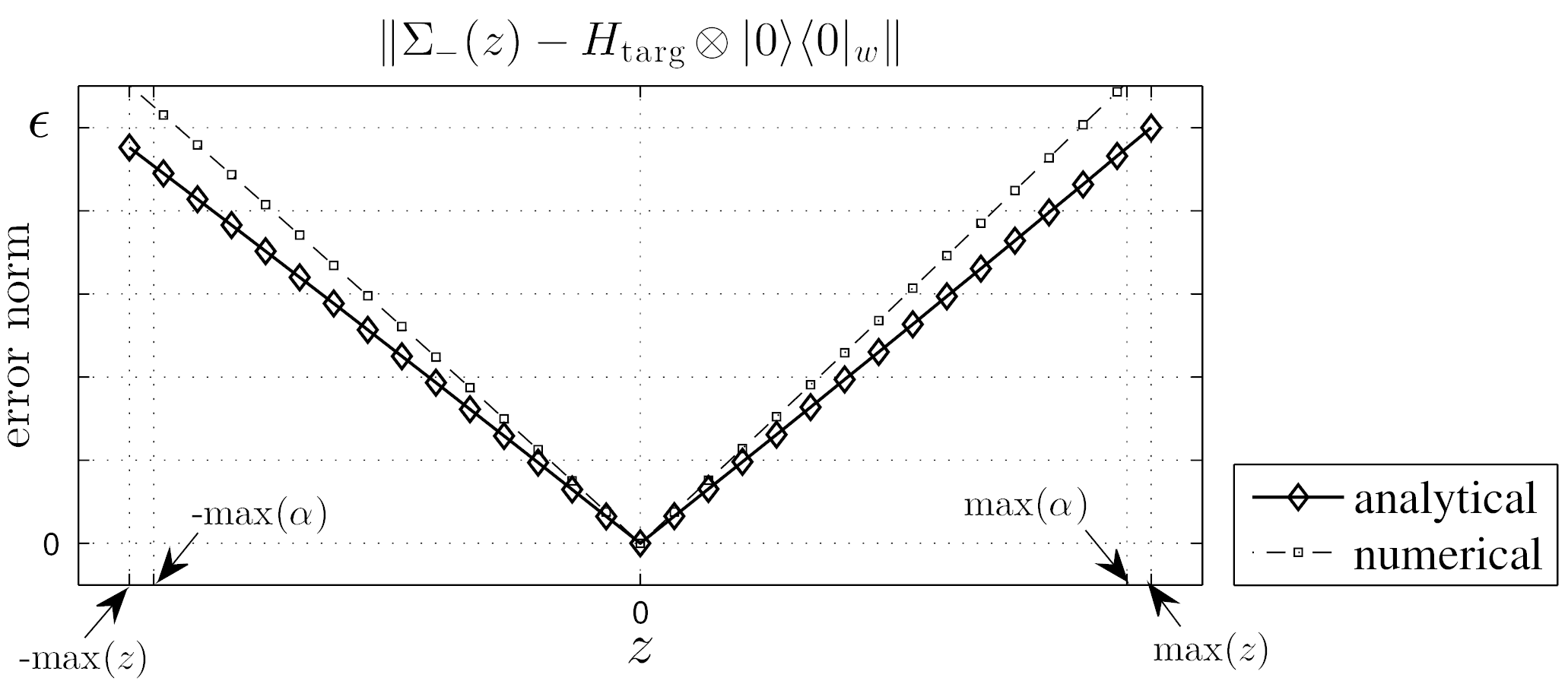}
\centerline{\makebox[-0.2cm][l]{ }(a)}
\makebox[1.6cm][l]{ }\includegraphics[scale=0.15]{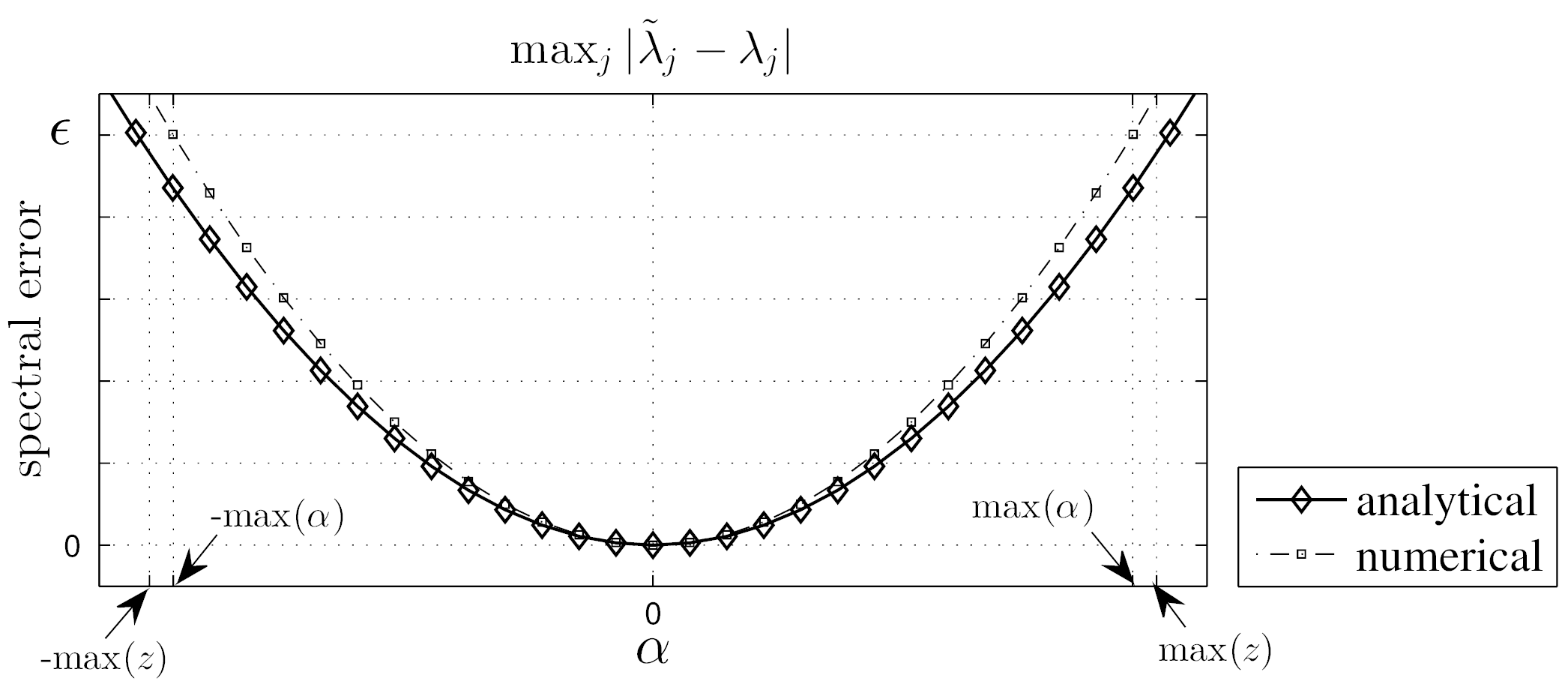}
\centerline{\makebox[-0.2cm][l]{ }(b)}
\end{center}
\caption{\normalsize Numerical illustration of gadget theorem using a subdivision gadget. Here we use a subdivision gadget to approximate $\mathcal{H}_\text{targ}=\mathcal{H}_\text{else}+\alpha Z_1Z_2$ with $\|\mathcal{H}_\text{else}\|=0$ and $\alpha\in[-1,1]$. $\epsilon=0.05$. ``analytical'' stands for the case where the value of $\Delta$ is calculated using~\eqref{eq:2bodyDelta} when $|\alpha|=1$. ``numerical'' represents the case where $\Delta$ takes the value that yield the spectral error to be $\epsilon$. In (a) we let $\alpha=1$. $z\in[-\max z,\max z]$ with $\max z=\|\mathcal{H}_\text{else}\|+\max\alpha+\epsilon$. The operator $\Sigma_-(z)$ is computed up to the $3^\text{rd}$ order. Subplot (b) shows for every value of $\alpha$ in its range, the maximum difference between the eigenvalues $\tilde\lambda_j$ in the low-lying spectrum of $\mathcal{\tilde{H}}$ and the corresponding eigenvalues $\lambda_j$ in the spectrum of $\mathcal{H}_\text{targ}\otimes|0\rangle\langle{0}|_w$. (These plots are from \cite{Cao_2015}).}
\label{fig:2bodygadget}
\end{figure}

The subdivision gadget is introduced by Oliveira and Terhal \cite{OT06} in their proof that \textsc{2-local Hamiltonian on square lattice} is \textsc{{\sf QMA}-Complete}. Here we show an improved lower bound for the spectral gap $\Delta$ needed on the slack of the gadget.
A subdivision gadget simulates a many-body target Hamiltonian $\mathcal{H}_\text{targ}=\mathcal{H}_\text{else}+\alpha \cdot A\otimes B$ ($\mathcal{H}_\text{else}$ is a Hamiltonian of arbitrary norm, $\|A\|=1$ and $\|B\|=1$) by introducing an slack spin $w$ and applying onto it a penalty Hamiltonian $\mathcal{H}=\Delta|1\rangle\langle{1}|_w$ so that its ground state subspace $\mathscr{L}_-=\text{span}\{|0\rangle_w\}$ and its excited subspace $\mathscr{L}_+=\text{span}\{|1\rangle_w\}$ are separated by energy gap $\Delta$. In addition to the penalty Hamiltonian $\mathcal{H}$, we add a perturbation $V$ of the form
\begin{equation}\label{eq:2body_V}
V = \mathcal{H}_\text{else}+|\alpha||0\rangle\langle{0}|_w  + \sqrt{\frac{|\alpha|\Delta}{2}}(\text{sgn}(\alpha)A-B)\otimes X_w.
\end{equation}
Hence if the target term $A\otimes B$ is $k$-local, the gadget Hamiltonian $\mathcal{\tilde{H}}=\mathcal{H}+V$ is at most $(\lceil{k/2}\rceil+1)$-local, accomplishing the locality reduction. Assume $\mathcal{H}_\text{targ}$ acts on $n$ qubits.
Prior work \cite{OT06} shows that $\Delta=\Theta(\epsilon^{-2})$ is a sufficient condition for the lowest $2^n$ levels of the gadget Hamiltonian $\mathcal{\widetilde{H}}$ to be $\epsilon$-close to the corresponding spectrum of $\mathcal{H}_\text{targ}$. However, by bounding the infinite series of error terms in the perturbative expansion, we are able to obtain a tighter lower bound for $\Delta$ for error $\epsilon$.  Hence we arrive at our first result (details will be presented later in this section), that it suffices to let
\begin{equation}\label{eq:2body_D}
\Delta\ge\left(\frac{2|\alpha|}{\epsilon}+1\right)(2\|\mathcal{H}_\text{else}\|+|\alpha|+\epsilon).
\end{equation}

In Figure\ \ref{fig:delta_compare_sub} we show numerics indicating the minimum $\Delta$ required as a function of $\alpha$ and $\epsilon$. In Figure\ \ref{fig:delta_compare_sub}a the numerical results and the analytical lower bound in~\eqref{eq:2body_D} show that for our subdivision gadgets, $\Delta$ can scale as favorably as $\Theta(\epsilon^{-1})$. For the subdivision gadget presented in \cite{OT06},  $\Delta$ scales as $\Theta(\epsilon^{-2})$. Though much less than the original assignment in \cite{OT06}, the lower bound of $\Delta$ in~\eqref{eq:2body_D}, still satisfies the condition of Theorem \ref{th:perturbation}. In Figure\ \ref{fig:delta_compare_sub} we numerically find the minimum value of such $\Delta$ that yields a spectral error of exactly $\epsilon$. 
\begin{figure}[p]
\begin{center}
\begin{minipage}{0.4\textwidth}
\setlength{\unitlength}{1cm}
\begin{picture}(10,7.5)
\put(-4.85,0.42){\includegraphics[scale=0.13]{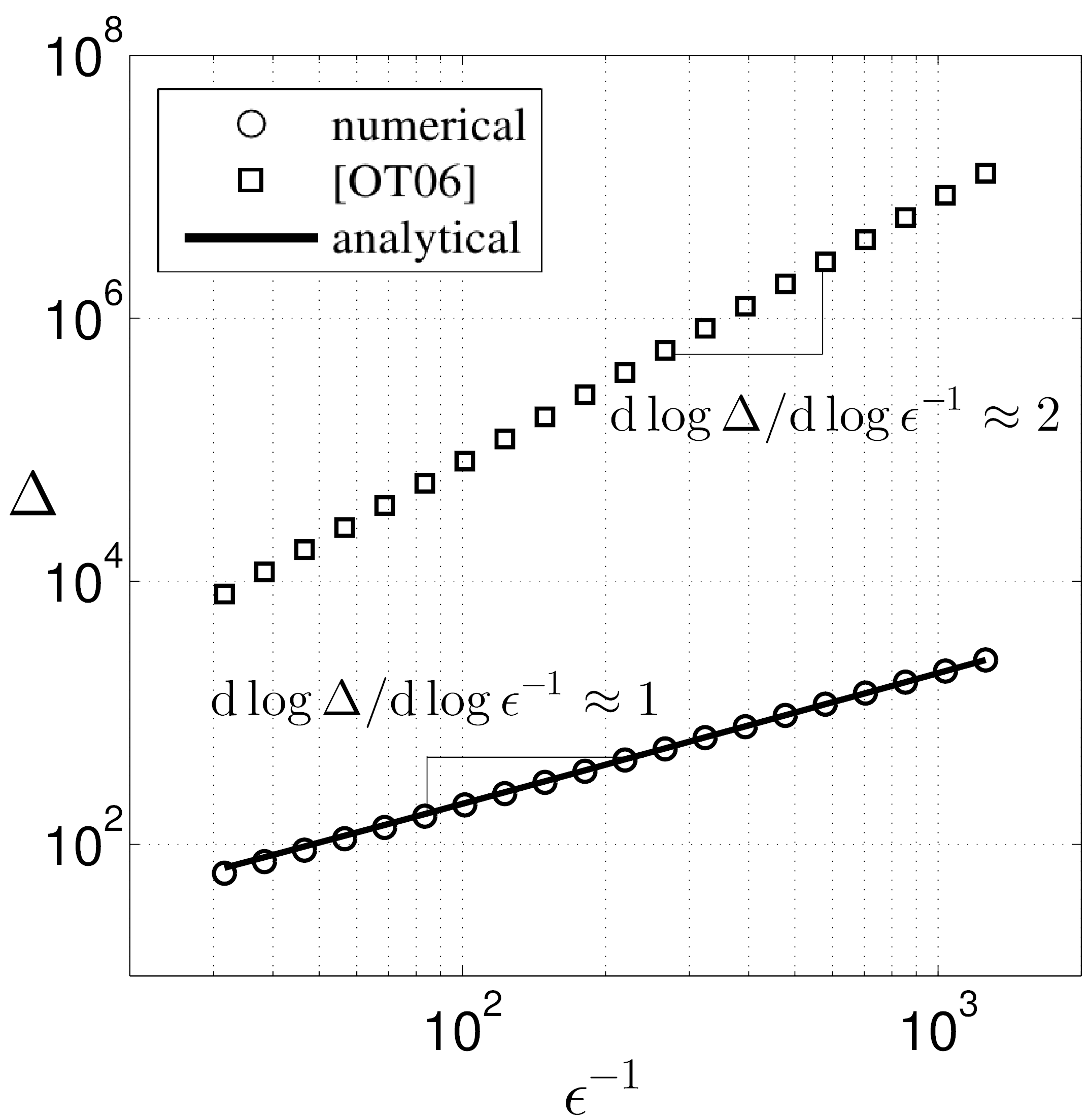}}
\put(2.65,0.28){\includegraphics[scale=0.188]{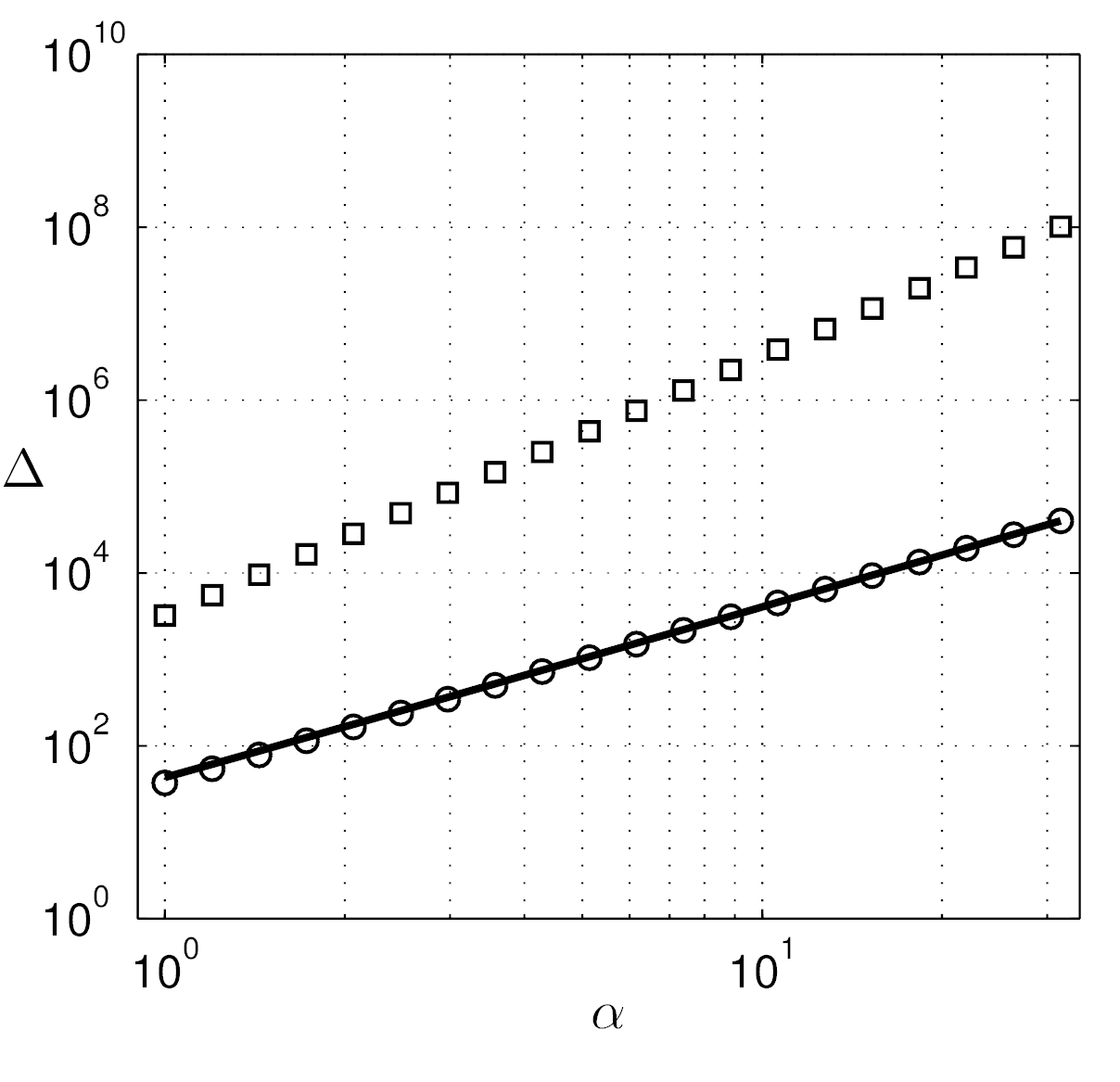}}
\put(-1.1,0){\text{(a)}}
\put(6.8,0){\text{(b)}}
\end{picture}
\end{minipage}
\end{center}
\caption{Comparison between our subdivision gadget with that of Oliveira and Terhal \cite{OT06}. The data labelled as ``numerical'' represent the $\Delta$ values obtained from the numerical search such that the spectral error between $\mathcal{H}_\text{targ}$ and $\mathcal{\widetilde{H}_-}$ is $\epsilon$. The data obtained from the calculation using~\eqref{eq:2body_D} are labelled as ``analytical''. ``[OT06]'' refers to values of $\Delta$ calculated according to the assignment by Oliveira and Terhal \cite{OT06}. In this example we consider $\mathcal{H}_\text{targ}=\mathcal{H}_\text{else}+\alpha Z_1Z_2$. (a) Gap scaling with respect to $\epsilon^{-1}$. Here $\|\mathcal{H}_\text{else}\|=0$ and $\alpha=1$. (b) The gap $\Delta$ as a function of the desired coupling $\alpha$. Here $\|\mathcal{H}_\text{else}\|=0$,  $\epsilon=0.05$.}
\label{fig:delta_compare_sub}
\end{figure}

$\quad$\\

Following joint work with with Cao and others, we will consider analysis of the improved subdivision gadget \cite{Cao_2015}. The previously known subdivision gadgets in the literature assume that the gap in the penalty Hamiltonian $\Delta$ scales as $\Theta(\epsilon^{-2})$ (see for example \cite{OT06,BDLT08}). Here we employ a method which uses infinite series to find the upper bound to the norm of the high order terms in the perturbative expansion. We find that in fact $\Delta=\Theta(\epsilon^{-1})$ is sufficient for the error to be within $\epsilon$. 

The key aspect of developing the gadget is that given $\mathcal{H}=\Delta|1\rangle\langle{1}|_w$, we need to determine a perturbation $V$ to perturb the low energy subspace \[\mathscr{L}_-=\text{span}\{|\psi\rangle\otimes|0\rangle_w,\makebox[0.05cm]{}\text{ $\ket{\psi}$ is any state of the system excluding the slack spin $w$}\}\] such that the low energy subspace of the gadget Hamiltonian $\mathcal{\tilde{H}}=\mathcal{H}+V$ approximates the spectrum of the entire operator $\mathcal{H}_\text{targ}\otimes|0\rangle\langle{0}|_w$ up to error $\epsilon$.  Here we will define $V$ and work backwards to show that it satisfies Theorem \ref{th:perturbation}. We let
\begin{equation}\label{eq:V}
V=\mathcal{H}_\text{else}+\frac{1}{\Delta}({\kappa}^2A^2+{\lambda}^2B^2)\otimes|0\rangle\langle{0}|_w+({\kappa}A+{\lambda}B)\otimes X_w
\end{equation}
\noindent{}where ${\kappa}$, ${\lambda}$ are constants which will be determined such that the dominant contribution to the perturbative expansion which approximates $\mathcal{\tilde{H}_-}$ gives rise to the target Hamiltonian $\mathcal{H}_\text{targ}=\mathcal{H}_\text{else}+\alpha \cdot A\otimes B$. In~\eqref{eq:V} and the remainder of the section, by slight abuse of notation, we use $\kappa A+\lambda B$ to represent $\kappa(A\otimes\eye_\mathcal{B})+\lambda(\eye_\mathcal{A}\otimes B)$ for economy. Here $\eye_\mathcal{A}$ and $\eye_\mathcal{B}$ are identity operators acting on the subspaces $\mathcal{A}$ and $\mathcal{B}$ respectively. The partitions of $V$ in the subspaces are
\begin{equation}
\begin{array}{c}
\displaystyle
V_+=\mathcal{H}_\text{else}\otimes|1\rangle\langle{1}|_w,\quad V_-=\left(\mathcal{H}_\text{else}+\frac{1}{\Delta}({\kappa}^2A^2+{\lambda}^2B^2)\eye\right)\otimes|0\rangle\langle{0}|_w,\\[0.1in] V_{-+}=({\kappa}A+{\lambda}B)\otimes|0\rangle\langle{1}|_w,\quad V_{+-}=({\kappa}A+{\lambda}B)\otimes|1\rangle\langle{0}|_w.
\end{array}
\end{equation}

\noindent{}We would like to approximate the target Hamiltonian $\mathcal{H}_\text{targ}$ {and so} expand {the} self-energy in~\eqref{eq:selfenergy} up to $2^\text{nd}$ order. Note that $\mathcal{H}_-=0$ and $G_+(z)=(z-\Delta)^{-1}|1\rangle\langle{1}|_w$. Therefore the self energy $\Sigma_-(z)$ can be expanded as
\begin{align}\label{eq:Sz_2}
\Sigma_-(z) & =  \displaystyle V_-+\frac{1}{z-\Delta}V_{-+}V_{+-}+\sum_{k=1}^\infty\frac{V_{-+}V_+^kV_{+-}}{(z-\Delta)^{k+1}} ={}\\[0.1in]
 & =  \displaystyle 
\underbrace{\left(\mathcal{H}_\text{else}-\frac{2{\kappa}{\lambda}}{\Delta} A\otimes B\right)\otimes|0\rangle\langle{0}|_w}_{\mathcal{H}_\text{eff}}+{}
\notag
\\
&\qquad+\underbrace{\frac{z}{\Delta(z-\Delta)}({\kappa}A+{\lambda}B)^2\otimes|0\rangle\langle{0}|_w+\sum_{k=1}^\infty\frac{V_{-+}V_+^kV_{+-}}{(z-\Delta)^{k+1}}}_\text{error term}.
\notag
\end{align}
{By selecting} ${\kappa}=\text{sgn}(\alpha)(|\alpha|\Delta/2)^{1/2}$ and ${\lambda}=-(|\alpha|\Delta/2)^{1/2}$, the leading order term in $\Sigma_-(z)$ becomes $\mathcal{H}_\text{eff}=\mathcal{H}_\text{targ}\otimes|0\rangle\langle{0}|_w$. {We must now} show that the condition of Theorem \ref{th:perturbation} is satisfied i.e.\ for a small real number $\epsilon>0$, $\|\Sigma_-(z)-\mathcal{H}_\text{eff}\|\le\epsilon,\forall z\in[\min z,\max z]$ where $\max z=\|\mathcal{H}_\text{else}\|+|\alpha|+\epsilon=-\min z$. {Essentially this amounts to choosing a value of $\Delta$ to cause the error term in~\eqref{eq:Sz_2} to be $\le\epsilon$}. In order to derive a tighter lower bound for $\Delta$, we bound the norm of the error term in~\eqref{eq:Sz_2} by letting $z\mapsto\max z$ and {from} the triangle inequality for operator norms:
\begin{align}\label{eq:Sz_derive}
\displaystyle
\left\|\frac{z}{\Delta(z-\Delta)}(\kappa A+\lambda B)^2\otimes|0\rangle\langle{0}|_w\right\| & \leq  \displaystyle \frac{\max{z}}{\Delta(\Delta-\max{z})}\cdot 4\kappa^2=\frac{2|\alpha|\max{z}}{\Delta-\max{z}} \\[0.1in]
\displaystyle
\left\|\sum_{k=1}^\infty\frac{V_{-+}V_+^kV_{+-}}{(z-\Delta)^{k+1}}\right\| & \leq \displaystyle \sum_{k=1}^\infty\frac{\|V_{-+}\|\cdot\|V_+\|^k\cdot\|V_{+-}\|}{(\Delta-\max{z})^{k+1}}\leq{}\notag \\[0.1in]
& \leq \displaystyle \sum_{k=1}^\infty\frac{2|\kappa|\cdot\|\mathcal{H}_\text{else}\|^k\cdot 2|\kappa|}{(\Delta-\max{z})^{k+1}}={}\notag
\\
&=\sum_{k=1}^\infty\frac{2|\alpha|\Delta\|\mathcal{H}_\text{else}\|^k}{(\Delta-\max{z})^{k+1}}.
\notag
\end{align}
Using $\mathcal{H}_\text{eff}=\mathcal{H}_\text{targ}\otimes|0\rangle\langle{0}|_w$, from \eqref{eq:Sz_2} we see that
\begin{align}
\label{eq:5}
&\|\Sigma_-(z)-\mathcal{H}_\text{targ}\otimes|0\rangle\langle{0}|_w\|  \leq \displaystyle \frac{2|\alpha|\max z}{\Delta-\max z}+\sum_{k=1}^\infty\frac{2|\alpha|\Delta\|\mathcal{H}_\text{else}\|^k}{(\Delta-\max z)^{k+1}} 
={}\\
\label{eq:Sz_bound}
 &\quad=  \displaystyle \frac{2|\alpha|\max z}{\Delta-\max z}+\frac{2|\alpha|\Delta}{\Delta-\max z}\cdot\frac{\|\mathcal{H}_\text{else}\|}{\Delta-\max z-\|\mathcal{H}_\text{else}\|}.
\end{align}
Here going from~\eqref{eq:5} to~\eqref{eq:Sz_bound} we have assumed the convergence of the infinite series in~\eqref{eq:5}, which adds the reasonable constraint that
\[
\Delta>|\alpha|+\epsilon+2\|\mathcal{H}_\text{else}\|. 
\]
To ensure that
\[
\|\Sigma_-(z)-\mathcal{H}_\text{targ}\otimes|0\rangle\langle{0}|_w\|\leq\epsilon
\]
it is sufficient to let expression~\eqref{eq:Sz_bound} be $\leq\epsilon$, which implies that
\begin{equation}\label{eq:2bodyDelta}
\Delta\geq\left(\frac{2|\alpha|}{\epsilon}+1\right)(|\alpha|+\epsilon+2\|\mathcal{H}_\text{else}\|)
\end{equation}
which is $\Theta(\epsilon^{-1})$, a tighter bound than $\Theta(\epsilon^{-2})$ in the literature \cite{BDLT08,KKR06,OT06}. This bound is illustrated with a numerical example (Figure\ \ref{fig:2bodygadget}). From the data labelled as ``analytical'' in Figure\ \ref{fig:2bodygadget}a we see that the error norm $\|\Sigma_-(z)-\mathcal{H}_\text{eff}\|$ is within $\epsilon$ for all $z$ considered in the range, which satisfies the condition of the theorem for the chosen example. In Figure\ \ref{fig:2bodygadget}b, the data labelled ``analytical'' show that the spectral difference between $\mathcal{\tilde{H}_-}$ and $\mathcal{H}_\text{eff}=\mathcal{H}_\text{targ}\otimes|0\rangle\langle{0}|_w$ is indeed within $\epsilon$ as the theorem promises. Furthermore, note that the condition of Theorem \ref{th:perturbation} is only sufficient, which justifies why in Figure\ \ref{fig:2bodygadget}b for $\alpha$ values at $\max\alpha$ and $\min\alpha$ the spectral error is strictly below $\epsilon$. This indicates that an even smaller $\Delta$, although below the bound we found in~\eqref{eq:2bodyDelta} to satisfy the theorem, could still yield the spectral error within $\epsilon$ for all $\alpha$ values in the range. The smallest value $\Delta$ can take would be one such that the spectral error is exactly $\epsilon$ when $\alpha$ is at its extrema. We numerically find this $\Delta$ (up to numerical error which is less than $10^{-5}\epsilon$) and as demonstrated in Figure\ \ref{fig:2bodygadget}b, the data labelled ``numerical" shows that the spectral error is indeed $\epsilon$ at $\max(\alpha)$ and $\min(\alpha)$, yet in Figure\ \ref{fig:2bodygadget}a the data labelled ``numerical" shows that for some $z$ in the range the condition of the Theorem \ref{th:perturbation},
\begin{equation}
  \|\Sigma_-(z)-\mathcal{H}_\text{targ}\otimes|0\rangle\langle{0}|_w\|\leq\epsilon,  
\end{equation}
no longer holds. In Figure\ \ref{fig:2bodygadget} we assume that $\epsilon$ is kept constant. In Figure\ \ref{fig:delta_compare_sub}a we compute both analytical and numerical $\Delta$ values for different values of $\epsilon$.

\section{YY-gadget from XX, ZZ, X, Z}\label{sec:yy} 

\begin{theorem}[Cao-Kais-Biamonte PRA 91:012315 (2015)]
The Hamiltonian
\begin{equation}
   H_{\text{ZZXX}}=\sum_{i}h_i Z_i+\sum_{i}\Delta_i X_i+\sum_{i,j}J_{ij} Z_i Z_j+\sum_{i,j}K_{ij} X_i X_j. 
\end{equation}
emulates a $Y\otimes Y$ interaction with $\delta=\mathcal{O}(\epsilon^{-4})$ given one slack qubit. \end{theorem}

Following joint work with with Cao and others, we will summarize the $YY$ gadget \cite{Cao_2015}. The gadgets which we have presented so far are intended to reduce the locality of the target Hamiltonian. Here we present another type of gadget, called ``creation'' gadgets \cite{BL08}, which simulate the type of effective couplings that are not present in the gadget Hamiltonian. Many creation gadgets proposed so far are modifications of existing reduction gadgets. For example, the ZZXX gadget in \cite{BL08}, which is intended to simulate $Z_iX_j$ terms using Hamiltonians of the form
\begin{equation}\label{eq:dwave2}
\begin{array}{ccl}
\mathcal{H}_{ZZXX} & = & \displaystyle \sum_i\Delta_iX_i+\sum_ih_iZ_i+\sum_{i, j}J_{ij} Z_iZ_j+\sum_{i, j}K_{ij}X_i X_j,
\end{array}
\end{equation}
is essentially a 3- to 2-body gadget with the target term $A\otimes B\otimes C$ being such that the operators $A$, $B$ and $C$ are $X$, $Z$ and identity respectively. Therefore the analyses on 3- to 2- body reduction gadgets that we have presented for finding the lower bound for the gap $\Delta$ are also applicable to this ZZXX creation gadget.

Note that YY terms can be easily realized via bases rotation if single-qubit Y terms are present in the Hamiltonian in~\eqref{eq:dwave2}. Otherwise it is not \emph{a priori} clear how to realize YY terms using $\mathcal{H}_{ZZXX}$ in~\eqref{eq:dwave2}. We will now present the first YY gadget which starts with a universal Hamiltonian of the form~\eqref{eq:dwave2} and simulates the target Hamiltonian $\mathcal{H}_\text{targ}=\mathcal{H}_\text{else}+\alpha Y_i Y_j$. The basic idea is to use the identity $X_iZ_i=\imath Y_i$ and induce a term of the form $X_iZ_iZ_jX_j=Y_iY_j$ at the $4^\text{th}$ order. Introduce slack qubit $w$ and apply a penalty $\mathcal{H}=\Delta|1\rangle\langle{1}|_w$. With a perturbation $V$ we could perform the same perturbative expansion as previously. Given that the $4^\text{th}$ order perturbation is $V_{-+}V_+V_+V_{+-}$ up to a scaling constant. we could let single $X_i$ and $X_j$ be coupled with $X_w$, which causes both $X_i$ and $X_j$ to appear in $V_{-+}$ and $V_{+-}$. Furthermore, we couple single $Z_i$ and $Z_j$ terms with $Z_w$. Then $\frac{1}{2}(\eye +Z_w)$ projects single $Z_i$ and $Z_j$ onto the $+$ subspace and causes them to appear in $V_+$. For $\mathcal{H}_\text{targ}=\mathcal{H}_\text{else}+\alpha Y_1Y_2$, the full expressions for the gadget Hamiltonian is the following: the penalty Hamiltonian $\mathcal{H}=\Delta|1\rangle\langle{1}|_w$ acts on the slack qubit. The perturbation $V=V_0+V_1+V_2$ where $V_0$, $V_1$, and $V_2$ are defined as
\begin{equation}\label{eq:yy_V1}
\begin{split}
V_0 & =  \displaystyle \mathcal{H}_\text{else}+\mu({ Z_1+ Z_2})\otimes{|1\rangle\langle{1}|_w}+\mu( X_1-\text{sgn}(\alpha) X_2)\otimes X_w \\[0.05in]
V_1  & =  \displaystyle \frac{2\mu^2}{\Delta}(\eye\otimes|0\rangle\langle{0}|_w+X_1X_2) \\[0.05in]
V_2  & =  \displaystyle -\frac{2\mu^4}{\Delta^{3}}Z_1Z_2.
\end{split}
\end{equation}
with $\mu = (|\alpha|\Delta^3/4)^{1/4}$. For a specified error tolerance $\epsilon$, we have constructed a YY gadget Hamiltonian of gap scaling $\Delta=O(\epsilon^{-4})$ and the low-lying spectrum of the gadget Hamiltonian captures the spectrum of $\mathcal{H}_\text{targ}\otimes|0\rangle\langle{0}|_w$ up to error $\epsilon$.

The YY gadget implies that a wider class of Hamiltonians such as 
\begin{equation}
\begin{array}{ccl}
\mathcal{H}_{ZZYY} & = & \displaystyle \sum_i h_iX_i + \sum_i l_iZ_i + \sum_{i, j}J_{ij}Z_iZ_j+ \sum_{i, j}K_{ij}Y_iY_j
\end{array}
\end{equation}
and
\begin{equation}
\begin{array}{ccl}
\mathcal{H}_{XXYY} & = & \displaystyle \sum_i h_iX_i + \sum_i l_iZ_i + \sum_{i, j}J_{ij}X_iX_j+ \sum_{i, j}K_{ij}Y_iY_j
\end{array}
\end{equation}
can be simulated using the Hamiltonian of the form in~\eqref{eq:dwave2}. Therefore using the Hamiltonian in~\eqref{eq:dwave2} one can in principle simulate any finite-norm real valued Hamiltonian on qubits. Although by the {\sf QMA}-completeness of $\mathcal{H}_{ZZXX}$ one could already simulate such Hamiltonian via suitable embedding, our YY gadget provides a more direct alternative for the simulation.

Following joint work with with Cao and others, we will analize the $YY$ gadget \cite{Cao_2015}. The results in \cite{BL08} shows that Hamiltonians of the form in~\eqref{eq:dwave2} supports universal adiabatic quantum computation and finding the ground state of such a Hamiltonian is \textsc{{\sf QMA}-complete}. This form of Hamiltonian is also interesting because of its relevance to experimental implementation \cite{2006cond.mat..8253H}. Here we show that with a Hamiltonian of the form in~\eqref{eq:dwave2} we could simulate a target Hamiltonian $\mathcal{H}_\text{targ}=\mathcal{H}_\text{else}+\alpha Y_1Y_2$. Introduce an slack $w$ and define the penalty Hamiltonian as $\mathcal{H}=\Delta|1\rangle\langle{1}|_w$. Let the perturbation $V=V_0+V_1+V_2$ be
\begin{equation}\label{eq:yy_V}
\begin{split}
V_0 & =  \mathcal{H}_\text{else}+\kappa({Z_1+ Z_2})\otimes{|1\rangle\langle{1}|_w}+\kappa(X_1-\text{sgn}(\alpha) X_2)\otimes X_w \\[0.05in]
V_1  & =  2\kappa^2\Delta^{-1}[|0\rangle\langle{0}|_w-\text{sgn}(\alpha) X_1X_2] \\[0.05in]
V_2  & =  -4\kappa^4\Delta^{-3}Z_1Z_2.
\end{split}
\end{equation}
Then the gadget Hamiltonian $\mathcal{\tilde{H}}=\mathcal{H}+V$ is of the form in~\eqref{eq:dwave2}. Here we choose the parameter $\kappa=(|\alpha|\Delta^3/4)^{1/4}$. In order to show that the low lying spectrum of $\mathcal{\tilde{H}}$ captures that of the target Hamiltonian, define $\mathscr{L}_-=\text{span}\{|\psi\rangle$  such that $\mathcal{\tilde{H}}|\psi\rangle=\lambda|\psi\rangle,\lambda<\Delta/2\}$ as the low energy subspace of $\mathcal{\tilde{H}}$ and $\mathscr{L}_+=\eye-\mathscr{L}_-$. Define $\Pi_-$ and $\Pi_+$ as the projectors onto $\mathscr{L}_-$ and $\mathscr{L}_+$ respectively. 

With these notations in place, here we show that the spectrum of $\mathcal{\tilde{H}_-}=\Pi_-\mathcal{\tilde{H}}\Pi_-$ approximates the spectrum of $\mathcal{H}_\text{targ}\otimes|0\rangle\langle{0}|_w$ with error $\epsilon$. To begin with, the projections of $V$ into the subspaces $\mathscr{L}_-$ and $\mathscr{L}_+$ can be written as
\begin{equation}\label{eq:yy_Vproj}
\begin{split}
V_- & =  \displaystyle\bigg(\mathcal{H}_\text{else}+\underbrace{\frac{\kappa^2}{\Delta}(X_1-\text{sgn}(\alpha) X_2)^2}_{(a)}\underbrace{-\frac{4\kappa^4}{\Delta^3}Z_1Z_2}_{(b)}\bigg)\otimes|0\rangle\langle{0}|_w \\[0.05in]
V_+ & =  \displaystyle\left(\mathcal{H}_\text{else}+\kappa(Z_1+Z_2)-\frac{2\kappa^2}{\Delta}\text{sgn}(\alpha) X_1X_2 -\frac{4\kappa^4}{\Delta^3}Z_1Z_2\right)\otimes|1\rangle\langle{1}|_w \\[0.05in]
V_{-+} & =  \kappa(X_1 - \text{sgn}(\alpha) X_2)\otimes|0\rangle\langle{1}|_w \\[0.05in]
V_{+-} & =  \kappa(X_1 - \text{sgn}(\alpha) X_2)\otimes|1\rangle\langle{0}|_w
\end{split}
\end{equation}
Given the penalty Hamiltonian $\mathcal{H}$, we have the operator valued resolvent $G(z)=(z\eye-\mathcal{H})^{-1}$ that satisfies $G_+(z)=\Pi_+G(z)\Pi_+=(z-\Delta)^{-1}|1\rangle\langle{1}|_w$. Then the low lying sector of the gadget Hamiltonian $\mathcal{\tilde H}$ can be approximated by the perturbative expansion~\eqref{eq:selfenergy}. For our purposes we will consider terms up to the $4^\text{th}$ order:
\begin{equation}
\begin{split}
\Sigma_-(z)&=V_-+\frac{1}{z-\Delta}V_{-+}V_{+-}+
\frac{1}{(z-\Delta)^2}V_{-+}V_+V_{+-}+{}\\
&\qquad+\frac{1}{(z-\Delta)^3}V_{-+}V_+V_+V_{+-}+\sum_{k=3}^\infty\frac{V_{-+}V_+^kV_{+-}}{(z-\Delta)^{k+1}}.
\end{split}
\end{equation}
Now we explain the perturbative terms that arise at each order. The $1^\text{st}$ order is the same as $V_-$ in~\eqref{eq:yy_Vproj}. The $2^\text{nd}$ order term gives
\begin{equation}\label{eq:yy_2nd}
\frac{1}{z-\Delta}V_{-+}V_{+-}=
\underbrace{\frac{1}{z-\Delta}\cdot\kappa^2(X_1-\text{agn}(\alpha) X_2)^2}_{(c)}\otimes|0\rangle\langle{0}|_w.
\end{equation}
At the $3^\text{rd}$ order, we have
\begin{equation}\label{eq:yy_3rd}
\begin{split}
\displaystyle &\frac{1}{(z-\Delta)^2}V_{-+}V_+V_{+-}  = {}\\
&\quad\displaystyle =\biggl(\frac{1}{(z-\Delta)^2}\cdot\kappa^2(X_1-\text{agn}(\alpha) X_2)\mathcal{H}_\text{else}( X_1-\text{sgn}(\alpha) X_2)\times\\ 
&\quad\qquad\cdot\underbrace{\frac{1}{(z-\Delta)^2}\frac{4\kappa^4}{\Delta}(X_1X_2-\text{sgn}(\alpha)\eye)}_{(d)}\biggr)\otimes|0\rangle\langle{0}|_w+ O(\Delta^{-1/4}).
\end{split}
\end{equation}
The $4^\text{th}$ order contains the desired YY term:
\begin{equation}\label{eq:yy_4th}
\begin{split}
&    
\displaystyle \frac{1}{(z-\Delta)^3}V_{-+}V_+V_+V_{+-}  = {}\\
&\displaystyle \quad=\bigg(\underbrace{\frac{1}{(z-\Delta)^3}\cdot 2\kappa^4(X_1-\text{sgn}(\alpha) X_2)^2}_{(e)}-\underbrace{\frac{1}{(z-\Delta)^3}4\kappa^4Z_1Z_2}_{(f)}+{} \\[0.05in]
&\quad\qquad + \displaystyle \frac{4\kappa^4\text{sgn}(\alpha)}{(z-\Delta)^3}Y_1Y_2\bigg)\otimes|0\rangle\langle{0}|_w+O(\|\mathcal{H}_\text{else}\|\cdot\Delta^{-3/4})+{}\\
&\quad\qquad+O(\|\mathcal{H}_\text{else}\|^2\cdot\Delta^{-1/2})
\end{split}
\end{equation}
Note that with the choice of $\kappa=(|\alpha|\Delta^3/4)^{1/4}$, all terms of $5^\text{th}$ order and higher are of norm $O(\Delta^{-1/4})$. In the $1^\text{st}$ order through $4^\text{th}$ order perturbations the unwanted terms are labelled as $(a)$ through $(f)$ in  \eqref{eq:yy_Vproj}, \eqref{eq:yy_2nd}, \eqref{eq:yy_3rd}, and \eqref{eq:yy_4th}. Note how they compensate in pairs: the sum of $(a)$ and $(c)$ is $O(\Delta^{-1/4})$. The same holds for $(d)$ and $(e)$, $(b)$ and $(f)$. Then the self energy is then
\begin{equation}\label{eq:yy_selfenergy}
\Sigma_-(z)=(\mathcal{H}_\text{else}+\alpha Y_1Y_2)\otimes|0\rangle\langle{0}|_w+O(\Delta^{-1/4}).
\end{equation}
Let $\Delta=\Theta(\epsilon^{-4})$, then by the Gadget Theorem (\ref{th:perturbation}), the low-lying sector of the gadget Hamiltonian $\mathcal{\tilde{H}_-}$ captures the spectrum of $\mathcal{H}_\text{targ}\otimes|0\rangle\langle{0}|_w$ up to error $\epsilon$. 

The fact that the gadget relies on $4^\text{th}$ order perturbation renders the gap scaling relatively larger than it is in the case of subdivision or 3- to 2-body reduction gadgets. However, this does not diminish its usefulness in all applications.


\chapter{Conclusion and Open Problems} \label{chap:conclusion}

I believe the presented results, as they were carefully selected, form the foundation of the modern theory of quantum information processing. While many aspects of the topic have been well developed, the ending point of this story is just a beginning of what I think will prove to be an even deeper and perhaps an even more exciting direction in this changing field. 

Anticipated computational resources to determine ground state energy and calculate energy relative to a state have been conjectured.  In Table \ref{table:zoo} I have summarized what is known/conjectured regarding efficiently checkable minimisation problems.  Therein `Restricted Ising' denotes problems known to be in {\sf P}. ($^\star$) denotes conjectures. Electronic structure problem instances have constant maximum size so are assumed to be in {\sf BQP} whereas the ZZXX model is {\sf QMA}-complete.

I now plan to address the following open questions.  Perhaps I can hope that readers of this thesis will also work on these questions.  I hope that the materials presented here will serve these purposes, or otherwise serve your research.  Some of the questions I am now considering include:

\noindent {To understand the power of low-depth circuits, particularly the approximate trainability of circuits (e.g.~ with respect to the hardware efficient ansatz \cite{2017Natur.549..242K} or the checkerboard ansatz).}  

\begin{enumerate}
    \item It was recently shown in \cite{2018arXiv180906957H} that variational circuits can approximate t-designs.  In general, how well short circuits can approximate t-designs remains open. (see related results in \cite{2020arXiv200500544U}) 
    \item A systematic analysis of the epsilon neighborhood of a variational family with respect to ansatz depth is still lacking.
    \item Furthermore, a theory to avoid barren plateaus appears to be distant yet possibly can follow from recent findings.  The development of such a theory is of tantamount importance. 
\end{enumerate}

\noindent {To develop a means of active error mitigation for the variational model of quantum computation. While post-processing error decoders \cite{McClean2020} have been developed recently, as have compilation of error control codes into ansatz states see e.g.~\cite{2019arXiv191105759X, 2017arXiv171102249J}, novel approaches to variational error mitigation are emerging \cite{PhysRevX7021050, mcardle2019error} yet still lacking development.} 
\begin{enumerate}
    \item Does the circuit to variational mapping~\cite{UVQC} to approximate the output of general quantum circuits offer a means towards active error correction in variational algorithms? 
    
    \item As variational algorithms rely on local measurements, can information gained in this process influence consequential measurement basis choices and operations to correct errors? 
    
    \item Furthermore, the impacts of noise on variational algorithms (while studied in some detail) represents largely an unknown domain.  For example, recently QAOA experiments by Google were limited to depth 3 ansatz levels due to noise levels. 
    
    \item To study existing variational quantum algorithms and to determine, in physical and mathematical terms, their inherent limitations and potential advantages for tasks such as machine learning \cite{2017Natur.549..195B} and quantum simulation. 
\end{enumerate}

\noindent {To understand variable concentrations.  E.g.~It was recently proven in \cite{2019arXiv191008187F} for the Sherington-Kirkpactrick model, how generally to variable concentrations extend with respect to the algorithmic 3SAT phase transition?}

\begin{center}
\begin{table}\caption{{ Hamiltonian complexity micro zoo}}
  \begin{tabular}{|p{5.4cm}|p{5.4cm}|p{5.4cm}|}
  \hline
      \textbf{Problem Hamiltonian} & {\bf Finding Ground Energy} (Classical / Quantum) & {\bf Calculating State Energy} (Classical / Quantum) \\
      \hline
      {\sc 1-Local Hamiltonian} & Polynomial & Polynomial\\\hline
      {\sc 2-Local Ising} & Exp  & Polynomial \\\hline
      Electronic Structure & $^\star$Exp & $^\star$Exp / Polynomial\\\hline
      ZZXX Model & Exp  & $^\star$Exp / $^\star$Polynomial\\\hline 
  \end{tabular}\label{table:zoo}
  \end{table}
\end{center}

\chapter*{List of abbreviations}
\addcontentsline{toc}{chapter}{List of symbols}


\noindent {\sf AND} --- Logical conjunction. Logic gate.  

\noindent \textbf{AQC} --- Adiabatic quantum computation.  

\noindent \textbf{cbit, c-Bit} --- A deterministic bit able to store logical zero or otherwise logical one. 

\noindent {\sf CN}, $k$-{\sf CN}, {\sf CNOT} --- A controlled {\sf NOT} activated when the {\sf AND} of all $k$ controls is logical 1. 

\noindent \textbf{CNF} --- Conjunctive Normal Form. A formula is in conjunctive normal form if it is an {\sf AND} of {\sf OR}s. It is often called a {\it product of sums}. 

\noindent \textbf{DNF} --- Disjunctive Normal Form.  A formula is in disjunctive normal form is a Boolean-logical formula consisting of an {\sf OR} of {\sf AND}s.  It is often called a {\it sum of products}. 

\noindent \textbf{$k$-body/$k$-local Hamiltonian} --- A Hermitian operator acting on a finite number of qubits where each term acts non-trivially on at most $k$-qubits. 

\noindent \textbf{ML} --- Machine Learning. 

\noindent \textbf{pbit, p-bit} --- A {\it probabilistic} bit able to store expected values of recovering logical zero and logical one. 

\noindent {\sf NAND} --- Logical Negation of {\sf AND}. Logic gate. 

\noindent {\sf NOR} --- Logical Negation of {\sf OR}. Logic gate. 

\noindent {\sf NOT} --- Logical Negation. Inverter. Represented by Pauli matrix $X$. 

\noindent \textbf{NISQ} --- Noisy Intermediate Scale Quantum Technology. 

\noindent {\sf OR} --- Logical Disjunction. Logic gate.  

\noindent \textbf{SAT} --- Boolean Satisfiability Problem (see also 2-, 3- and $k$-{\sf SAT}). 

\noindent \textbf{T.D.S.E.} --- Time Dependent Schr{\"o}dinger's Equation. 

\noindent \textbf{T.I.S.E.} --- Time Independent Schr{\"o}dinger's Equation.

\noindent \textsf{QAOA} --- Quantum Approximate Optimization Algorithm.

\noindent \textsf{QAOA} --- Quantum Alternating Operator Ansatz. 

\noindent \textbf{QFT} --- Quantum Fourier Transform. 

\noindent \textbf{QMA} --- Complexity Class Quantum Merlin Author (analog of {\sf NP}).

\noindent \textbf{QML} --- Quantum Machine Learning. 

\noindent \textbf{S.E.} --- Schr{\"o}dinger's Equation. 

\noindent \textbf{VQE} --- Variational Quantum Eigensolver (Quantum Algorithm). 

\noindent {\sf XOR} --- Logical Exclusive {\sf OR}. Logic gate.  

\noindent \textbf{$2$-, $3$- and $k$-SAT} --- Boolean Satisfiability Problem with specified variables per clause. 

\newpage 
\chapter*{List of symbols}
\addcontentsline{toc}{chapter}{List of abbreviations}

\begin{longtable}{ll}
\S & section or chapter \\ 
$\bydef$ &  defined as \\
$\in$ &  belongs to (a set) \\
$\notin$ & does not belong to (a set) \\
$\cap$ & intersection of sets  \\
$\cup$ & union of sets  \\
$\emptyset$, $\varnothing$ & empty set \\
$\mathbb{B}$ &  the single space of booleans $\{0,1 \}$ \\
$\mathbb{B}^n$ & space of $n$-long boolean numbers $\{0,1\}^n$\\
$\mathbb{N}$ & set of natural numbers which we assume \\
& includes zero ($\mathbb{N} \bydef \mathbb{N}\cup \{0\}$)\\
$\mathbb{Z}$ & set of integer numbers \\
$\mathbb{Q}$ & set of rational numbers \\
$\mathbb{R}$ & set of real numbers \\
$\mathbb{R}^{+}$ & set of nonnegative real numbers \\
$\mathbb{C}$ & set of complex numbers \\
$\mathbb{C}_{/ \{0\}}$  & set of complex numbers with zero excluded\\
$\mathbb{R}^n$ & space of column vectors with $n$ real components \\
$\mathbb{C}^n$ & space of column vectors with $n$ complex components \\
$i$ & $ \sqrt{-1}$ \\
$\mathfrak{Re}$ & projection; real part of the complex number \\
$\mathfrak{Im}$ & projection; part of the complex number \\
$ | z |$ & modulus of complex number $z$ \\
$T \subset S$ & subset $T$ of set $S$ \\
$S \cap T$ & intersection of sets $S$ and $T$ \\
$S \cup T$ & union of sets $S$ and $T$ \\
$f(S)$ & image of set $S$ under mapping $f$ \\
$f \circ g$ & composition of two mappings $(f \circ g)(x) = f(g(x))$ \\
$\ket{x}$ & column vector in $\mathbb{C}^n $\\
$\bra{x}$ & complex conjugate transpose of $\ket{x} $\\
$||\cdot||$ & norm \\
$\braket{x}{y}$ & scalar product (inner product) in $\mathbb{C}^n$\\
$\det(A)$ & determinant of square matrix $A$ \\
$\Tr(A)$& trace of square matrix $A$ \\
$\rank(A)$& rank of matrix $A$ \\
$A^T$& transpose of matrix $A$ \\
$\overline{A}$& conjugate of matrix $A$ \\
$\spn  \{ \ket{i}, \ket{j},...,\ket{k} \} $ & vector space formed by linear combinations of $\ket{i}, \ket{j}, ..., \ket{k}$ \\
$\spn  \{ A, B,...,C \} $ & vector space formed by linear combinations of $A, B, ..., C$ \\
$A^{\dagger}$& conjugate transpose of matrix $A$ \\
$A^{-1}$& inverse of matrix $A$ \\
$\mathds{1}$ & identity matrix \\
$AB$ & matrix product of $m \times n$ matrix $A$ and $n \times p$ matrix $B$\\
$A\bullet B$ & Hadamard (entrywise) product of $m \times n$ matrices $A$ and $B$\\
$[A,B] \bydef AB - BA$ & commutator for square matrices $A$ and $B$\\
$\{A,B\} \bydef AB + BA$ & anticommutator for square matrices $A$ and $B$\\
$A \otimes B$  &  Kronecker (a.k.a.~tensor) product \\
$A \oplus B$  &  direct sum of matrices/operators/spaces \\
$\mathscr{L}\big({\mathscr A},{\mathscr B}\big)$  &  the space of linear maps between spaces ${\mathscr A}$ and ${\mathscr B}$ \\
$\mathscr{L}\big({\mathscr A}\big)$  & the space of linear maps from ${\mathscr A}$ to itself \\ 
$\delta_{jk}$ &  Kronecker delta with $\delta_{jk} = 1$ for $j=k $ and  $\delta_{jk}=0$ for $j\neq k$\\
iff & if and only if \\ 
$A | B$ & The predicate $A$ conditioned {\it such that} $B$ is true  \\
$\mathfrak u$ & A Lie algebra of a unitary group, e.g.~${\mathfrak su}(2)$ \\
$\bf U$ & A unitary (Lie) group, e.g.~${\bf SU}(2)$ 

\end{longtable}

\begin{definition}[Linear Extension] 
A complex linear extension of a $d$-element set spans $\mathbb{C}^d$
\end{definition}

\begin{definition}[Bell Basis] The standard Bell basis is denoted as follows. 
\begin{enumerate}
    \item $\sqrt{2}\ket {\Phi^+} = \ket{00}+\ket{11}$ 
    \item $\sqrt{2}\ket {\Phi^-} = \ket{00}-\ket{11}$ 
    \item $\sqrt{2}\ket {\Psi^+} = \ket{01}+\ket{10}$ 
    \item $\sqrt{2}\ket {\Psi^-} = \ket{01}-\imath \ket{10}$ 
\end{enumerate}
\end{definition}

\begin{remark}[Using $\cdot$ for a product with scalar(s)]
Though typically omitted, we sometimes use $\cdot$ to denote multiplication by a scalar. 
\end{remark}

\begin{remark}
We will interchange the notation $\sigma_0 \equiv \openone$, $\sigma_1 \equiv X$, $\sigma_2 \equiv Y$, $\sigma_3 \equiv Z$.
\end{remark}

\begin{definition}[Pauli Matrices]
We use the Pauli basis (or group; group algebra) defined by the following matrices and relations. 
\begin{equation}
X =
\begin{pmatrix}
    0   & 1 \\
    1   & 0 
\end{pmatrix}
\hspace{30pt}
Y =
\begin{pmatrix}
    0   & -\imath \\
    \imath   & 0 
\end{pmatrix}
\hspace{30pt}
Z =
\begin{pmatrix}
    1   & 0 \\
    0   & -1 
\end{pmatrix}
\end{equation}
\begin{equation}
X Y = iZ 
\end{equation}
\begin{equation}
X^2 = Y^2 = Z^2 = \eye = -\imath XYZ 
\end{equation}
\end{definition}

This group called the Pauli group together with the tensor product contains $4 \cdot 4^n$. This structure arises in angular momentum theory. Some properties are
\begin{enumerate}
\item $\left[O_i, O_j \right] = 2 \imath \varepsilon^{ijk} O_k$ (Lie product)
\item complex conjugation 
$$\forall \omega \in \{ 1, 2, 3 \} \hspace{10pt} O_i O_{\omega}^{*}O_i = -O_{\omega} $$
\item $Tr(O_i O_j) = 2 \delta_{ij}$ 
\item $O_i O_j = \delta_{ij} \eye + \imath \varepsilon^{ijk} O_k $ (geometric product)
\item For $i=j$ $ \{O_i, O_j \} = 2 \delta_{ij} \eye $
\end{enumerate}

\begin{example}[Single qubit density operator]
A polarisation vector is defined as 
$$ \bm{P} = (p_1, p_2, p_3)$$
where $\forall i$ $p_i \in \mathbb{R}$. We also define 
$$\bm{O} = (O_1, O_2, O_3) $$
where the $O_i \in \{X, Y, Z \}$
A density operator can be written as 
\begin{eqnarray}
\rho = \rho^{\dagger}, \hspace{10pt} \rho \geq 0 \hspace{10pt} \Tr \rho = 1, 
\end{eqnarray}
with $\rho: \mathbb{C}^2 \to \mathbb{C}^2$. The density operator is another way to write a quantum state. This density operator can be written as
\begin{equation}
\rho = \frac{1}{2} \eye + \bm{O} \cdot \bm{P} = \frac{1}{2} \eye + \sum_i p_i O_i = \frac{1}{2} + p_1 X + p_2 Y + p_3 Z
\end{equation}
We also define an inner product  $(-,-): (\mathbb{C}^2 \to \mathbb{C}^2) \times (\mathbb{C}^2 \to \mathbb{C}^2) \to \mathbb{C}$
$$(\bm{O} \cdot \bm{A}, \bm{O} \cdot \bm{B}) = \Tr \left[ (\bm{O} \cdot \bm{A}) (\bm{O} \cdot \bm{B}) \right] = \bm{A} \cdot \bm{B}$$
\end{example}

\begin{remark}
One readily verifies that $(\bm{O} \cdot \bm{A}, \bm{O} \cdot \bm{B}) = \bm{A} \cdot \bm{B} + i \bm{O}(\bm{A} \wedge \bm{B})$. 
\end{remark}

\begin{remark}
One can readily establish that for $\bm{A} = \bm{B}$ show that $\bm{A} \wedge \bm{B}$ vanishes. Hence $\lambda (\bm{A} \cdot \bm{O}) = \pm |A| $ ($\pm \sqrt{\sum_j p_j^2}$) and both roots appear as $\Tr \bm{A} \cdot \bm{O} = 0$. 
\end{remark} 

\begin{remark}
We consider a quantum state $\psi$ and a Hilbert space ${\mathscr A}$. We adopt the slight abuse of notation that:  
\begin{enumerate}
    \item $\psi \subset {\mathscr A}$ is shorthand for  $\psi \in {\mathscr B} \subset {\mathscr A}$ where $\psi$ is necessarily restricted to a proper subset ${\mathscr B}$ of ${\mathscr A}$, 
    \item $\psi \subseteq {\mathscr A}$ is shorthand for $\psi \in {\mathscr B} \subseteq {\mathscr A}$ where $\psi$ is restricted to a subset ${\mathscr B}$ that might be equivalent to ${\mathscr A}$, 
    \item $\psi \in {\mathscr A}$ means that $\psi$ can take any value in ${\mathscr A}$.
\end{enumerate}
\end{remark}

\begin{remark}
We adopt the convention that: 
\begin{equation}\label{eqn:arg1}
         \argmin_{\psi \in {\mathscr A}} \bra{\psi}{\mathcal H}\ket{\psi} = \ket{\psi'} 
\end{equation}
while it is more accurate that, 
\begin{equation}\label{eqn:arg2}
    \ket{\psi'} \in \argmin_{\psi \in {\mathscr A}} \{\bra{\psi}{\mathcal H}\ket{\psi}\} = \text{Span}\{ \ket{\psi}|\braket{\psi} =1, \bra{\psi}{\mathcal H}\ket{\psi}= E_\star\}. 
\end{equation}
\end{remark}

\subsection*{Switching Algebra Postulates} 


 \noindent Identity Elements\\
$ x  \lor 0 = 0 $\\
$ x \land 1 = x $

 \noindent Commutativity\\
$ x  \lor y = y \lor x $\\
$ x \land y = y \land x $

 \noindent Complements\\
$ x  \lor  \overline{x} = 1 $\\
$ x \land \overline{x} = 0 $

 \noindent Absorption\\
$ x \lor (x \land y) = x$\\
$ x \land (x \lor y) = x$\\
$ x \lor (\overline{x} \land y) = x \lor y$\\
$ x \land (\overline{x} \lor y) = x \land y$\\
$ (x \land y) \lor (x \land \overline{y}) = x$\\
$ (x \lor y) \land (x \lor \overline{y}) = x$ 

 \noindent Consensus\\
$(x \land y) \lor (\overline{x} \land z) \lor (y \land z) = (x \land y) \lor (\overline{x}\land z)$
$(x \lor y) \land (\overline{x} \lor z) \land (y \lor z) = (x \lor y) \land (\overline{x}\lor z)$

 \noindent Associativity\\
$(x \lor y) \lor z = x \lor (y \lor z)$\\
$(x \land y) \land z = x \land (y \land z)$

 \noindent Distributivity\\
$x \lor (y \land z) = (x \lor y) \land (x \lor z) $\\
$x \land (y \lor z) = (x \land y) \lor (x \land z) $

 \noindent Idempotency\\
$x \lor x = x$\\
$x \land x = x$

 \noindent Null elements\\
$x \lor 1 = 1$\\
$x \land 0 = 0$

 \noindent Involution\\
$ \overline{(\overline{x})} = \overline{\overline{x}} = x $

 \noindent De Morgan's Laws  \\ 
 \noindent
$ \lnot(x \lor y) = \overline{x} \land \overline{y} $\\
$ \lnot(x \land y) = \overline{x} \lor \overline{y} $

 \noindent De Morgan's Principle of Duality\\ Any theorem or postulate in Boolean algebra remains true if: \\ 0 $ \leftrightarrow $ 1\\
$ \land \leftrightarrow \lor $

\subsection*{XOR algebra and the Zhegalkin polynomial}
Here we recall the algebraic normal form (ANF) for Boolean
polynomials, commonly known as PPRMs, (Positive Polarity Reed Muller Forms). 

\begin{definition}
The XOR-algebra forms a commutative ring with presentation
$M=\{\B,\wedge,\oplus\}$ where the following product is called XOR
\begin{equation}
\text{---}\oplus\text{---}:\B\times \B\mapsto \B: (a,b)\rightarrow
a+b-ab~\text{mod}~2
\end{equation}
and conjunction is given as
\begin{equation}
\text{---}\wedge\text{---}:\B\times \B\mapsto \B: (a,b)\rightarrow a\cdot b, 
\end{equation}
where $a\cdot b$ is regular multiplication over the reals (restricted to $0, 1$).  One defines left negation
$\neg(\text{---})$ in terms of $\oplus$ as
$\neg(\text{---})\equiv$
\begin{equation}
  \text{1}\oplus(\text{---}):\B \mapsto \B: a\rightarrow 1-a.
\end{equation}
In the XOR-algebra, (i-v) hold. 
\begin{enumerate}
\item[(i)] $a\oplus 0 = a$, 
\item[(ii)] $a\oplus 1 = \neg a$,
\item[(iii)] $a\oplus a = 0$,
\item[(iv)] $a\oplus \neg a = 1$ and 
\item[(v)] $a\vee b = a\oplus b\oplus
(a\wedge b)$  
\end{enumerate}
\end{definition}
Hence, $0$ is the unit of the operation XOR and $1$ is the unit of AND. The 5th
rule (v)
reduces to $a\vee b = a\oplus b$ whenever $a\wedge b=0$, which is the case for
disjoint ($\text{mod}~2$) sums.  The 
truth table for AND/XOR follows
\begin{center}
\begin{tabular}{c|c|c|c}
$~x_1~$ & $~x_2~$ & $f(x_1,x_2)=x_1\wedge x_2$ & $f(x_1,x_2)=x_1\oplus x_2$ \\ \hline
0 & 0 & 0 & 0\\
0 & 1 & 0 & 1\\
1 & 0 & 0 & 1\\
1 & 1 & 1 & 0
\end{tabular}
\end{center}

\begin{definition}
\label{def:FPRM-ps1}
Any Boolean equation may be uniquely expanded to the fixed polarity
Reed-Muller form as
\begin{eqnarray}
\label{eqn:rm_exp-ps1}
&&f(x_1,x_2,...,x_k) = c_0\oplus c_1 x_1^{\sigma_1}\oplus c_2 x_2^{\sigma_2}\oplus
\cdots \oplus c_n x_n^{\sigma_n}\oplus\nonumber\\
&&~~~~~~~~c_{n+1}x_1^{\sigma_1} x_n^{\sigma_{n}}\oplus \cdots \oplus
c_{2k-1}x_1^{\sigma_1} x_2^{\sigma_2},...,x_k^{\sigma_k},
\end{eqnarray}
where selection variable $\sigma_i\in \{0,1\}$, literal
$x_i^{\sigma_i}$ represents a variable or its negation and any $c$
term labeled $c_0$ through $c_j$ is a binary constant $0$ or $1$. In
\eqref{eqn:rm_exp-ps1}
only fixed polarity variables appear such that
each is in either uncomplemented or complemented form.
\end{definition}

Let us now consider derivation of the form from 
Definition~\ref{def:FPRM-ps1}.
As illustrative example, we avoid keeping track of
indices in the $n$ node case, by considering the case where $n\equiv 2^n=8$.

\begin{example}
The vector $\underline{c}=(c_0,c_1,c_2,c_3,c_4,c_5,c_6,c_7,)^\intercal$ represents
all possible outputs of any function $f(x_1,x_2,x_3)$ over
algebra $\mathbb{Z}_2\times \mathbb{Z}_2\times \mathbb{Z}_2$.  We wish to
construct a normal form in terms of the vector $\underline{c}$, where each $c_i\in
\{0,1\}$, and therefore $\underline{c}$ is a selection vector
that represents the output of the function $f:\B\times\B\times\B\rightarrow
\B:(x_1,x_2,x_3)\mapsto f(x_1,x_2,x_3)$. One may expand $f$ as 
\eqref{eqn:generic-ps1}. 
\begin{eqnarray}
\label{eqn:generic-ps1}
f(x_1,x_2,x_3) &=& (c_0\cdot \neg x_1\cdot \neg x_2\cdot \neg
x_3)\vee(c_1\cdot \neg x_1\cdot \neg x_2\cdot x_3)\vee(c_2\cdot
\neg x_1\cdot x_2\cdot \neg x_3)\nonumber\\
&&\vee(c_3\cdot \neg x_1\cdot x_2\cdot x_3)\vee(c_4\cdot x_1\cdot
\neg x_2\cdot \neg x_3)\vee(c_5\cdot x_1\cdot
\neg x_2\cdot x_3)\nonumber\\
&&\vee(c_6\cdot x_1\cdot x_2\cdot \neg x_3)\vee(c_7\cdot x_1\cdot
x_2\cdot x_3)
\end{eqnarray}
Since each disjunctive term is disjoint the logical OR operation may
be replaced with the logical XOR operation.  By making the substitution $ \neg
a=a\oplus 1$ for all variables and rearranging terms one arrives at
the normal form 
\eqref{eqn:generic2-ps1}.\footnote{For instance, $\neg x_1\cdot
\neg x_2\cdot\neg x_3=(1\oplus x_1)\cdot(1\oplus x_2)\cdot(1\oplus
x_3)=(1\oplus x_1\oplus x_2\oplus x_2\cdot x_3)\cdot(1\oplus x_3)=
1\oplus x_1\oplus x_2\oplus x_3\oplus x_1\cdot x_3\oplus x_2\cdot
x_3\oplus x_1\cdot x_2\cdot x_3$.}
\begin{eqnarray}
\label{eqn:generic2-ps1}
&& f(x_1,x_2,x_3) = \\
&& c_0\oplus(c_0\oplus c_4)\cdot x_1\oplus(c_0\oplus
c_2)\cdot x_2\oplus(c_0\oplus c_1)\cdot x_3\oplus(c_0\oplus
c_2\oplus c_4\oplus c_6)\cdot x_1\cdot
x_2\nonumber\\
&& \oplus(c_0\oplus c_1\oplus c_4 \oplus c_5)\cdot x_1\cdot
x_3\oplus(c_0 \oplus c_1 \oplus c_2 \oplus c_3)\cdot x_2\cdot x_3\nonumber\\
&&\oplus (c_0\oplus c_1\oplus c_2\oplus c_3\oplus c_4 \oplus c_5
\oplus c_6\oplus c_7)\cdot x_1\cdot x_2\cdot x_3 \nonumber 
\end{eqnarray}
The set of algebraically independent polynomials, $\{x_1,x_2,x_3,x_1\cdot
x_2,x_1\cdot x_3,x_2\cdot x_3,x_1\cdot x_2\cdot x_3\}$ combined with
a set of scalars from \eqref{eqn:generic2-ps1} spans the eight
dimensional space of the hypercube representing the Algebra.
A similar form holds for arbitrary $n$.  This expansion is typically credited to Zhegalkin. 
\end{example}

\subsection*{Karnaugh maps}\label{appendix:kmap}

The \emph{Karnaugh map} is a tool to facilitate the reduction of terms needed to express Boolean functions. Many excellent texts and online tutorials cover the use of Karnaugh maps such as the wikipedia entry (\href{http://en.wikipedia.org/wiki/Karnaugh_map}{\tt http://en.wikipedia.org}) and the articles linked to therein.  Here we briefly introduce K-maps to make the thesis self contained.

\begin{definition}
We let $\mathbb{B}=\{0,1\}$. 
\end{definition}

Karnaugh maps (see Figure~\ref{fig:quad} for three examples), or more compactly K-maps, are organized so that the truth table of a given equation, such as a Boolean equation ($f:\mathbb{B}^n\rightarrow \mathbb{B}$) or pseudo Boolean form ($f:\mathbb{B}^n\rightarrow \mathbb{R}$), is arranged in a grid form and between any two adjacent boxes only one domain variable can change value.

This ordering results as the rows and columns are sequenced according to Gray code---a binary numeral system where two successive values differ in only one digit.  For example, the 4-bit Gray code is sequenced as:
\begin{eqnarray*}
&&(0000, 0001, 0011, 0010, 0110, 0111, 0101, 0100, 1100,\\
&&1101, 1111, 1110, 1010, 1011, 1001, 1000).
\end{eqnarray*}
By arranging the truth table of a given function in this way, a K-map can be used to derive a minimized function.

To use a K-map to minimize a Boolean function one \emph{covers} the 1s on the map by rectangular \emph{coverings} containing a number of boxes equal to a power of 2.  For example, one could circle a map of size $2^n$ for any constant function $f=1$.  Figure~\ref{fig:quad} (a) and (b) contain three circles each --- all of 2 and 4 boxes respectively. After the 1's are covered, a term in a \emph{sum of products expression} is produced by finding the variables that do not change throughout the entire covering, and taking a 1 to mean that variable ($x_i$) and a 0 as its negation ($\overline{x_i}$). Doing this for every covering yields a function which matches the truth table.

For instance consider Figure~\ref{fig:quad} (a) and (b).  Here the boxes contain simply labels representing the decimal value of the corresponding Gray code ordering.  The circling in Figure~\ref{fig:quad} (a) would correspond to the truth vector (ordered $z_\star, x_1$ then $x_2$)
\begin{equation}\label{eqn:tt1}
\left(0,0,0,1,0,1,1,1\right)^\top.
\end{equation}
The cubes 3 and 7 circled in Figure~\ref{fig:quad} correspond to the sum of products term $x_1x_2$.  Likewise (5,7) corresponds to $z_\star x_2$ and finally (7,6) corresponds to $z_\star x_1$. The sum of products representation of (\ref{eqn:tt1}) is simply
\begin{equation*}
f(z_\star,x_1,x_2)=x_1x_2\vee z_\star x_2\vee z_\star x_1.
\end{equation*}
Let us repeat the same procedure for Figure~\ref{fig:quad} (b) by again assuming the circled cubes correspond to 1's in the functions truth table.  In this case one finds $z_\star$ for the circling of cubes ladled (4,5,7,6), $x_2$ for (1,3,5,7) and $x_1$ for (3,2,7,6) resulting in the function
\begin{equation*}
f(z_\star,x_1,x_2)=x_1\vee z_\star \vee x_2.
\end{equation*}

\chapter*{Glossary of terms}
\addcontentsline{toc}{chapter}{Glossary of terms}

\noindent {\bf adiabatic model of quantum computation};  A model of quantum computation that relies on the adiabatic theorem to ensure that a physical process will minimize a target Hamiltonian which embeds a problem instance solution in its ground state. \\ 

\noindent {\bf adjoint}; The complex conjugate transpose of the operation (inverse of a unitary). \\ 

\noindent {\bf advantage, quantum}; For a given problem, the improvement in run time for a quantum computer versus a conventional computer running the best known conventional algorithm. \\ 

\noindent {\bf annealing, quantum/simulated/physical};  Annealing algorithms mimic (simulate) the cooling process.  Alternatively, physical annealing utilizes a physical process to replace simulated annealing with a physical annealing process.  Likewise, quantum annealing utilizes quantum effects to help accelerate the physical annealing process.   \\ 

\noindent {\bf ansatz}. An ansatz is the establishment of the starting equation(s), the theorem(s), or the value(s) describing a mathematical or physical problem or solution.  It is often determined by a so called, educated guess. \\ 

\noindent {\bf bit, classical};  A deterministic ideal memory/register storing one binary unit (logical zero or logical one). \\ 

\noindent {\bf bit, probabilistic};  A probabilistic bit is a controllable memory/register that is able to store probabilistic values of logical zero and logical one.  For example, storing the state $(1-p)\ket{0}+ p\ket{1}$ for probability $p$, implies that the expected value of recovering $\ket{1}$ ($\ket{0}$) is $p$ ($1-p$). This is typically ensured through redundancy (i.e.~multiple deterministic bits can mimic a probabilistic bit). \\ 

\noindent {\bf Bloch sphere}; The Bloch sphere is a geometrical representation of the pure state space of a two-level quantum mechanical system (qubit), named after the physicist Felix Bloch. \\ 

\noindent {\bf coherence, quantum};  The (quantum) coherence of a qubit, roughly is the ability to maintain superposition over time. (Not to be confused with optics.) \\ 

\noindent {\bf complement, logical}; The logical negation of a Boolean variable, sending $x$ to $\bar{x}=1-x$. \\ 

\noindent {\bf costate, dual state, effect}; The result of a measurement.  The mathematical dual (dagger) of a state vector. A measurement outcome operator. A linear map from the complex numbers to the dual of a vector space.  A linear map from a vector space to the complex numbers. \\ 

\noindent {\bf Dirichlet operator}; A Hamiltonian that is both self-adjoint and infinitesimal stochastic. \\

\noindent {\bf EPR pair, Einstein–Podolsky–Rosen};	Also known as a Bell State. An entangled bipartite quantum state.  The state is typically written as $\sqrt{2}\ket {\Phi^+} = \ket{00}+\ket{11}$, see e.g.~\ref{eqn:phi}. \\ 

\noindent {\bf gadget Hamiltonian}; A construction(s) to emulate desired Hamiltonian coupling terms. \\ 

\noindent {\bf gate model, quantum};  The defacto model of universal quantum computation described by quantum logic gates acting on registers of qubits.  \\

\noindent {\bf measurement based or one-way model};  The measurement based model of quantum computation preforms universal quantum computation by local measurements on an initially entangled resource state, usually a {\it cluster state} or {\it graph state}. \\

\noindent {\bf monotonic function/quantity}; To be varying in such a way as to either (i) never decrease or (ii) never increase. \\

\noindent {\bf slack space, slack bits, ancillary}; Additional bits/qubits used to assist information processing tasks. \\ 

\noindent {\bf spin, classical/quantum}; Sometimes called {\it nuclear spin} or {\it intrinsic spin} is the quantum version form of angular momentum carried by elementary particles.  \\ 

\noindent {\bf supremacy, quantum}; An adversarial game consisting of a calculation on a quantum computer that cannot be in practice be performed on any foreseeable conventional computer. \\ 

\noindent {\bf tuple};	Deliminator (comma) separated types grouped together by parenthesis. \\

\noindent {\bf uniform circuit, polynomial-time}; 
A family of (Boolean) circuits $\{C_n:n \in \mathbb{N}\}$ generated by an algorithm $\mathcal M$, such that 
\begin{enumerate}
    \item $\mathcal M$ successfully terminates in polynomial time in $n$, 
    \item $\forall n \in \mathbb{N}$, ${\mathcal M}$ outputs a description of $C_n$ given input $1^{\times n}$. 
\end{enumerate}

\printbibliography
\addcontentsline{toc}{chapter}{Bibliography}

\newpage 
\listoffigures
\addcontentsline{toc}{chapter}{List of Figures}
\newpage 
\listoftables
\addcontentsline{toc}{chapter}{List of Tables}






\printindex
\end{document}